\documentclass[12pt]{report}

\usepackage[utf8]{inputenc}
\usepackage[T1]{fontenc}
\usepackage[polish, english]{babel}

\author{\theauthor}
\title{\thetitle}

\def\mytitle{\newpage
\thispagestyle{empty}
\begin{centering}

\includegraphics[width=6cm, trim={0.2cm, 0.2cm, 0.2cm, 0.2cm}, clip]{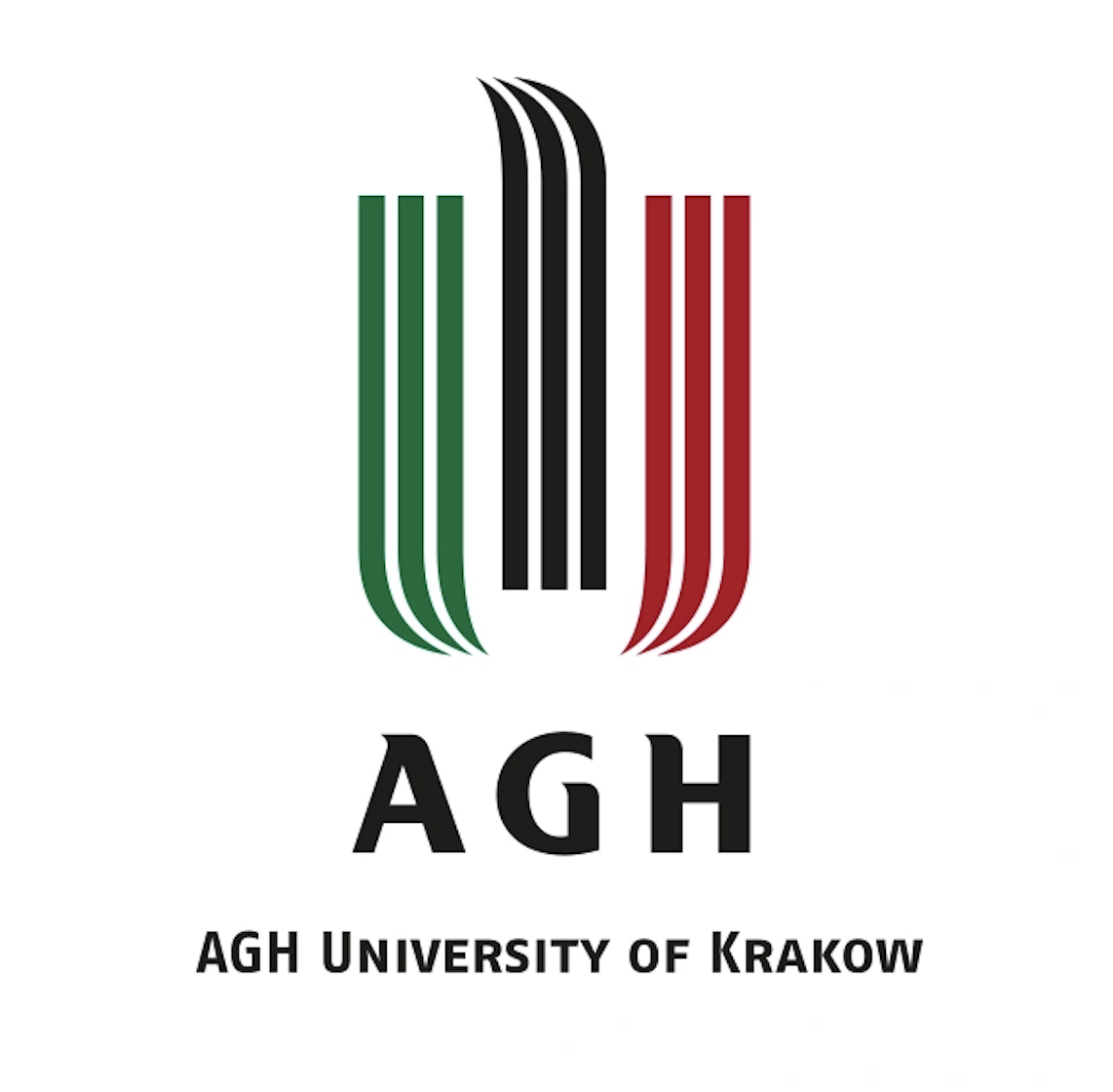} \\
\vspace{0.9cm}
\large
\textbf{FIELD OF SCIENCE:\quad Natural sciences} \\ 
\vspace{0.4cm}
SCIENTIFIC DISCIPLINE:\quad Computer and information sciences \\
\vspace{1.2cm}
\LARGE
\textbf{DOCTORAL THESIS}\\
\vspace{1.2cm}
\LARGE
Map of Elections\\
\vspace{1.9cm}
\large
Author:\quad Stanisław Andrzej Szufa \\
\vspace{0.6cm}
Supervisor:\quad prof. dr~hab.~in\.z.~Piotr Faliszewski \\
\vspace{0.6cm}
Completed in:\quad AGH University, Faculty of Computer Science \\
\vspace{1cm}
Krakow, 2024 \\
\end{centering} 
\newpage
}

\usepackage{natbib}
\usepackage{hyperref}
\usepackage{amsmath,amsthm,amsfonts}
\usepackage{mathtools}
\usepackage[dvipsnames]{xcolor}
\usepackage{pgfplots}
\usepackage{booktabs}
\usepackage{units}
\usepackage{kbordermatrix}
\usepackage{graphicx}
\usepackage{tikz}
\usepackage{caption}
\usepackage{subcaption}
\usepackage{amsthm}
\usepackage{amsfonts}
\usepackage{nicefrac}
\usepackage{times}
\usepackage{algorithm}
\usepackage{algorithmic}
\usepackage{paralist}
\usepackage{multirow}
\usepackage{comment}
\usetikzlibrary{external}
\usepackage{cleveref}
\usepackage{mathdots}
\usepackage{footnote}
\usepackage{etoolbox}

\theoremstyle{definition}
\newtheorem{proposition}{Proposition}[chapter]
\newtheorem{observation}{Observation}[chapter]

\newtheorem{corollary}{Corollary}[chapter]
\newtheorem{remark}{Remark}[chapter]
\newtheorem{definition}{Definition}[chapter]
\newtheorem{theorem}{Theorem}[chapter]

\newtheorem{example}{Example}[chapter]
\AtBeginEnvironment{example}{\pushQED{\qed}}
\AtEndEnvironment{example}{\popQED\endexample}

\newcommand{\pref}{\succ}

\newcommand{\score}{{{\mathrm{score}}}}

\newcommand{\reals}{{\mathbb{R}}}
\newcommand{\p}{{\mathrm{P}}}
\newcommand{\np}{{\mathrm{NP}}}

\newcommand{\wone}{{\mathrm{W[1]}}}

\newcommand{\calB}{\mathcal{B}}
\newcommand{\calM}{\mathcal{M}}

\newcommand{\calS}{\mathcal{S}}

\newcommand{\calD}{\mathcal{D}}
\newcommand{\calL}{\mathcal{L}}
\newcommand{\calP}{\mathcal{P}}
\newcommand{\calT}{\mathcal{T}}

\newcommand{\POS}{{{d_\mathrm{pos}}}}
\newcommand{\LPOS}{{{d_\mathrm{pos}^{\ell_1}}}}
\newcommand{\PAIR}{{{d_\mathrm{pair}}}}
\newcommand{\BOR}{{{d_\mathrm{Borda}}}}
\newcommand{\EMD}{{{{\mathrm{emd}}}}}
\newcommand{\discrete}{{{\mathrm{disc}}}}
\newcommand{\DISC}{{{d_\mathrm{disc}}}}
\newcommand{\hamming}{{{\mathrm{H}}}}
\newcommand{\SWAP}{{{d_\mathrm{swap}}}}
\newcommand{\swap}{{{\mathrm{swap}}}}
\newcommand{\SPEAR}{{{d_\mathrm{Spear}}}}
\newcommand{\spearman}{{{\mathrm{Spear}}}}
\newcommand{\MR}{{{\mathrm{MR}}}}
\newcommand{\TMR}{{{\mathrm{TMR}}}}

\newcommand{\cost}{{{{\mathrm{cost}}}}}
\newcommand{\pos}{{{{\mathrm{pos}}}}}

\newcommand{\emd}{{{\mathrm{emd}}}}
\newcommand{\sort}{{{\mathrm{sort}}}}
\newcommand{\normphi}{{{\mathrm{norm}\hbox{-}\phi}}}

\newcommand{\rID}{{{\mathrm{rID}}}}
\newcommand{\UN}{{{\mathrm{UN}}}}
\newcommand{\AN}{{{\mathrm{AN}}}}
\newcommand{\ID}{{{\mathrm{ID}}}}
\newcommand{\id}{{{\mathrm{id}}}}
\newcommand{\an}{{{\mathrm{an}}}}
\newcommand{\ST}{{{\mathrm{ST}}}}

\newcommand{\stt}{{{\mathrm{st}}}}
\newcommand{\sgn}{{{\mathrm{sgn}}}}
\newcommand{\Euc}{{{\mathrm{Euc}}}}

\newcommand{\did}{d\hbox{-}\mathrm{ID}}
\newcommand{\discid}{d_{disc}\hbox{-}\mathrm{ID}}

\newcommand{\freq}{{{\mathrm{freq}}}} 
\usepackage{ mathrsfs }

\usepackage[most]{tcolorbox}

\newtcolorbox{conclusionbox}[1][]{
    colback=white,
    colframe=black,
    title=Main Conclusions,
    sharp corners,
    #1,
}
\newtcolorbox{conclusionsbox}[1][]{
    colback=white,
    colframe=black,
    title=Main Conclusions,
    sharp corners,
    #1,
}
\newtcolorbox{contributionbox}[1][]{
    colback=white,
    colframe=black,
    title=Main Contributions,
    sharp corners,
    #1,
}
\newtcolorbox{contributionsbox}[1][]{
    colback=white,
    colframe=black,
    title=Main Contributions,
    sharp corners,
    #1,
}

\def\R{\mathbb{R}}

\def\2vec#1#2{\left(\begin{array}{c}{#1}\\{#2}\end{array}\right)}

\usepackage{helvet}
\usepackage{courier}
\DeclareCaptionStyle{ruled}{labelfont=normalfont,labelsep=colon,strut=off}
\usepackage{environ}

\newcommand{\repeatproposition}[1]{  \begingroup
  \renewcommand{\theproposition}{\ref{#1}}  \expandafter\expandafter\expandafter\proposition
  \csname repproposition@#1\endcsname
  \endproposition
  \endgroup
  \setcounter{theorem}{\value{theorem}-1}
}

\NewEnviron{repproposition}[1]{  \global\expandafter\xdef\csname repproposition@#1\endcsname{    \unexpanded\expandafter{\BODY}  }  \expandafter\proposition\BODY\unskip\label{#1}\endproposition
}

\newcommand{\repeattheorem}[1]{  \begingroup
  \renewcommand{\thetheorem}{\ref{#1}}  \expandafter\expandafter\expandafter\theorem
  \csname reptheorem@#1\endcsname
  \endtheorem
  \endgroup
  \setcounter{theorem}{\value{theorem}-1}
}

\NewEnviron{reptheorem}[1]{  \global\expandafter\xdef\csname reptheorem@#1\endcsname{    \unexpanded\expandafter{\BODY}  }  \expandafter\theorem\BODY\unskip\label{#1}\endtheorem
}

\usepackage{thm-restate}

\newcommand{\av}{\mathit{av}}

\newcommand{\app}{{\mathrm{app}}}
\DeclareMathOperator{\ham}{ham}

\usepackage[toc,page]{appendix}

\newcommand{\new}[1]{{\color{black} #1}}

\usetikzlibrary{chains,positioning,decorations.pathreplacing}

 \newcommand{\drawunabove}[2]{
    \draw (#1+0.5, #2+1) node[anchor=south] {UN};
    \fill[black!25!white] (#1+0,#2+0)  rectangle (#1+1,#2+1);
    \draw (#1+0,#2+0) rectangle (#1+1,#2+1);
  }
  \newcommand{\drawidabove}[2]{
    \draw (#1+0.5, #2+1) node[anchor=south] {ID};
    \fill[black!25!white] (#1+0,#2+1)  -- (#1+0.2, #2+1) -- (#1+1, #2+0.2) -- (#1+1, #2) -- (#1+1-0.2, #2) -- (#1, #2+1-0.2) -- cycle;
    \draw (#1+0,#2+0) rectangle (#1+1,#2+1);
  }

    \newcommand{\drawanaboveleft}[2]{
    \draw (#1-0.75, #2+0.5) node[anchor=south] {AN};
    \fill[black!25!white] (#1+0,#2+0)  -- (#1+0.2, #2+0) -- (#1+1, #2+1-0.2) -- (#1+1, #2+1) -- (#1+1-0.2, #2+1) -- (#1, #2+0.2) -- cycle;
    \fill[black!25!white] (#1+0,#2+1)  -- (#1+0.2, #2+1) -- (#1+1, #2+0.2) -- (#1+1, #2) -- (#1+1-0.2, #2) -- (#1, #2+1-0.2) -- cycle;
    \draw (#1+0,#2+0) rectangle (#1+1,#2+1);
  }
    \newcommand{\drawanaboveright}[2]{
    \draw (#1+1.75, #2+0.5) node[anchor=south] {AN};
    \fill[black!25!white] (#1+0,#2+0)  -- (#1+0.2, #2+0) -- (#1+1, #2+1-0.2) -- (#1+1, #2+1) -- (#1+1-0.2, #2+1) -- (#1, #2+0.2) -- cycle;
    \fill[black!25!white] (#1+0,#2+1)  -- (#1+0.2, #2+1) -- (#1+1, #2+0.2) -- (#1+1, #2) -- (#1+1-0.2, #2) -- (#1, #2+1-0.2) -- cycle;
    \draw (#1+0,#2+0) rectangle (#1+1,#2+1);
  }
    \newcommand{\drawstbelowleft}[2]{
    \draw (#1-0.75, #2-0.5) node[anchor=south] {ST};
    \fill[black!25!white] (#1+0,#2+1)  rectangle (#1+0.5, #2+0.5);
    \fill[black!25!white] (#1+0.5,#2+0.5)  rectangle (#1+1, #2+0);
    \draw (#1+0,#2+0) rectangle (#1+1,#2+1);
  }
      \newcommand{\drawstbelowright}[2]{
    \draw (#1+1.75, #2-0.5) node[anchor=south] {ST};
    \fill[black!25!white] (#1+0,#2+1)  rectangle (#1+0.5, #2+0.5);
    \fill[black!25!white] (#1+0.5,#2+0.5)  rectangle (#1+1, #2+0);
    \draw (#1+0,#2+0) rectangle (#1+1,#2+1);
  }
  
  \usepackage[textwidth=1.5cm, textsize=tiny]{todonotes}
 
 \newcommand{\drawunabovepair}[2]{
    \draw (#1+0.5, #2+1) node[anchor=south] {UN};
    \draw (#1+0.5, #2-1.15) node[anchor=south] {AN};
    \fill[black!12!white] (#1+0,#2+0)  rectangle (#1+1,#2+1);
    \draw (#1+0,#2+0) rectangle (#1+1,#2+1);
    \draw (#1+0,#2+1) -- (#1+1,#2+0); 
  }
   \newcommand{\drawidabovepair}[2]{
    \draw (#1+0.5, #2+1) node[anchor=south] {ID};
    \fill[black!25!white] (#1+0,#2+1) -- (#1+1,#2+0) -- (#1+1, #2+1) -- cycle;
    \draw (#1+0,#2+0) rectangle (#1+1,#2+1);
    \draw (#1+0,#2+1) -- (#1+1,#2+0); 
  }
 \newcommand{\drawstbelowpair}[2]{
    \draw (#1+0.5, #2-1.15) node[anchor=south] {ST};
    \fill[black!25!white] (#1+0.5,#2+0.5)  rectangle (#1+1, #2+1);
    \fill[black!12!white] (#1+0,#2+1)  rectangle (#1+0.5, #2+0.5);
    \fill[black!12!white] (#1+0.5,#2+0.5)  rectangle (#1+1, #2+0);
    \draw (#1+0,#2+0) rectangle (#1+1,#2+1);
    \draw (#1+0,#2+1) -- (#1+1,#2+0); 
    \draw (#1+0.5,#2+0) -- (#1+0.5,#2+1);
  }
  
 \newcommand{\drawunaboveborda}[2]{
    \draw (#1+0.5, #2+0.75) node[anchor=south] {UN};
    \draw (#1+0.5, #2-0.85) node[anchor=south] {AN};
    \fill[black!25!white] (#1+0,#2+0.25)  rectangle (#1+1,#2+0.5);
    \draw (#1+0,#2+0.25) rectangle (#1+1,#2+0.75);
  }
   \newcommand{\drawidaboveborda}[2]{
    \draw (#1+0.5, #2+0.75) node[anchor=south] {ID};
    \fill[black!25!white] (#1+0,#2+0.25) -- (#1+1,#2+0.25) -- (#1+0, #2+0.75) -- cycle;
    \draw (#1+0,#2+0.25) rectangle (#1+1,#2+0.75);
  }
  
 \newcommand{\drawstbelowborda}[2]{
    \draw (#1+0.5, #2-1.15) node[anchor=south] {ST};
    \fill[black!25!white] (#1+0,#2+0.25)  rectangle (#1+0.5, #2+0.625);
    \fill[black!25!white] (#1+0.5,#2+0.25)  rectangle (#1+1, #2+0.375);
    \draw (#1+0,#2+0.25) rectangle (#1+1,#2+0.75);
  }

 \newcommand{\drawunswap}[2]{
    \draw (#1+0.5, #2+0) node[anchor=south] {\scriptsize UN};
    \draw (#1+0,#2+0) rectangle (#1+1,#2+1);
  }

  \newcommand{\drawidswap}[2]{
    \draw (#1+0.5, #2+0) node[anchor=south] {\scriptsize ID};
    \draw (#1+0,#2+0) rectangle (#1+1,#2+1);
  }
  \newcommand{\drawanswap}[2]{
    \draw (#1+0.5, #2+0) node[anchor=south] {\scriptsize AN};
    \draw (#1+0,#2+0) rectangle (#1+1,#2+1);
  }
  \newcommand{\drawstswap}[2]{
    \draw (#1+0.5, #2+0) node[anchor=south] {\scriptsize ST};
    \draw (#1+0,#2+0) rectangle (#1+1,#2+1);
  }

\usepackage{cleveref}
\crefname{figure}{figure}{figures}
\hypersetup{
		pdfencoding=auto, 
		psdextra,
		colorlinks=true,
		citecolor=black,
		linkcolor=black,
		urlcolor=black
	}
\usepackage{subcaption}

\begin{document}

\mytitle

\quad
\thispagestyle{empty}
\clearpage
\selectlanguage{polish}
\begin{abstract}

W niniejszej rozprawie skupiam się na badaniu zagadnień związanych z obliczeniową teorią wyboru społecznego (ang. \textit{computational social choice}). Dyscyplina ta koncentruje się na analizie zbiorowego podejmowania decyzji -- w szczególności na jej obliczeniowych aspektach. Ze zbiorowym podejmowaniem decyzji mamy do czynienia między innymi w kontekście wyborów. Przykładowo mogą to być wybory prezydenckie bądź wybory parlamentarne. 
Wyborami można również nazwać proces wyłaniania zwycięzcy w konkursie Chopinowskim. 
Formalnie, przez wybory rozumiemy zbiór kandydatów oraz zbiór wyborców posiadających pewne preferencje względem tychże kandydatów. Teoria wyborów mierzy się z szeregiem problemów, takich jak stwierdzanie kto wygrywa dane wybory, ocenianie marginesów zwycięstwa, analiza różnego rodzaju manipulacji wynikiem, badanie własności aksjomatycznych (np. kryterium Condorceta) i wiele innych.

W swoich badaniach przede wszystkim skupiam się na analizie różnych statystycznych modeli preferencji, czyli modeli pozwalających generować wybory (zestawy głosów).  Analizowane przeze mnie modele preferencji są powszechnie wykorzystywane przez społeczność zajmującą się algorytmicznymi aspektami wyborów. Ich lepsze zrozumienie pozwoli w przyszłości na trafniejsze dobieranie modeli zależnie od sytuacji (przykładowo do symulacji obliczeniowych) oraz na bardziej racjonalne planowanie eksperymentów obliczeniowych.

Każde wybory, zarówno te prawdziwe, jak i te wygenerowane przez modele statystyczne, możemy utożsamić z punktem w pewnej wielowymiarowej przestrzeni. Pojawia się pytanie, jak porównywać ze sobą różne wybory? W szczególności jak mierzyć odległości pomiędzy nimi? 

Próbując odpowiedzieć na powyższe pytania wprowadzam narzędzie nazywane mapą wyborów -- graficzną reprezentację ułatwiającą zrozumienie przestrzeni wyborów. Na początku przygotowujemy zestaw wyborów. Następnie, zgodnie z zadaną metryką, obliczamy odległości pomiędzy każdą parą wyborów. Na koniec, bazując na obliczonych odległościach, osadzamy wszystkie wybory (punkty) w dwuwymiarowej przestrzeni euklidesowej, tak aby odległości euklidesowe, jak najlepiej odzwierciedlały, te obliczone przy pomocy metryki. Mapa wyborów – to nie pojedyncza mapa, a narzędzie pozwalające tworzyć różne warianty mapy dla różnych modeli i parametrów. Mapa wyborów pozwala lepiej zrozumieć zarówno istniejące modele, jak i prawdziwe wybory. Dzięki mapie udało się dokonać wielu istotnych spostrzeżeń. 


\end{abstract}

\quad
\thispagestyle{empty}
\clearpage
\selectlanguage{english}
\begin{abstract}

In the following thesis, we study the topics related to computational social choice theory. This discipline focuses on the analysis of collective decision-making, in particular, on its computational aspects. We deal with collective decision-making, for example, in the context of elections. For instance, it can be a presidential or a parliamentary election. By an election we can also call the process of selecting the winner in the Chopin Competition. Formally, by an election we mean a set of candidates and a set of voters who have certain preferences over these candidates. Election theory faces a number of problems, such as determining the winner or winners in an election, calculating margins of victory, analyzing various types of manipulation, studying axiomatic properties (e.g., the Condorcet winner criterion), and many others.

In our research, we mainly focus on the analysis of various statistical models of preferences, i.e., models for generating elections (votes). The preference models we have analyzed are widely used by the community dealing with the algorithmic aspects of elections. Their better understanding will allow for more accurate selection of models depending on the situation (for example, for computational simulations) and for more rational planning of computational experiments.

Each election, both real and generated from statistical cultures, can be associated with a point in a multidimensional space. The question arises, how to compare different elections with each other? In particular, how do we measure the distances between them?

In an attempt to answer these questions, we introduce a framework called a map of elections, a graphical representation that makes it easier to understand the election space. First, we prepare a set of elections. Then, according to a given metric, we calculate the distances between each pair of elections. Finally, based on the calculated distances, we embed all elections (points) in a two-dimensional Euclidean space so that the Euclidean distances reflect as closely as possible those computed with the metric. The map of elections is not a single map, but a framework that allows us to create different map variants for different models and parameters. Moreover, the map of elections allows us to better understand both existing statistical cultures and real-life elections. Thanks to the map, it was possible to make many intriguing observations.

\end{abstract}

\tableofcontents

\chapter{Introduction}
\label{ch:introduction}

When talking about elections, most people focus on the election's outcome, or, in other words, on the winners. However, a raw election without a voting rule and without a winner is a very interesting object in itself. It is surprising how much we can say about an election without electing anyone at all. For example, if in an election each voter votes differently, we would say that such an election is very diverse. On the other hand, if in an election all voters vote in exactly the same way (i.e., all votes are identical), then we would say that the voters are in perfect agreement.


Usually by an election we will refer to an object that consists of a set of candidates and a collection of voters that have some preferences over these candidates. In most cases, we assume that each voter strictly ranks all the candidates, from the most to the least appealing one. 

The most popular examples of elections are political elections, such as presidential and parliamentary ones. But, in fact, we do not only deal with elections in politics. Many surveys or sport competitions can be seen as elections as well. However, perhaps in a less intuitive way. To clarify how one can treat a sport competition as an election, let us give an example. Given Formula 1 races from a certain year, we can treat each single race as a single vote, where the position of each candidate in the ballot is the place he or she won in the race. The driver who finishes the race first will be ranked in the first place in the ballot, the second driver will be ranked second, and so on. The total number of votes will be equal to the number of races that took place in a given year. We can understand various other competitions, such as Tour de France or Giro d’Italia as elections in a similar manner.



\section{Map of Elections}
One of the questions that arises naturally is \textit{when are two elections similar?} Or stated the other way around, \textit{when are two elections different?} And if they are different, we might want to know how different they are.  For instance, are presidential elections in Poland and France very different? Or maybe the structure of the elections is similar and only the names on the ballots differ? Let us have a look at the following toy example. We have two elections, each consisting of three voters and three candidates. In the first election, children give their preferences about animals, and in the second one, adults give their preferences about food:
  \begin{align*}
    Veronica \colon &Pig \pref Snail \pref Rabbit, \\
    John  \colon &Snail \pref Pig \pref Rabbit, \\
    Nicholas  \colon &Rabbit \pref Snail \pref Pig. \\
    & \\
    Newman \colon &Risotto \pref Salad \pref Pizza, \\
    Veil  \colon &Pizza \pref Salad \pref Risotto, \\
    Johnson  \colon &Salad \pref Pizza \pref Risotto. \\
  \end{align*}

Are these two elections similar? At first glance, probably not that much. But if we forget about the names of the voters and the names of the alternatives, these two elections become identical. To see this, let us assume that Veronica is Veil, John is Johnson, Nicholas is Newman, Pig is Pizza, Snail is Salad and Rabbit is Risotto. Mathematically speaking, we simply have three preference orders, each of them appearing exactly once. The order of votes is irrelevant:
  \begin{align*}
    (p \pref s \pref r, \ \ s \pref p \pref r, \ \ r \pref s \pref p)
  \end{align*}

To speak more generally, given two elections, the first problem which we will face is verifying whether these elections are isomorphic, that is, if it is possible to rename the candidates and the voters in such a way that these elections become identical. Verifying whether two elections are isomorphic can be done in polynomial time. However, if two elections are not isomorphic, the second problem arises, that is, how to define and compute the distance between them. 

An efficient way of computing distances between elections is important. However, even if we knew that the distance between two elections is equal to five, we still would not know much about these elections. Is five a lot or not?

To solve this problem, we introduce another crucial component of this dissertation, the concept called the  \textit{map of elections}, to which this work owes its title. The idea is as follows. First, we generate numerous elections from various statistical cultures (that is, models that serve for generating random instances of elections). Second, we compute the distances between each pair of elections. Third, we embed these distances in a two-dimensional Euclidean space using an embedding algorithm.  Finally, we obtain a map. Map-representation of elections makes it easier to understand their numerous properties. 

To make it even easier, we mark four characteristic points on our map. First, we have an \textit{identity} election, where all voters agree on a single preference order. Then, we have a \textit{uniformity} election, where the votes are as diverse as possible. Finally, we have \textit{stratification} and \textit{antagonism} elections, the description of which we will omit in the introduction for simplicity (all four points will be described in detail in Chapter~\ref{ch:stat_cult}). We call these points \textit{the compass} because they help us navigate through the map; so when a given point (an election) lands in a certain part of the map, we can say something meaningful about this election. Identity and uniformity are the two most extreme points, representing order versus chaos, respectively. In all of our metrics, the distance between identity and uniformity is the largest possible in the whole space of elections. For example, if the distance between two particular elections is five, but the distance between identity and uniformity is six, then these two elections are far away. However, if the distances between identity and uniformity were fifty, then we can argue that these elections are quite similar.

To give the reader the flavor of what this thesis is about, we present an example of a map of elections in~\Cref{fig:map_introduction_example}. Each dot corresponds to a single election. The closer two particular dots are on the map, the more similar are elections that they represent, and if two dots are of the same color, it means that they come from the same distribution, i.e., statistical culture. How to generate elections from a given model will be described in detail in~\Cref{ch:stat_cult}. Nonetheless, without going into the technical details of particular models, we can see that for most of the models, the elections generated from that model are very similar to each other. However, it is not entirely true for, for example, blue points, which represent the Mallows model---a popular model which we will now briefly describe. The Mallows model is parametrized by a dispersion parameter, which defines the correlation between the votes within an election. The larger is the parameter, the less correlated are the votes. If this parameter is equal to zero, we have an extreme correlation and all votes are identical. When this parameter is equal to one, we witness full chaos and no correlation at all.
Going back to our map in~\Cref{fig:map_introduction_example}, depending on the dispersion parameter, elections generated from the Mallows model (the blue points) can occupy quite different places. If we sample elections from the Mallows model with numerous different values of the dispersion parameter, we obtain what we call a path from one of the extreme points, identity, to another extreme point, uniformity. To conclude, depending on the dispersion parameter, we can generate drastically different elections. Nevertheless, for a fixed parameter, all generated elections will be similar to each other.

\begin{figure}
    \centering
    \includegraphics[width=10cm, trim={0.2cm 0.2cm 0.2cm 0.2cm}, clip]{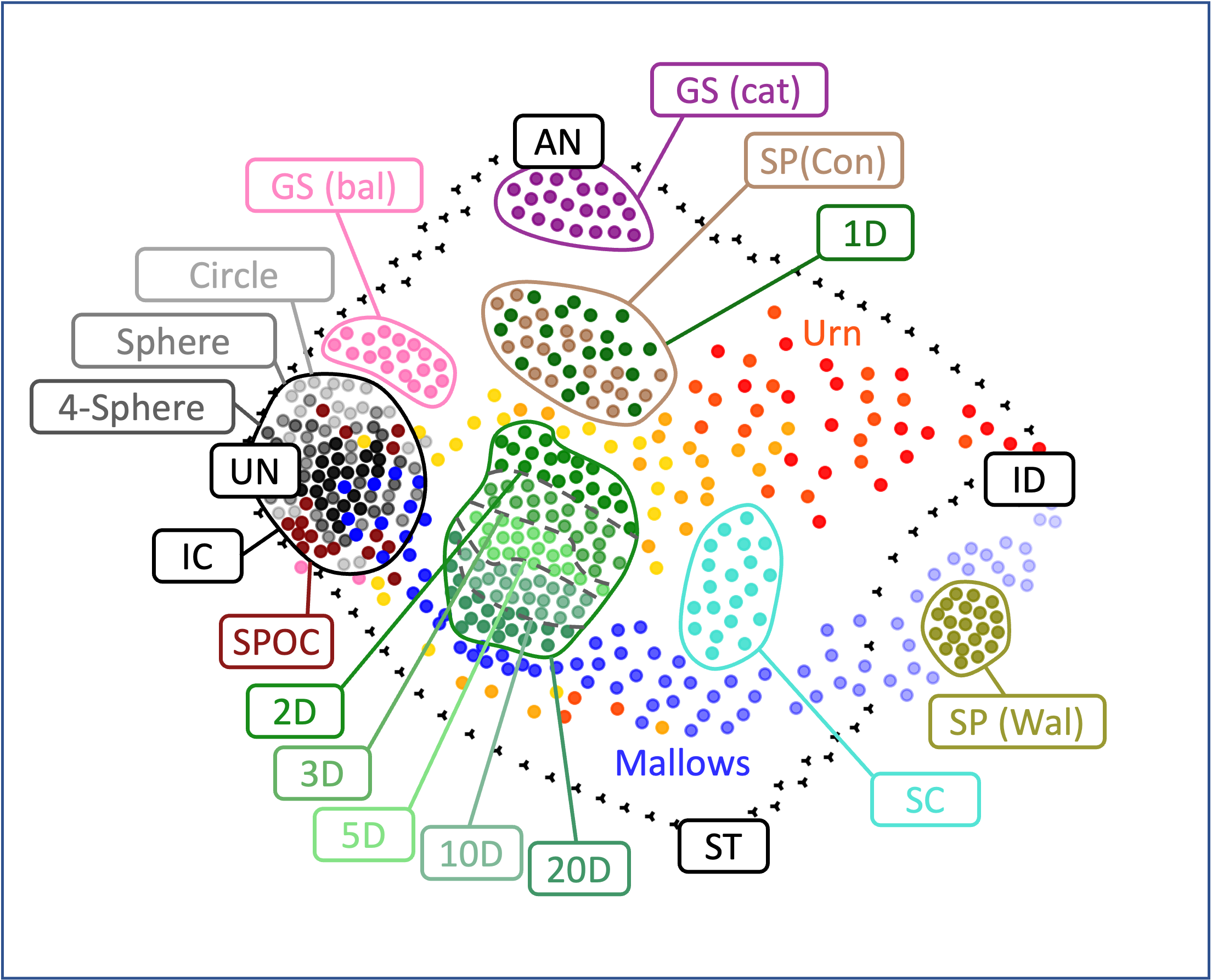}
    \caption{An example of Map of Elections.}
    \label{fig:map_introduction_example}
\end{figure}

Our analyzes of distances between elections started a new line of research within computational social choice, resulting in numerous papers and, hopefully, many more to come. In this dissertation, we focus on ordinal elections, but the \textit{map of elections} framework can be easily generalized to \textit{map of instances}, which can be used for many other types of objects that are studied within computational social choice, such as approval elections (which we discuss in detail in~\Cref{ch:approval}), stable roommates instances, stable marriages instances, participatory budgeting instances, or fair division ones.

\section{Motivation}

Although many papers on computational social choice are theoretical, the number of experimental works is rapidly growing. And there are many questions that can only be answered by experimentation. We start by giving an example related to the Condorcet winner. We call a candidate a Condorcet winner if such a candidate is preferred by more than half of the voters when compared one-to-one with any other candidate. In some elections none of the candidates is a Condorcet winner; however, if such a candidate exists, many people claim that he or she should become an overall winner of the election. We say that a voting rule satisfies the Condorcet winner criterion if whenever a Condorcet winner exists, this rule selects him or her as the winner. From a theoretical point of view, we can divide rules into two groups, those that satisfy the Condorcet winner criterion, and those that do not. Unfortunately, the real world is not black and white. It might be the case that some of the rules that do not satisfy the Condorcet winner criterion, but do not satisfy it due to very few unrealistic instances, on which they fail to select the Condorcet winner. This moves us to the second problem---what does it mean that an instance is unrealistic? It is hard to answer this question in general. But if we speak about particular types of elections, we can try to give an answer. For example, in the context of political elections, we usually have many more voters than candidates, so an instance with $1000$ candidates and $10$ voters probably is not very realistic. Another way of verifying whether a given instance is realistic is by comparing it with real-life data from a given context, e.g., political. Going back to our Condorcet winner criterion, instead of two groups, we rather have a spectrum of rules. And to distinguish between rules that almost always select a Condorcet winner (if such a candidate exists) and those that fail it more frequently, we need experiments.

Another thing that we can only partially describe with raw theory is the time needed to perform particular tasks. For example, the time needed to compute a winner or a winning committee under a certain voting rule.
In~\Cref{fig:time_introduction_example} we present an introductory example of a map of elections where each point's color depicts the time needed to compute the winning committee under the Harmonic-Borda multiwinner voting rule. As we can see, the longest time is needed for instances similar to those from impartial culture, while the shortest time is needed for those similar to identity.
For most of the rules, we know their time-complexity, however, usually it relates only to the worst case. So again, it might be the case that the rule in practice is fast, but due to some unfortunate instances, the worst-case complexity is far from polynomial. It is also interesting to know whether, if two instances of elections are similar, it takes the same amount of time to compute the winners of these elections under the considered voting rule.

The next potential benefit from this thesis is a general improvement on experiments done across the computational social choice. In numerous experiments, people use different models with different parameters that seem to be selected quite arbitrarily. A better understanding of statistical cultures and the nature of elections is crucial for conducting better experiments. 
We believe that the \textit{map of elections} framework, proposed by us, can help in choosing synthetically generated elections to use in experiments when evaluating a given voting rule
or a social choice phenomenon.

\begin{figure}[t]
    \centering
    \includegraphics[width=10cm, trim={0.2cm 0.2cm 0.2cm 0.2cm}, clip]{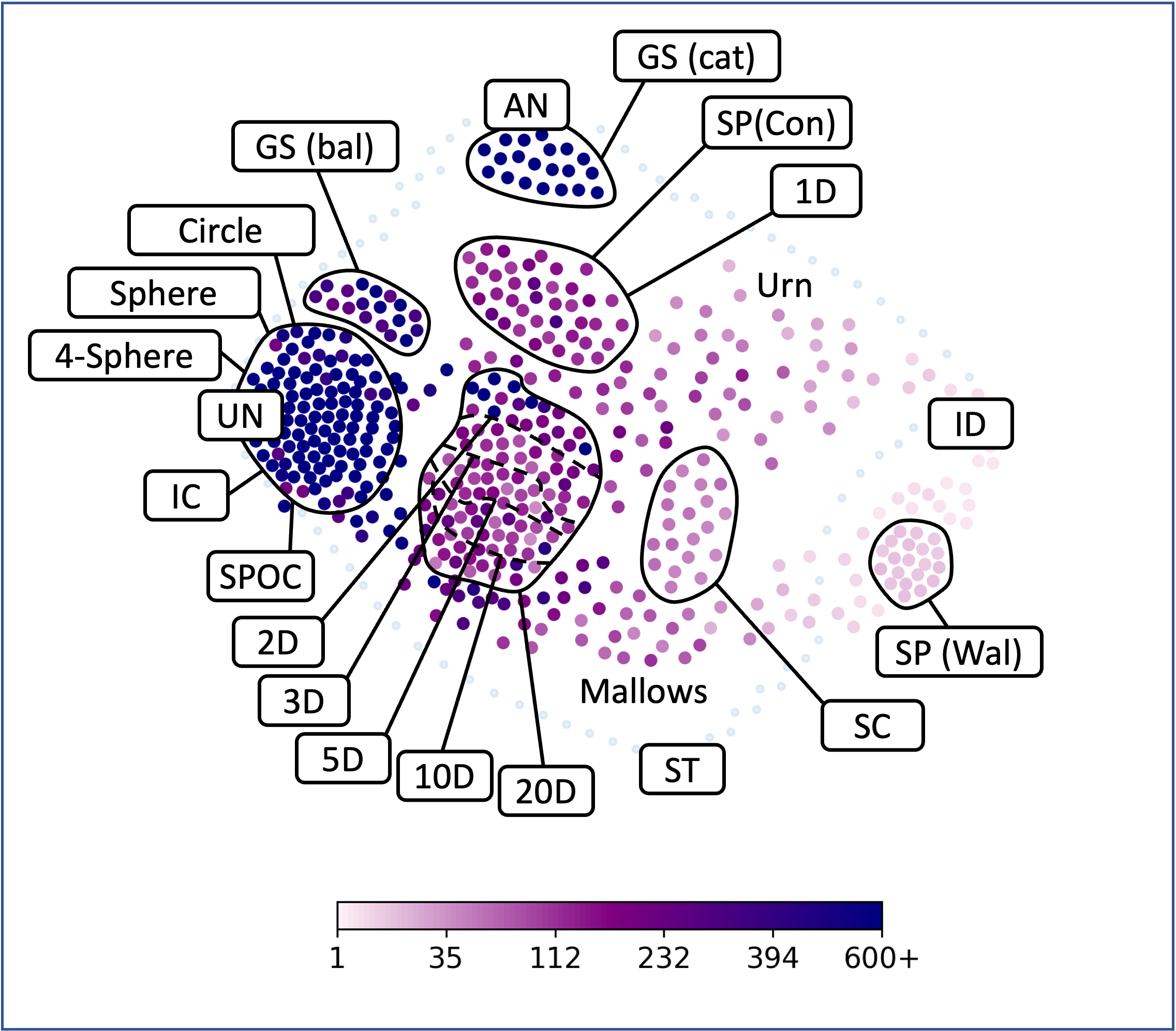}
    \caption{ILPs runtime (in seconds) for Harmonic-Borda voting rule.}
    \label{fig:time_introduction_example}
\end{figure}

There are many statistical cultures, for example, the Mallows model, the urn model, or, the impartial culture and there are many questions worth asking here. First, it would be valuable to know how different from each other the elections generated from a given model are. Next, how different are the statistical models from each other. For parameterized models, it would also be important to know how their parameters influence them, and which ranges of parameters correspond to realistic instances.

Another motivation regards real-life elections. So, one way of getting data is by generating it according to a certain statistical model. This gives us flexibility in selecting arbitrarily the number of candidates and the number of voters. However, it is also very interesting to analyze real-life data. How do real-life elections relate to synthetic data? Are real-life elections similar to each other?

Finally, our analysis will help us better understand the space of elections in itself and will tell us how different two elections can be.
In the following, we briefly describe the structure of the dissertation.
\section{Structure}
 
We start by describing statistical cultures for sampling ordinal elections and provide some insight into the inner structure of such elections.
Next, we focus on numerous distances between elections. First, we introduce isomorphic distances (that is, distances under which only isomorphic elections are at distances zero). Then, we move on to nonisomorphic distances.
Later, we evaluate the map of elections framework, and study its potential applications.
After that, we discuss the distances between elections of different sizes.
Finally, we focus on elections with approval ballots and present maps of approval elections. Below we briefly describe the content of each chapter one by one.

\begin{description}
    
    \item[Preliminaries.] Introduction of basic definitions and notation. 

    \item[Statistical Cultures.] We provide a description of statistical cultures known in the literature, and how to sample elections from these cultures. Next, we present \emph{maps of preferences}, where we look at relations between votes within a single election. Informally speaking, it is a microscope view of an election, giving us insight into the structure of the votes.
    

    \item[Distances.] In the first part of this chapter, we define the \textsc{Election Isomorphism} problem and introduce three isomorphic distance. We say that a distance is isomorphic if, for any two elections that are not isomorphic, the distance between them is larger than zero. These three metrics are the Swap distance, the Spearman distance, and the discrete distance. Unfortunately, only the discrete distance can be computed in polynomial time. As for the Swap and Spearman distances, the complexity mostly comes from the fact that we have to find optimal matchings of voters and of candidates at the same time. What is surprising is that, for Swap distance, even if the voters' matching is given, the problem remains NP-hard. 
    
    In the second part, we introduce the nonisomorphic positionwise distance, which can be computed in polynomial time. We argue that this particular distance is very practical. Although it is a pseudometric, and it is losing some precision when compared to, e.g., the Swap distance, it is much faster to compute and still carries a lot of information. Within this chapter, we also discuss two other nonisomorphic distances, the pairwise distance and the Bordawise distance. Finally, we compare all isomorphic and nonisomorphic distances.
    
    Moreover, throughout the chapter we present various maps of elections---one for each metric. 
    
        \item[Applications.]
         We focus on practical applications of the \textit{map of elections} framework. We consider several embedding algorithms (i.e., ways of putting a set of points in a low-dimensional Euclidean space)  and evaluate their performance. We focus on Fruchterman-Reingold
         force-directed algorithm, a novel variant of Kamada-Kawai algorithm, Multi Dimensional Scaling and a few others. Next, we evaluate popular voting rules for ordinal elections. First, we focus on single-winner voting rules such as Plurality, Borda, Copeland, and Dodgson. All these rules assign a certain score to each candidate and the candidate with the highest score is declared a winner. For each election, we can compute such highest score, and then color the map proportionally to that score, that is, color each point on the map proportionally to the highest score in the election that that point depicts. We also analyze multiwinner voting rules such as Chamberlin--Courant and Harmonic-Borda. Both these rules assign a certain score to each committee, and the committee with the highest score is declared as winning. So, again, we can color the map, however, this time we color it proportionally to the score of the best committee. Besides coloring the map by a score, we also color it by the runtime of the algorithm for a given voting rule. Such time-focused coloring gives an insight into which types of elections are harder (i.e., take more time) and which ones are easier (i.e., take less time) to compute. Then, we analyze real-life instances of elections. In particular, we focus on the data provided within PrefLib---a popular \textit{preference library} that contains various real-life datasets. We study such instances as political elections in Dublin, Glasgow, and Aspen; voting of Electoral Reform Society; surveys about different types of sushi and about pictures on T-Shirts; numerous sport competitions and many others. Finally, we consider a \emph{skeleton map}---a special type of a map of elections, which can be computed analytically.
    
    \item[Subelections.] For the classical \textsc{Election Isomorphism} problem, we always consider elections of the same size, that is, with the same number of voters and the same number of candidates in both elections. In this chapter, we introduce \textsc{Subelection Isomorphism}, where we relax the assumption about the sizes of elections. In the \textsc{Subelection Isomorphism} problem we are given two elections, a smaller and a larger one, and we ask if it is possible to remove some candidates and voters from the larger election so that it becomes isomorphic to the smaller one. 
    
    Moreover, we consider a family of \textsc{Maximum Common Subelection} problems, where given two elections we ask for the largest election, which is a subelection of both given elections at the same time. First, we provide the computational complexity for all variants, and later we present several experimental results on both synthetic and real-life data.
    
    \item[Approval Elections.]
    We consider approval elections, where instead of ranking all the candidates, voters approve subsets of them. In other words, each voter partition all the candidates into two sets, those that he or she approves and those that he or she does not. As for ordinal elections, we introduce distances between such elections (the isomorphic Hamming distance and the nonisomorphic approvalwise distance). We present several novel statistical cultures, and argue why we recommend using them. Again, we show maps of preferences. Finally, we present maps of approval elections and conduct experiment such as, for example, analysis of cohesiveness level or behavior of voting rules.
    
    \item[Summary.]
    In the last chapter, we recapitulate the main contributions of this dissertation, and show directions for possible extensions and future work.
    
\end{description}

\section{Conference Publications}

Most of the results presented in this dissertation have already been presented at various conferences. In the following, we attach the list of publications chronologically, by the date of publication, on which this thesis is based. All the results that are included in the thesis are due to Stanisław Szufa.

\begin{enumerate}
    \item \emph{How Similar Are Two Elections?} \\ Piotr Faliszewski, Piotr Skowron, Arkadii Slinko, \textbf{Stanisław Szufa}, Nimrod Talmon (AAAI-2019).
    \item \emph{Drawing a Map of Elections in the Space of Statistical Cultures} \\ \textbf{Stanisław Szufa}, Piotr Faliszewski, Piotr Skowron, Arkadii Slinko, Nimrod Talmon (AAMAS-2020).
    \item \emph{Putting a Compass on the Map of Elections} \\ Niclas Boehmer, Robert Bredereck, Piotr Faliszewski, Rolf Niedermeier, \textbf{Stanisław Szufa} (IJCAI-2021).
    \item \emph{The Complexity of Subelection Isomorphism Problems} \\ Piotr Faliszewski, Krzysztof Sornat, \textbf{Stanisław Szufa} (AAAI-2022).
    \item \emph{Understanding Distance Measures Among Elections} \\ Niclas Boehmer, Piotr Faliszewski, Rolf Niedermeier, \textbf{Stanisław Szufa}, Tomasz W\c{a}s (IJCAI-2022).
    \item \emph{How to Sample Approval Elections?} \\ \textbf{Stanisław Szufa}, Piotr Faliszewski, Łukasz Janeczko, Martin Lackner, Arkadii Slinko, Krzysztof Sornat, Nimrod Talmon (IJCAI-2022).
    \item \emph{Expected Frequency Matrices of Elections: Computation, Geometry, and Preference Learning} \\ Niclas Boehmer, Robert Bredereck, Edith Elkind, Piotr Faliszewski, \textbf{Stanisław Szufa} (NeurIPS-2022).
\end{enumerate}

\new{
Below we list the results from the thesis that are not included in any of the publications described above.
\begin{itemize}
    \item Maps of Ordinal Preferences~(\Cref{ordinal_map_pref}).
    \item Evaluation of different embedding algorithms (i.e., analysis of distortion and monotonicity; \Cref{ch:applications:sec:embedding}).
    \item Comparison of the performerce of different voting rules~(\Cref{ch:applications:sec:rules}).\footnote{Minor results were also published in the work of~\cite{szu-fal-sko-sli-tal:c:map}}
    \item Experiments on real-life data in the context of subeletions~(\Cref{real-life-sub-exp})
    \item Maps of Approval Preferences~(\Cref{approval_map_pref}).
\end{itemize}
}

\new{
As to the connections between the chapters and the papers, they are as follows.
\begin{itemize}
    \item \Cref{ch:distances} is based on the papers \emph{How Similar Are Two Elections?}~\citep{fal-sko-sli-szu-tal:c:isomorphism}, \emph{Drawing a Map of Elections in the Space of Statistical Cultures}~\citep{szu-fal-sko-sli-tal:c:map}, \emph{Putting a Compass on the Map of Elections}~\citep{boe-bre-fal-nie-szu:c:compass}, \emph{Understanding Distance Measures Among Elections}~\citep{boe-fal-nie-szu-was:c:understanding}. 
    \item \Cref{ch:applications} is based on the papers \emph{Drawing a Map of Elections in the Space of Statistical Cultures}~\citep{szu-fal-sko-sli-tal:c:map} and \emph{Expected Frequency Matrices of Elections: Computation, Geometry, and Preference Learning}~\citep{boehmer2022expected}.
    \item \Cref{ch:subelections} is based on the paper \emph{The Complexity of Subelection Isomorphism Problems}~\citep{faliszewski2022complexity}.
    \item \Cref{ch:approval} is based on the paper \emph{How to Sample Approval Elections?}~\citep{SFJLSS22}.
\end{itemize}
}

\subsection*{Acknowledgements}
The research presented in this dissertation was supported by the National Science Centre, Poland (NCN) grant No 2018/29/N/ST6/01303 and by European Research Council (ERC) under the European Union’s Horizon 2020 research and innovation programme (grant agreement No 101002854).

\begin{center}
    \includegraphics[width=5cm]{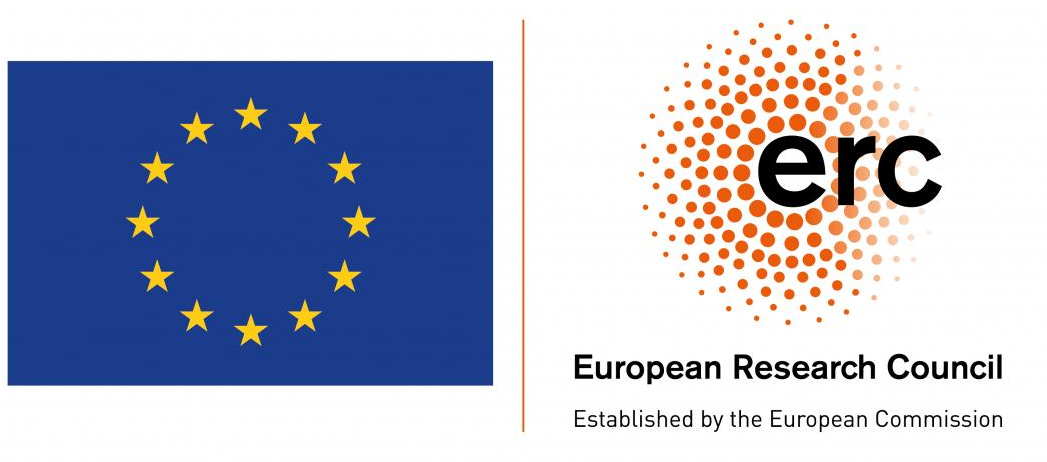}
\end{center}

\chapter{Preliminaries}
\label{ch:preliminaries}


In this chapter, we describe the basic concepts and notation. We define basic metrics between vectors and basic metrics between votes, which will serve us for computing distances between elections, to be introduced in \Cref{ch:distances}. We briefly describe six different embedding algorithms. Finally, we explain the concept of a~\emph{map}.


For a given positive integer~$t$, we write~$[t]$ to denote the set~$\{1,2,\dots,t\}$, and we write~$[t]_0$ as an abbreviation for~$[t] \cup \{0\}$. By~$ \mathbb{R}_+$ we denote the set of nonnegative real numbers.
By~$S_n$ we mean the set of all permutations over~$[n]$. Given two equal-sized sets~$A$ and~$B$, by~$\Pi(A,B)$ we denote the set of all one-to-one mappings from A to B.
For a vector~$x$,~$\overline{x}$ denotes the arithmetic average of
the values from~$x$.



\section{Elections}

An election~$E = (C,V)$ consists of a set of candidates~$C = \{c_1,\ldots, c_m\}$ and a collection of voters~$V = (v_1, \ldots, v_n)$, where each voter~$v$ has a \emph{preference
order} (sometimes referred to as a \emph{vote}), also denoted as~$v$ (the exact meaning will always be clear
from the context, and this convention will simplify our discussions).
We write~$\calL(C)$ to denote the set of all preference orders over~$C$. Every subset~$\calD$ of~$\calL(C)$ is called a domain (of
preference orders over~$C$) and, in particular,~$\calL(C)$ itself is
the general domain. 
The preference orders 
always come from some domain~$\calD$ (the general domain, unless stated otherwise).
Given two candidates~$c_i, c_j \in C$, we write~$c_i \pref_v c_j$ (or, equivalently,~$v \colon c_i \pref c_j$) to denote that voter~$v$ prefers~$c_i$ to~$c_j$. We extend this notation to more than two candidates in a
natural way. For example, we write~$v \colon c_1 \pref c_2 \pref \cdots \pref c_m$ to indicate that voter~$v$ likes~$c_1$ best, then~$c_2$, and so on, until~$c_m$. If we put
some set~$S$ of candidates in such a description of a preference
order, then we mean listing its members in some arbitrary (but fixed, global)
order. Including~$\overleftarrow{S}$ means listing the members of~$S$
in the reverse order.



Consider two sets of candidates,~$C$ and~$D$, of the same
cardinality. 
Let~$\sigma$ be a bijection from~$C$ to~$D$.  We extend~$\sigma$ to act on preference orders~$v$ in~$\mathcal{L}(C)$ in the
natural way:~$\sigma(v)\in \mathcal{L}(D)$
is the preference
order such that for each~$c, c' \in C$ it holds that~$v \colon c \pref c' \iff \sigma(v) \colon \sigma(c) \pref
\sigma(c')$.

For an election~$E = (C,V)$, where~$V = (v_1, \ldots, v_n)$, a
candidate set~$D$, and a bijection~$\sigma$ from~$C$ to~$D$, by~$\sigma(E)$ we mean election with candidate set~$D$ and voter
collection~$(\sigma(v_1), \ldots, \sigma(v_n))$.
Similarly, given a permutation~$\pi \in S_n$, by~$\pi(V)$ we mean~$(v_{\pi(1)}, \ldots, v_{\pi(n)})$.

\section{Distances}

Formally, for a set~$X$ a function~$d \colon X \times X \rightarrow \mathbb{R}_+ \cup \{0\}$ is a metric if for each~$x, y, z \in X$ it holds that:

\begin{enumerate}
    \item~$d(x,y) = 0$ if and only if~$x = y$,
    \item~$d(x,y) = d(y,x)$,
    \item~$d(x,z) \leq d(x,y) + d(y,z)$.
\end{enumerate}    

\noindent
A~pseudometric relaxes the first condition
to the requirement that~$d(x,x) = 0$ for each~$x \in X$. In
particular, for a pseudometric~$d$ it is possible that~$d(x,y) = 0$
when~$x \neq y$.

\subsubsection{Distances Between Vectors}
For some metrics, as an intermediate step, we will be computing the distances between vectors of real numbers.

Let~$x = (x_1, \ldots, x_n)$ and~$y = (y_1, \ldots, y_n)$ be
two real-valued vectors. Then 
%
%
for a given~$p\in \mathbb{R}$, their~$\ell_p$-distance is 
$\ell_p(x,y) = (|x_1-y_1|^p + \cdots + |x_n-y_n|^p)^{\nicefrac{1}{p}}$.

Given two real-valued vectors $x$ and
$y$, we write $\EMD(x, y)$
to denote the \emph{earth mover's distance} (EMD) between them. Intuitively,
this is the minimal cost of turning $x$ into $y$, where the cost of
moving a value $\Delta$ from position $i$ to position $j$ in the
vector is $\Delta \cdot |i - j|$. Our EMD
distance can be computed using a well-known greedy polynomial-time
algorithm.

Given a real-valued vector~$z = (z_1, \ldots, z_n)$, we
write~$\hat{z}$ to denote its prefix-sum variant, i.e., an~$n$-dimensional vector such that for each~$i \in [n]$, its~$i$-th
entry is~$\hat{z}_i = z_1 + z_2 + \cdots +  z_i$.
If the entries of~$x$ and~$y$ sum up to the same value and contain only nonnegative
entries, then their \emph{earth mover's distance} alternatively, can be defined as:
\[
  \EMD(x,y) = \ell_1(\hat{x},\hat{y}).
\]
%
%
Both definitions, presented above, are equivalent (\cite{rubner2000earth}).


\subsubsection{Distances Between Votes}
\label{ch:preliminaries:dist_between_votes}
We focus on the following three distances between preference orders
(below, let~$C$ be a set of candidates and let~$u$ and~$v$ be two
preference orders from~$\calL(C)$):

\begin{description}\label{swapSpearDef}
\item[Discrete Distance.]  The discrete distance between~$u$ and~$v$,~$d_\discrete(u,v)$, is~$0$ when~$u$ and~$v$ coincide and
  is~$1$ otherwise.

\item[Swap Distance.] The swap distance between~$u$ and~$v$ (also known as Kendall's Tau distance in statistics),
  denoted~$d_\swap(u,v)$, is the smallest number of swaps of
  consecutive candidates that need to be performed within~$u$ to
  transform it into~$v$.

\item[Spearman Distance.] The Spearman's distance (also known as the
  Spearman's footrule or the displacement distance) measures the total
  displacement of candidates in~$u$ relative to their positions in~$v$. Formally, it is defined as:
\[
   d_\spearman(u,v) = \sum_{c\in C} |\text{pos}_v(c)-\text{pos}_u(c)|.
\]
\end{description}

\begin{example}\label{ex:2}
  Consider an election~$E = (C,V)$, where~$C = \{a,b,c,d,e\}$,~$V = (u,v)$, and the votes are:
  \begin{align*}
    \small
    u\colon&  a \pref b \pref c \pref d \pref e, \\
    v\colon&  b \pref a \pref e \pref c \pref d.
  \end{align*}
  Then,~$d_\discrete(u,v)=1$, because the votes are not identical;~$d_\swap(u,v)=3$, as to transform~$u$ into~$v$ we can first swap~$(a,b)$, then~$(d,e)$, and finally~$(c,e)$; and~$d_\spearman(u,v)=1+1+1+1+2=6$;~$1$ from candidate~$a$ because~$|pos_u(a)-pos_v(a)|=1$, also~$1$ from candidate~$b$ because~$|pos_u(b)-pos_v(b)|=1$, similarly~$1$ from~$c$, and~$1$ from~$d$, and finally~$2$ from~$e$ because~$|pos_u(e)-pos_v(e)|=2$.
\end{example}

We only consider distances over preference orders that are defined for
all sets of candidates (as is the case for~$d_\discrete$,~$d_\swap$,
and~$d_\spearman$). For an in-depth discussion regarding distances
between elections, we point 
to the literature on distance rationalizability of voting
rules~\citep{nit:j:closeness,mes-nur:b:distance-realizability,elk-fal-sli:j:dr}
and, in particular, to the survey of~\cite{elk-sli:b:rationalization}.

\section{Correlation}

The Pearson Correlation Coefficient (PCC) measures the level of linear correlation between two random variables
and takes values between~$-1$ and~$1$. Its absolute value gives the
level of correlation, and the sign indicates positive or negative
correlation. For two vectors~$x = (x_1, \ldots, x_t)$ and~$y = (y_1, \ldots, y_t)$,
their PCC is defined as:
\[\textstyle
     \mathrm{PCC}(x,y) = \frac{\sum_{i=1}^t(x_i-\overline{x})(y_i-\overline{y})} {\left({\sum_{i=1}^t(x_i-\overline{x})^2 \sum_{i=1}^t(y_i-\overline{y})^2}\right)^{\nicefrac{1}{2}}}. 
\]

\noindent


\noindent
In~\Cref{fig:pcc} we present the examples of the PCC values. For values close to $0$ we do not have correlation at all. For values close to $1$ ($-1$) we have strong positive (negative) correlation.

We use PCC as it is one of the standard, well-known ways of calculating the correlation.

\begin{figure}[t]
    \centering

    \begin{subfigure}[b]{0.19\textwidth}
        \centering
        \includegraphics[width=3cm]{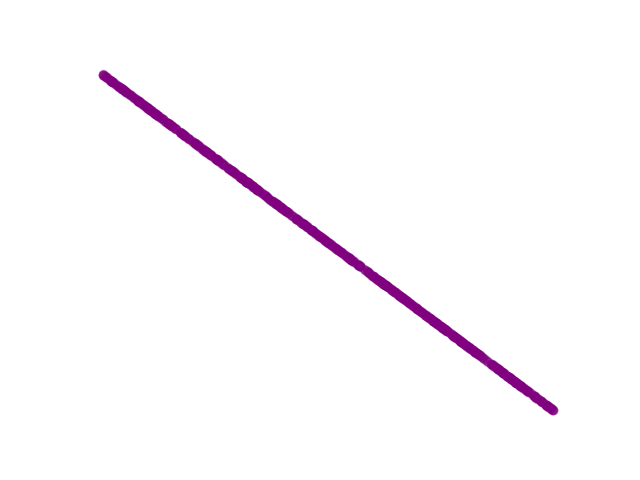}
        \caption{PCC = $-1$}
    \end{subfigure}
    \begin{subfigure}[b]{0.19\textwidth}
        \centering
        \includegraphics[width=3cm]{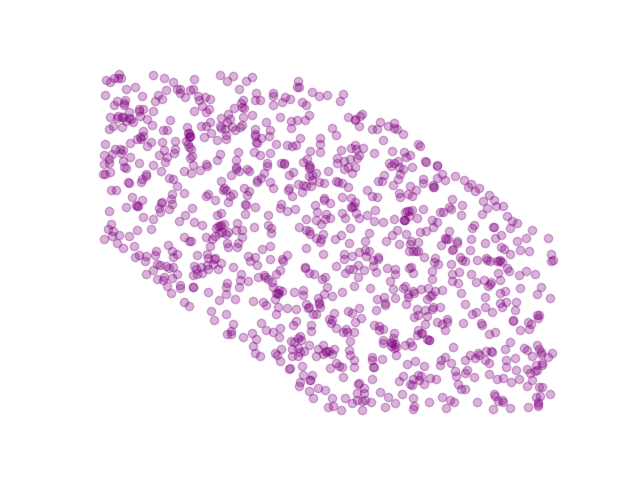}
        \caption{PCC = $-0.5$}
    \end{subfigure}
        \begin{subfigure}[b]{0.19\textwidth}
        \centering
        \includegraphics[width=3cm]{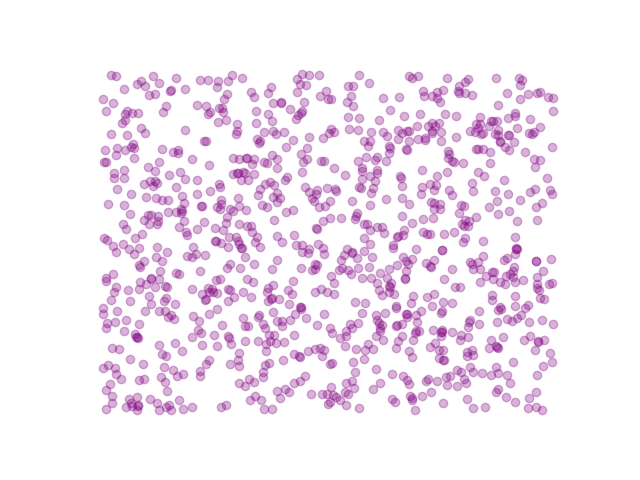}
        \caption{PCC = $0$}
    \end{subfigure}
        \begin{subfigure}[b]{0.19\textwidth}
        \centering
        \includegraphics[width=3cm]{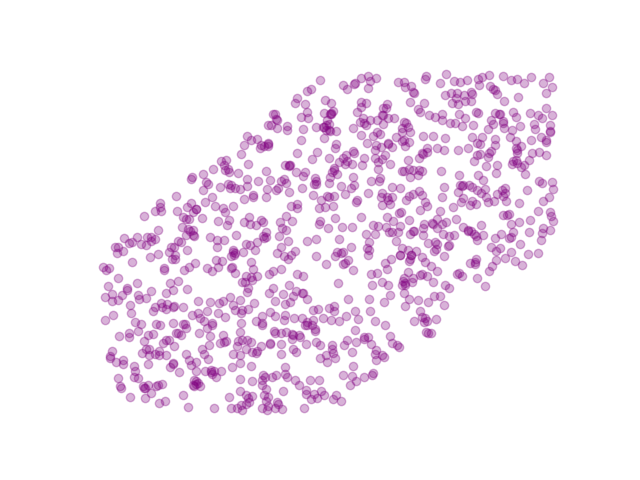}
        \caption{PCC = $0.5$}
    \end{subfigure}
        \begin{subfigure}[b]{0.19\textwidth}
        \centering
        \includegraphics[width=3cm]{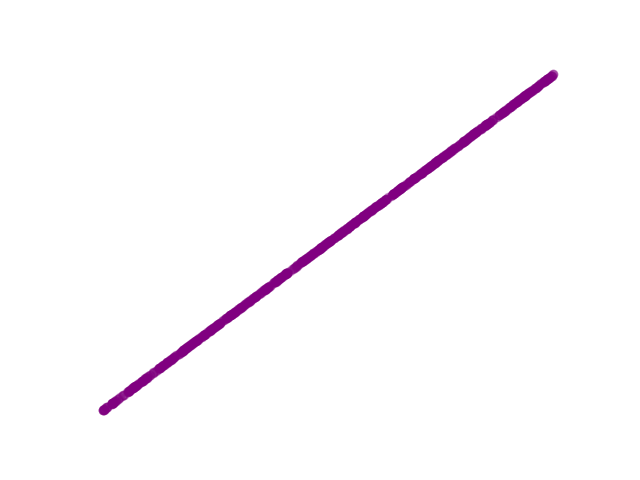}
        \caption{PCC = $1$}
    \end{subfigure}
    
    \caption{Examples of PCC values. In each picture we visualize two vectors $(x_1,\dots,x_n)$ and $(y_1,\dots,y_n)$ as points $(p_1,\dots,p_n)$, where each point $p_i$ has coordinates $x_i, y_i$.}
    \label{fig:pcc}
\end{figure}


\section{Integer Linear Programming}

Integer Linear Programming is a way of solving optimization problems, where the original problem is being represented as a linear objective function and a set of constraints expressed as linear inequalities, where all variables are integers. ILP algorithms are particularly useful when the considered problem is NP-hard.


In our experiments, we use two popular (and free under academic license) ILP solvers. One provided by \emph{Gurobi Optimiztion} and the other one provided by \emph{IBM ILOG CPLEX Optimization Studio}.

In some experiments, we will focus on the running time of certain algorithms (e.g., the running time needed to computing the winning committee under given multi-winner voting rule). Whenever we discuss the running time of a particular algorithm, we assume that the computation for a single instance was run with CPLEX on a single thread (Intel(R) Xeon(R) Platinum 8280 CPU @ 2.70GH) of a 448 thread machine with 6TB of RAM, with exception for experiments done in~\ref{ch:subelections}, which were performed on a single thread on Apple MacBook Air with M1 processor and 8 GB RAM.

\section{Embeddings}\label{desc:embed}
Sometimes, given a set of points and their distance matrix (i.e., square matrix containing all pairwise distances between points), we want to embed these points in a low (i.e., two or three) dimensional space. To do this, we can use a wide variety of techniques. We use the following six methods (which we briefly describe below): 
\emph{Principal Component Analysis}\footnote{We use Python implementation from \emph{sklearn.decomposition.PCA} package.}~\citep{minka2000automatic},
\emph{(metric) Multidimensional Scaling}\footnote{We use Python implementation from \emph{sklearn.manifold.MDS} package.}~\citep{kruskal1964multidimensional,de2005modern}, 
\emph{t-Distributed Stochastic Neighbor Embedding}\footnote{We use Python implementation from 
\emph{sklearn.manifold.TNSE} package.}~\citep{van2008visualizing,van2010fast}, 
\emph{Locally Linear Embedding}\footnote{We use Python implementation from \emph{sklearn.manifold.LocallyLinearEmbedding} package.}~\citep{donoho2003hessian,zhang2006mlle}, 
\emph{Fruchterman-Reingold}\footnote{We use Python implementation from \emph{networkx.spring\_layout} package.}~\citep{fruchterman1991graph}, 
and \emph{Kamada-Kawai}\footnote{We use Python implementation from \emph{mapel.core} package}~\citep{kamada1989algorithm,mt:sapala}:

\begin{description}
\item [Principal Component Analysis (PCA)] is a linear dimensionality reduction algorithm that aims at extracting crucial information from a high-dimensional space. It is based on eigenvalues and eigenvectors of the distance matrix. 
\item [(metric) Multidimensional Scaling (MDS)] unlike PCA, is a nonlinear dimensionality reduction method. It mainly focuses on maintaining the original distances, by minimizing the stress function, where the stress function is the square root of the normalized squared misrepresentations (i.e., differences between original distances and Euclidean distances after the embedding).
\item [t-Distributed Stochastic Neighbor Embedding (t-SNE)] is a statistical method based on Kullback–Leibler divergence. It may not properly preserve densities or distances. Like MDS, it is nonlinear.
\item [Locally Linear Embedding (LLE)] is also a nonlinear dimensionality reduction method. While embedding the points, instead of trying to maintain properly the distances between all of them, it is only focusing on maintaining the distances between points that are close (i.e., the original distance between them is small) to each other.
\item [Fruchterman-Reingold (FR)] is a force-directed graph drawing algorithm (i.e., it aims at drawing the graphs in a pleasant and appealing way). It works in analogy to physical springs as edges between points. It uses both attracting and repulsing forces between points. The aim of the method is to draw an appealing graph. It tries to distribute all the points more or less evenly across the given space.
\item [Kamada-Kawai (KK){\normalfont,}] like FR, is a force-directed graph drawing algorithm. The difference between FR and KK is that FR focuses more on producing a~pleasant picture, while KK focuses more on maintaining proper distances. We use a variant of KK proposed by \cite{mt:sapala} that lowers the probability of the result being stuck in the local minima.
\end{description}
One of the main disadvantages of force-directed algorithms is that they are slower compared to the other methods described above. We provide a detailed comparison of the embedding algorithms in~\Cref{ch:applications:sec:embedding}.

\section{Map of Objects}
Given a set of objects (for example, a set of elections or a set of votes), by a \emph{map} of these objects, we refer to a two-dimensional graphical representation of that set. To create such a map, we first compute distances between each pair of objects, and then, based on these distances, we create a two-dimensional embedding, where each point depicts a single object. We expect similar objects to be embedded close to one another.
    Note that objects might be located in a high-dimensional non-Euclidean space, hence, it will not always be possible to embed properly all the points, i.e., maintain all the distances.
    
We briefly discuss a toy example of such a map. Let us assume that we have five items named~$a, b, c, d,$ and~$e$, and the distance matrix shown in \Cref{fig:toy}a. Looking at the matrix, we expect~$a, b,$ and~$c$ to be located relatively close to each other and to form more or less a triangle, and we expect~$d$ and~$e$ to be located even closer to each other (because the distance between them is the smallest one in the whole matrix). Moreover, we expect that the~$abc$ triangle would be rotated in such a way that~$c$ would be pointing towards~$de$. That is exactly what we observe in the embedding presented in \Cref{fig:toy}b, and that is our map.

\begin{figure}[t]

\begin{subfigure}[b]{0.49\columnwidth}
 \centering

        \begin{tabular}{c|ccccc}
          \toprule
               &~$a$ &~$b$ &~$c$ &~$d$ &~$e$  \\
          \midrule
         ~$a$  &~$-$ &~$2$ &~$2$ &~$4$ &~$4$    \\
         ~$b$  &~$2$ &~$-$ &~$2$ &~$4$ &~$4$    \\
         ~$c$  &~$2$ &~$2$ &~$-$ &~$3$ &~$3$    \\
         ~$d$  &~$4$ &~$4$ &~$3$ &~$-$ &~$1$    \\
         ~$e$  &~$4$ &~$4$ &~$3$ &~$1$ &~$-$    \\
          \bottomrule
        \end{tabular}
        \vspace{0.45cm}
    \caption{Distance Matrix}
    \end{subfigure}
    \begin{subfigure}[b]{0.49\columnwidth}
            \centering
            \includegraphics[width=4.5cm]{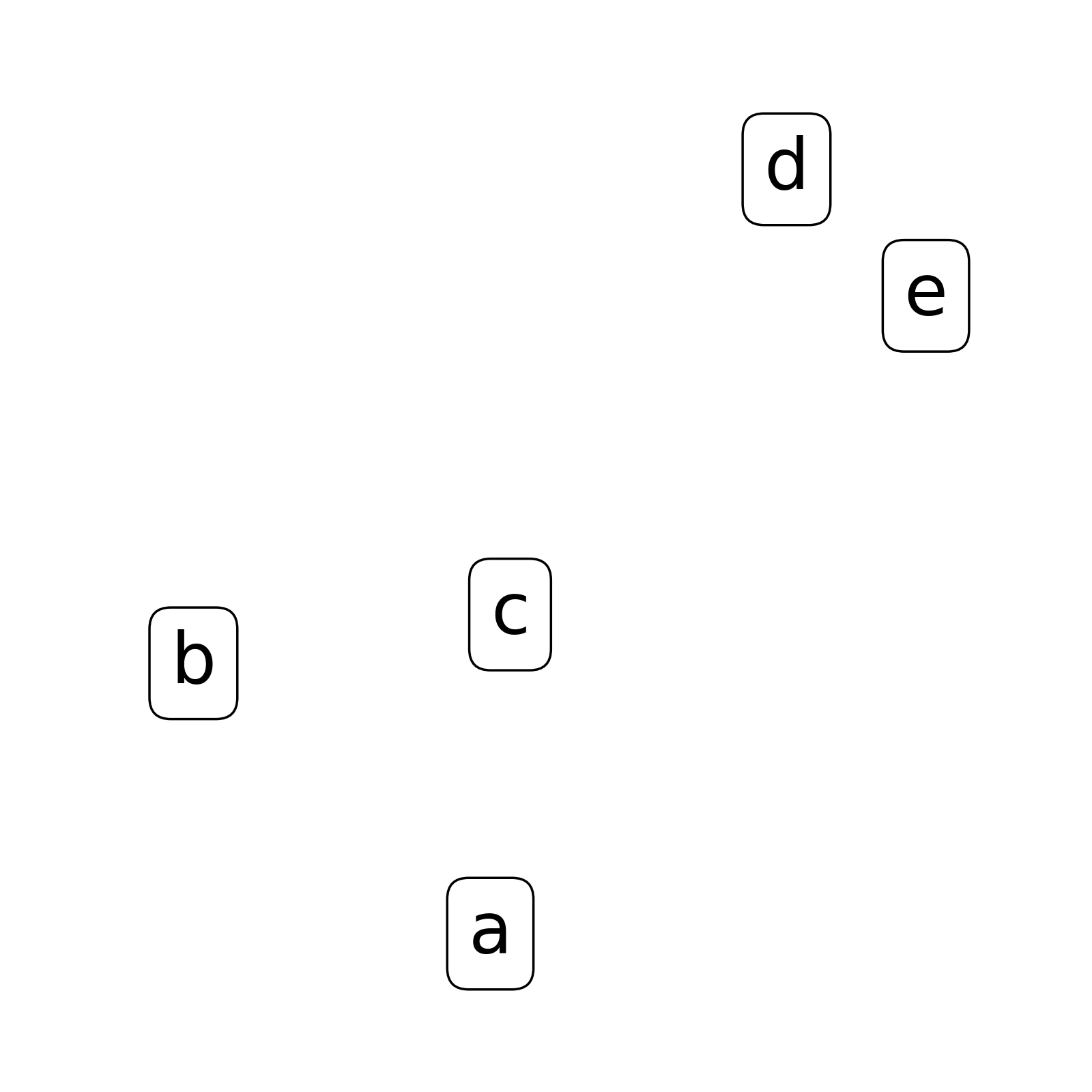}
          
    \caption{Map}        
    \end{subfigure}
    
      \caption{Toy example with matrix distances (left) and the map (right).}
\label{fig:toy}

\end{figure}

\chapter{Statistical Cultures}
\label{ch:stat_cult}

In this chapter we describe the statistical cultures that we use for generating instances of elections. First, we describe general models, and then we move on to structured domains. At the end of the chapter, we present maps of preferences, a simple yet interesting visualization of instances generated from various models described below. Maps of preferences help in understanding the structure of the votes and show (dis)similarities between different models.

\section{General Models}

Below, we define the most popular general\footnote{By general we mean that any election can be sampled from such statistical culture.} statistical cultures that we will use in this dissertation. When discussing elections, by $m$ we denote the number of candidates, and by $n$ we denote the number of voters.

\subsubsection{Impartial Culture and Related Models}
Under the impartial culture (IC) model, every preference order appears
with the same probability. That is, to generate a vote, we choose a
preference order uniformly at random.

Under the impartial anonymous culture (IAC) model, we require that each
\emph{voting situation} appears with the same probability~\citep{iac1,iac2}. A voting
situation specifies how many votes with a given preference order are
present in a profile; thus, IAC generates anonymized
preference profiles uniformly at random.

The impartial anonymous neutral culture (IANC) additionally abstracts away from the names of the
candidates~\citep{ege-gir:j:isomorphism-ianc}. This means that for a given numbers of candidates and voters, the number of different IANC elections is equal to the number of equivalence classes under any isomorphic distance, such as, for example, the swap distance.

\subsubsection{Pólya-Eggenberger Urn Model} The Pólya-Eggenberger
urn model \citep{berg1985paradox,mcc-sli:j:similarity-rules} is parametrized with a nonnegative number~$\alpha$, the level of contagion, and
proceeds as follows:
  Initially, we have an urn with one copy of each
of the~$m!$ possible preference orders. To generate a vote, we draw a
preference order from the urn uniformly at random (this is
the generated vote) and return it to the urn together with
additional~$\alpha m!$ copies. The larger~$\alpha$ is,
the more correlated are the generated votes.
For~$\alpha = 0$, the model is equivalent to IC, for~$\alpha = \nicefrac{1}{m!}$, 
it is equivalent to IAC, and for~$\alpha = \infty$
all votes are identical.
%

\subsubsection{Mallows Model} The Mallows
model \citep{mal:j:mallows} is parameterized by a \emph{ dispersion parameter}~$\phi \in
[0,1]$ and a center preference order~$v$ (we choose it uniformly at
random and then use it for all generated votes). We generate each vote
independently at random, where the probability of generating vote~$u$
is proportional to~$\phi^{d_\swap(v,u)}$. For~$\phi = 1$, the model is
equivalent to IC, while for~$\phi = 0$ all generated
votes are identical to the center vote~$v$.  See the work of Lu and
Boutilier for an effective algorithm for sampling from the Mallows
model~\citep{lu-bou:j:sampling-mallows}.

In our experiments, we consider a new parameterization
introduced by \cite{boe-bre-fal-nie-szu:c:compass}. It uses a \emph{normalized dispersion parameter}~$\normphi$, which is converted to a value of~$\phi$ so that the expected swap distance
between the central vote~$v^*$ and a sampled vote~$v$ is~$\frac{\normphi}{2}$ 
times the maximum swap distance between two votes. We refer to Mallows model with normalized dispersion parameter as Normalized Mallows (Norm-Mallows) model.


Besides the basic Norm-Mallows model, we consider a combination of two Norm-Mallows models, where~$\omega \in (0, 0.5]$ fraction of the votes are reversed, i.e., after sampling all the votes from the basic Norm-Mallows model, we reverse the first
$\lfloor \omega n \rfloor$ of them\footnote{It is equivalent to sampling $\lfloor \omega n \rfloor$ votes from the basic Norm-Mallows model, and then sampling the rest of the votes from the same Norm-Mallows model but with a reversed central ballot.}.
We refer to this variant as the weighted Norm-Mallows model. 

\section{Structured Domains}

In this section we focus on structured domains. We describe several properties of elections such as single-peakedness, single-crossingness, and group-separability, and discuss how to sample elections having such properties. Moreover, we study the Euclidean-based models.

\subsection{Single-Peaked Elections}

\emph{Single-peaked} preferences, introduced by~\cite{bla:b:polsci:committees-elections}, capture settings where it is possible to order the candidates in such a way
that as we move along this order, each voter's appreciation of the candidates first increases and then decreases. One typical
example of such an order is the classic left-to-right spectrum of
political opinions.



\begin{definition}
  Let~$v$ be a vote over~$C$ and let~$\lhd$ be the societal axis over~$C$. 
  We say that~$v$ is {\em single-peaked with respect to~$\lhd$}
  if for every~$t \in [|C|]$ its~$t$ top-ranked candidates form an
  interval within~$\lhd$.  An election is {\em single-peaked with
    respect to~$\lhd$} if all its votes are. An election is {\em
    single-peaked (SP)} if it is single-peaked with respect to some
  axis.
\end{definition}





\begin{example}
  Consider an election with the set of candidates~$C = \{a,b,c,d,e\}$ and votes:
 \begin{align*}
  v_1& \colon a \pref b \pref c \pref d \pref e, \\
  v_2& \colon e \pref d \pref c \pref b \pref a, \\
  v_3& \colon b \pref c \pref a \pref d \pref e. \\
\end{align*}
  This election is single-peaking with respect to axis~$a,b,c,d,e$. Moreover, it is a unique axis with respect to which this election is single-peaked.
\end{example}

We also consider the {\em single-peaked on a circle domain (SPOC)},
introduced by \cite{pet-lac:j:spoc}.  A vote is {\em SPOC with respect to
an axis~$c_1 \lhd \cdots \lhd c_m$} if it is single-peaked with respect
to some axis of the form:
\[ c_i \lhd c_{i+1} \lhd \cdots \lhd c_m \lhd c_1 \lhd \cdots \lhd
  c_{i-1}; \] 
  An election is SPOC with respect to an axis if all of its votes are (the value of~$i$ may differ from one vote to another). SPOC votes may capture, for example, preferences regarding
meeting times when people are in different time zones.

\subsubsection{Sampling}
We consider two ways of generating single-peaked elections, one
studied by~\cite{wal:t:generate-sp} and one studied by~\cite{con:j:eliciting-singlepeaked};
hence, we refer to them as the \emph{Walsh model} and the \emph{Conitzer
  model}. 
In both models, we first choose the axis (uniformly at
random). To generate a vote, we proceed as follows:
\begin{enumerate}
\item Under the Walsh model, we choose a single-peaked preference
  order (under the given axis) uniformly at random.
  \cite{wal:t:generate-sp} provided a sampling algorithm for
  this task. This model is also sometimes referred to as \emph{impartial culture over single-peaked votes}.

\item To generate a vote under the Conitzer model for the axis~$c_1
  \lhd c_2 \lhd \cdots \lhd c_m$, we first choose some candidate~$c_i$
  (uniformly at random) to be ranked on top (so, at this point,~$c_i$
  is the only ranked candidate). Then, we perform~$m-1$ steps
  as follows: Let~$\{c_j, c_{j+1}, \ldots,
  c_{k}\}$ be the set of the currently ranked candidates. We choose
  the next-ranked candidate from the set~$\{c_{j-1}, c_{k+1}\}$
  uniformly at random.
  This model is also sometimes referred to as the \emph{random peak} model.
\end{enumerate}
To generate a single-peaked on a circle vote, we use the Conitzer
model, except that we take into account that the axis is cyclical (note that this process generates each possible SPOC vote with equal probability, so, in fact, we can say that we use impartial culture over SPOC votes).
\smallskip

\subsection{Single-Crossing Elections}
We also consider \emph{single-crossing} elections, introduced by \cite{mir:j:single-crossing} and \cite{rob:j:tax} 
in the context of taxation.  
\begin{definition}[\cite{mir:j:single-crossing}, \cite{rob:j:tax}]
  An election~$E = (C,V)$ is single crossing if it is possible to
  order the voters in such a way that for each pair of candidates~$a,b
  \in C$, the set of voters that prefer~$a$ to~$b$ either forms a
  prefix or a suffix of this order.
\end{definition}

\newpage
\begin{example}
  Consider election with the set of candidates~$C = \{a,b,c,d\}$, and votes:
 \begin{align*}
  v_1& \colon a \pref b \pref c \pref d, \\
  v_2& \colon a \pref b \pref d \pref c, \\
  v_3& \colon d \pref a \pref c \pref b, \\
  v_4& \colon d \pref c \pref b \pref a.
\end{align*}
  This election is single-crossing because each pair of candidates is crossing at most once. 
  In particular, pair~$\{c,d\}$ is crossing between votes~$v_1$ and~$v_2$, pairs~$\{a,d\}, \{b,c\}$, and~$\{b,d\}$ are crossing between votes~$v_2$ and~$v_3$, and pairs~$\{a,b\}$ and~$\{a,c\}$ are crossing between votes~$v_3$ and~$v_4$. 
\end{example}


We say that a set of preference orders~$\calD$ is a \emph{single-crossing
  domain} if every election where each voter has a preference order
from~$\calD$ is single-crossing.
For a recent discussion 
of single-crossing domains, see, e.g., 
the work of~\cite{pup-sli:j:single-crossing}.


\subsubsection{Sampling}
We would like to generate single-crossing elections uniformly at
random, 
but we are not aware of an efficient sampling algorithm for this task.
Thus, to generate a single-crossing election, we first generate a
single-crossing domain~$\calD$ and then draw~$n$ votes from it
uniformly at random. To generate this 
domain for a
candidate set~$C = \{c_1, \ldots, c_m\}$, we use the following
procedure:
\begin{enumerate}
\item We let~$v$ be a preference order~$c_1 \pref c_2 \pref \cdots
  \pref c_m$ and we output~$v$ as the first member of our domain.
\item We repeat the following steps until we output~$c_m \pref c_{m-1}
  \pref \cdots \pref c_1$: 
  \begin{enumerate}
      \item  We draw candidate~$c_j$ uniformly at
  random and we let~$c_i$ be the candidate ranked right ahead of~$c_i$
  in~$v$ (if~$c_j$ is ranked on top, then we repeat);
    \item If~$i < j$
  then we swap~$c_i$ and~$c_j$ in~$v$ and output the new preference
  order.
  \end{enumerate}
\item We randomly permute the names of the candidates.
\end{enumerate}
Our 
domains always have cardinality~$(\nicefrac{1}{2})m(m-1)+1$.

\subsection{Group-Separable Elections}
Next, we consider \emph{group-separable} elections, introduced by
\cite{ina:j:group-separable, ina:j:simple-majority}. 
An election is group-separable if each set~$A$ of at least two candidates can be
partitioned into two nonempty subsets,~$A'$ and~$A''$, such that each
voter either prefers all members of~$A'$ to all members of~$A''$ or
the other way round.
For our purposes, it will be convenient to use the
 tree-based definition of \cite{kar:j:group-separable} which is equivalent to the previous one.  Let~$C = \{c_1, \ldots, c_m\}$ 
 be a set of candidates and consider a
rooted, ordered tree~$\calT$ whose leaves are elements of~$C$. The
\emph{frontier} of this tree is the preference order that ranks the
candidates in the order in which they appear in the tree from left to
right. A preference order is \emph{consistent} with a given tree if
it can be obtained as its frontier by reversing the order in which the
children of some nodes appear.

\begin{definition}\label{def:gs}
  An election~$E = (C,V)$ is \emph{group-separable} if there is a
  rooted, ordered tree~$\calT$ whose leaves are members of~$C$, such
  that each vote in~$V$ is consistent with~$\calT$.
\end{definition}

\noindent The trees from~\Cref{def:gs} form a subclass of
\emph{clone decomposition trees}, which 
are examples of PQ-trees
\citep{elk-fal-sli:c:decloning,boo-lue:j:consecutive-ones-property}.

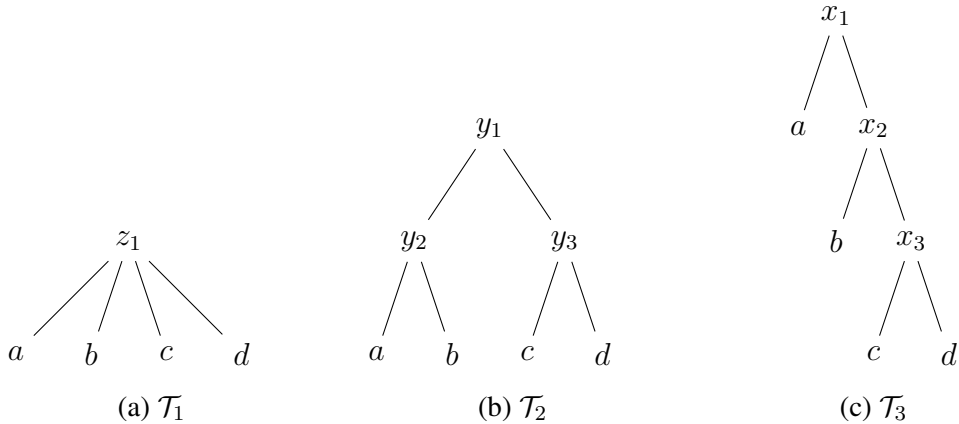
\begin{figure}
  \newcommand{\myscale}{1}
  \centering
  
  \begin{subfigure}[b]{0.3\columnwidth}
    \begin{tikzpicture}[scale=\myscale]
    \node{$z_1$}[sibling distance = 1cm]
        child {node {$a$}}
        child {node {$b$}}
        child {node {$c$}}
        child {node {$d$}};
        \end{tikzpicture}
    \caption{$\calT_1$}
  \end{subfigure}\hfill
  \begin{subfigure}[b]{0.3\columnwidth}
    \begin{tikzpicture}[scale=\myscale]
      \node{$y_1$}[sibling distance = 2cm]
        child {node {$y_2$} [sibling distance = 1cm]
            child {node {$a$}}
            child {node {$b$}}
            }
        child {node {$y_3$} [sibling distance = 1cm]
            child {node {$c$}}
            child {node {$d$}}
            };
    \end{tikzpicture}
    \caption{$\calT_2$}
  \end{subfigure}\hfill
  \begin{subfigure}[b]{0.3\columnwidth}
    \centering
  \begin{tikzpicture}[scale=\myscale]
     \node{$x_1$}[sibling distance = 1cm]
        child {node {$a$}}
        child {node {$x_2$}
            child {node {$b$}}
            child {node {$x_3$}
                child {node {$c$}}
                child {node {$d$}}
                }
            };
      \end{tikzpicture}
    \caption{$\calT_3$}
  \end{subfigure}
  
  \caption{\label{fig:trees}Three examples of clone decomposition trees.}
\end{figure}

\newpage
\begin{example}
  Consider the set of candidates~$C = \{a, b,c,d\}$, trees~$\calT_1$,~$\calT_2$, 
  and~$\calT_3$ from~\Cref{fig:trees}, and votes:
 \begin{align*}
  v_1& \colon a \pref b \pref c \pref d, \\
  v_2& \colon c \pref d \pref b \pref a, \\
  v_3& \colon b \pref d \pref c \pref a.
\end{align*}
  Vote~$v_1$ is consistent
  with each of the trees,~$v_2$ is consistent with~$\calT_2$ (reverse the children of~$y_1$
  and~$y_2$) and with~$\calT_3$ (reverse the children of~$x_1$
  and~$x_2$) , and~$v_3$ is consistent with~$\calT_3$ (reverse the
  children of~$x_1$ and~$x_3$).
\end{example}

In many cases, we will be interested in two particularly characteristic trees, i.e., balanced and caterpillar ones. \emph{Balanced tree} is a complete, full binary tree (if the number of candidates/leaves is equal to a power of two, then this is a perfect tree). \emph{Caterpillar tree} is a binary tree where each inner node's left child is a leaf, and the right child is either an inner node or a leaf.
$\calT_2$ from~\Cref{fig:trees} is an example of a balanced tree, whereas~$\calT_3$ is an example of a caterpillar tree.

\subsubsection{Sampling}
Given a certain tree, to generate a vote, we simply reverse each internal node with probability~$0.5$ and then take the frontier as our vote. We repeat this procedure independently to generate as many votes as required in the election.

\subsection{Euclidean Elections}

Finally, Euclidean preferences, discussed in detail, e.g., by \cite{enelow1984spatial,enelow1990advances}, are based on a
similar idea as the single-peaked ones, but are defined geometrically:
Each candidate and each voter corresponds to a point in a Euclidean
space and voters form their preferences by ranking the candidates with
respect to their distance. That is, if the point of voter~$v$ is closer
to the point of candidate~$c$ than to that of candidate~$d$ then~$v$ prefers~$c$ to~$d$.

\begin{definition}\label{def:t-euclidean}
  Let~$t$ be a positive integer.  An election~$E = (C,V)$ is~$t$-Euclidean if it is possible to associate each candidate and each
  voter with his or her ideal point in a~$t$-dimensional Euclidean space~$\mathbb{R}^t$ in such a way that the following holds: For each
  voter~$v$ and each two candidates~$a, b \in C$,~$v$ prefers~$a$ to~$b$ if and only if~$v$'s point is closer to the point of~$a$ than to
  the point of~$b$.
\end{definition}

\subsubsection{Sampling}
To generate the Euclidean election, we simply sample ideal points of candidates and voters from a given space and then, based on these ideal points, we create the votes. Given a certain space, we sample from it uniformly at random. In particular, we consider the following models:
\begin{itemize}
    
    \item \emph{Interval} -- points are sampled uniformly at random from a $1$-dimensional interval.
    \item \emph{Disc} -- points are sampled uniformly at random from a $2$-dimensional disc.
    \item \emph{Square} -- points are sampled uniformly at random from a $2$-dimensional square.
    \item \emph{Cube} -- points are sampled uniformly at random from a $3$-dimensional cube.
    \item \emph{$n$-Cube} -- points are sampled uniformly at random from an $n$-dimensional hyper cube.
    \item \emph{Circle} -- points are sampled uniformly at random from a circle.
    \item \emph{Sphere} -- points are sampled uniformly at random from an ordinary sphere in a $3$-dimensional Euclidean space.
    \item \emph{$n$-Sphere} -- points are sampled uniformly at random from an $n$-sphere in an $(n+1)$-dimensional Euclidean space.

\end{itemize}

It is well known that Interval elections are both 
single-peaked and single-crossing. We also note that in Circle elections, the voters have SPOC
preferences.


\section{Compass Elections}
Next, we provide four characteristic elections, to which we refer as \emph{compass elections}. We believe that they capture some notions of ``extremes'' and are qualitatively different from each other. These four compass points are as follows.
\begin{description}
\item[Identity.] In the identity elections, denoted~$\ID$, all voters have the
  same, fixed preference order---which we sample uniformly at random.
\item[Antagonism.] In the antagonism elections, denoted~$\AN$, half of the voters
  rank the candidates in one way and half of the voters rank them in
  the opposite way.
\item[Uniformity.] In the uniformity elections, denoted~$\UN$, each possible vote
  appears the same number of times.
\item[Stratification.] In the stratification elections, denoted~$\ST$, the candidates
  are partitioned into two equal-sized sets~$A$ and~$B$.  Each
  possible preference order where all members of~$A$ are ranked ahead
  of~$B$ appears the same number of times.
\end{description}

In practice, to generate the identity election, we sample one vote uniformly at random, and all votes are its copies. To generate the antagonism election, we sample one vote uniformly at random, and half of the votes are its copies, while the other half are copies of the reverse vote. To get ideal uniformity and stratification, we would need exponentially many votes (i.e., with respect to the number of candidates), so due to limited
number of votes, for uniformity we just sample an election from impartial culture---as an approximation of the uniformity, and for stratification, to generate a vote we sample the first half of the vote from impartial culture (based on the first half of the candidates), and then we sample the second half of the vote also from impartial culture (but based on the second half of the candidates).



  


\section{Map of Preferences}\label{ordinal_map_pref}
To get a better understanding of our statistical cultures, in this section we present a {\it map of preferences}\footnote{In principle, the map of preferences is very similar to the map of elections, where each point on the map, instead of depicting a single election, is depicting a single vote. Historically, we introduced the maps of elections prior to the maps of preferences}. For a given election, to generate its map of preferences, we proceed as follows. First, we compute the swap distance between each pair of votes. Then, based on these distances, we create a two-dimensional embedding using the MDS algorithm (see \Cref{desc:embed}). Each dot corresponds to a single vote. The closer two dots are on the map, the more similar are the votes that they represent (or, more precisely, the smaller is their swap distance).

We generated elections with~$10$ candidates and~$1000$ voters from~$25$ different models\footnote{We take at most one election from a given model, the only exceptions are parametrized models such as the Norm-Mallows and urn models, from which we take several elections with different parameters.}, described before. These models include impartial culture, urn model with~$\alpha \in \{0.05, 0.2, 1\}$, Walsh and Conitzer models, SPOC, single-crossing model, balanced and caterpillar group-separable models, Interval, Square, Cube, 10-dimensional Hypercube, Circle, and Sphere Euclidean models, and Norm-Mallows model with~$\phi \in \{0.05, 0.2, 0.5\}$, and~$\omega \in \{0, 0.25, 0.5\}$ (with each possible combination of~$\phi$ and~$\omega$). Moreover, we added three compass elections, i.e., $\ID$, $\AN$, and $\ST$; we skipped $\UN$ because the result is almost identical to the IC map. The results are presented in~\Cref{fig:microscope}. For clarity, if there are more than 30 copies of the same vote, we denote it by adding a purple disc---the larger the disc, the more copies there are. 

\begin{figure}
    \centering
    \includegraphics[width=14cm]{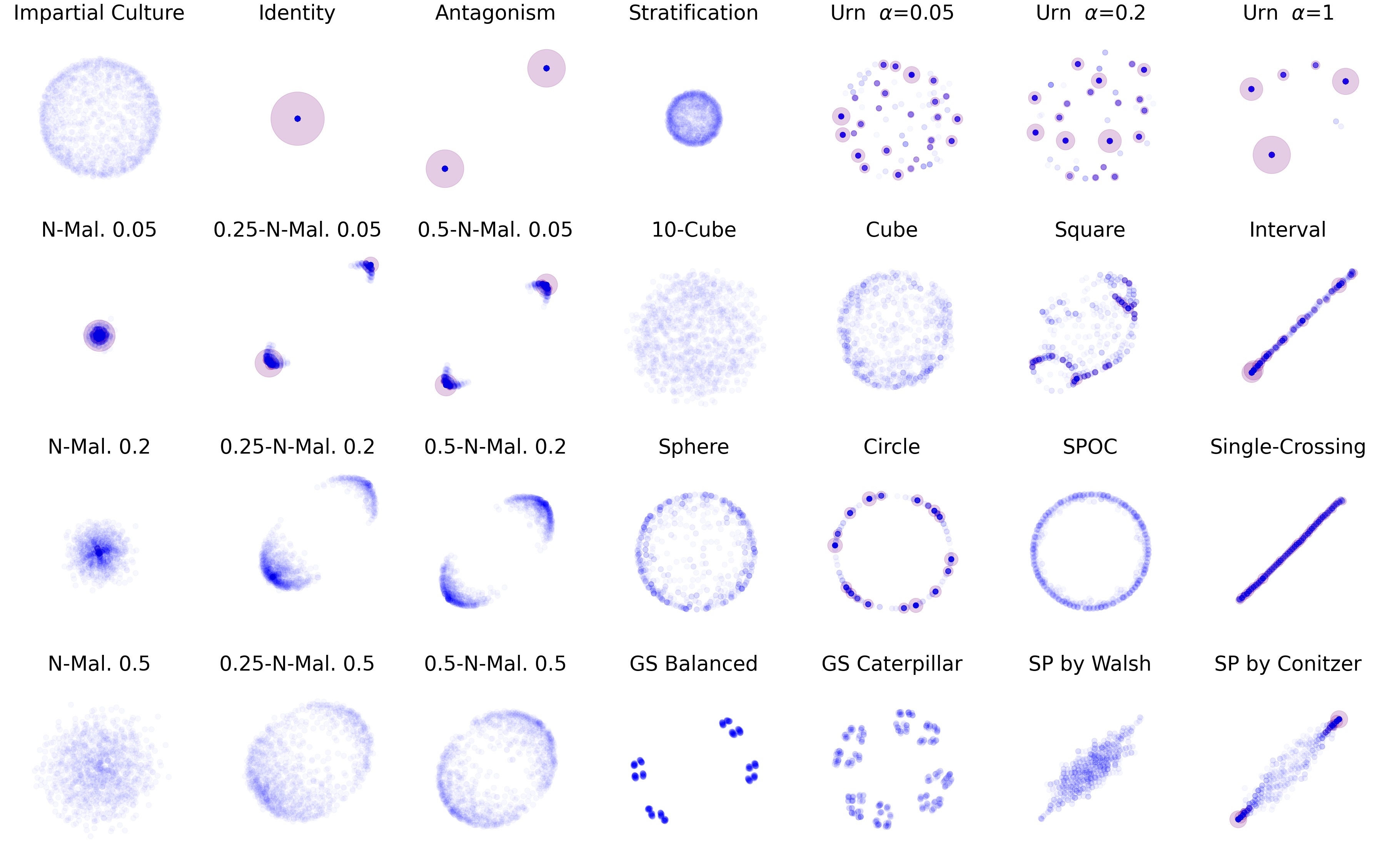}
    \caption{Maps of Preferences ($10$ candidates,~$1000$ voters).}
    \label{fig:microscope}
\end{figure}

We start our analysis by looking at the impartial culture election. Votes are more or less uniformly spread, with slightly higher density near the edge. In multidimensional space the votes would form a permutohedron; however, here we are limited to an embedding in two-dimensional space, so proportionally more votes land on the edge. 

Then, we have $\ID$ followed by $\AN$ and $\ST$. As expected, for $\ID$ we have a single point in the center because all votes are identical, and for $\AN$ we have two points located at the largest possible distance because we have only two types of votes (i.e., $500$ times vote $v$, and $500$ times its reversed copy). For $\ST$, we observe a similar picture to the one for the IC election, however, the diameter is much smaller. This is because in an $\ST$ election all the voters agree that half of the candidates are better than the other half, hence, the largest possible distance between two votes is equal to half of the largest possible distance between two votes from IC.

Next, we have three elections from the urn model. The larger is the~$\alpha$ parameter, the smaller is the number of different votes, leading to fewer points on the map. To be more precise, below we provide the formula for the
(upper bound on the) expected number of different votes under the urn model, with assumption that~$n \leq m!$.

\begin{proposition}
Given parameter of contagion~$\alpha$ and number of voters~$n$ the expected number of different votes under the urn model is upper-bounded by~$\sum_{i=1}^{n} \frac{1}{1+(i-1)\alpha}$.
\end{proposition}

\begin{proof}
The probability of having a new vote in the first iteration is $1$; in the second iteration it is at most  $\frac{1}{1+\alpha}$; in the third iteration it is at most $\frac{1}{1+2\alpha}$ and so on. In general, in the $i$th iteration we have at most probability $\frac{1}{1+(i-1)\alpha}$ of sampling a vote from the original urn, and probability $\frac{(i-1)\alpha}{1+(i-1)\alpha}$ of repeating one of the previous votes.
Therefore, the expected number of different votes in $n$ iterations is upper-bounded by~$\sum_{i=1}^{n} \frac{1}{1+(i-1)\alpha}$. It is an upper-bound because we ignore the case where while sampling a vote from the original urn, we sample a vote that we have already sampled before.
\end{proof}

For the normalized Mallows model, the shorthand captions in the pictures are of the form~$\omega$-$N$-$Mal.$~$\phi$. For standard Normalized Mallows, as expected, we have a central point (corresponding to the central order) and the further away we move from that point, the fewer votes we have. On the other hand, for the weighted variant with~$\omega \in \{0.25, 0.5\}$ we observe two antagonistic groups. The central ranking and its reverse are at the largest possible distance. Any noise on one of them is shifting a given vote closer to the other group.

Next, we move on to structured domains. We start with the single-crossing model. The map for the single-crossing is one straight line. It is because the single-crossing domain is defined by a sequence of swaps, so for each vote the sum of its distances to the two most extreme votes is constant. Moreover, there cannot be two different votes that are at the same distances from the extremes, because it would contradict the fact that the domain is defined by a sequence of swaps. The map for the Interval model looks very similar. Note that every election from the Interval model is also single-crossing. Interestingly, the votes from the Interval election look less evenly distributed than those from the single-crossing election. It is so, because in the Interval election candidates' points are sampled randomly, so since there are only ten of them, by chance they can be distributed unevenly over the interval, which leads to an uneven distribution of preference orders. For the single-crossing model such a thing cannot occur.

When we shift from the Interval model to the Square, Cube, and finally the 10-dimensional Hypercube ones, the maps become gradually more and more similar to that for impartial culture. The same is true when we shift from Circle to Sphere, etc. However, hypersphere elections converge faster toward impartial culture model than hypercube ones. For example, even 10-dimensional Hypercube is still something in between~$N$-$Mal.$~$0.5$ and impartial culture.

In a single-peaked election there are two possible \emph{extreme votes}, i.e., one identical with the societal axis, and the second one, identical with the reversed societal axis. For single-peaked models, we observe an interesting difference between the Walsh and Conitzer approaches. For the Walsh model, the points are more uniformly spread, while for the Conitzer model, we obtain somewhat antagonistic single-peaked elections. In fact, for Conitzer model, the probability of sampling an extreme vote is~$2 \frac{1}{m}$ (for~$m>1$). Therefore, for the presented example~$\sim200$ votes will be extreme ones ($\sim100$ per each extreme). On the other hand, for the Walsh model, the probability of sampling an extreme vote is~$2^{-(m-2)}$ (for~$m>1$), so for the presented example~$\sim0.4$ vote will be an extreme one.

Although, the voters in an election from the Circle model have SPOC preferences, votes from the Circle model are less evenly distributed than those from the SPOC model. This is a similar case to that of Interval and single-crossing elections.
    
For the balanced and caterpillar group-separable models, we see the divisions of points into subgroups, which corresponds to the inner nodes of the trees. For the group-separable caterpillar variant, they are spread across a larger space than for the group-separable balanced variant.
    

\section{Summary}
In this chapter, we introduced some of the most popular statistical cultures that are used in experiments in computational social choice. Next, we described four compass elections: identity, uniformity, antagonism, and stratification. Finally, using the {\it map of preferences} framework, we gave the reader the intuition about how elections from different models look like.


\chapter{Distances Among Elections}
\label{ch:distances}

\section{Introduction}
How similar are two elections? In this chapter we suggest how one can go about answering this question.
We introduce the \textsc{Election Isomorphism} problem and a family of
its approximate variants, which measure the degree of similarity
between two elections by using distances over preference orders.


In the \textsc{Election Isomorphism} problem we are given two elections,
$E_1$ and~$E_2$, both with the same numbers of candidates and the same
numbers of voters, and we ask if it is possible to transform one into the other by renaming the candidates and reordering the voters.
While this problem is similar in spirit to the famous \textsc{Graph
  Isomorphism} problem (whose complexity status remains elusive; see
the report of~\cite{bab-daw-sch-tor:j:graphi-isomorphism} and further
discussion on Babai's home page for recent progress on the
problem), the structure of elections with ordinal ballots is such that
it is very easy to provide a polynomial-time algorithm for
\textsc{Election Isomorphism}.  On the other hand, for approval-based
elections, 
\textsc{Election Isomorphism} is at least as hard as \textsc{Graph Isomorphism}---a graph can be
  encoded as an approval election in a simple way. However, more details about the approval-based elections will be given in \Cref{ch:approval}.

We are also interested in approximate variants of the
\textsc{Election Isomorphism} problem, which turn out to define
distances over elections. 
We extend the distance between preference orders to 
whole elections in a way that respects both anonymity and neutrality.
Namely, we ask if, via appropriate
renaming of the candidates and reordering the voters, it is 
possible to bring a given election within some small 
distance of another given one. 
  We note that approximate \textsc{Graph Isomorphism} problems are also
studied in the
literature \citep{arv-koe-kuh-vas:c:approximate-graph-isomorphism,gro-rat-woe:c:approximate-isomorphism}. Although, in
spirit, they are very similar to our problems, they differ on the
technical level.

We focus on three isomorphic distances (i.e., distances under which only isomorphic elections are at distances zero), that is, the swap, Spearman, and discrete distances. Unfortunately, both the swap and Spearman distances are quite complex and take a lot of time to compute even for relatively small instances of elections. On the other hand, the discrete distance is faster, yet not very informative. So, in one way or another, all three distances are of limited practical value when comparing elections with, for example,~$100$ candidates and~$100$ voters. This conclusion leads to the development of various ``nonisomorphic'' distances. We call them nonisomorphic because sometimes, even if two elections are not isomorphic, these distances might return zero. All our nonisomorphic distances instead of operating on complete elections, work on their aggregate representations---compressed forms of elections. It can be seen as a tradeoff, when we accept losing some information about elections in exchange for a better performance with regard to the running time. However, as we will show in \Cref{sec:non_iso_dist}, not for all nonisomorphic distances this tradeoff pays off.

The structure of this chapter is as follows. First, we focus on \textsc{Election Isomorphism}, and isomorphic distances. Second, we move to the aggregate representations of elections and nonisomorphic distances based on these representations. In these parts we largely focus on the complexity of computing our distances.
Then, we compare both isomorphic and nonisomorphic distances altogether: We discuss the relation between compass elections (i.e., the four characteristic elections which were initially presented in \Cref{ch:stat_cult}). Finally, we study correlation between distances, numbers of equivalence classes under each of them, and (what is most interesting) we compare the maps that our distances produce.

  
 

\section{Election Isomorphism}
In this section we define the notion of election isomorphism,
illustrate its usefulness, and show that testing if two elections are
isomorphic is a polynomial-time computable task. We start with a
formal definition.

\begin{definition}
  We say that elections~$E=(C,V)$ and~$E'=(C',V')$, where~$|C| = |C'|$,~$V = (v_1, \ldots, v_n)$, and~$V' = (v'_1, \ldots, v'_n)$, are
  isomorphic if there is a bijection~$\sigma\colon C \to C'$ and a
  permutation~$\nu\in S_n$ such that~$\sigma(v_i)=v'_{\nu(i)}$ for all~$i\in [n]$.
\end{definition}


\begin{example}
  Consider elections~$E = (C,V)$ and~$E' = (C',V')$, 
  such that~$C = \{a,b,c\}$,~$C' = \{x,y,z\}$,~$V = (v_1,v_2,v_3)$,~$V' = (v'_1, v'_2,v'_3)$, with the following preference orders:
\begin{align*}
\centering
 &  v_1 \colon a \pref b \pref c, &&  v'_1 \colon y \pref x \pref z,  & \\
 &  v_2 \colon b \pref a \pref c, &&  v'_2 \colon x \pref y \pref z,  & \\
 &  v_3 \colon c \pref a \pref b, &&   v'_3 \colon z \pref x \pref y. & \\
\end{align*}
$E$ and~$E'$ are isomorphic, by mapping candidates~$a$ to~$x$,~$b$ to
$y$, and~$c$ to~$z$, and voters~$v_1$ to~$v'_2$,~$v_2$ to~$v'_1$, and
$v_3$ to~$v'_3$.
\end{example}
The idea of election isomorphism has already appeared in the
literature, though without using this name and usually as a tool to
achieve some specific goal.  For example,~\cite{ege-gir:j:isomorphism-ianc} refer to two
isomorphic elections as members of the same \emph{anonymous and
  neutral equivalence class (ANEC)} and study the problem of sampling
representatives of ANECs uniformly at random.~\cite{has-end:c:diversity-indices} use the election
isomorphism idea in their analysis of preference diversity indices.

In the \textsc{Election Isomorphism} problem we are given two
elections and we ask if they are isomorphic.  Surprisingly, 
the problem has an easy polynomial-time algorithm.


\begin{proposition}\label{thm:ei-complexity}
  \textsc{Election Isomorphism} is in~$\p$.
\end{proposition}
\begin{proof}
  Let~$E = (C,V)$ and~$E' = (C',V')$ be two input elections where
 $C = \{c_1, \ldots, c_m\}$,~$C' = \{c'_1, \ldots, c'_m\}$,
 $V = (v_1, \ldots, v_n)$ and~$V = (v'_1, \ldots, v'_n)$.  Without loss
  of generality, let us assume that~$v_1$'s preference order is
 $v_1 \colon c_1 \pref c_2 \pref \cdots \pref c_m.$ For each 
 $v'_j$ there is a bijection~$\sigma_j$ from~$C$ to~$C'$ such that for
  the preference order of~$v'_j$ we have~$\pos_{v'_j}(\sigma_j(c_i))=i$.
  For each~$\sigma_j$, we
  build a bipartite graph where~$v_1, \ldots, v_n$ are the vertices on
  the left,~$v'_1, \ldots, v_n'$ are the vertices on the right, and
  there is an edge between~$v_i$ and~$v_\ell$ if~$\sigma_j(v_i) = v_\ell$; we accept if this graph has a perfect
  matching for some~$\sigma_j$ and we reject otherwise.

  The algorithm runs in polynomial time because there 
  are~$n$~$\sigma_j$'s to try, and computing perfect matchings is a
  polynomial-time computable task.  The correctness follows from the fact
  that we need to map~$v_1$ to some vote in~$E'$ and we try all
  possibilities.
\end{proof}

Before moving to isomorphic distances, for a moment we will stop and discuss the single-peaked and single-crossing domains and their relation to isomorphism.

\subsubsection{Maximal Domains}
As an extended example of the usefulness of the isomorphism idea,
we consider the single-peaked and single-crossing domains.
They received extensive attention within (computational) social choice;
we point the reader to the survey of~\cite{elkind2022preference} for more details.

A single-peaked (single-crossing) domain is maximal if it is not
contained in any other single-peaked (single-crossing) domain. Each
maximal single-peaked domain~$\calD \subseteq \calL(C)$ contains
$2^{|C|-1}$ preference orders (\cite{mon:survey} attributes this fact to a 1962 work of Kreweras).
Since we can view a domain as an election that includes a single copy of
every preference order from the domain, our notion of isomorphism
directly translates to the case of domains, and we can formalize the fundamental difference between single-peakedness and
single-crossingness.

\begin{proposition}
  Each two maximal single-peaked domains over candidate sets of the
  same size are isomorphic. 
  \end{proposition}
  
\begin{proof}
  It suffices to note that 
  if~$\calD$ and~$\calD'$ are two maximal single-peaked domains (over
  candidate sets~$\{x_1, \ldots, x_m\}$ and~$\{y_1, \ldots, y_m\}$,
  respectively), with axes~$>_1$ and~$>_2$, such that: 
  \begin{align*}
   x_1 >_1 \cdots >_1 x_m && \text{and} && y_1 >_2 \cdots >_2 y_m,
  \end{align*}
  then a bijection that maps each~$x_i$ to~$y_i$ witnesses that the
  two domains are isomorphic.
\end{proof}

According to \cite{slinko2021characterization}, the number of maximal nonisomorphic single-crossing domains is equivalent to the number of weak Bruhat orders\footnote{https://oeis.org/A005118}.

\begin{corollary}
  There are~$\nicefrac{\binom{m}{2}!}{1^{n-1}\cdot3^{n-2}\cdot\,\cdots\,\cdot(2n-3)^1}$ maximal single-crossing
  domains over the same set of candidates that are not isomorphic.
\end{corollary}

This means that there is a significant difference between the single-peaked and single-crossing domains.

\section{Isomorphic Distances}
\label{sec:iso_dist}

We use the isomorphism idea to build distances between
elections that respect voter anonymity (so the order of the
voters in an election is irrelevant) and candidate neutrality (so the
names of the candidates are nothing more than 
temporary identifiers).

We focus on the following three distances, swap, Spearman, and discrete, which were described in detail in~\Cref{ch:preliminaries:dist_between_votes}.

As a reminder, by~$S_n$, we mean the set of all permutations over~$[n]$. Moreover, for two sets~$A$,~$B$ of the same cardinality, by~$\Pi(A,B)$ we denote the set of all one-to-one mappings from~$A$ to~$B$. 
Below we give our main definition.

\begin{definition}
  Let~$d$ be a distance between preference orders.
%
  Let~$E = (C,V)$ and~$E' = (C',V')$ be two elections, 
  where~$|C| = |C'|$,~$V = (v_1, \ldots, v_n)$ 
  and~$V' = (v'_1, \ldots, v'_n)$. We define the~$d$-isomorphism distance
  between~$E$ and~$E'$ as:
  
    \begin{equation*} 
    d(E,E') = \min_{\nu \in S_n}\min_{\sigma \in \Pi(C,C')}\sum_{i=1}^n  d(\sigma(v_i),v'_{\nu(i)})
  \end{equation*}
\end{definition}

We sometimes refer to the bijection~$\sigma$ as the candidate matching
and to the permutation~$\nu$ as the voter matching, and sometimes
instead of~$\nu$, we use bijection~$\tau \in \Pi(V,V')$ (depending on 
what is more convenient).  The name,~$d$-isomorphism distance, is
justified by the fact that if~$d(E,E') = 0$ for some two
elections (and~$d$ is a metric over preference orders), then these
elections are isomorphic.

  Note that in the above definition, we view elections as both
  anonymous and neutral.  This is why we apply the minimum operator
  over all permutations of the voters and over all bijections between
  the candidates.


\subsection{Computational Complexity}

We now turn to the complexity of computing isomorphism distances.
Formally, our problem is defined as follows.

\begin{definition}
  Let~$d$ be a distance over preference orders.  In the
  $d$-\textsc{Isomorphism Distance} problem (the~$\did$ problem) we
  are given two elections,~$E = (C,V)$ and~$E' = (C',V')$ 
  such that~$|C| = |C'|$ and~$|V| = |V'|$, and an integer~$k$.  We ask if~$d(E,E') \leq k$.
\end{definition}

We are also interested in two variants of this problem, the
\textsc{$\did$ with Candidate Matching} problem, where the bijection
$\sigma$ between the candidate sets is given (and fixed), and the
\textsc{$\did$ with Voter Matching} problem, where the voter
permutation~$\nu$ is given (and fixed). The former problem is in~$\p$
for polynomial-time computable distances, but, as we will see later,
this is not always true for the latter.

The summary of results is presented in Table~\ref{tab:iso_complexity}. Now, we will move on to analyzing all nine variants.

\begin{proposition}\label{pro:matching}
  For a polynomial-time computable~$d$, the problem \textsc{$\did$ with Candidate
    Matching} is in~$\p$.
\end{proposition}

\begin{proof}
  Let~$E$ and~$E'$ be our input elections and let~$\sigma$ be the
  input matching between candidates from~$E$ and~$E'$. To compute the distance between elections, it suffices to do the following. First, compute a distance between every pair of votes (one from~$\sigma(E)$ and another from~$E'$), Then, build a corresponding bipartite graph, where vertices on the left
  are the voters from~$\sigma(E)$, the vertices on the right are the voters
  from~$E'$, and all possible edges exist, weighted by the distances
  between the votes they connect. Finally, find the smallest-weight
  matching. The weight of the matching gives the value of the
  distance, and the matching itself gives the permutation~$\nu$).
\end{proof}
%
%
Using an argument very similar to that in the proof of Proposition~\ref{thm:ei-complexity},
we show that~$\discid$ problem is in~$\p$.


\begin{table}[t]
    \centering
    \begin{tabular}{c|ccc}
      \toprule
      &      & \sc with voter & \sc with candidate  \\
     $d$ &~$\did$ & \sc matching   & \sc matching    \\
      \midrule
     $d_\discrete~$ &~$\p$  &~$\p$    &~$\p$   \\
     $d_\spearman~$ &~$\np$-complete$^\dag$  &~$\p$    &~$\p$   \\
     $d_\swap~$    &~$\np$-complete &~$\np$-complete &~$\p$  \\
      \bottomrule
    \end{tabular}
    
    \caption{\label{tab:iso_complexity}The complexity of computing isomorphic distances.~$\dag$ this result is not a contribution of this dissertation.}
\end{table}




\begin{proposition}\label{thm:aei-hamming}
  The~$\discid$ problem is in~$\p$.
\end{proposition}

\begin{proof}
  Given two elections~$E = (C,V)$ and~$E' = (C',V')$, 
  where~$|C| = |C'|$,~$V = (v_1, \ldots, v_n)$ 
  and~$V' = (v'_1, \ldots, v'_n)$, for each pair of votes~$(v_i,v'_j)$ we
  construct a mapping~$\sigma_{ij}\colon C\to C'$ so that 
$
  \pos_{v_i}(c)=\pos_{v'_j}(\sigma_{ij}(c))
$
for each~$c\in C$.
We  choose~$\sigma_{ij}$ that leads to the smallest~$d_\discrete$
  distance (we compute these distances using the
  \textsc{$\discid$ with Candidate Matching} problem).

  The correctness of the algorithm follows from the observation that
  the largest possible value of~$d_\discrete(E,E')$ is~$n-1$; we
  can always ensure that at least one vote from~$E$ matches perfectly
  a vote from~$E'$. Thus, there must be two votes 
  for which~$\sigma_{ij}$ is the optimal candidate matching. 
\end{proof}

Using the same reasoning as above, we can also easily show the following.

\begin{corollary}
  The \textsc{$\discid$ with Voter Matching} problem is in~$\p$.
\end{corollary}

The elections for which the~$d_\discrete$ distance
is small are, in fact, nearly identical (up to renaming of the
candidates and reordering the voters). In consequence, we do not
expect such elections to frequently appear in real-life (for example, for two elections with~$n$
voters and a relatively large number of candidates, generated according
to the impartial culture model, we would expect their
$d_\discrete$ distance to typically be~$n-1$).  Thus, we need
more fine-grained distances, such as~$d_\swap$ and
$d_\spearman$.  Unfortunately, they are~$\np$-hard to compute
and, indeed, for~$d_\swap$ we inherit this result from the Kemeny
rule.

The $d_\swap$-\textsc{ID} problem generalizes the problem of finding a Kemeny ranking 
(roughly speaking,
to find a Kemeny ranking for a given election, it suffices to find
the smallest swap-based isomorphism distance between this election
and a ``constant'' one, where all the voters report identical
preference orders).


\begin{proposition}\label{proposition:swapnphard}
  The~$d_\swap$-\textsc{ID} problem is~$\np$-complete, even for elections
  with four voters.
\end{proposition}

\begin{proof}
  Membership in~$\np$ is easy to see. We give a reduction from the
  \textsc{Kemeny Score} problem. In the \textsc{Kemeny Score} problem
  we are given an election~$E = (C,V)$ and an integer~$k$, and we ask
  if there exists a preference order~$p$ over~$C$ such that~$\sum_{v
    \in V}d_\swap(v,p) \leq k$. The problem is~$\np$-complete~\citep{bar-tov-tri:j:who-won} and remains~$\np$-complete
  even for the case of four
  voters~\citep{dwo-kum-nao-siv:c:rank-aggregation}.  We reduce it to
  the~$d_\swap$-\textsc{ID} problem in a straightforward way: Given
  election~$E = (C,V)$ and~$k$, our reduction outputs election~$E$, a
  newly constructed election~$E' = (C',V')$, and an integer~$k$, 
  where~$C' = \{c'_1, \ldots, c'_{|C|}\}$ and every voter in~$V'$ has
  identical preference order~$v':c'_1 \succ \cdots \succ
  c'_{|C|}$.

  The reduction runs in polynomial time.
  Let us now argue that it is correct.  Let~$V = (v_1, \ldots, v_n)$
  and let~$V'$ consist of~$n$ copies of~$v'$.  We note that~${d_\swap}(E,E') = \min_{\sigma \in \Pi(C,C')}\sum_{i=1}^n
  d_\swap(\sigma(v_i),v') = \min_{\sigma' \in \Pi(C',C)}\sum_{i=1}^n
  d_\swap(v_i,\sigma'(v'))$, which is at most~$k$ if and only if there
  exists a preference order~$p\in \calL(C)$ such that~$\sum_{v \in V}d_\swap(v,p) \leq k$.
\end{proof}

Since the above reduction works even for elections with four voters,
having a matching between the voters cannot make the problem simpler
(this also follows from the fact that in our reduction one election
consists of identical votes).

\begin{corollary}
  \textsc{$d_\swap$-ID with Voter Matching} is~$\np$-complete.
\end{corollary}

The situation for~$d_\spearman$-ID is somewhat different. In this case
Litvak's rule \citep{litv:j:dist-cons}, defined analogously to 
the Kemeny rule, but for the Spearman distance, is polynomial-time
computable~\citep{dwo-kum-nao-siv:c:rank-aggregation} and we can lift
this result to the case of \textsc{$d_\spearman$-ID with Voter Matching}. Without the voter matching,~$d_\spearman$-ID is
$\np$-complete.

\begin{proposition}\label{pro:spear-with-voters}
  \textsc{$d_\spearman$-ID with Voter Matching} is in~$\p$.
\end{proposition}

\begin{proof}
  Let~$E = (C,V)$ and~$E' = (C',V')$ be two elections, 
  where~$|C| = |C'|$,~$V = (v_1, \ldots, v_n)$, and~$V' = (v'_1, \ldots, v'_n)$, and let~$\nu \in S_n$ be the given
  voter matching. 
  For a bijection~$\sigma\colon C \rightarrow C'$, the Spearman
  distance between~$E$ and~$E'$ is~$\sum_{i=1}^n d_\spearman(\sigma(v_i),v'_{\nu(i)})$, which is:
  \begin{align*}
    \sum_{c' \in C'} \sum_{i=1}^n |\pos_{v_i}(\sigma^{-1}(c')) - \pos_{v'_{\nu(i)}}(c')|.
  \end{align*}
  In consequence, the cost induced by matching candidates~$c \in C$
  and~$c' \in C'$ is~$\cost(c,c') = \sum_{i=1}^n |\pos_{v_i}(c) -
  \pos_{v'_{\nu(i)}}(c')|$. To solve our problem, it suffices to find
  a minimum cost perfect matching in a bipartite graph where
  candidates from~$C$ are the vertices on the left, candidate from~$C'$ 
  are the vertices on the right, and for each~$c \in C$ and~$c' \in C'$ 
 we have an edge from~$c$ to~$c'$ with cost~$\cost(c,c')$.
\end{proof}

The final missing result was proved by~\cite{fal-sko-sli-szu-tal:c:isomorphism}

\begin{theorem}[\cite{fal-sko-sli-szu-tal:c:isomorphism}]\label{theorem:spearmannphard}
  The~$d_\spearman$-\textsc{ID} problem is~$\np$-complete.
\end{theorem}

\subsection{ILP}
We provide integer linear programs (ILPs) for computing $d_\spearman$ and $d_\swap$.



\begin{proposition}
  There is an ILP for~$d_\spearman$. 
\end{proposition}

\begin{proof}

  Let~$E = (C,V)$ and~$E' = (C',V')$ be the elections we wish to
  compute the distance for, with~$C = \{c_1, \dots, c_m\}$,~$C' = \{c'_1, \dots, c'_m\}$,~$V = \{v_1, \dots, v_n\}$, and~$V' = \{v'_1, \dots, v'_n\}$.
%
%
%
  For each~$k, k' \in [n]$, we define 
  a binary variable~$N_{k, k'}$ with the intention that value~$1$ indicates that
  voter~$v_k$ is matched to voter~$v'_{k'}$. 
  Similarly, for each~$i, i' \in [m]$, we define 
  a binary variable~$M_{i, i'}$ with the intention that value~$1$ means that candidate~$c_i$ is
  matched to candidate~$c'_{i'}$.
  For each~$k, k' \in [n]$ and each~$i,i' \in [m]$, we define a binary
  variable~$P_{k,k',i,i'}$ with the intention that~$P_{k,k',i,i'} = N_{k,k'}\cdot M_{i,i'}$.
  We introduce the following constraints:
  \begin{align}
    \label{ilp:1}
    &\textstyle\sum_{k' \in [n]} N_{k, k'} = 1, \forall k\in[n]; \text{ }
    \textstyle\sum_{k \in [n]} N_{k, k'} = 1, \forall k' \in[n] \\
    \label{ilp:2}
    &\textstyle\sum_{i' \in [m]} M_{i,i'} = 1, \forall i\in[m]; \text{ }
    \textstyle\sum_{i \in [m]} M_{i,i'} = 1, \forall i' \in[m]\! \\
    \label{ilp:3}
    &\textstyle\sum_{k' \in [n],i' \in [m]} P_{k,k',i,i'} = 1, \text{\quad}\forall i \in [m], k \in [n]\\
    \label{ilp:4}
    &\textstyle\sum_{k \in [n],i \in [m]} P_{k,k',i,i'} = 1, \text{\quad}\forall i' \in [m], k' \in [n]\\
    \label{ilp:5}
    &P_{k, k', i, i'} \leq N_{k, k'}, \text{\quad\quad\quad\quad\quad\ } \forall i,i' \in [m], k,k' \in [n] \\
    \label{ilp:6}
    &P_{k, k', i, i'} \leq M_{i, i'}, \text{\quad\quad\quad\quad\quad\ } \forall i,i' \in [m], k,k' \in [n] 
  \end{align}
  Constraints~\eqref{ilp:1} and~\eqref{ilp:2} ensure that variables~$N_{k,k'}$ and~$M_{i,i'}$ describe matchings between voters and
  candidates, respectively.  Constraints~\eqref{ilp:3}--\eqref{ilp:6}
  implement the semantics of the~$P_{k,k',i,i'}$ variables (the
  former two ensure that for a given vote/candidate pair, there is
  exactly one vote/candidate pair in the other election to which they
  are matched; the latter two ensure connection between the~$P_{k,k',i,i'}$ variables and the~$N_{k,k'}$ and~$M_{i,i'}$ variables).
%
%
%
%
  The optimization goal is to minimize~$\textstyle\sum_{k, k' \in[n],\ i, i' \in [m]} P_{k, k',i,i'} \cdot
  |\pos_{v_k}(c_i) - \pos_{v'_{k'}}(c'_{i'})|$
  (which, for values~$P_{k,k',i,i'}$ that satisfy the constraints of
  the program, defines the Spearman distance for
  the given matchings). Values~$|\pos_{v_k}(c_i) - \pos_{v'_{k'}}(c'_{i'})|$ are precomputed.
\end{proof}

\begin{proposition}
  There is an ILP for~$\SWAP$.
\end{proposition}

\begin{proof}

  The proof for~$\SWAP$ is very similar to the one for~$\SPEAR$. For~$\SWAP$ we need all the constraints presented for~$\SPEAR$ and three more, so, we focus only on the additional ones.
  For each~$k, k' \in [n]$ and each~$i,i',j,j' \in [m]$, we define a binary
  variable~$R_{k,k',i,i',j,j'}$ with the intention that~$R_{k,k',i,i',j,j'} = N_{k,k'}\cdot M_{i,i'}\cdot M_{j,j'}$. Note that, we assume that~$i < j$,~$i' \neq j'$.
  For~$\SPEAR$ it suffices to have four indices (two for voters and two for candidates), because to compute the Spearman distance between two matched votes we only need to iterate over each pair of matched candidates. However, for the~$\SWAP$ we need six indices (two for voters and four for candidates), because to compute the swap distances between two matched votes we need to iterate over each pair of pairs of candidates.

  We introduce the following constraints:
  \begin{align}
    \label{ilp:13}
    &R_{k, k', i,i',j,j'} \leq P_{k,k',i,i'}, \text{\quad\quad\quad} \forall \substack{i',j'\in[m] \\ i < j, i' \neq j'} \in [m], k,k' \in [n] \\
    \label{ilp:14}
    &R_{k, k', i,i',j,j'} \leq P_{k,k',j,j'}, \text{\quad\quad\quad} \forall \substack{i',j'\in[m] \\ i < j, i' \neq j'} \in [m], k,k' \in [n] \\
    \label{ilp:15}
    &\textstyle\sum_{k, k' \in [n], \substack{i',j'\in[m] \\ i < j, i' \neq j'}} R_{k, k', i,i',j,j'} = n \cdot \binom{m}{2} 
  \end{align}
  
  Constraints~\eqref{ilp:13} and~\eqref{ilp:14} ensure that~$R_{k, k', i,i',j,j'}$ can be true only if~$P_{k,k',i,i'}$ and~$P_{k,k',j,j'}$ are true. And constraints~\eqref{ilp:15} ensure that proper number of~$R$ variables are equal to one. 

  The optimization goal is to minimize:
  
\centering
$
\textstyle\sum_{k, k' \in [n], \substack{i',j'\in[m] \\ i < j, i' \neq j'}} R_{k, k', i,i',j,j'} \times 
\begin{cases}
1 \text{ if } pos_{v_k}(c_{i}) > pos_{v'_{k'}}(c'_{j}) \\ \text{ \ \ and }  pos_{v_k}(c_{i'}) < pos_{v'_{k'}}(c'_{j'})\\
1 \text{ if } pos_{v_k}(c_{i}) < pos_{v'_{k'}}(c'_{j})  \\ \text{ \ \ and } pos_{v_k}(c_{i'}) > pos_{v'_{k'}}(c'_{j'})\\
0 \text{ otherwise }   \\
\end{cases} 
$ \\

where~$pos$ values are precomputed.
\end{proof}

While the ILPs described above find optimal solutions, they can be
quite slow to solve for any but the smallest instances. Thus, in practice when we want to compute particular distances, instead of ILPs, we have to use a brute-force (BF) algorithm.

\subsubsection{Comparison}

To compare ILP and BF approaches, we conducted a simple experiment in which we computed the Spearman and swap distances for small numbers of candidates and voters and compared the time needed to find the optimal solution. For the Spearman distance, we use elections with $3,\dots,8$ candidates and $8$ voters, while for the swap distances, we use elections with $3,4,5$ candidates and $5$ voters. The results are presented in~\Cref{tab:ilp_bf_spearman} (for Spearman) and~\Cref{tab:ilp_bf_swap} (for swap). 
In each cell, we have the average time (in seconds) needed to compute a single distance between two random impartial culture elections using ILP and BF. The presented values are averages over~$1000$ iterations. 
The differences are extreme, with BF approach being, literally speaking, thousands times faster. When computing the swap distance with~$5$ candidates and~$5$ voters, the BF approach was more than~$300 000$ times faster than the ILP approach.

\begin{table}[h]
    \centering
    \begin{tabular}{c|cccccc}
      \toprule
      Method & 3 & 4 & 5 & 6 & 7 & 8 \\
      \midrule
      ILP &~$0.12$s  &~$0.61$s  &~$1.69$s & ~$5.49$s &~$15.86$s &~$42.82$s \\
      BF &~$<0.01$s &~$<0.01$s  &~$<0.01$s &~$<0.01$s &~$0.01$s &~$0.05$s  \\
      \bottomrule
    \end{tabular}
    \caption{\label{tab:ilp_bf_spearman}Time needed to compute the Spearman distance between two random impartial culture instances. All the values are presented in seconds. We consider elections with $3,\dots,8$ candidates and $8$ voters.}
\end{table}

\begin{table}[h]
    \centering
    \begin{tabular}{c|ccc}
      \toprule
      Method & 3 & 4 & 5 \\
      \midrule
      ILP &~$0.55$s  &~$12.15$s  &~$97.54$s \\
      BF &~$<0.01$s &~$<0.01$s  &~$<0.01$s  \\
      \midrule
    \end{tabular}
    \caption{\label{tab:ilp_bf_swap}Time needed to compute the swap distance between two random impartial culture instances. All the values are presented in seconds. We consider elections with $3,4,5$ candidates and $5$ voters.}
\end{table}


\subsection{Visualization of the Distances}
In this section we present "cross maps" of preferences, a similar experiment to the one presented in \Cref{ordinal_map_pref}.
The main difference is that, in~\Cref{ordinal_map_pref} on each single picture we presented one election, while now on each single picture we present two elections embedded jointly.

Given two elections, we compute the mapping between the candidates from these elections, such that it minimizes the swap distance between them. Next, given the mapping, we proceed as for previous maps of preferences and simply compute the swap distance between each pair of votes from both elections.

In \Cref{fig:dicroscope_part_1} we present cross maps for the eight following models: impartial culture, antagonism, the Norm-Mallows model with $\normphi = 0.2$, the 0.25-Norm-Mallows model with $\normphi = 0.2$, the urn model with $\alpha=0.2$, and SPOC. We generated $16$ elections (two from each model). Eight of them are as columns (red ones), and the other eight of them are as rows (blue ones). When presented jointly, all red points represent the \emph{column} election and all blue points represent the \emph{row} one.

We analyze the results row by row. The votes from impartial culture occupy the whole space, and hence other models, when combined with it, should look similar to how they look alone. We observe this for all instances with the exception for those from $0.25$-Norm-Mallows model, which are shifted towards the edge.

In the next row we have antagonism, which is slightly ``squeezing'' all other instances. As expected, when combined with $0.25$-Norm-Mallows, two extreme AN votes match the centers of two mallows groups. In the following discussion, we refer to these groups as the smaller group and the larger group.

\begin{figure}
    \centering
    \includegraphics[width=14cm]{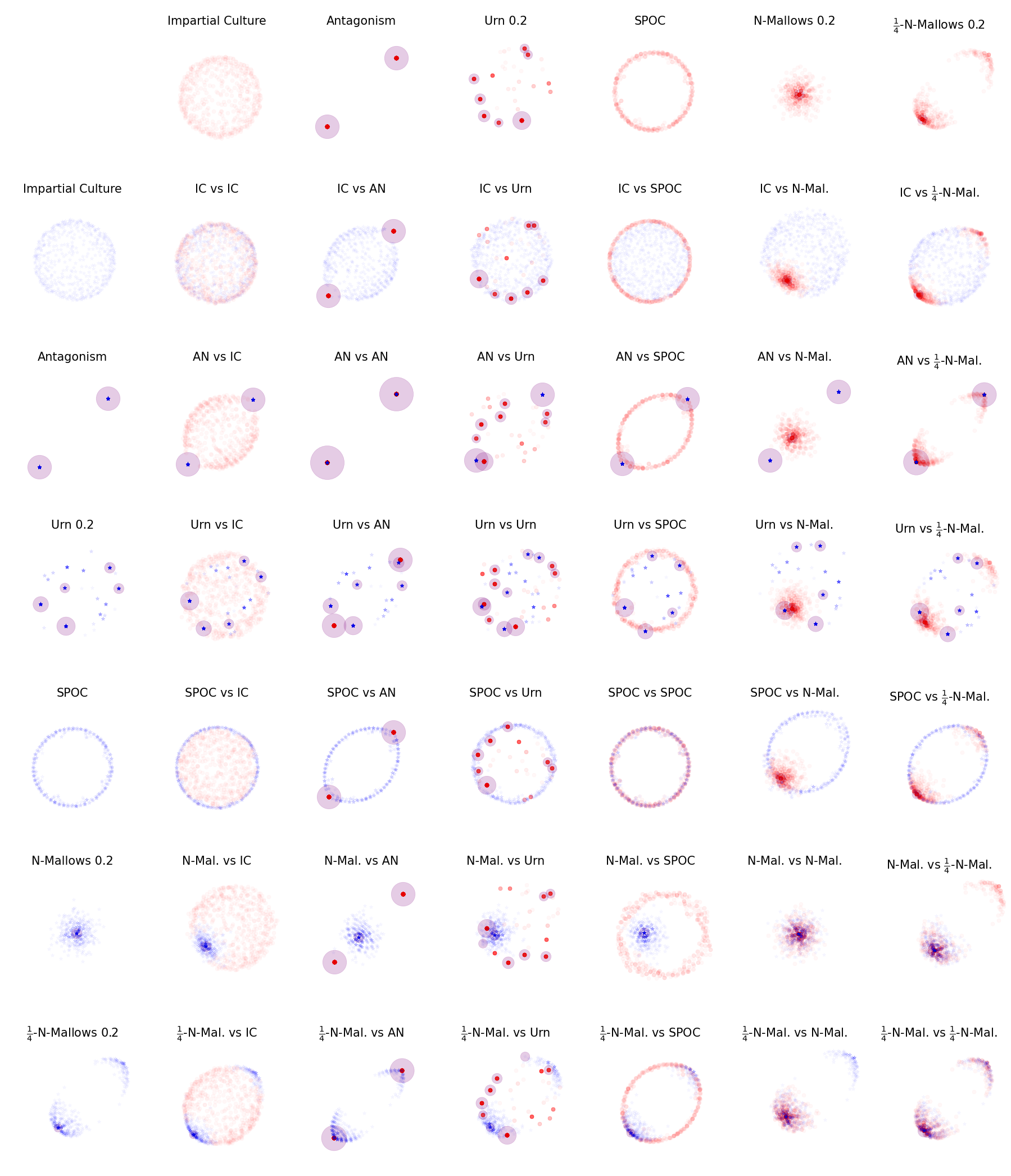}
    \caption{Cross maps of preferences ($10$ candidates, $2 \times 500$ voters).}
    \label{fig:dicroscope_part_1}
\end{figure}

\begin{figure}
    \centering
    \includegraphics[width=14cm]{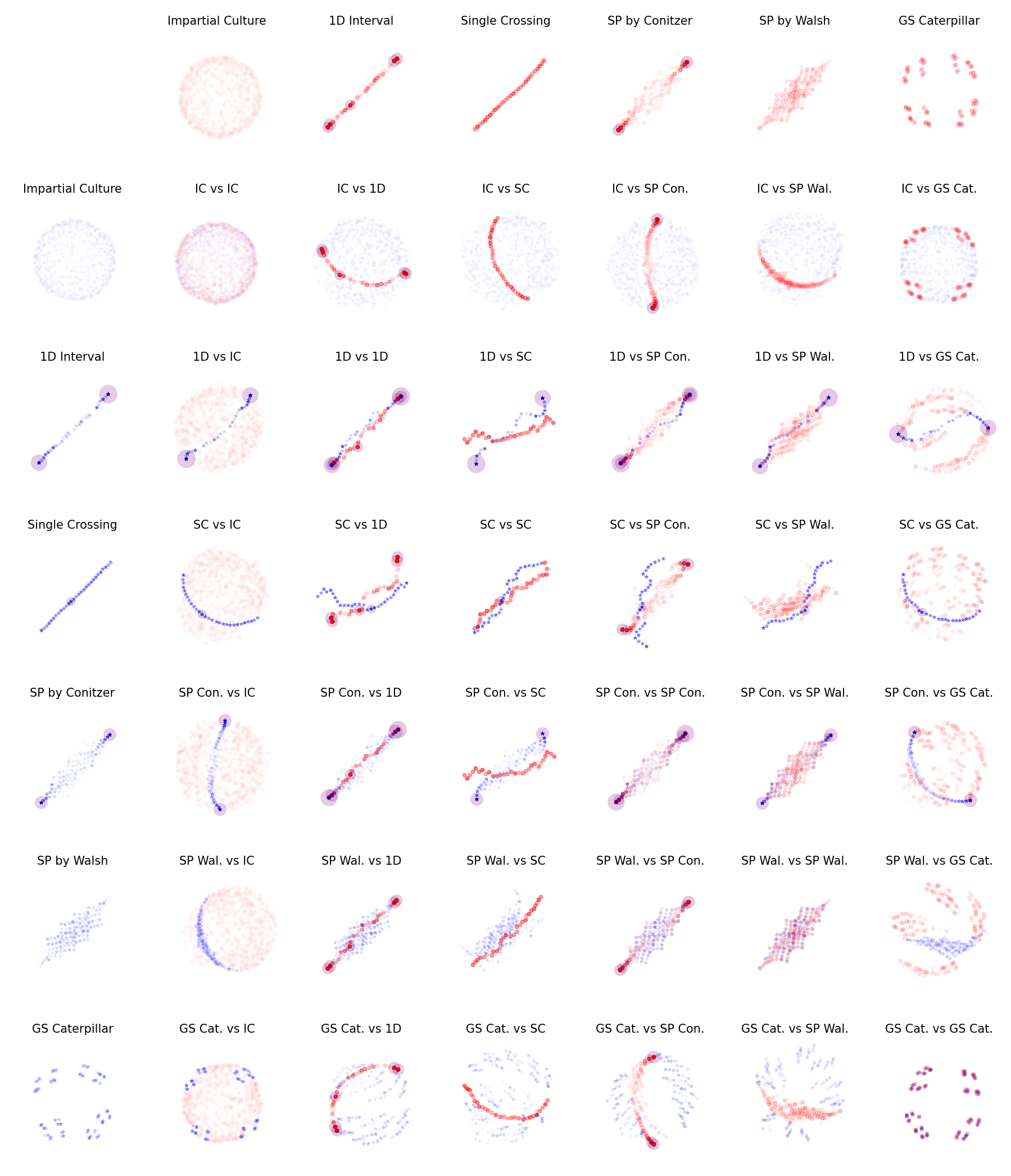}
    \caption{Cross maps of preferences ($10$ candidates, $2 \times 500$ voters).}
    \label{fig:dicroscope_part_2}
\end{figure}

At the diagonal, we have pairs of elections from the same model embedded jointly. Interestingly, most of the elections, when embedded jointly with another election from the same model, produce a very similar picture to those, when they are embedded separately. Moreover, many points from both elections overlap. We can see this most clearly for the $\AN$ elections. When we embed two $\AN$ elections, they fully overlap (which should not be surprising because they are isomorphic). Similarly, if we look, for example, at two SPOC elections embedded jointly, they also strongly overlap. However, it is not the case for all the models, for example, two elections from the urn model are relatively independent of each other, and points from both elections occupy quite different places in the picture.

Then we have SPOC, which is maintaining its circular shape in all pictures. However, sometimes when combined with Norm-Mallows, it is getting less sharp.
When Norm-Mallows is combined with 0.25-Norm-Mallows, it is placing its center over the larger group of 0.25-Norm-Mallows.

In \Cref{fig:dicroscope_part_2} we present another set of cross maps, for the eight following models: impartial culture, 1D Interval, single-crossing, single-peaked by Walsh, single-peaked by Conitzer, and GS Caterpillar.

The single-crossing, Walsh, Conitzer, and 1D Interval models all have the same oblong shape. Interestingly, when embedding jointly Walsh and Conitzer or 1D Interval elections, the oblong shapes are put together nicely, i.e., one over the other. Nonetheless, when single-crossing elections were embedded jointly with other longitudinal instances, they form a crossing-over shape.

Caterpillar group-separable elections are changing a lot depending on what other instances they are embedded with. When combined with impartial culture, they present a very similar shape to the one when embedded alone. However, when embedded jointly with an oblong-shaped instance, they disperse significantly.

\vspace{0.2cm}
\begin{conclusionbox}
\begin{itemize}
    \item Presented maps of preferences confirm the general intuition behind different statistical cultures and the relations between them. Later on, when we see the map of elections in~\Cref{fig:main_matrix_iso}, if two elections are similar (i.e., their red and blue points are close to covering each other in~\Cref{fig:dicroscope_part_1,fig:dicroscope_part_2}) then, indeed, they are next to each other on the map.
\end{itemize}
\end{conclusionbox}

\subsection{Isomorphic Maps of Elections}\label{sec:iso_maps_of_elections}

Next we present our first maps of elections. We start with the description of the concept, and later we present the technical details.
To build a map, we proceed as follows. First, we generate a number of instances of elections. Second, we compute a certain distance between each pair of them. Third, we embed these distances in a two-dimensional Euclidean space. Voilà, we obtain a map of elections. Now, we will go over these three steps with more technical details.

We assembled a number of elections generated using statistical cultures from~\Cref{ch:stat_cult} and four compass
elections that capture four different types of (dis)agreement among voters, identity, uniformity, antagonism, and stratification (see~\Cref{ch:stat_cult}). We expect good metrics to put these compass elections far apart.

\begin{table}[t]
  \centering
{
  \begin{tabular}{lcc}
    \toprule
    Model & Number of Elections \\
    
    \midrule
    Impartial Culture           & 20 \\
    \midrule
    Urn       & 60 \\  
    Mallows    & 60 \\
    \midrule
    Group-Separable (Balanced)     & 20 \\
    Group-Separable (Caterpillar)    & 20 \\
    Single-Peaked (Conitzer)  & 20 \\
    Single-Peaked (Walsh)      & 20 \\
    SPOC (Conitzer)                & 20 \\
    Single-Crossing           & 20 \\
    \midrule
    Interval        & 20 \\
    Disc         & 20 \\
    Cube         & 20 \\
    Circle         & 20 \\
    \midrule
    Compass ($\ID$,~$\AN$,~$\UN$,~$\ST$) & 4 \\
    \bottomrule
  \end{tabular} }
    \caption{\label{tab:testbed}List of selected statistical cultures and numbers of elections from these cultures accordingly.}

\end{table}

We list the exact distributions, and numbers of generated elections used in the map in~\Cref{tab:testbed}.
All in all, we generated~$340+4$ elections, each with~$10$ candidates and $50$ voters, 
some from very popular statistical cultures, and some
from less typical ones, such as the SPOC; plus four compass elections.

Regarding the parameters, for the Norm-Mallows model we choose~$\normphi$ uniformly at random. For the urn model, we choose~$\alpha$ according to the Gamma distribution\footnote{Popular probability distribution parametrized by shape and scale parameters.} with shape parameter~$k=0.8$ and scale parameter~$\theta=1$ (we will discuss this in detail in~\Cref{ch:applications}).
For all Euclidean models, we sample the ideal points of candidates and voters uniformly at random from:
\begin{itemize}
    \item $[0,1]$ interval for the Interval model;
    \item disc with radius~$r=0.5$ and center in~$(0,0)$ for the Disc model;
    \item $[0,1]^3$ cube for the Cube model;
    \item circle with radius~$r=0.5$ and center in~$(0,0)$ for the Circle model.
\end{itemize}

\begin{figure}[t]
   
    \begin{subfigure}[b]{0.325\textwidth}
     \centering
    \includegraphics[width=4.4cm]{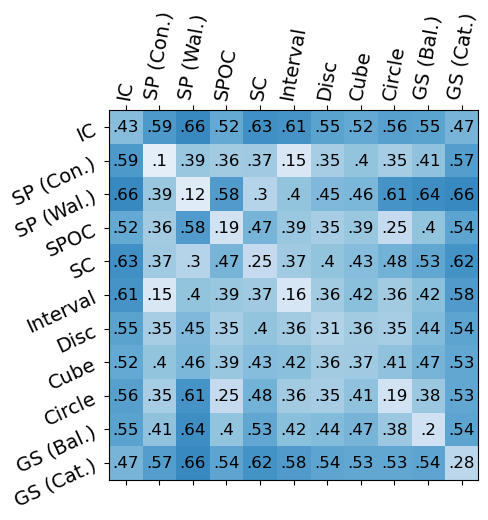}
    \caption{Swap}
    \end{subfigure}
    \begin{subfigure}[b]{0.325\textwidth}
     \centering
     \includegraphics[width=4.4cm]{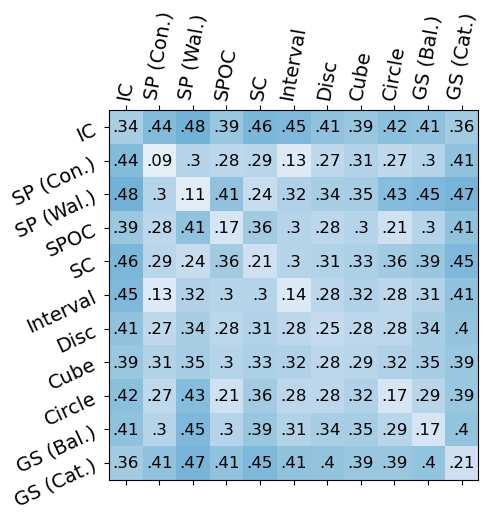}
    \caption{Spearman}
    \end{subfigure}
    \begin{subfigure}[b]{0.325\textwidth}
     \centering
     \includegraphics[width=4.4cm]{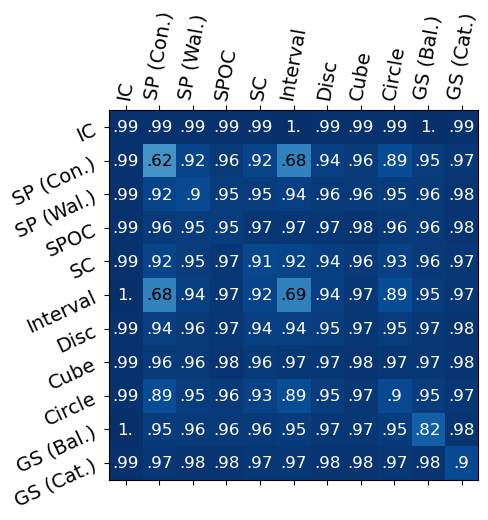}
    \caption{Discrete}
    \end{subfigure}
\caption{The average distances between elections from given cultures (normalized by the largest distance).}
\label{fig:main_matrix_iso}
\end{figure}

As a second step, we computed the swap (Spearman/discrete) distance between each pair of the generated elections. For each set of elections,
we give their average distance to the elections from the other sets
(or to the elections within the set, on the diagonal), normalized by the largest distance.
We show statistics regarding (some of) these distances in~\Cref{fig:main_matrix_iso}. 

With the concrete values of the swap (Spearman/discrete) distances in hand, we
computed a mapping of the generated elections to a 2D space, so that
the Euclidean distances between the points in this mapping
reflect the original distances between the elections.
To compute the embedding, we used a variant of the Kamada-Kawai algorithm, recently proposed by\cite{mt:sapala}, which is based on the work of \cite{kamada1989algorithm}. More details about the different embedding algorithms are presented in \Cref{ch:applications}.

\begin{figure}
    \centering
    \begin{subfigure}[b]{0.49\textwidth}
        \centering
        \includegraphics[width=6.4cm, trim={0.2cm 0.2cm 0.2cm 0.2cm}, clip]{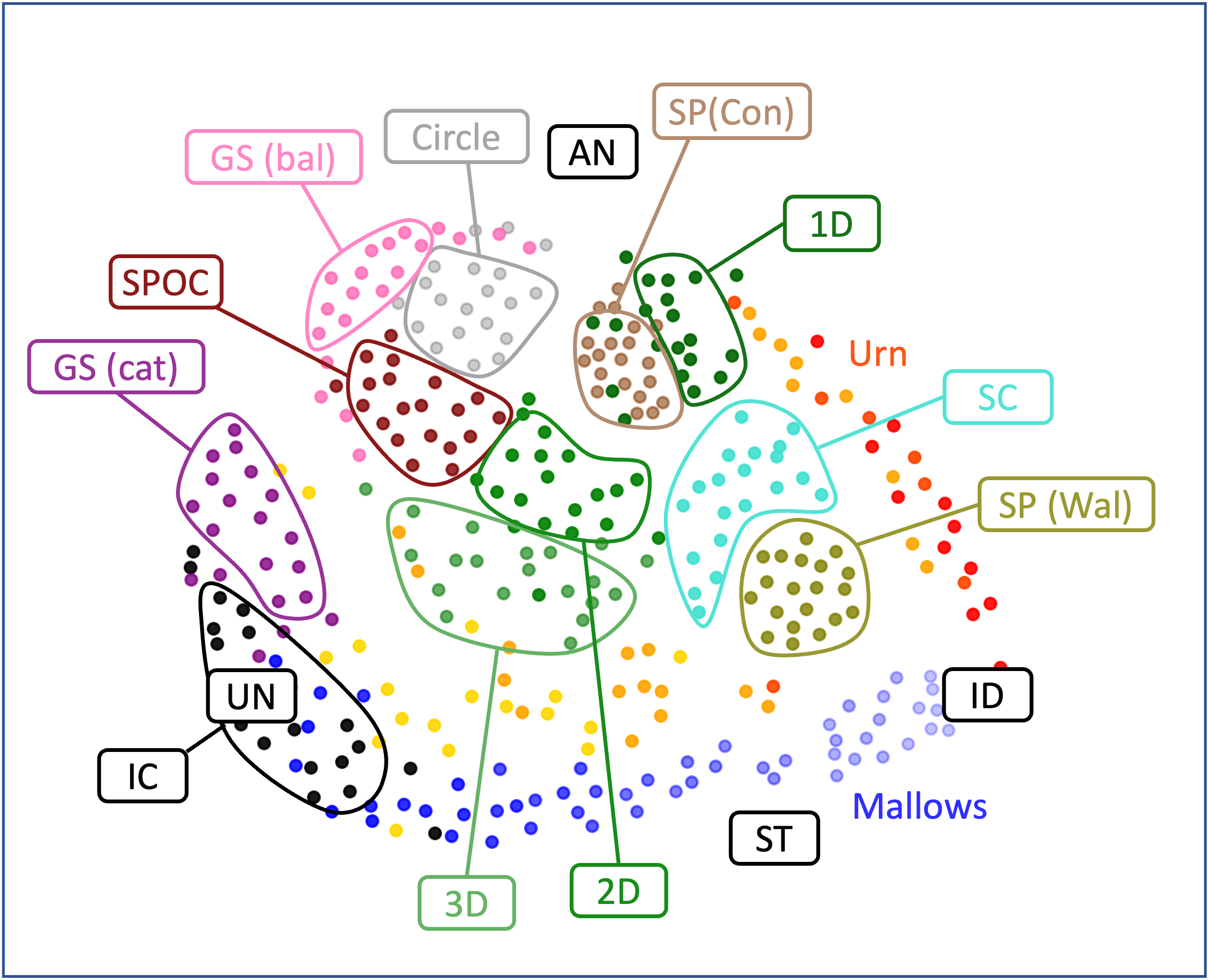}
        \caption{Swap}
    \end{subfigure}
    \begin{subfigure}[b]{0.49\textwidth}
        \centering
        \includegraphics[width=6.4cm, trim={0.2cm 0.2cm 0.2cm 0.2cm}, clip]{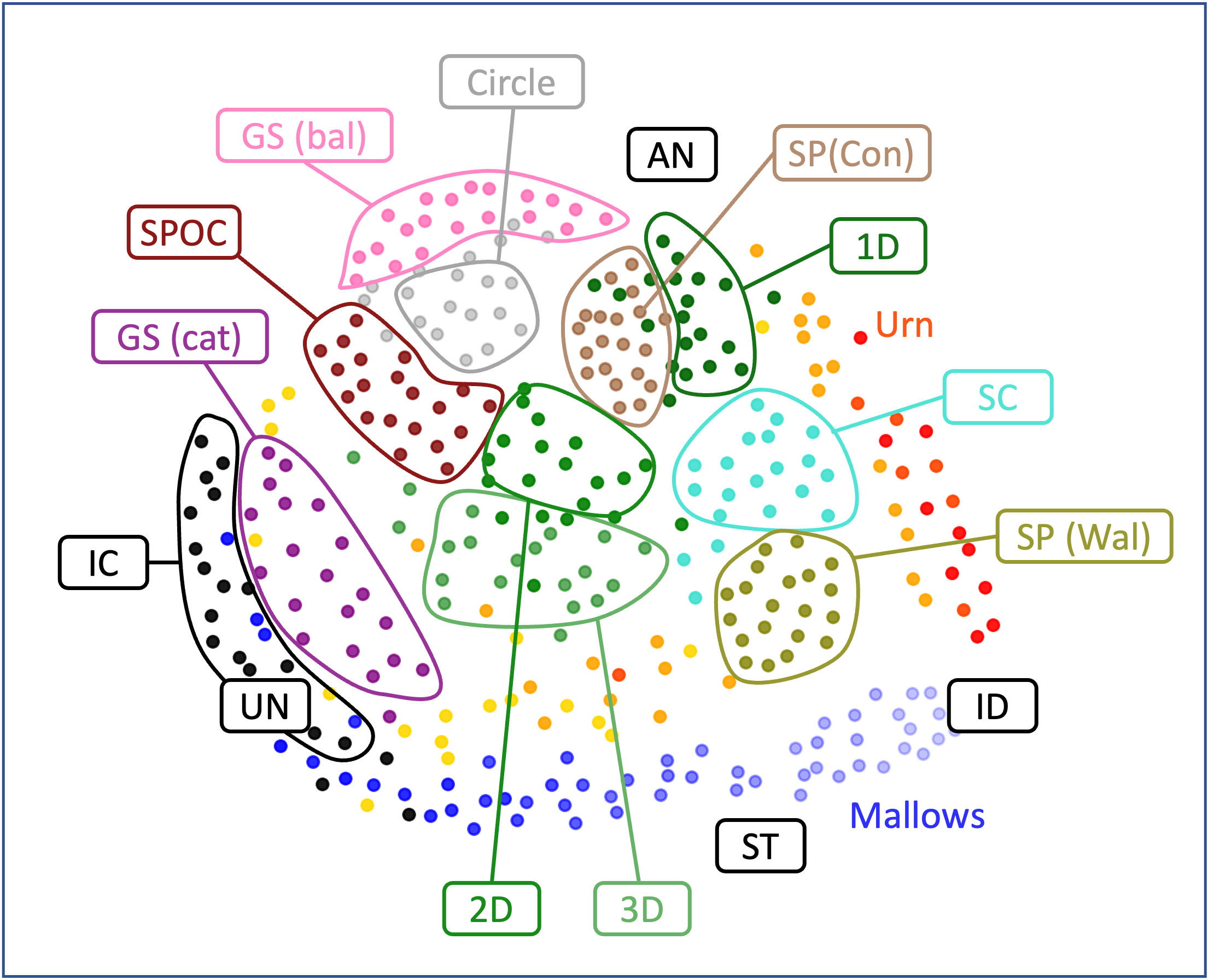}
        \caption{Spearman}
    \end{subfigure}
    
    \begin{subfigure}[b]{0.49\textwidth}
        \centering
        \includegraphics[width=6.4cm, trim={0.2cm 0.2cm 0.2cm 0.2cm}, clip]{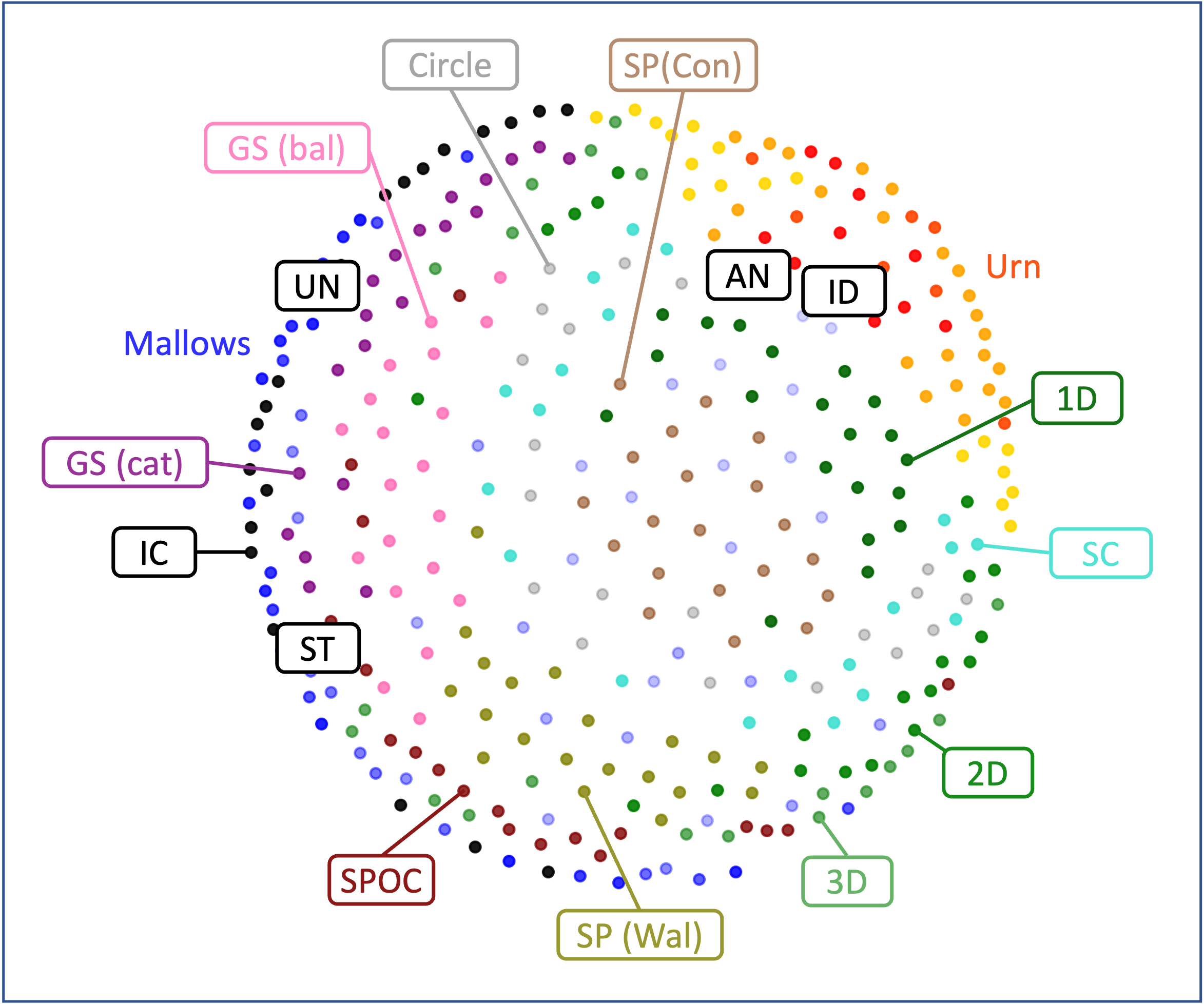}
        \caption{Discrete}
    \end{subfigure}
    \caption{Maps of elections based on isomorphic distances.}
    \label{fig:main_maps_iso}
\end{figure}

We present the visualization we obtained for this embeddings in Figure~\ref{fig:main_maps_iso} and refer to
them as our maps of elections. 
We first focus on the left map, which is based on the swap distances. Three compass elections, i.e., identity, uniformity, and antagonism, form a triangle that is almost equilateral and roughly limits the space. At the bottom, we see elections from the Norm-Mallows model\footnote{For the Norm-Mallows elections, the larger is the transparency of a given point, the smaller is the $\normphi$ value.} that form a path from the identity to the uniformity elections. The higher the~$\normphi$ parameter, the closer we get to the impartial culture.
Elections from the Pólya-Eggenberger urn model\footnote{We mark the urn elections with small, medium, and large $\alpha$ parameters by yellow, orange, and red colors, respectively} are distributed over a large area, in comparison to elections from other models; in principle, the larger the~$\alpha$ parameter, the closer they are to the identity, however, unlike for the Norm-Mallows model, two urn elections with the same~$\alpha$ parameter can be very different from one another.

What is surprising is that the way of sampling single-peaked elections is strongly influencing their location on the map. Elections from the Conitzer model are not that close to elections from the Walsh one. However, they are very close to elections from the Interval model. Similarly, the elections from the SPOC model lie next to those from the Circle model. For Euclidean elections, the higher the dimension, the closer they are to the impartial culture ones.

Almost the whole left upper part, that is, the part between the identity and antagonism, is occupied by elections  where, on average, all candidates have very similar Borda scores\footnote{In particular, standard deviation for the caterpillar group-separable, balanced group-separable, SPOC, IC and Sphere elections on average equals $18.45$, while for other elections on average it equals $76.71$, i.e., four times more.}. In other words, all the candidates perform similarly. We call this part the \emph{Borda balance} area, and we will return to it in~\Cref{ch:applications:sec:score}.



Now, if we look at the map that is based on the Spearman distances, we see that it is very similar to the one based on the swap distances. At the same time, the map created based on the discrete distance is clearly different. We will not exaggerate if we say that the discrete map is of limited usefulness.

\begin{figure}[t]
    \centering
    \begin{subfigure}[b]{0.49\textwidth}
        \centering
          \includegraphics[width=6.6cm]{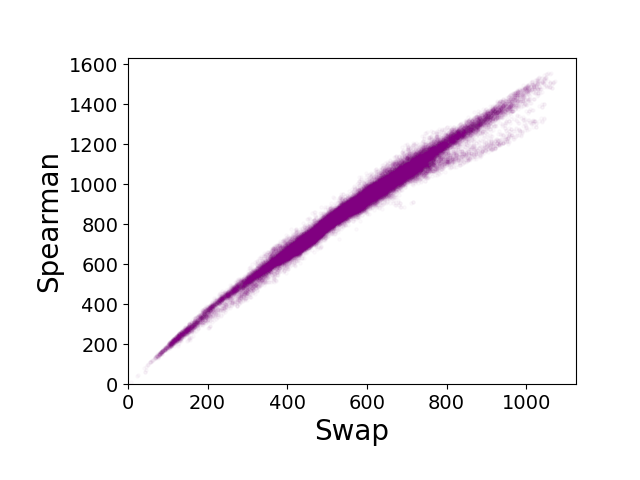}
        \caption{Swap vs Spearman}
    \end{subfigure}
    \begin{subfigure}[b]{0.49\textwidth}
        \centering
          \includegraphics[width=6.6cm]{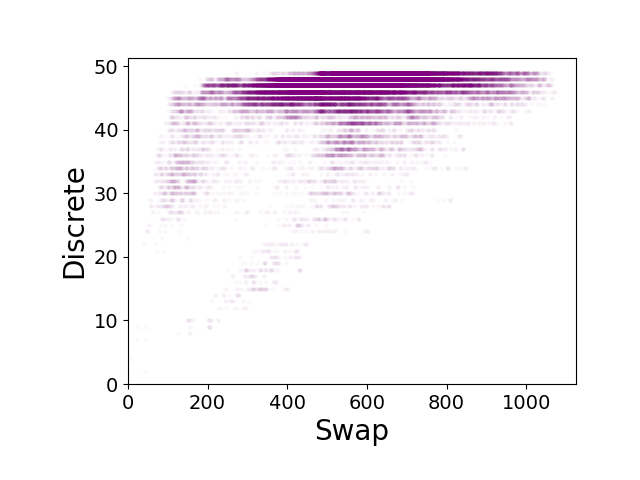}
        \caption{Swap vs Discrete}
    \end{subfigure}

    \caption{Correlation between the isomorphic distances.}
    \label{fig:corr_iso}
\end{figure}

In \Cref{fig:corr_iso} we present the correlation plots for our isomorphic distances. Each purple dot represents the distance between a pair of elections. As we can see, the swap and Spearman distances are very strongly correlated (having Pearson correlation coefficient~$0.99$), while the swap and discrete distances are vaguely correlated (having Pearson correlation coefficient equal to~$0.33$).

While the map
based on the discrete distance is not very appealing, the maps based on the swap
and Spearman distances give us an interesting insight into the space of statistical
cultures. Unfortunately, the computations of these distances, even for instances
with only~$10$ candidates and~$50$ voters, are quite demanding. Due to this fact, if
we wanted a map with a larger number of candidates, such as~$20$ or~$100$, 
we would need a new distance, which could be computed faster. And that is what we
are going to discuss in the next section. Moreover, due to the very strong similarity between the swap and Spearman distances, in the later part of this dissertation we focus only on the swap distance.

\vspace{0.2cm}
\begin{conclusionbox}
\begin{itemize}
   \item The Spearman distance is very strongly correlated with the swap distance.
   \item Elections from the same statistical culture tend to lie next to each other on the maps.
\end{itemize}
\end{conclusionbox}

\section{Nonisomorphic Distances}
\label{sec:non_iso_dist}

Next, we introduce several nonisomorphic distances. For each of these distances, we start by giving its formal definition, then, we show that it is a psuedometric, and finally we discuss its computational complexity. Before that, we present several aggregate representations of elections, which are nothing else but simplified forms of elections, and which will be useful for defining our distances.

\subsection{Aggregate Representations}
Let~$E = (C,V)$ be an election with~$C = \{c_1, \ldots, c_m\}$ and~$V = (v_1, \ldots, v_n)$.
Below, we present the following aggregate representations of~$E$:
\begin{description}

\item[Weighted Majority Relation.] For each two candidates~$c, d \in C$,~$\calM_E(c,d)$ is the
  number of voters that prefer~$c$ to~$d$ in election~$E$. We call it
  the weighted majority relation and represent it as an~$m \times m$
  matrix where rows and columns correspond to the candidates (the
  diagonal is undefined).
  A relative weighted majority relation is a weighted majority
  relation from whose entries we subtract~$\nicefrac{n}{2}$.
    
\item[Position Matrix.] For a candidate~$c \in C$ and a position~$i \in [m]$,~$\calP_E(c,i)$ is the number of voters from~$E$ that rank~$c$ on
  position~$i$;~$\calP_E(c) = (\calP_E(c,1), \ldots, \calP_E(c,m))$ is
  a (column) position vector of~$c$.
  We view~$\calP_E$ as a matrix with columns~$\calP_E(c_1), \ldots, \calP_E(c_m)$ and call it a position matrix.

    
\item[Borda Score Vector.] For a candidate~$c \in C$,~$\calB_E(c) = \sum_{i=1}^n \big(m-\pos_{v_i}(c)\big)$ is the Borda score of~$c$ in~$E$. Then~$\calB_E=(\calB_E(c_1), \dots, \calB_E(c_m))$ is the Borda score vector,
  whose entries correspond to the candidates.
\end{description}

Note that in each of these aggregate representations, we are losing certain information about the election, i.e., there may be two distinct elections that have the same aggregate representation. Next, we provide a simple example that shows how these aggregate representations look in practice.

\begin{example}
  Consider an election~$E = (C,V)$, 
  where~$C = \{a,b,c\}$,~$V = (v_1,v_2,v_3,v_4)$, and the votes are:
  \begin{align*}
    \small
    v_1\colon&  a \pref b \pref c, \\
    v_2\colon&  b \pref c \pref a, \\ 
    v_3\colon&  b \pref a \pref c, \\ 
    v_4\colon&  c \pref a \pref b.
  \end{align*}
  Aggregate representations~$\calM_E$,~$\calP_E$, and~$\calB_E$
  for election~$E$ are as follows:
  \begin{align*}
    \small
    \calM_E = 
    \kbordermatrix{ & a & b & c   \\
    a\!\!\!\! &               - & 2 & 2\\
    b\!\!\!\! &               2 & - & 3\\
    c\!\!\!\! &               2 & 1 & -\\
    },&&
         \small
          \calP_E = 
         \kbordermatrix{ & a & b & c   \\
    1\!\!\!\! &               1 & 2 & 1\\
    2\!\!\!\! &               2 & 1 & 1\\
    3\!\!\!\! &               1 & 1 & 2\\
    },&&
%
%
         \small
         \calB_E =
         \kbordermatrix{ & a & b & c   \\
    &               4 & 5 & 3\\
    }.
  \end{align*}
\end{example}

By a \textit{realization} of an aggregated representation, we refer to an election that has a given aggregated representation.

\subsection{Positionwise Distance}

The first nonisomorphic distance that we will discuss is based on analyzing how frequently the candidates are ranked at particular positions, and we call it the \emph{positionwise distance}. 
(This distance is based on the earth mover's distance~(EMD) introduced in~\Cref{ch:preliminaries}, and on position matrices.)
The definition is as follows.


%

\begin{definition}\label{def:poswise} 
  Let~$E_1 = (C_1,V_1)$ and~$E_2 = (C_2,V_2)$ be two elections such
  that~$|C_1| = |C_2|$.  For a bijection~$\delta  \colon C_1 \rightarrow C_2$, we define~$\delta$-$\POS(E_1,E_2) = \sum_{c \in
    C_1} \EMD(\calP_E(c_1), \calP_E(c_2))$.
  The \emph{positionwise distance} between elections~$E_1$ 
  and~$E_2$,~$\POS(E_1, E_2)$, is the minimum of the~$\delta$-$\POS(E_1,E_2)$ values, taken over~$\delta$.
\end{definition}

We use earth mover's distance in~\Cref{def:poswise} because it captures the
idea that being ranked on the top position is more similar to being
ranked on the second position than to being ranked on the bottom one.
Alternately, instead of using EMD, one can use, e.g., the~$\ell_1$ distance. By~$\POS$ we refer to EMD-positionwise distances (which we treat as the default variant) and by~$\LPOS$ we refer to the~$\ell_1$-positionwise distance, the variant of the distance where we replace EMD with~$\ell_1$.

\begin{example}\label{ex:poswise}
  Consider two elections,~$E_1$ and~$E_2$, over candidate 
  sets~$C_1 = \{a,b,c\}$ and~$C_2 = \{x,y,z\}$. Election~$E_1$ contains
  voters~$v_1, v_2, v_3$ and election~$E_2$ contains 
  voters~$u_1, u_2, u_3$:
  \begin{align*}
    v_1 \colon & a \pref b \pref c, & v_2 \colon & b \pref a \pref c, & v_3 \colon & b \pref c \pref a, \\
    u_1 \colon & x \pref y \pref z, & u_2 \colon & z \pref x \pref y, & u_3 \colon & y \pref x \pref z.                             \end{align*}

\noindent
   The vectors (i.e., columns in the position matrix) associated with each of our candidates are as follows:
  \begin{align*}
    \calP_{E_1}(a)& = (1, 1, 1), &
    \calP_{E_1}(b)& = (2, 1, 0), &
    \calP_{E_1}(c)& = (0, 1, 2),\\
    \calP_{E_2}(x)& = (1, 2, 0), &
    \calP_{E_2}(y)& = (1, 1, 1), &
    \calP_{E_2}(z)& = (1, 0, 2).
  \end{align*}
  We see that~$\EMD(\calP_{E_1}(a),\calP_{E_2}(y)) = 0$,~$\EMD(\calP_{E_1}(b),\calP_{E_2}(x)) = 1$ because to transform~$\calP_{E_1}(b)$ into~$\calP_{E_2}(x)$, we need to move value~$1$
  from the first position to the second one (so we multiply~$1$ by~$1$), and~$\EMD(\calP_{E_1}(c),\calP_{E_2}(z)) = 1$. Thus for~$\delta(a)=~y$,~$\delta(b)=~x$, and~$\delta(c) = z$ we have~$\delta$-$\POS(E_1,E_2) = 2$ and, in fact,~$\POS(E_1,E_2)=2$.
\end{example}

The positionwise distance is not a metric, because the distance between two nonisomorphic elections can be zero. However, it is a pseudometric.

\begin{proposition}
  The positionwise distance is a pseudometric.
\end{proposition}

\begin{proof}
  We show that the positionwise distance satisfies the triangle inequality
  (the other requirements for being a pseudometric are easy to verify).
%
%
%
  Consider three elections with candidate sets of equal size,~$E_1 =
  (C_1,V_1)$,~$E_2 = (C_2,V_2)$, and~$E_3 = (C_3,V_3)$. Let~$\delta$
  and~$\sigma$ be the permutations of the candidates that minimize the distances between~$E_1$ and~$E_2$ and between~$E_2$ and~$E_3$, respectively. We have that:
  \begin{align*}
    \POS(E_1, E_3) &= \textstyle \sum_{c \in C_1} \EMD( \calP_{E_1}(c), \calP_{E_3}(\sigma(\delta(c))) \\
    &\textstyle \leq \sum_{c \in C_1} \EMD( \calP_{E_1}(c), \calP_{E_2}(\delta(c)) \\
    &\textstyle + \sum_{c \in C_2} \EMD( \calP_{E_1}(\delta(c)), \calP_{E_2}(\sigma(\delta(c))) \\
    &\textstyle =\POS(E_1, E_2) + \POS(E_2, E_3) \text{.}
  \end{align*}
  The first inequality follows from the definition of the positionwise
  distance, the second one---from the fact that EMD is a metric.
\end{proof}

%






One of the advantages of the positionwise distance is the fact that it can be computed in polynomial-time.

\begin{proposition}
  There exists a polynomial-time algorithm for computing the positionwise
  distance. 
\end{proposition}

\begin{proof}
  Let~$E_1 = (C_1,V_1)$
  and~$E_2 = (C_2,V_2)$ be two elections where~$|C_1| = |C_2|$.  The
  value of~$\POS(E_1,E_2)$ is equal to the minimum-cost matching in
  the bipartite graph whose vertex set is~$C_1 \cup C_2$ and which has
  the following edges: For each~$c_1 \in C_1$ and each~$c_2 \in C_2$
  there is an edge with the cost equal to the EMD between~$c_1$'s and~$c_2$'s 
  candidate distribution vectors (these weights
  can be computed independently for each pair of candidates). Such
  minimum-cost matchings can be computed in polynomial time.
\end{proof}

While computing a position matrix of an election is straightforward, the reverse direction is less clear. 
We observe that each~$m\times m$ position matrix has a corresponding
$m$-candidate election with at most~$m^2-2m+2$ distinct preference
orders. This was shown by \citet[Theorem 7]{leep1999marriage} (they
speak of ``semi-magic squares'' and not ``position matrices'' and
show a decomposition of a matrix into permutation matrices, which
correspond to votes in our setting). In other words, given a position matrix we can compute its realization in polynomial-time.


\begin{observation}
  Given a position matrix~$X$, one can compute
  in~$O(m^{4.5})$ time  an election~$E$ that contains at most~$m^2-2m+2$
  different votes such that~$\calP(E) =
  X$. 
\end{observation}

\subsection{Pairwise Distance}

Next, we define the \emph{pairwise distance}, which is inspired by
the class of Condorcet-consistent
voting rules and relies on analyzing the results of head-to-head
majority contests between the candidates.

\begin{definition}
  Let~$E_1 = (C_1,V_1)$ and~$E_2 = (C_2,V_2)$ be two elections such
  that~$|C_1| = |C_2|$. For a bijection~$\delta \colon C_1 \rightarrow C_2$, we define~$\delta$-$\PAIR(E_1,E_2) =
  \sum_{(c,d) \in C_1 \times C_1} \big| \calM_{E_1}(c,d) -
  \calM_{E_2}(\delta(c),\delta(d))\big|$.  The \emph{pairwise distance}
  between elections~$E_1$ and~$E_2$,~$\PAIR(E_1, E_2)$, is the
  minimum value of the~$\delta$-$\PAIR(E_1,E_2)$ values, taken over~$\delta$.
\end{definition}

\begin{example}
  Let us consider the two elections from~\Cref{ex:poswise}. Weighted majority relations look as follows:
  \begin{align*}
  \calM_{E_1} = \kbordermatrix{
    & a & b & c  \\
    a & - & 1 & 2  \\
    b & 2 & - & 3  \\
    c & 1 & 0 & - 
  }\qquad
  \calM_{E_2} = \kbordermatrix{
    & x & y & z  \\
    x & - & 2 & 2 \\
    y & 1 & - & 2  \\
    z & 1 & 1 & - 
  }
  \end{align*}
  For~$\delta(a) = y$,~$\delta(b) = x$, and~$\delta(c) = z$, the~$\delta$-$\PAIR(E_1,E_2) = 2$,
  and  this is also the value of~$\PAIR(E_1, E_2)$.
\end{example}

Note that using EMD for the pairwise distance would not be very useful, because each value in the matrix is in some sense independent of the values surrounding it.

Similarly to the positionwise distance, pairwise distance is a pseudometric.

\begin{proposition}
    The pairwise distance is a pseudometric.
\end{proposition}
\begin{proof}
  Clearly, the pairwise distance is symmetric, and for each election~$E$ it holds that~$\PAIR(E, E) = 0$. The triangle inequality
  follows by the same reasoning as in the case of the positionwise
  distance. In particular, we define~$\delta$ and~$\sigma$
  analogously as in that proof. Then:

\begin{align*}
\PAIR(E_1, E_3) &\leq   \sum_{(c,d) \in C_1 \times C_1} \big| \calM_{E_1}(c,d) -
  \calM_{E_3}(\sigma(\delta(c)),\sigma(\delta(d))\big| \\
& \leq \sum_{(c,d) \in C_1 \times C_1} \big| \calM_{E_2}(\delta(c),\delta(d)) -
  \calM_{E_3}(\sigma(\delta(c)),\sigma(\delta(d))\big|\\
&+ \sum_{(c,d) \in C_1 \times C_1} \big| \calM_{E_1}(c,d) -
  \calM_{E_2}(\delta(c),\delta(d)\big|  \\
&=\PAIR(E_1, E_2) + \PAIR(E_2, E_3) \text{.}
\end{align*} 
This completes the proof.
\end{proof}

Both the positionwise distance and the pairwise distance
satisfy our
minimal requirements; they both are pseudometrics
defined to be neutral/anonymous.  Yet, we can compute the positionwise
distances in polynomial-time, but the pairwise distance is intractable
(indeed, it is similar to the~$\np$-complete \textsc{Approximate Graph
  Isomorphism} 
problem~\citep{arv-koe-kuh-vas:c:approximate-graph-isomorphism,gro-rat-woe:c:approximate-isomorphism}).

\begin{proposition}[\cite{szu-fal-sko-sli-tal:c:map}]
  The decision variant of the problem of computing the pairwise distance is~$\np$-complete.
\end{proposition}

Nonetheless, we can compute the pairwise distance by formulating it as
an integer linear program. 
In practice, this allows us to compute distances
between elections of up to around 20 candidates.

\begin{proposition}
  There is an ILP for~$\PAIR$.
\end{proposition}

\begin{proof}
  Let~$E = (C,V)$ and~$E' = (C',V')$ be the elections we wish to
  compute the distance for, with~$C = \{c_1, \dots, c_m\}$,~$C' = \{c'_1, \dots, c'_m\}$,~$V = \{v_1, \dots, v_n\}$, and~$V' = \{v'_1, \dots, v'_n\}$.
For each~$i, i' \in [m]$, we define a binary variable~$M_{i, i'}$ with the intention that value~$1$ indicates that
  candidate~$c_i$ is matched to candidate~$c'_{i'}$. 
For each~$i,i',j,j' \in [m], i \neq j, i' \neq j'$, we define a binary
  variable~$P_{i,i',j,j'}$ with the intention that~$P_{i,i',j,j'} = M_{i,i'}\cdot M_{j,j'}$.
  We introduce the following constraints:
  \begin{align}
    \label{pair_ilp:1}
    &\textstyle\sum_{i'\in[m]} M_{i,i'} = 1, \forall i\in[m]; \\
    \label{pair_ilp:2}
    &\textstyle\sum_{i\in[m]} M_{i,i'} = 1, \forall i'\in[m]; \text{ } \\
    \label{pair_ilp:3}
    &\textstyle\sum_{\substack{i',j'\in[m] \\ i \neq j, i' \neq j'}} P_{i,i',j,j'} = 1, \forall i,j\in[m]; \\
    \label{pair_ilp:4}
    &\textstyle\sum_{\substack{i',j'\in[m] \\ i \neq j, i' \neq j'}} P_{i,i',j,j'} = 1, \forall i',j'\in[m]; \text{ } \\
    \label{pair_ilp:5}
    &P_{i,i',j,j'} \leq M_{i,j}, \forall \substack{i,i',j,j'\in[m] \\ i \neq j, i' \neq j'}; \\
    \label{pair_ilp:6}
    &P_{i,i',j,j'} \leq M_{k,l}, \forall \substack{i,i',j,j'\in[m] \\ i \neq j, i' \neq j'}.
  \end{align}
    Constraints~\eqref{pair_ilp:1} and~\eqref{pair_ilp:2} ensure that variables~$M_{i,j}$ describe matchings between the candidates. Constraints~\eqref{pair_ilp:3}--\eqref{pair_ilp:6}
  implement the semantics of the~$P_{i,i',j,j'}$ variables (the
  former two ensure that there is one-to-one matching between pairs of candidates; the latter two ensure connection between the~$P_{i,i',j,j'}$ variables and the~$M_{i,i'}$ and~$M_{j,j'}$ variables).
  
  The optimization goal is to minimize:
  \[
  \textstyle\sum_{ \substack{i,i',j,j'\in[m] \\ i \neq j, i' \neq j'} } P_{i,i',j,j'}  \cdot
  |\calM_{E}(i,j) - \calM_{E'}(i',j')|.
  \]
  Values~$\calM_{E}(i,j)$ and~$\calM_{E'}(i',j')$ are precomputed.
\end{proof}

  




Unlike for the positionwise distance, for the pairwise distance it is hard to recover an election with a given weighted majority relation.

\begin{theorem}[\cite{boe-fal-nie-szu-was:c:understanding}]
  Given an~$m \times m$ matrix~$M$, it is~$\np$-complete to decide if
  there is an election~$E$ with~$\mathcal{M}_E=M$.
\end{theorem}
\subsection{Bordawise Distance}

We introduce one more metric, similar in spirit to the positionwise and
pairwise ones, but defined on top of the election's Borda score vectors.  Given
two equal-sized elections~$E$ and~$E'$, their Bordawise
distance is:
\[
  \BOR(E,E') = \emd( \sort(\calB_E), \sort(\calB_{E'}) ),
\]
where for a vector~$x$,~$\sort(x)$ means the vector obtained from~$x$ by
sorting it in the nonincreasing order.
The Bordawise metric is defined to be as simple as possible, while trying
to still be meaningful. For example, 
sorting the score vectors ensures that two isomorphic elections are at
distance zero and removes the use of an explicit matching between the
candidates.

%

\begin{example}
  The Borda score vectors of the elections from~\Cref{ex:poswise} are
  \begin{align*}
  \calB_{E_1} =
         \kbordermatrix{ & a & b & c   \\
    &               3 & 5 & 1\\
    }, \\
  \calB_{E_2} =
         \kbordermatrix{ & x & y & z   \\
    &               4 & 3 & 2\\
    },
  \end{align*}
  and the distance between them is~$\emd\big((5,3,1),(4,3,2) \big) = 2$.
\end{example}

\begin{observation}
  The Bordawise distance is a pseudometric.
\end{observation}

EMD is a distance itself, and the Bordawise distance simply computes the EMD between two Borda score vectors, so it must satisfy the triangle inequality and symmetry as well, and the distance between two identical vectors is zero. However, it might be the case that different elections will produce the same Borda score vector, so there will be two different elections at distance zero. Therefore, the Bordawise distance is a pseudometric.


\begin{observation}
  There is a polynomial-time algorithm for computing the Bordawise distance.
\end{observation}

Converting an election into a Borda score vector requires polynomial time, and computing EMD between two vectors uses polynomial time as well.


Unfortunately, for Borda score vectors (as for weighted majority
relations) it is hard to decide whether there exists a realization.

\begin{theorem}
Given a vector~$x$ of nonnegative integers, it is~$\np$-complete to
decide if there is an election~$E$ with~$\calB_E = x$.
\end{theorem}

\begin{proof}
  \citet{DBLP:journals/scheduling/YuHL04} showed that given a sequence
  of positive integers~$a_1, \dots, a_m$ such that~$a_1 \geq a_2 \geq \cdots \geq a_m$,~$\sum_{i=1}^m a_i=m(m+1)$, and such that for each~$i \in [m]$ we have~$2\leq a_i \leq 2m$, it is~$\np$-complete to decide if there are two permutations~$\phi, \phi' \in S_m$ such that for all~$i \in [m]$ it holds that~$\phi(i)+\phi'(i)=a_i$. We reduce this problem to the one from the
  statement of the theorem by forming a vector~$x = (a_1-2, \ldots, a_m-2)$.

  If there are two permutations~$\phi$ and~$\phi'$ that satisfy the
  conditions of \citeauthor{DBLP:journals/scheduling/YuHL04}'s
  problem, then we form a two-voter election~$E = (C,V)$ as follows:
  We let~$C = [m]$ and we form two votes,~$v$ and~$v'$.  For each
  candidate~$i \in C$, the first (the second) voter ranks~$i$ on
  position~$m-\phi(i)+1$ ($m-\phi'(i)+1$); note that the produced
  votes rank exactly one candidate in each position because~$\phi$ and~$\phi'$ are permutations. Then, the Borda score of each~$i \in C$ is~$\phi(i)-1+\phi'(i)-1=a_i-2$.

  For the other direction, assume that there is an election~$E = (C,V)$ with Borda score vector~$x$. 
  Then,~$E$ must contain exactly two voters because otherwise the sum of the candidates'
  scores would either be too large or too small. W.l.o.g., we assume
  that~$C = [m]$ and that each candidate~$i \in C$ has Borda score~$a_i-2$. 
  Let~$v$ and~$v'$ be the two votes in~$E$.  We form a
  permutation~$\phi$ so that for each~$i \in C$ we have~$\phi(i) = m-\pos_v(i)+1$, We form~$\phi'$ analogously, but using~$v'$ instead of~$v$. It follows that for each~$i \in [m]$ we have~$\phi(i)+\phi'(i)=(a_i-2)+2=a_i$. This completes the proof.
\end{proof}

\subsection{Maps of Elections Using Nonisomorphic Distances}

\begin{figure}[t]
    \centering
    
    \begin{subfigure}[b]{0.49\textwidth}
        \centering
        \includegraphics[width=5.6cm, trim={0.2cm 0.2cm 0.2cm 0.2cm}, clip]{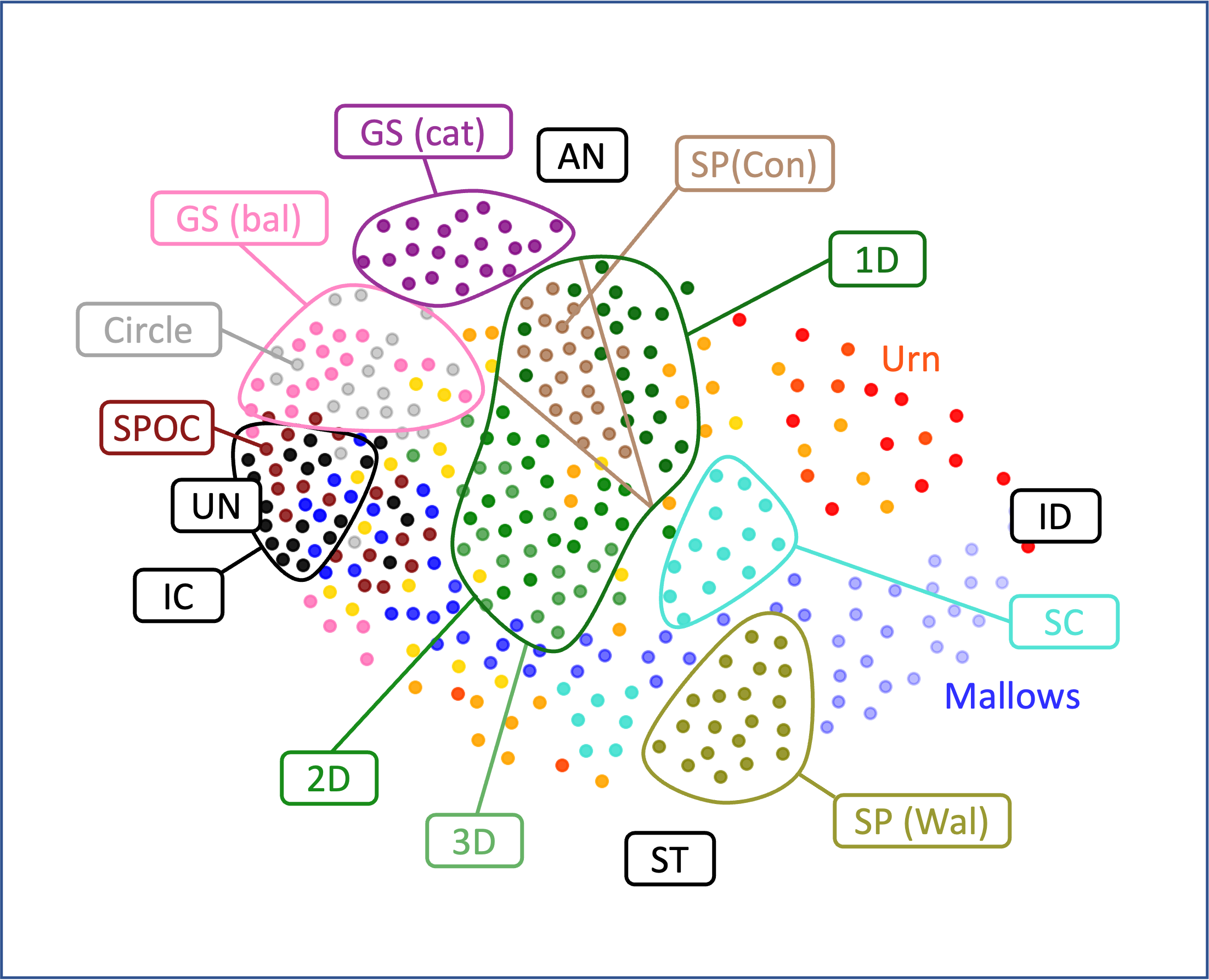}
        \caption{EMD-positionwise}
    \end{subfigure}
    \begin{subfigure}[b]{0.49\textwidth}
        \centering
        \includegraphics[width=5.6cm, trim={0.2cm 0.2cm 0.2cm 0.2cm}, clip]{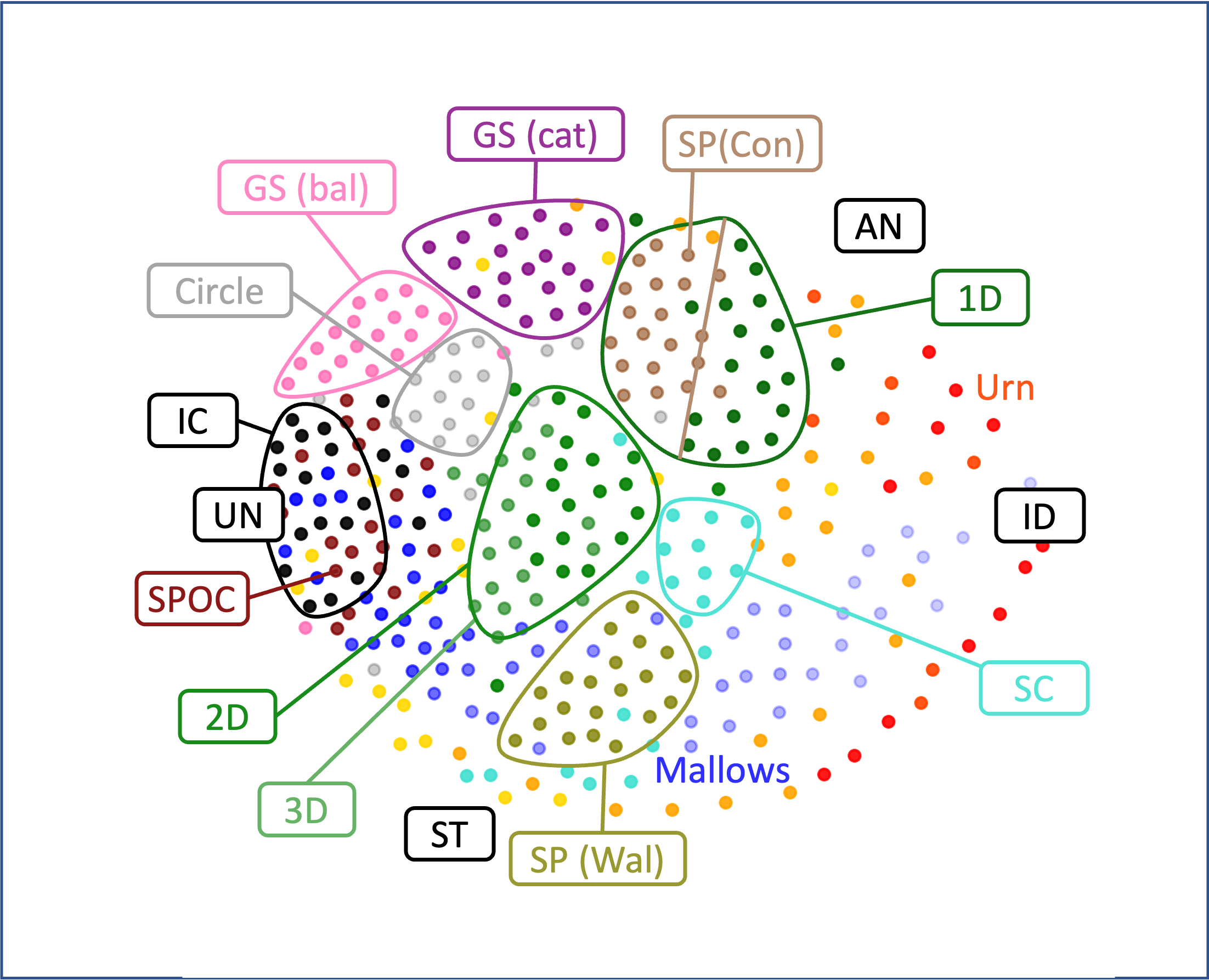}
        \caption{$\ell_1$-positionwise}
    \end{subfigure}
    
    \begin{subfigure}[b]{0.49\textwidth}
        \centering
        \includegraphics[width=5.6cm, trim={0.2cm 0.2cm 0.2cm 0.2cm}, clip]{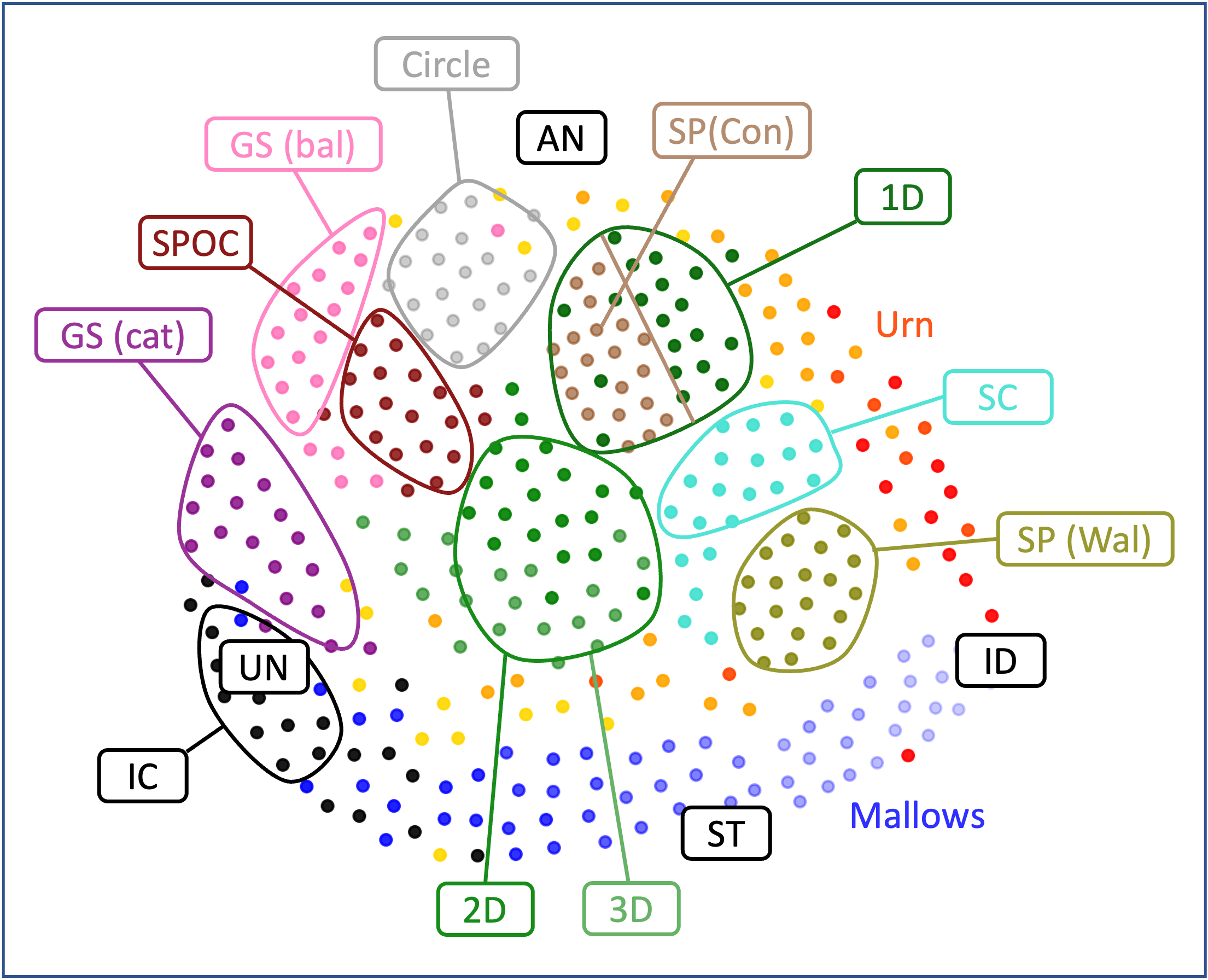}
        \caption{Swap}
    \end{subfigure}

    \caption{Comparison of maps of elections.}
    \label{fig:main_maps_noniso_1}
\end{figure}

This section is analogous to \Cref{sec:iso_maps_of_elections}, but this time we focus on the maps based on nonisomorphic distances and compare them with those for the swap distance.

We use the same elections as before. (Details of the dataset were described in \Cref{tab:testbed}). Just as a reminder, all elections consist of~$10$ candidates and~$50$ voters. However, the embedding algorithm differs from the one used in \Cref{sec:iso_dist}. In this chapter we decided to use the algorithm of
\citet{fruchterman1991graph} to place the points\footnote{More details about the differences between embeddings will be presented in \Cref{ch:applications}}.

We will start by focusing on maps based on the positionwise distance. In \Cref{fig:main_maps_noniso_1} we present two maps for the EMD- and~$\ell_1$- variants of the positionwise distance, and one map for the isomorphic swap distance, which will serve as a reference point. As we can see, EMD- and~$\ell_1$- variants are quite similar, and at first glance it is hard to say which one is better. By being better, we mean that the map is more similar to the one produced based on the swap distance.

There are three significant differences between the positionwise variants and the swap one. First, let us have a look at group-separable elections. Under the swap distance, balanced group-separable elections are closer to~$\AN$ than the caterpillar group-separable elections, while for the positionwise variants, the caterpillar elections are closer to~$\AN$ than the balanced ones. Second, for positionwise maps,~$\ST$ appears to be one of the extreme points, while for the swap distance map, the space seems to span between~$\AN$,~$\ID$, and~$\UN$, while~$\ST$ is not that crucial. Third, the swap distance is far better at distinguishing between SPOC and impartial culture elections.

\begin{figure}[t]
    \centering
    
    \begin{subfigure}[b]{0.49\textwidth}
        \centering
        \includegraphics[width=5.2cm, trim={0.2cm 0.2cm 0.2cm 0.2cm}, clip]{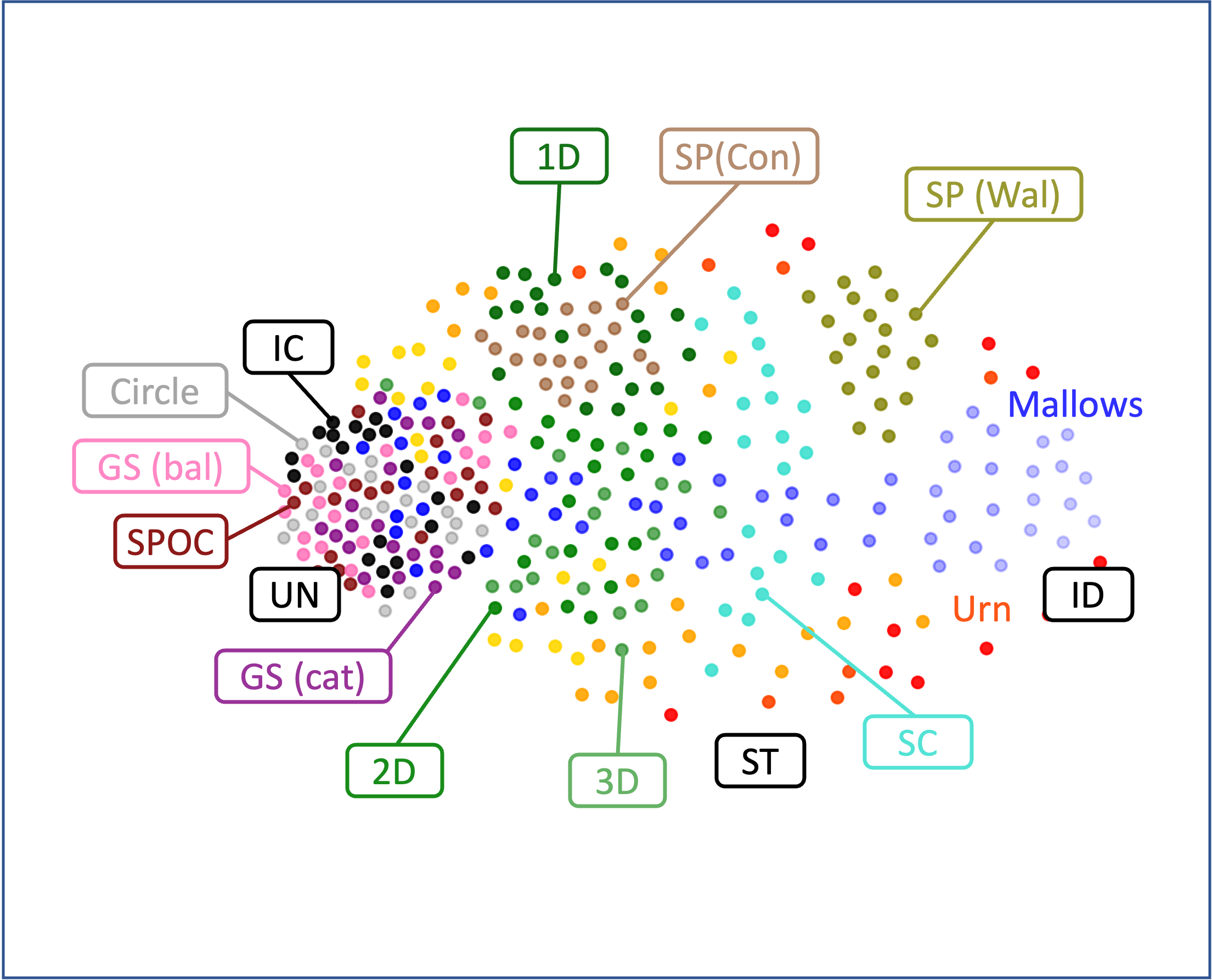}
        \caption{$\ell_1$-pairwise}
    \end{subfigure}
    \begin{subfigure}[b]{0.49\textwidth}
        \centering
        \includegraphics[width=5.2cm, trim={0.2cm 0.2cm 0.2cm 0.2cm}, clip]{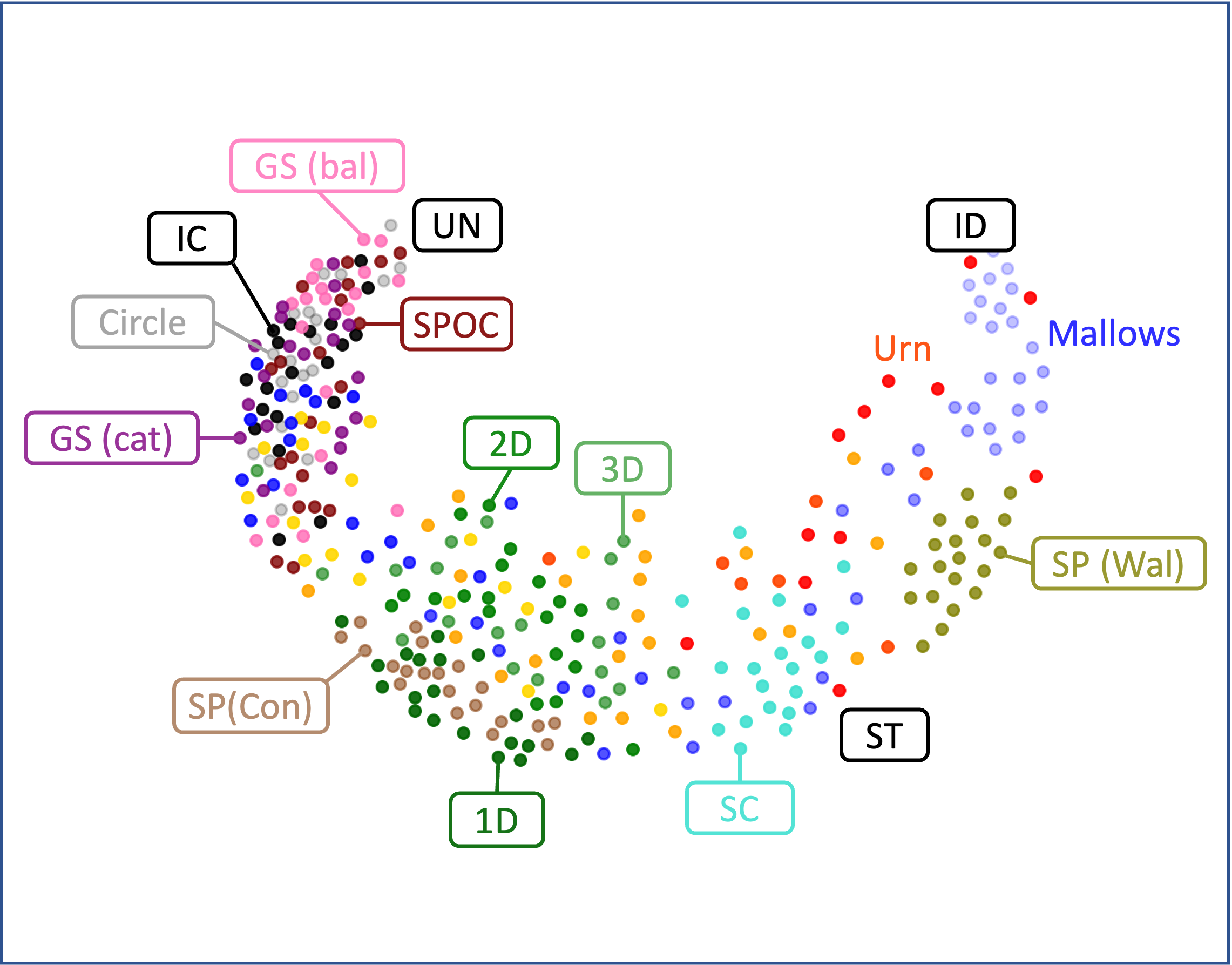}
        \caption{EMD-Bordawise}
    \end{subfigure}

    \caption{Comparison of maps of elections.}
    \label{fig:main_maps_noniso_2}
\end{figure}
\vspace{-0.45cm}

Next, we move to two more maps, i.e., the maps based on the pairwise and Bordawise distances. The results are presented in \Cref{fig:main_maps_noniso_2}. 
Both the Borda score vectors and the weighted majority relations do not distinguish between uniformity and antagonism elections (i.e., under both the Bordawise and pairwise distances, the distance between $\UN$ and $\AN$ is zero). Unfortunately, for the Bordawise distance, the situation is even more drastic. If the Borda score of all the candidates is more or less equal, then such elections will be almost identical under the Bordawise distance.
Note that in the maps based on the swap or positionwise distances, in all elections that lie in the upper left part of the map (somewhere between~$\UN$ and~$\AN$), all the candidates (on average) have very similar Borda scores. As to the pairwise distance, in spite of the fact that the whole Borda balance area is collapsing onto~$\UN$, the rest of the map looks relatively fine, that is, it roughly resembles the map based on the swap distance. 
In \Cref{fig:main_matrix_noniso} we present the average distances between elections from each pair of statistical cultures (we omitted the urn and Mallows elections because they are parametrized and comparing the average value would be meaningless). Each cell gives the average distance (according to a given metric) between the elections generated from respective models. All values are normalized by the largest possible distance under the given metric, i.e., the distance between~$\ID$ and~$\UN$ (we will return to the problem of calculating the largest possible distance, for a given metric, in~\Cref{subsec:analysis_compass}).

\vspace{0.2cm}
\begin{conclusionbox}
The main conclusion of~\Cref{sec:non_iso_dist} is the following.
Maps based on the positionwise distances show a lot of similarities to the map based on the swap distance, while, at the same time, being much easier to generate, due to computational complexity of respective distances.
\end{conclusionbox}

\begin{figure}[]
   
    \begin{subfigure}[b]{0.49\textwidth}
        \centering
        \includegraphics[width=6.6cm]{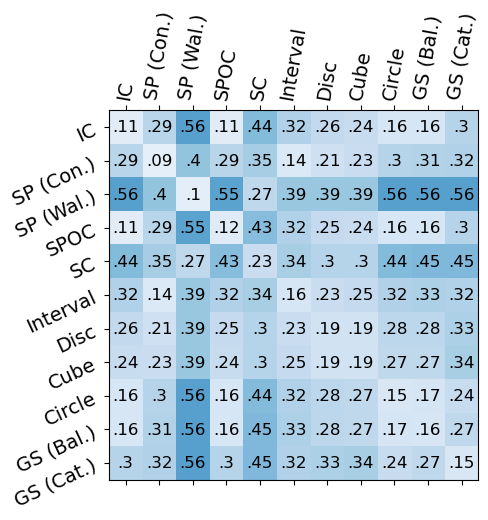}
        \caption{EMD-positionwise}
    \end{subfigure}
    \begin{subfigure}[b]{0.49\textwidth}
        \centering
        \includegraphics[width=6.6cm]{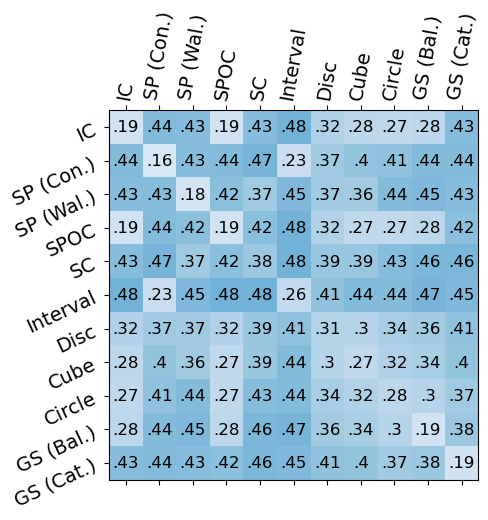}
        \caption{$\ell_1$-positionwise}
    \end{subfigure}
    
    \begin{subfigure}[b]{0.49\textwidth}
        \centering
        \includegraphics[width=6.6cm]{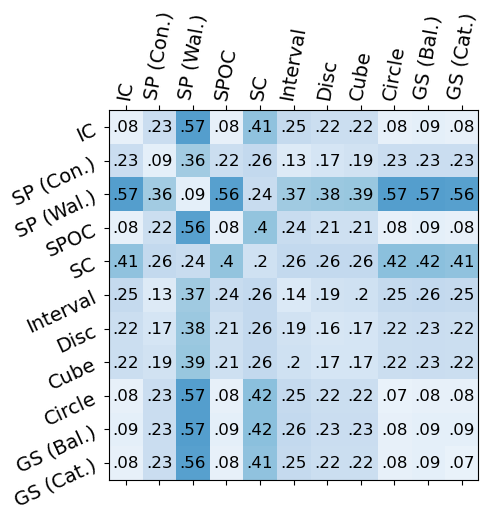}
        \caption{$\ell_1$-pairwise}
    \end{subfigure}
    \begin{subfigure}[b]{0.49\textwidth}
        \centering
        \includegraphics[width=6.6cm]{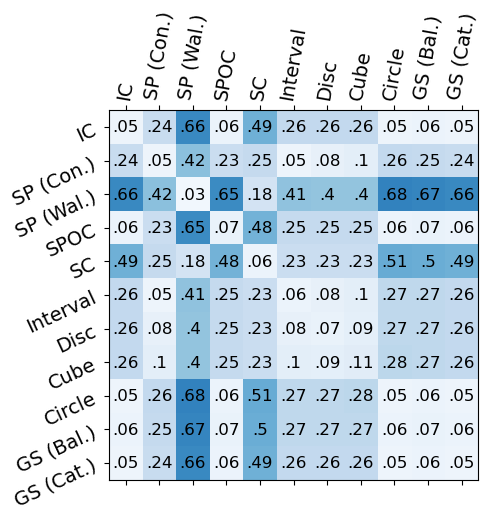}
        \caption{EMD-Bordawise}
    \end{subfigure}
    
\caption{The average distances between the elections from given cultures.}
\label{fig:main_matrix_noniso}
\end{figure}

\vspace{-0.2cm}
\section{Comparison}
In this section we compare nonisomorphic and isomorphic distances with each other. We start with an analysis of the compass. The relation between compass elections differs depending on the distance chosen. Next, we focus on the maps of elections and compare maps based on nonisomorphic distances with the map based on the swap distance---which we treat as an ideal one. Finally, we discuss the correlation between metrics, and conclude by discussing equivalence classes of our distances. 

\subsection{Analysis of the Compass}\label{subsec:analysis_compass}

For isomorphic distances, we can easily create instances of identity and antagonism elections. For the uniformity and stratification ones, we need exponentially many voters with respect to the number of candidates. Hence we usually use their approximations, as we described it in \Cref{ch:stat_cult}. Luckily, for nonisomorphic distances such as the positionwise, pairwise, and Bordawise ones, we can represent the compass perfectly, using the aggregate representations. We are going to describe all four characteristic points, and their aggregate representations, for each of the metrics described in the previous section. Moreover, we are going to present the distances between these characteristic points. We focus on the following variants: EMD-positionwise, $\ell_1$-positionwise, $\ell_1$-pairwise, and EMD-Bordawise. All our nonisomorphic distances are independent of the number of voters, and can be computed between elections with different numbers of voters. From isomorphic distances we study the swap and discrete ones---for them, to compare two elections, we need exactly the same numbers of voters and candidates in both elections.

The proofs of all the propositions from this section are in the~\Cref{apdx:supp}, due to their tediously technical character and limited interest.

\subsubsection{EMD-positionwise}

We start with the EMD-positionwise distance. Sometimes instead of using position matrix, it is more convenient to use its normalized variant, which we define as follows.

Consider an election $E = (C,V)$ with $C = \{c_1, \ldots, c_m\}$ and
$V = (v_1, \ldots, v_n)$. For each candidate $c_j$ and position
$i \in [m]$, we define $\#\freq_E(c_j,i)$ to be the fraction of the
votes from $V$ that 
rank $c_j$ in position $i$. We define the column vector
$\#\freq_E(c_j)$ to be $(\#\freq_E(c_j,1), \ldots, \#\freq_E(c_j,m))$
and matrix $\#\freq(E)$ to consist of vectors
$\#\freq_E(c_1), \ldots, \#\freq_E(c_m)$. We refer to $\#\freq(E)$ as the {\em frequency
  matrix} of election~$E$.  Frequency matrices are bistochastic, i.e.,
their entries are nonnegative and each of their rows and columns sums
up to one. Note that if we take position matrix and divide all its entries by the number of voters we immediately obtain the frequency matrix of the same election.\footnote{When we use frequency matrices instead of position matrices all the distances between such matrices are scaled by the factor of $\nicefrac{1}{n}$.}

Two most important matrices are the identity matrix~$\ID$ and the
uniformity matrix~$\UN$. The identity matrix corresponds to elections where each voter has the
same preference order, i.e., there is a common ordering of the
candidates from the most to the least desirable one. For $\ID$, we have ones on the diagonal and zeros elsewhere, as presented below.
\begin{align*}
 \ID_m = \begin{bmatrix}
    1 & 0 & \cdots & 0\\
    0 & 1 & \cdots & 0\\
    \vdots & \vdots & \ddots & \vdots \\
    0 & 0 & \cdots & 1
 \end{bmatrix}.\!\!
\end{align*}

\noindent

In contrast, the uniformity matrix captures elections where each candidate is
ranked on each position equally often, i.e., where, in aggregate, all
the candidates are viewed as equally good. Uniformity elections are
quite similar to the IC ones
and, in the limit, indistinguishable from them. 
Yet, for a fixed number of voters, typically IC elections are at some
(small) positionwise distance from uniformity. For $\UN$, each entry is equal to~$\nicefrac{1}{m}$.

\begin{align*}
 \UN_m = \begin{bmatrix}
    \nicefrac{1}{m} & \nicefrac{1}{m} & \cdots & \nicefrac{1}{m}\\
    \nicefrac{1}{m} & \nicefrac{1}{m} & \cdots & \nicefrac{1}{m}\\
    \vdots & \vdots & \ddots & \vdots \\
    \nicefrac{1}{m} & \nicefrac{1}{m} & \cdots & \nicefrac{1}{m}
 \end{bmatrix}.                   
\end{align*}

\noindent

\medskip

The next matrix, \emph{stratification}, is defined as
follows (we assume that~$m$ is even):
\[
  \ST_m = \begin{bmatrix}
    \UN_{\nicefrac{m}{2}} & 0 \\
    0 & \UN_{\nicefrac{m}{2}}
  \end{bmatrix}.
\]
Stratification matrices correspond to elections where the voters agree
that half of the candidates are 
more desirable than the other half, but, in aggregate, are unable to
distinguish between the qualities of the candidates in each group.
\medskip

For the final matrix, we need one more piece of notation.
Let~$\rID_m$ be the matrix obtained by reversing the order of the columns of the identity matrix~$\ID_m$. We define the \emph{antagonism} matrix,~$\AN_m$, to be
$
\textstyle
\nicefrac{1}{2} \ID_m+\nicefrac{1}{2} \rID_m.
$

\[
 \textstyle
 \AN_m = \frac{1}{2}\begin{bmatrix}
    1 & 0 & \cdots & 0\\
    0 & 1 & \cdots & 0\\
    \vdots & \vdots & \ddots & \vdots \\
    0 & 0 & \cdots & 1
 \end{bmatrix}
 +
 \frac{1}{2}\begin{bmatrix}
    0 & 0 & \cdots & 1\\
    \vdots & \vdots & \iddots & \vdots \\
    0 & 1 & \cdots & 0\\
    1 & 0 & \cdots & 0
 \end{bmatrix}.
\]

Such matrices are generated, for example, by
elections where half of the voters rank the
candidates in one way, and half of the voters rank them in the
opposite one, so there is a clear conflict. 
In some sense, stratification and antagonism are based on
similar premises. Under stratification, the group of candidates is
partitioned into halves with different properties, whereas in
antagonism (for the case where half of the voters rank the candidates in the same order)  the voters are partitioned. However, the nature of the
partitioning is, naturally, quite different.\medskip



\begin{restatable}{proposition}{distemdpos} 
\label{pr:emd-pos-dist}
    If~$m$ is divisible by~$4$, then it holds that:
  \begin{enumerate}
  \item~$\POS(\ID_m,\UN_m) = \frac{1}{3}(m^2-1)$,
  \item~$\POS(\ID_m,\AN_m) = \POS(\UN_m,\ST_m) = \frac{m^2}{4}$,
  \item
   ~$\POS(\ID_m,\ST_m) = \POS(\UN_m,\AN_m)  = 
    \frac{2}{3}(\frac{m^2}{4}-1)$,
  \item~$\POS(\AN_m,\ST_m) = \frac{13}{48} m^2 - \frac{1}{3}$.
  \end{enumerate}
\end{restatable}

What is worth emphasizing is the fact that~$\POS(\ID,\UN)$ is the largest possible distance in the whole space; more precisely, there do not exist any other pair of elections that are at larger distance than~$\ID$ and~$\UN$.

The same is true for all other distances that we will analyze within this section. For more details, see the work of \cite{boe-fal-nie-szu-was:c:understanding}). 

\begin{theorem}[\cite{boe-fal-nie-szu-was:c:understanding}]
  For each two elections~$X$ and~$Y$, each over~$m$ candidates, it holds that~$\POS(X,Y) \leq \POS(\ID_m,\UN_m)$.
\end{theorem}


To normalize the distances from \Cref{pr:emd-pos-dist}, we divide them by~$D(m) = \POS(\ID_m, \UN_m)$.
For each two matrices~$X$ and~$Y$ among our four compass matrices, we let~$\POS(X,Y) = \lim_{m \rightarrow
  \infty}\nicefrac{\POS(X_m,Y_m)}{D(m)}$. A simple computation shows
the following (see also the drawing on the right side; we sometimes omit
the subscript~$m$ for simplicity):

\begin{minipage}[b]{0.45\textwidth}
  \centering
  \begin{align*}
  \POS(\ID,\UN) &= 1,\\
  \POS(\ID,\AN) &= \POS(\UN,\ST) = \nicefrac{3}{4},\\
  \POS(\AN,\ST) &= \nicefrac{13}{16},\\
  \POS(\ID,\ST) &= \POS(\UN,\AN) = \nicefrac{1}{2}.\\
\end{align*}
\end{minipage}
\begin{minipage}[b]{0.5\textwidth}
  \centering
    \begin{tikzpicture}[xscale=0.5, yscale=0.5]
        \clip (-0.5, -3) rectangle (9, 3.5);
        \drawunabove{0}{0}
        \drawidabove{8}{0}
        \drawanaboveright{3}{2}
        \drawstbelowleft{5}{-2}
        \draw (1,0.5) -- (8,0.5);
        \draw (3,0.5) node[anchor=south] {$1$};
        \draw (1,1) -- (3,2.5);
        \draw (2,1.75) node[anchor=south] {$\frac{1}{2}$};
        \draw (4,2.5) -- (8,1);
        \draw (6,1.75) node[anchor=south] {$\frac{3}{4}$};
        \draw (4,2) -- (5,-1);
        \draw (4.7,0.75) node[anchor=south] {$\frac{13}{16}$};
        \draw (1,0) -- (5,-1.5);
        \draw (2,-0.5) node[anchor=north] {$\frac{3}{4}$};
        \draw (6,-1.5) -- (8,0);
        \draw (7.2,-0.5) node[anchor=north] {$\frac{1}{2}$};
      \end{tikzpicture}
\end{minipage}

\subsubsection{$\boldsymbol{\ell_1}$-positionwise}
For the~$\ell_1$-positionwise variant, all compass matrices are exactly the same as for the~$\EMD$-positionwise, so we move directly to computing distances between them.

\begin{restatable}{proposition}{distlonepos} 
    If $m$ is divisible by~$4$, then it holds that:
    \begin{enumerate}
        \item $\LPOS(\ID_m,\UN_m) = 2(m-1)$ 
        \item $\LPOS(\UN_m,\AN_m) = \LPOS(\AN_m,\ST_m) = \LPOS(\ID_m,\ST_m) = 2(m-2)$ 
        \item $\LPOS(\UN_m,\ST_m) = \LPOS(\ID_m,\AN_m) = m$
    \end{enumerate}
\end{restatable} 
 
As before, we normalize these distances by dividing them by the largest possible distance,~$\LPOS(\ID_m,\UN_m)$, and then compute the limits.
 
\medskip



 
 \begin{minipage}[b]{0.45\textwidth}
  \centering
  \begin{align*}
  \LPOS(\ID,\UN) &= \LPOS(\AN,\ST) \\
             &= \LPOS(\ID,\ST) \\
             &= \LPOS(\UN,\AN) = 1.\\
  \LPOS(\ID,\AN) &= \LPOS(\UN, \ST) = \nicefrac{1}{2},\\
\end{align*}
\end{minipage}
\begin{minipage}[b]{0.5\textwidth}
  \centering
    \begin{tikzpicture}[xscale=0.5, yscale=0.5]
        \clip (-0.5, -3) rectangle (9, 3.5);
        \drawunabove{0}{0}
        \drawidabove{8}{0}
        \drawanaboveleft{5}{2}
        \drawstbelowright{3}{-2}
        \draw (1,0.5) -- (8,0.5); 
        \draw (3.5,0.5) node[anchor=south] {$1$}; 
        \draw (1,1) -- (5,2.5); 
        \draw (2.5,1.5) node[anchor=south] {$1$}; 
        \draw (6,2.5) -- (8,1); 
        \draw (7.25,1.5) node[anchor=south] {$\frac{1}{2}$};  
        \draw (5,2) -- (4,-1); 
        \draw (5.2,0.75) node[anchor=south] {$1$}; 
        \draw (1,0) -- (3,-1.5); 
        \draw (1.5,-0.5) node[anchor=north] {$\frac{1}{2}$}; 
        \draw (4,-1.5) -- (8,0); 
        \draw (6.5,-0.5) node[anchor=north] {$1$}; 
      \end{tikzpicture}
\end{minipage}

If we compare the EMD and~$\ell_1$ variants, we will see that, for EMD,~$\POS(\ID,\UN)$ is dominating all other distances, while for~$\ell_1$,~$\POS(\AN,\ST)$, ~$\POS(\ID,\ST)$, and~$\POS(\UN,\AN)$ are almost as large as~$\POS(\ID,\UN)$.


\subsubsection{$\boldsymbol{\ell_1}$-pairwise}
For the pairwise distance we only consider the~$\ell_1$ variant, so usually instead of~$\ell_1$-pairwise we simply write pairwise. 

As before, we start by defining weighted majority relations for our compass elections, 
normalized by the number of voters.\footnote{When we use normalized weighted majority relations instead of  unnormalized ones, all the distances are scaled by the factor of $\nicefrac{1}{n}$.} 
For the identity, we simply have a matrix with ones above the diagonal and zeros below.
     \begin{align*}
\ID_m &= {\scriptsize\begin{bmatrix}
    - & 1 & \cdots & 1 & 1\\
    0 & - & \cdots & 1 & 1\\
    \vdots & \vdots & \ddots & \vdots & \vdots\\
    0 & 0 & \cdots & - & 1 \\
    0 & 0 & \cdots & 0 & -
  \end{bmatrix}},
   \end{align*}
   
Now, we observe something interesting. Both uniformity and antagonism produce exactly the same weighted majority relation with undefined values on the diagonal and~$0.5$ values everywhere else. In head-to-head comparisons between any two candidates there is always a tie. 

  \begin{align*}
\UN_m = \AN_m &= {\scriptsize\begin{bmatrix}
    - & 0.5 & \cdots & 0.5 & 0.5\\
    0.5 & - & \cdots & 0.5 & 0.5\\
    \vdots & \vdots & \ddots & \vdots & \vdots\\
    0.5 & 0.5 & \cdots & - & 0.5 \\
    0.5 & 0.5 & \cdots & 0.5 & -
  \end{bmatrix}}, 
   \end{align*}
   
Finally, we present the matrix for the stratification election. It consists of four squares. The upper-right square is filled with ones, the lower-left square is filled with zeros, while the upper-left and lower-right squares are undefined on the diagonal and have~$0.5$ values elsewhere.

     \begin{align*}
\ST_m &= {\scriptsize\begin{bmatrix}
- & 0.5 & \cdots & 0.5 & 1 &\cdots & 1 & 1\\
0.5 & - & \cdots & 0.5 & 1 &\cdots & 1 & 1 \\
\vdots & \vdots & \ddots & \vdots&  \vdots & \ddots & \vdots &   \vdots \\
0.5 & 0.5 & \cdots & - & 1 &\cdots & 1 & 1 \\
0 & 0 & \cdots & 0 & - & \cdots & 0.5 & 0.5 \\
\vdots & \vdots & \ddots & \vdots&  \vdots & \ddots & \vdots &   \vdots \\
0 & 0 & \cdots & 0 & 0.5 & \cdots & -  & 0.5\\
0 & 0 & \cdots & 0 & 0.5 & \cdots & 0.5  & -
\end{bmatrix}}
\end{align*}

\noindent
Next, we compute the pairwise distances of these matrices. Because~$\UN$ and~$\AN$ are identical, we omit distances between~$\AN$ and other points from the compass.

\begin{restatable}{proposition}{distpair} 
    It holds that:
  \begin{enumerate}
  \item $\PAIR(\ID_m,\UN_m) = \frac{1}{2}m(m-1)$
  \item $\PAIR(\UN_m,\ST_m) = \frac{1}{4}m^2$
  \item $\PAIR(\ID_m,\ST_m) = \frac{1}{4}m(m-2)$
  \end{enumerate}
\end{restatable}

As before, we normalize these distances by dividing them by the largest possible distance,~$\PAIR(\ID_m,\UN_m)$, and then compute the limits.

\begin{minipage}[b]{0.45\textwidth}
  \centering
  \begin{align*}
  \PAIR(\ID,\UN) &= 1 \\  
  \PAIR(\UN,\ST) &= \PAIR(\ID,\ST) = \frac{1}{2} \\
  &
    \end{align*}
\end{minipage}
\begin{minipage}[b]{0.5\textwidth}
  \centering
    \begin{tikzpicture}[xscale=0.5, yscale=0.5] 
        \clip (-0.5, -2.2) rectangle (9, 3.5);
        
        \draw (0.75,0.5) node[anchor=south] (A) {};
        \draw (8.25,0.5) node[anchor=south] (D) {};
        \draw [decorate,decoration={brace,amplitude=20}] (A.south east) -- (D.south west);
        
        \draw (4.25,0.5) node[anchor=south] (B) {};
        \draw [decorate,decoration={brace,amplitude=10,mirror}] (A.south east) -- (B.south west);
        
        \draw (4.75,0.5) node[anchor=south] (C) {};
        \draw [decorate,decoration={brace,amplitude=10,mirror}] (C.south east) -- (D.south west);
        
        \drawunabovepair{0}{0}
        \drawidabovepair{8}{0}
        \drawstbelowpair{4}{0}
        \draw (1,0.5) -- (8,0.5); 
        \draw (4.5,1.9) node[anchor=south] {$1$}; 
        \draw (1,0.5) -- (4,0.5); 
        \draw (2.5,-0.1) node[anchor=north] {$\frac{1}{2}$}; 
        \draw (5,0.5) -- (8,0.5); 
        \draw (6.5,-0.1) node[anchor=north] {$\frac{1}{2}$}; 
      \end{tikzpicture}
\end{minipage}


\subsubsection{EMD-Bordawise}
For the Bordawise distance we only consider the EMD variant, so usually instead of EMD-Bordawise we simply write Bordawise. Moreover, we normalize all the values in Borda score vectors by $n$.\footnote{It means that all the distances between such normalized vectors are scaled by the factor of $\nicefrac{1}{n}$.}

The Borda score vector of the identity election is as follows:
\begin{align*}
\ID_m &= [{\scriptstyle(m-1),(m-2),\dots,1,0}]%
\end{align*}

\noindent
As was the case for the pairwise distance, here again uniformity and antagonism are indistinguishable and produce the same Borda score vector.
\begin{align*}
\UN_m = \AN_m &= [{\scriptstyle \frac{m-1}{2},\dots,\frac{m-1}{2} }]
\end{align*}

\noindent
Finally, we have the vector for the stratification election.
\begin{align*}
\ST_m &= [{\scriptstyle \frac{3(m-1)}{4},\dots,\frac{3(m-1)}{4},\frac{m-1}{4},\dots,\frac{m-1}{4}}]
\end{align*}

\noindent
Next we compute the Bordawise distances of these vectors:

\begin{restatable}{proposition}{distborda} 
    If~$m$ is even, it holds that:
  \begin{enumerate}
  \item $\BOR(\ID_m,\UN_m) = \frac{1}{12} \cdot m (m^2-1)$
  \item $\BOR(\UN_m,\ST_m)) = \frac{1}{16} \cdot m^2 (m-1)$
  \item $\BOR(\ID_m,\ST_m) = \frac{1}{48} \cdot m (m^2+3m-4)$
  \end{enumerate}
\end{restatable}

\medskip




As before, we normalize these distances by dividing them by the largest possible distance,~$\BOR(\ID_m,\UN_m)$, and then compute the limits.

\begin{minipage}[b]{0.45\textwidth}
  \centering
  \begin{align*}
  \BOR(\ID,\UN) &= 1 \\  
  \BOR(\UN,\ST) &= \frac{3}{4} \\
  \BOR(\ID,\ST) &= \frac{1}{4}
\end{align*}
\end{minipage}
\begin{minipage}[b]{0.5\textwidth}
  \centering
    \begin{tikzpicture}[xscale=0.5, yscale=0.5]
        \clip (-0.5, -1.8) rectangle (9, 3.5);
    
        \draw (0.75,0.5) node[anchor=south] (A) {};
        \draw (8.25,0.5) node[anchor=south] (D) {};
        \draw [decorate,decoration={brace,amplitude=20}] (A.south east) -- (D.south west);
        
        \draw (5.75,0.5) node[anchor=south] (B) {};
        \draw [decorate,decoration={brace,amplitude=10,mirror}] (A.south east) -- (B.south west);
        
        \draw (6.25,0.5) node[anchor=south] (C) {};
        \draw [decorate,decoration={brace,amplitude=10,mirror}] (C.south east) -- (D.south west);

        \drawunaboveborda{0}{0}
        \drawidaboveborda{8}{0}
        \drawstbelowborda{5.5}{0}
        \draw (1,0.5) -- (5.5,0.5); 
        \draw (6.5,0.5) -- (8,0.5); 
        \draw (4.5,1.9) node[anchor=south] {$1$}; 
        \draw (1,0.5) -- (5.5,0.5); 
        \draw (3.25,-0.1) node[anchor=north] {$\frac{3}{4}$}; 
        \draw (6.5,0.5) -- (8,0.5); 
        \draw (7.25,-0.1) node[anchor=north] {$\frac{1}{4}$}; 
      \end{tikzpicture}
\end{minipage}

The distances (and the whole picture) for the Bordawise distance are very similar to those of the pairwise distance. The only major difference is the placement of the stratification election. While under the pairwise distance it is located in the middle between identity and uniformity, for the Bordawise distance it is much closer to identity.

\subsubsection{Swap}
Unlike for the nonisomorphic distances, for the swap distances we do not have any aggregate form of elections and compute the swap distances on the original elections. Unfortunately, not all compass elections we can easily generate with any~$m$ and~$n$. For the identity election we have the simplest scenario, because for any~$n$ we can easily generate~$\ID_{m,n}$. For the antagonism election it is also simple: To generate~$\AN_{m,n}$, it suffices to assume that~$n$ is even. However, for the stratification and uniformity elections the situation is getting complicated, because for~$\ST_{m,n}$ we need~$\frac{m}{2}!|n$, and for~$\UN_{m,n}$ we need~$m!|n$.

\begin{proposition}
    If $m!|n$ it holds that:
  \begin{enumerate}
    \item $d_\swap(\ID_{m,n},\UN_{m,n}) = d_\swap(\ID_{m,n},\AN_{m,n}) = \frac{1}{4}n(m^2-m)~$
    \item $d_\swap(\ID_{m,n},\ST_{m,n}) = \frac{1}{8}n(m^2-2m)$
    \item $d_\swap(\UN_{m,n},\AN_{m,n}) = \Theta(nm^2)\text{ (see Remark 1 below)}$
    \item $d_\swap(\UN_{m,n},\ST_{m,n})= \frac{1}{8}nm^2$
    \item $d_\swap(\AN_{m,n},\ST_{m,n}) = \Theta(nm^2)\text{ (see Remark 1 below)}$ 
  \end{enumerate}
\end{proposition}

\begin{remark}
Unfortunately for~$d_\swap(\UN_{m,n},\AN_{m,n})$ and~$d_\swap(\AN_{m,n},\ST_{m,n})$ we do not have closed form formulas (and we are not sure if they exist). However, it holds that~$\nicefrac{1}{8} \ n(m^2 - 3m + 2) \le d_\swap(\UN_{m,n},\AN_{m,n}) \le \nicefrac{1}{4} \ n(m^2 -m)$ and also~$\nicefrac{1}{8} \ n(m^2 - 2m) \le d_\swap(\AN_{m,n},\ST_{m,n}) \le \nicefrac{1}{4} \ n(m^2 -m)$.
\end{remark}

As before, we normalize these distances by dividing them by the largest possible distance,~$d_\swap(\ID_{m,n},\UN_{m,n})$, and then compute the limits.

    \vspace{0.5em}

\begin{minipage}[b]{0.45\textwidth}
  \centering
  \begin{align*}
    \SWAP(\ID,\UN) &= \SWAP(\ID,\AN) = 1 \\
    \SWAP(\ID,\ST) &= \SWAP(\UN,\ST) =  \frac{1}{2} \\
    & \\
    & \\
\end{align*}
\end{minipage}
\begin{minipage}[b]{0.5\textwidth}
  \centering
    \begin{tikzpicture}[xscale=0.5, yscale=0.5]
        \clip (-0.5,-2) rectangle (9, 8);

        \draw (0.75,0.5) node[anchor=south] (A) {};
        \draw (8.25,0.5) node[anchor=south] (D) {};
        \draw [decorate,decoration={brace,amplitude=20,mirror}] (A.south east) -- (D.south west);

        \drawunswap{0}{0}
        \drawidswap{8}{0}
        \drawanswap{4}{6.5}
        \drawstswap{4}{0}
        \draw (1,0.5) -- (4,0.5);
        \draw (5,0.5) -- (8,0.5);
        \draw (4.5,-2) node[anchor=south] {$1$};
        \draw (1,1) -- (4,6.5);
        \draw (2.3,3.8) node[anchor=south] {$?$};
        \draw (5,6.5) -- (8,1);
        \draw (6.8,3.8) node[anchor=south] {$1$};
        \draw (4.5,1) -- (4.5,6.5);
        \draw (4.9,3) node[anchor=south] {$?$};
        \draw (2.5,2) node[anchor=north] {$\frac{1}{2}$};
        \draw (6.5,2) node[anchor=north] {$\frac{1}{2}$};
      \end{tikzpicture}
\end{minipage}
    \vspace{0.5em}

\subsubsection{Discrete}
For the discrete distance, the situation is analogous to the case of the swap distance.

\begin{proposition}
    If $m!|n$ it holds that:
  \begin{enumerate}
  \item $\DISC(\ID_{m,n},\UN_{m,n}) = n \frac{m!-1}{m!}$
  \item $\DISC(\ID_{m,n},\AN_{m,n}) = \frac{1}{2} n$
  \item $\DISC(\UN_{m,n},\AN_{m,n}) = n \frac{m!-2}{m!}$
  \item $\DISC(\UN_{m,n},\ST_{m,n})=n \frac{m! - ((m/2)!)^2}{m!}$ 
  \item $\DISC(\ID_{m,n},\ST_{m,n}) =  \DISC(\AN_{m,n},\ST_{m,n}) = n\frac{((m/2)!)^2-1}{((m/2)!)^2}$ 
  \end{enumerate}
\end{proposition}

As before, we normalize these distances by dividing them by the largest possible distance,~$\DISC(\ID_{m,n},\UN_{m,n})$, and then compute the limits.

\begin{minipage}[b]{0.45\textwidth}
  \centering
  \begin{align*}
    \DISC(\ID,\UN) &= \DISC(\ID,\ST) = \DISC(\UN,\AN) \\
                   &= \DISC(\UN,\ST) = \DISC(\AN,\ST)  = 1 \\
    \DISC(\ID,\AN) &= \frac{1}{2} \\
\end{align*}
\end{minipage}

Unlike for the other distances, for the discrete distance we do not present graphical representation due to its obscurity.

\subsection{Equivalence Classes}\label{subsec:eqclass}

Given a distance, two elections are in the same equivalence class if
their distance is zero.
An \emph{anonymous, neutral equivalence class (ANEC)} consists of all elections 
that are isomorphic to each other~\citep{ege-gir:j:isomorphism-ianc}.
While ANECs are 
the equivalence classes of the isomorphic distances (e.g., the swap one), the other distances are less precise and
their equivalence classes are unions of some ANECs.

\begin{table}[t]
    \centering
    \begin{tabular}{ c | c | c | c | c }
     $|C|\times|V|$ & ANECs & Positionwise  & Pairwise & Bordawise\\
    	\midrule
       $3 \times 3$    & 10 & 10 & 8 & 8 \\
       $3 \times 4$    & 24 & 23 & 17 & 13 \\
       $3 \times 5$    & 42 & 40 & 25 & 18 \\
    	\midrule
       $4 \times 3$    & 111 & 93 & 50 & 37 \\
       $4 \times 4$    & 762 & 465 & 200 & 76  \\
       $4 \times 5$    & 4095 & 1746 & 513 & 131  \\
       
        \end{tabular}
        \caption{\label{table:num_classes} Number of equivalence classes under our metrics.}
\end{table}

To get a feeling as to how much precision is lost
due to various aggregate representations,
%
%
in~\Cref{table:num_classes} we compare the numbers of ANECs and
the numbers of equivalence classes 
of the positionwise, pairwise, and Bordawise metrics, for small
elections; we computed the table using exhaustive
search\footnote{There are exact formulas for some columns in~\Cref{table:num_classes}, but not for all. See, e.g., the work
  of \citet{ege-gir:j:isomorphism-ianc}.}  (note that EMD- and
$\ell_1$-positionwise metrics have the same equivalence classes).

Among these metrics, the positionwise ones perform best and Bordawise performs
worst.
Next, we provide a partial theoretical
explanation for this observation.
We say that a metric~$d$ is at least as fine as a metric~$d'$ if for
each two elections~$A$ and~$B$,~$d(A,B) = 0$ implies that
$d'(A,B) = 0$ (i.e., each equivalence class of~$d$ is a subset of
some equivalence class of~$d'$). Metric~$d$ is finer
than~$d'$ if it is at least as fine as~$d'$ but~$d'$ is not at least
as fine as~$d$.

\begin{proposition}[\cite{boe-fal-nie-szu-was:c:understanding}]\label{prop:equclass}
The swap and discrete isomorphic distances are finer than the
  EMD/$\ell_1$-positionwise and pairwise ones, which both are finer than the
  Bordawise distance. Neither the EMD/$\ell_1$-positionwise distance is finer than the pairwise distance
  nor the other way round.
\end{proposition}

In~\Cref{fig:aov} we present a scheme that illustrates the relationship between different distances (i.e., the implications of \Cref{prop:equclass}).

\begin{figure}

   \begin{center}
     \begin{tikzpicture}
       \node[inner sep=3pt] (id)  at (1.5,2) {ANECs};
       \node[inner sep=3pt] (pad)  at (0,1) {Pairwise ECs};
       \node[inner sep=3pt] (pod)  at (3,1) {Positionwise ECs};
       \node[inner sep=3pt] (bo)  at (1.5,0) {Bordawise ECs};
       
        \draw[draw=black, line width=0.4mm,->]     (pad) edge node  {} (id);
        \draw[draw=black, line width=0.4mm,->]     (pod) edge node  {} (id);
        \draw[draw=black, line width=0.4mm,->]     (bo) edge node  {} (pad);
        \draw[draw=black, line width=0.4mm,->]     (bo) edge node  {} (pod);

      \end{tikzpicture}

    \end{center}
    \caption{Relationship between different metrics. An arc from metric $d$ to $d'$ means that $d'$ is at least as expressive as $d$. 
    } 
    \label{fig:aov}

\end{figure}
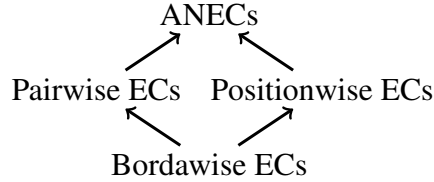

\begin{figure}
    \centering
    
    \begin{subfigure}[b]{0.49\textwidth}
        \centering
          \includegraphics[width=6.6cm]{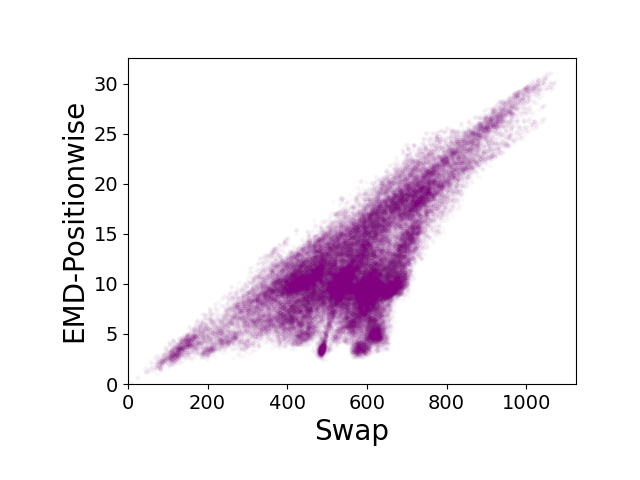}
        \caption{Swap vs EMD-positionwise}
    \end{subfigure}
    \begin{subfigure}[b]{0.49\textwidth}
        \centering
          \includegraphics[width=6.6cm]{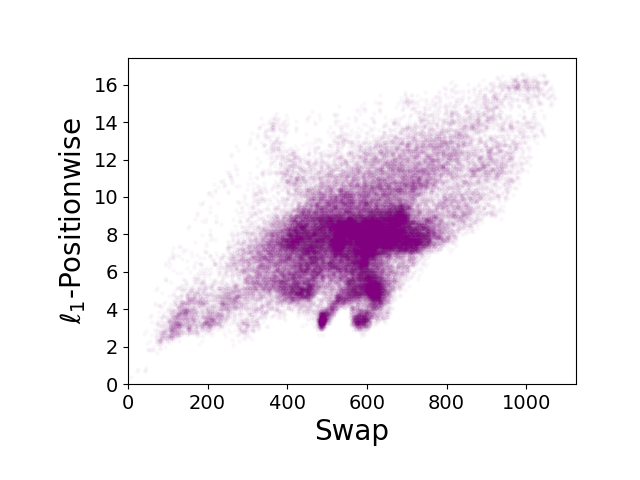}
        \caption{Swap vs~$\ell_1$-positionwise}
    \end{subfigure}
    
    \vspace{1em}
    
    \begin{subfigure}[b]{0.49\textwidth}
        \centering
          \includegraphics[width=6.6cm]{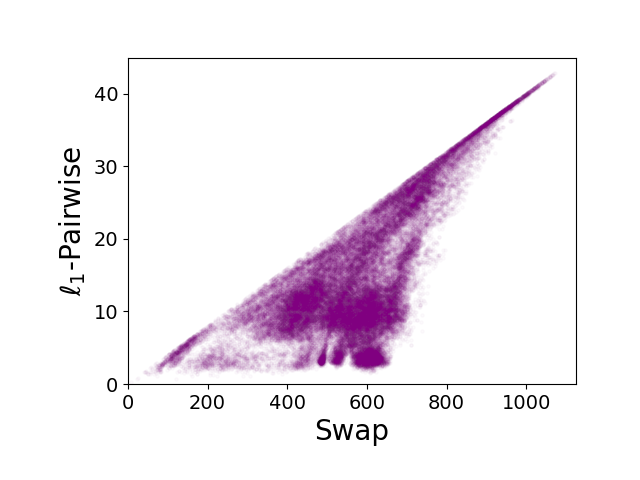}
        \caption{Swap vs~$\ell_1$-pairwise}
    \end{subfigure}
    \begin{subfigure}[b]{0.49\textwidth}
        \centering
          \includegraphics[width=6.6cm]{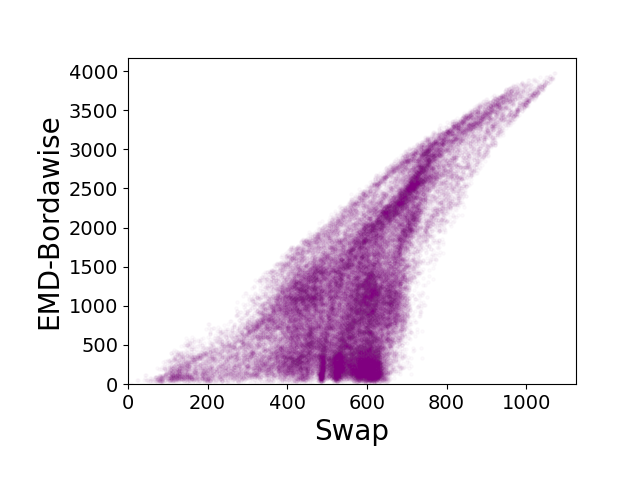}
        \caption{Swap vs EMD-Bordawise}
    \end{subfigure}

    \caption{Correlation between the nonisomorphic distances and the swap distances based on the synthetic dataset described in~\Cref{tab:testbed}.}
    \label{fig:correlationPlots}
\end{figure}

\subsection{Correlation}
In \Cref{fig:correlationPlots} we present the correlation plots for the nonisomorphic distances and the swap distance. (Recall that the synthetic dataset that we use consists of elections with $10$ candidates and $50$ voters sampled from~$13$ models; for~$11$ of them we generated~$20$ elections and for the Norm-Mallows and urn models we sampled~$60$ elections.)
As a complement to the correlation plots, we present two additional tables. First, we have~\Cref{table:pearson_correlation}, where we have computed Pearson correlation coefficient between the swap distances and other distances
Second, we have~\Cref{table:pearson_correlation_cultures}, where we have computed the PCC between the swap distances and those provided by the other metrics for each statistical culture independently (i.e., for each statistical culture we give the correlation coefficients between the swap distances of all pairs of elections from this culture and their distances according to our other metrics).

As we can see, the strongest correlation is witnessed by the EMD-positionwise distance (having PCC equal~$0.745$, see~\Cref{table:pearson_correlation}) followed by the EMD-Bordawise distance (having PCC equal~$0.713$). Then we have the~$\ell_1$-pairwise distance (having PCC equal~$0.708$), and $\ell_1$-positionwise distance (having PCC equal~$0.563$), and, finally, the worst correlation is witnessed by discrete distance (having PCC equal~$0.342$). The surprisingly high correlation for the EMD-Bordawise and $\ell_1$-pairwise distances apparently comes from the fact that this distance works well for elections from the Normalized Mallows and urn models, and in our dataset we had a lot of elections from these two models. In \Cref{table:pearson_correlation_cultures}, we can see that for almost all other models (with the exception of group-separable and impartial culture) the correlation is insignificant.

\begin{table}[t]
    \centering
    \begin{tabular}{ c | c | c | c | c | c }
       $|C| \times |V|$ & EMD-Pos.&~$\ell_1$-Pos. &~$\ell_1$-Pair. & EMD-Bordawise & Discrete\\
    	\midrule
       $3 \times 3$    & 0.942 & 0.748 & 0.860 & 0.817 & 0.614 \\
       $3 \times 4$    & 0.900 & 0.697 & 0.860 & 0.737 & 0.636 \\
       $3 \times 5$    & 0.920 & 0.759 & 0.843 & 0.747 & 0.680 \\ 
    	\midrule
       $4 \times 3$    & 0.850 & 0.577 & 0.735 & 0.675 & 0.402 \\
       $4 \times 4$    & 0.782 & 0.561 & 0.689 & 0.610 & 0.434 \\
       $4 \times 5$    & 0.772 & 0.567 & 0.672 & 0.606 & 0.432 \\[1mm]
    	\midrule
     $\substack{\mathrm{10 \times 50} \\ \text{(340 elections)}}$  \rule{0mm}{2.75mm}  & 0.745 & 0.563 & 0.708 & 0.713 & 0.342 \\
        \end{tabular}
        \caption{\label{table:pearson_correlation} Pearson correlation coefficients between the swap distance and the
        other ones computed for our datasets.}
    
\end{table}

\begin{table}[t]
    \centering
    \small
    \begin{tabular}{ c | c | c | c | c | c }
        Name & EMD-Pos.&~$\ell_1$-Pos. &~$\ell_1$-Pair. & EMD-Borda. & Discrete\\
    	\midrule
        Impartial Culture  & 0.481 & 0.114 & 0.525 & 0.471 & -0.039 \\
        SP by Conitzer  & 0.471 & 0.727 & -0.142 & -0.015 & 0.976 \\
        SP by Walsh  & 0.377 & 0.467 & -0.119 & 0.111 & 0.7 \\
        SPOC  & 0.297 & 0.409 & -0.074 & 0.079 & 0.622 \\
        Single-Crossing  & 0.252 & 0.248 & 0.123 & 0.098 & 0.625 \\
        Interval  & 0.242 & 0.219 & 0.101 & 0.088 & 0.606 \\
        Disc  & 0.337 & 0.317 & 0.203 & 0.149 & 0.636 \\
        Cube  & 0.406 & 0.347 & 0.311 & 0.286 & 0.67 \\
        Circle  & 0.406 & 0.329 & 0.335 & 0.287 & 0.651 \\
        Urn  & 0.84 & 0.86 & 0.803 & 0.772 & 0.102 \\
        Norm-Mallows  & 0.86 & 0.784 & 0.839 & 0.82 & 0.255 \\
        GS Balanced  & 0.863 & 0.793 & 0.844 & 0.822 & 0.259 \\
        GS Caterpillar  & 0.864 & 0.795 & 0.845 & 0.824 & 0.252 \\
        \midrule
        \end{tabular}
    \caption{\label{table:pearson_correlation_cultures} Pearson correlation coefficients between the swap distance and the
        other ones computed for each statistical culture used in our maps.}
    
\end{table}


\vspace{0.2cm}
\begin{conclusionbox}
Among studied distances (i.e., EMD-positionwise, $\ell_1$-positionwise, EMD-Bordawise, $\ell_1$-pairwise) the EMD-positionwise is most strongly correlated with the swap distance, with a clear advantage
over the other metrics. Hence, we recommend it for using in practice, especially when dealing with large elections (i.e., with many voters and candidates).
\end{conclusionbox}


\section{Summary}
The main objective of this chapter was to find meaningful ways of calculating the distances between elections with ordinal ballots. We believe that we have succeeded in fulfilling this task, or at least we have shown the direction in which to go.

We proposed three isomorphic distances, i.e., swap, Spearman, and discrete distances. Both the swap and Spearman distances are very precise, but slow to compute. Without surprise, the discrete distance proves to be quite useless, with most elections being at maximal (or almost maximal) distances from each other.

We also introduced several nonisomorphic distances. We have two variants of the positionwise distance, one using EMD and the other using~$\ell_1$ as the underlying norms. Although both variants are quite similar, we favor EMD over~$\ell_1$ due to its stronger correlation with the swap distance.
Then we have pairwise and Bordawise distances. Neither of them is convincing because they collapse the whole Borda balanced area into a single point (in particular, uniformity, and antagonism elections become indistinguishable).
Nonetheless, both are doing well enough at placing elections between the uniformity and identity. Regarding the amount of time needed to compute the distances, pairwise distance is relatively slow to compute, while Bordawise is extremely fast---but it is its only advantage.

As a major conclusion, we can say that if the number of candidates is limited (e.g., not larger than~$10$) then we recommend using the swap distance, as it is the most precise one. For elections with more candidates, we recommend the EMD-positionwise as it achieves the best trade-off between precision and time.

\vspace{0.2cm}
\begin{contributionbox}
\begin{itemize}
    \item Introduction of various distances that serve for measuring similarities between ordinal elections. 
    \item Introduction of the map of elections framework. 
    \item Detailed comparison between different distances, concluding that for small elections we suggest using swap distance, while for larger elections we suggest using the EMD-positionwise one.
\end{itemize}
\end{contributionbox}

\chapter{Applications}
\label{ch:applications}

\section{Introduction}
Creating a map of elections consists of the three following steps. First, we have to prepare the elections---we can either sample them from statistical models or select some real-life ones. Second, we compute the distances between each pair of elections---this gives us a distance matrix. Third, we embed the distance matrix in a two-dimensional Euclidean space. Each of these steps can be done in numerous ways, that is, there are many ways of generating elections, there are several distances to choose from, and finally we have to decide on a particular embedding algorithm. 

In the first part of this chapter, we argue that the way we design the map is reasonable. We present results for several different embeddings, and explain why we recommend using one over the other. In particular, we analyze the concepts of \textit{monotonicity} and \textit{distortion} of embeddings, which test the quality of a given embedding. We also discuss how changing the number of candidates is influencing the map, in other words, we answer the question of scalability of the map. 

Later, in the second part, we provide numerous practical examples of applications of the map. We study single-winner voting rules such as plurality, Borda, Copeland, and Dodgson, and multiwinner voting rules such as Chamberlin--Courant and Harmonic-Borda. We use the map to show the relationship between various voting rules and statistical cultures. In particular, we are curious if elections lying next to one another on the map behave in a similar manner---for example, the winning candidate/committee has a similar score, or computing the winning candidate or committee is taking similar amount of time. For example, for Dodgson rule the longest running time was witnessed by group-separable caterpillar elections, while for Harmonic-Borda it was~$4$-Sphere elections.
Moreover, for the Chamberlin--Courant and Harmonic-Borda rules, we compare the effectiveness of their approximation algorithms.
Then, we analyze a number of real-life instances of elections, and see where they land on the map. We study political elections, surveys, and sport competitions.
Finally, we briefly discuss the concept of a skeleton map, where instead of sampling numerous elections from a given distribution, we present only one frequency matrix that captures that statistical culture.

All the maps presented in this chapter are based on the positionwise distance. This means that we operate on matrices rather than elections. We have chosen the positionwise distance in order to be able to draw maps with large (up to~$100$) numbers of candidates.

\section{Setup}
We start by outlining the basic setup for our experiments, i.e., the set of elections on which the map is based. 

In \Cref{tab:embed_setup} we list all the models that we use in our maps, and also the numbers of elections sampled from each model. The exact number of candidates and voters will be specified for each experiment independently.

We also introduce new artificial families of elections, called \textit{paths}, which serve for making the map more stable and easier to interpret. Briefly speaking, we take convex combinations of the compass matrices and create paths between them. Below, we describe this concept in more detail.

\begin{table}[]
  \centering
{\footnotesize
  \begin{tabular}{lcc}
    \toprule
    Model & Number of Elections \\
    
    \midrule
    Impartial Culture           & 20 \\
    \midrule
    Single-Peaked (Conitzer)  & 20 \\
    Single-Peaked (Walsh)      & 20 \\
    SPOC                       & 20 \\
    Single-Crossing           & 20 \\
    \midrule
    Interval        & 20 \\
    Square         & 20 \\
    Cube         & 20 \\
    5-Cube         & 20 \\
    10-Cube         & 20 \\
    20-Cube         & 20 \\
    \midrule
    Circle         & 20 \\
    Sphere         & 20 \\
    4-Sphere   & 20 \\
    \midrule
    Group-Separable (Balanced)     & 20 \\
    Group-Separable (Caterpillar)    & 20 \\
    \midrule
    Urn       & 80 \\  
    Mallows    & 80 \\
    \midrule
    Compass ($\ID$,~$\AN$,~$\UN$,~$\ST$) & 4 \\
    \midrule
    Paths & 20$\times$4 \\
    \bottomrule
  \end{tabular} }
    \caption{\label{tab:embed_setup} Setup}
\end{table}

\subsubsection{Paths between Election Matrices}
We consider convex combinations of frequency matrices.
Given two such matrices,~$X$ and~$Y$, and~$\alpha \in [0,1]$, one
might expect that matrix~$Z = \alpha X + (1-\alpha)Y$ would lie at positionwise 
distance~$(1-\alpha)\cdot \POS(X,Y)$ from~$X$ and at positionwise distance~$\alpha \cdot \POS(X,Y)$ from~$Y$, so that we would~have:
\begin{align*}
     \POS(X,Y) = \POS(X, Z) 
                      + \POS(Z, Y).
\end{align*}
However, without further assumptions, this is not necessarily the
case. Indeed, if we take~$X = \ID_m$ and~$Y = \rID_m$, then~$0.5X+0.5Y = \AN_m$ and~$\POS(X,Y) = 0$, 
but~$\POS(X,0.5X+0.5Y) = \POS(\ID,\AN) > 0$.
However, if we arrange the two 
matrices~$X$ and~$Y$
so that their positionwise distance is achieved by the identity
permutation of their column vectors, then their convex combination lies
exactly between them. 

\begin{figure}[t]
    \centering

        \includegraphics[width=5.0cm, trim={0.2cm 0.2cm 0.2cm 0.2cm}, clip]{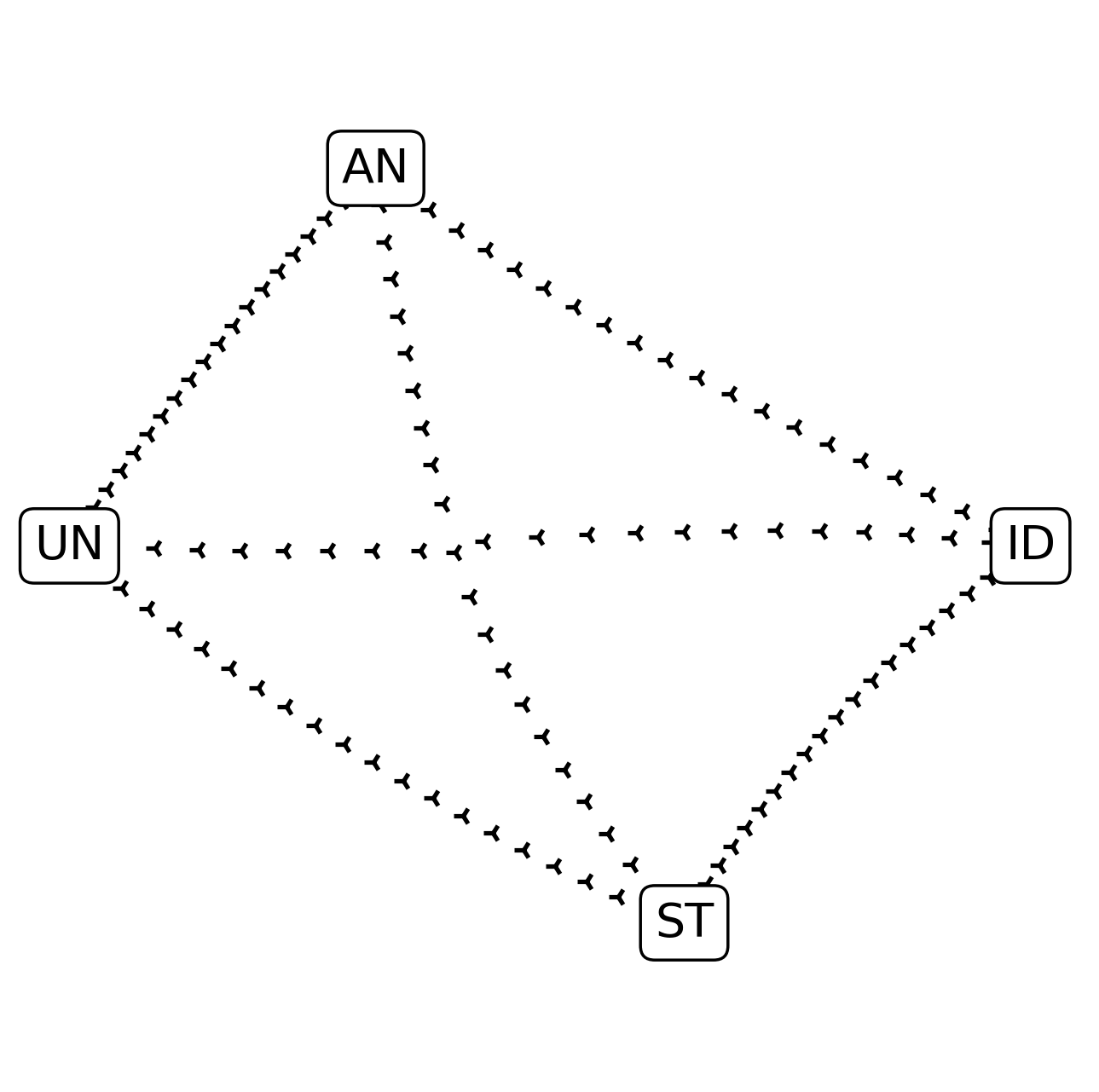}
    
    \caption{Paths between compass matrices.}
    \label{fig:paths}
\end{figure}

\begin{proposition}\label{pro:paths}
     Let~$X = (x_1, \ldots, x_m)$ and~$Y = (y_1, \ldots y_m)$ be two~$m \times m$ frequency matrices
     such that~$
     \POS(X,Y) = \textstyle \sum_{i=1}^m \EMD(x_i,y_i).~$
  Then, for each~$\alpha \in [0,1]$ it holds that~$\POS(X,Y) = \POS(X, \alpha X + (1-\alpha)Y) + \POS( \alpha X +
  (1-\alpha)Y, Y)$.
\end{proposition}
\begin{proof}
  Let~$Z = (z_1, \ldots, z_m) = \alpha X + (1-\alpha) Y$ be our convex
  combination of~$X$ and~$Y$.  We note two properties of the earth
  mover's distance. Let~$a$,~$b$, and~$c$ be three vectors that
  consist of nonnegative numbers, where the entries in~$b$ and~$c$ sum up to the same
  value. Then, it holds that~$\EMD(a+b,a+c) = \EMD(b,c)$. Further, for a
  nonnegative number~$\lambda$, we have 
  that~$\EMD(\lambda b, \lambda c) = \lambda\EMD(b,c)$.  Using these
  observations and the definition of the earth mover's distance, we
  note that:
  \begin{align*}
    \textstyle
    \POS(X,Z)  & \textstyle\leq \sum_{i=1}^m \EMD(x_i,z_i) \\
    &\!\!\!\!\!\!\!\!\!\!\!\!\!\!\!\! =   \textstyle\sum_{i=1}^m \EMD(x_i,\alpha x_i + (1-\alpha)y_i) \\
    &\!\!\!\!\!\!\!\!\!\!\!\!\!\!\!\! =   \textstyle\sum_{i=1}^m \EMD((1-\alpha)x_i, (1-\alpha)y_i) \\
    &\!\!\!\!\!\!\!\!\!\!\!\!\!\!\!\! =    \textstyle(1-\alpha) \sum_{i=1}^m \EMD(x_i,y_i) = (1-\alpha)\POS(X,Y).
  \end{align*}
  The last equality follows by our assumption regarding~$X$ and~$Y$. By an analogous reasoning, 
  we also have that~$\POS(Z,Y) \leq \alpha \POS(X,Y)$. By putting these two inequalities
  together, we have that:
  \[
    \POS(X,Z) + \POS(Z,Y) \leq \POS(X,Y).
  \]
  By the triangle inequality, we have that~$\POS(X,Y) \leq \POS(X,Z) + \POS(Z,Y)$
  and, so, we have that~$\POS(X,Z) + \POS(Z,Y) = \POS(X,Y)$.
\end{proof}

Using \Cref{pro:paths}, for any two compass matrices,
we can generate a sequence of matrices that form a path between them.
For example, matrix~$0.5\ID + 0.5\UN$ is exactly at the same distance
from~$\ID$ and from~$\UN$.  

In \Cref{fig:paths} we show a map
of elections that contains our four compass matrices
and, for each two of them, i.e., for each two~$X, Y \in \{\ID, \UN, \AN, \ST\}$, 
a set of~$20$ matrices obtained as their convex
combinations with values of~$\alpha$ uniformly spread over~$[0,1]$. The map was created using the MDS embedding.
Even though each path consisted of the same number of matrices, we see that proportions of the distances between the compass matrices are maintained. Recall that, if~$\POS(\ID,\UN) = 1$, then~$\POS(\ID,\AN) = \POS(\UN, \ST) = \nicefrac{3}{4}$,~$\POS(\AN,\ST) = \nicefrac{13}{16}$, and~$\POS(\ID,\ST) = \POS(\UN,\AN) = \nicefrac{1}{2}$.

\section{Embedding}\label{ch:applications:sec:embedding}

In this section, we compare various different embedding methods. In particular, we consider the following six  algorithms: multidimensional scaling (MDS), t-distributed stochastic neighbor embedding (t-SNE), locally linear embedding (LLE), a variant of the Kamada and Kawai algorithm (KK), principal component analysis (PCA), and the Fruchterman and Reingold algorithm (FR). Technical aspects of these methods were described in \Cref{ch:preliminaries}. Sometimes we use the term \textit{MDS map}, as an abbreviated form for \textit{the map that was created using the MDS embedding}. Whenever we write \textit{original distance} we refer to the EMD-positionwise distance between the elections,
and whenever we write \textit{embedded distance} we refer to the Euclidean distance between the points on the plane (which correspond to these elections) after embedding. Whenever we write \emph{normalized distance}, we refer to the distance divided by the distance between the identity and uniformity because it is the largest possible one.

\begin{figure}[]
    \centering

    \begin{subfigure}[b]{0.49\textwidth}
        \centering
        \includegraphics[width=6cm, trim={0.2cm 0.2cm 0.2cm 0.2cm}, clip]{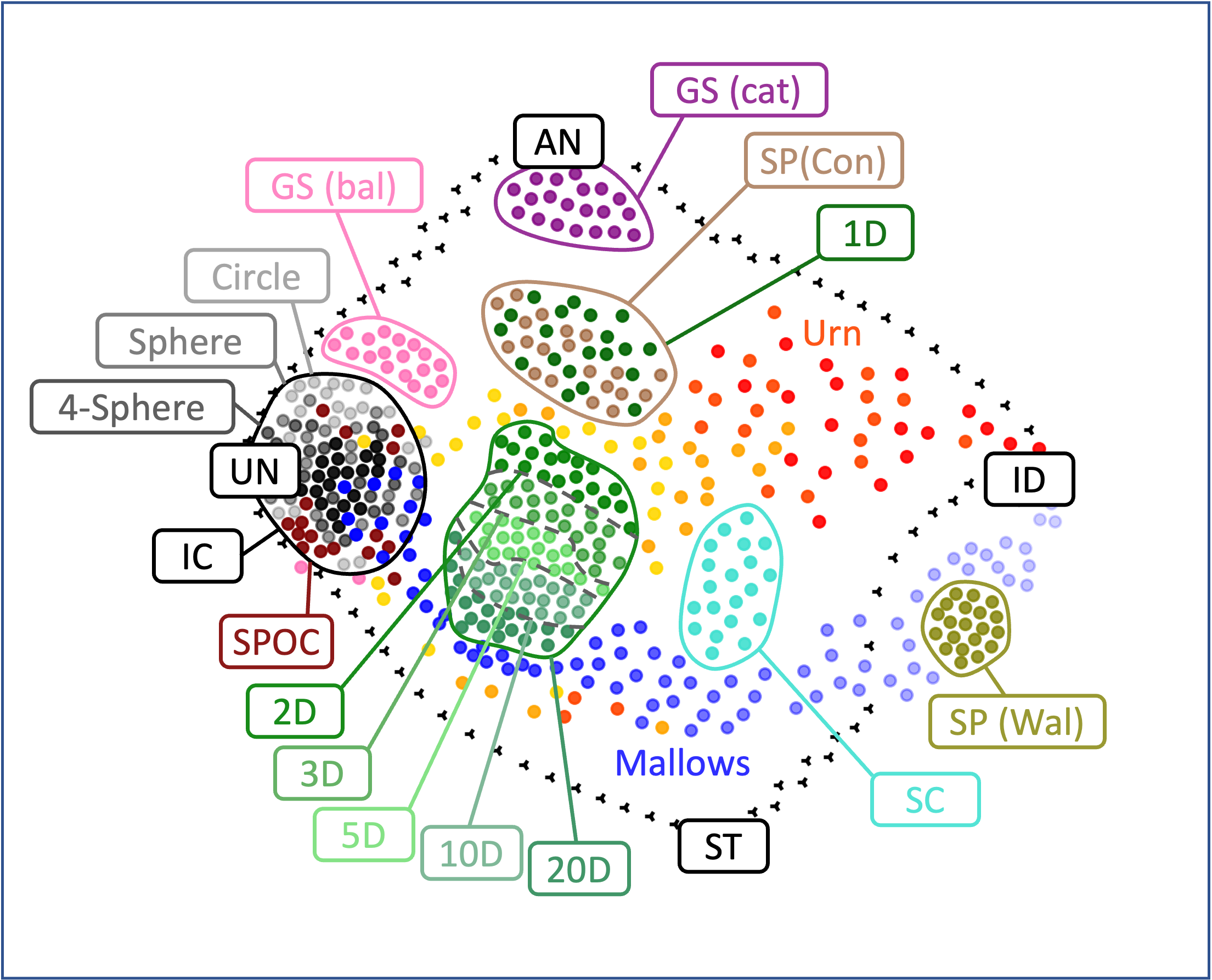}
        \caption{FR}
    \end{subfigure}
    \begin{subfigure}[b]{0.49\textwidth}
        \centering
        \includegraphics[width=6cm, trim={0.2cm 0.2cm 0.2cm 0.2cm}, clip]{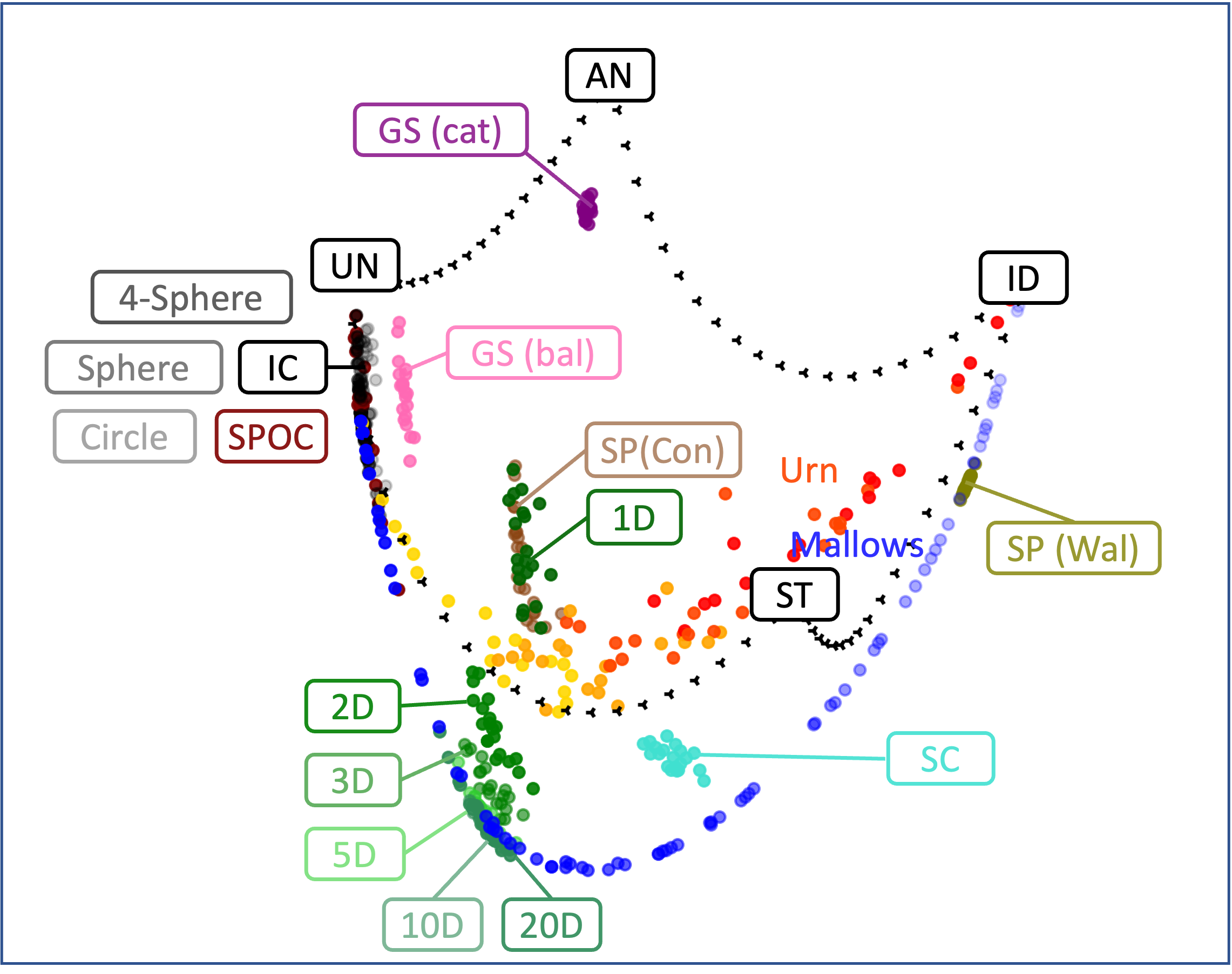}
        \caption{LLE}
    \end{subfigure}
    
     \vspace{1em}

    \begin{subfigure}[b]{0.49\textwidth}
        \centering
        \includegraphics[width=6cm, trim={0.2cm 0.2cm 0.2cm 0.2cm}, clip]{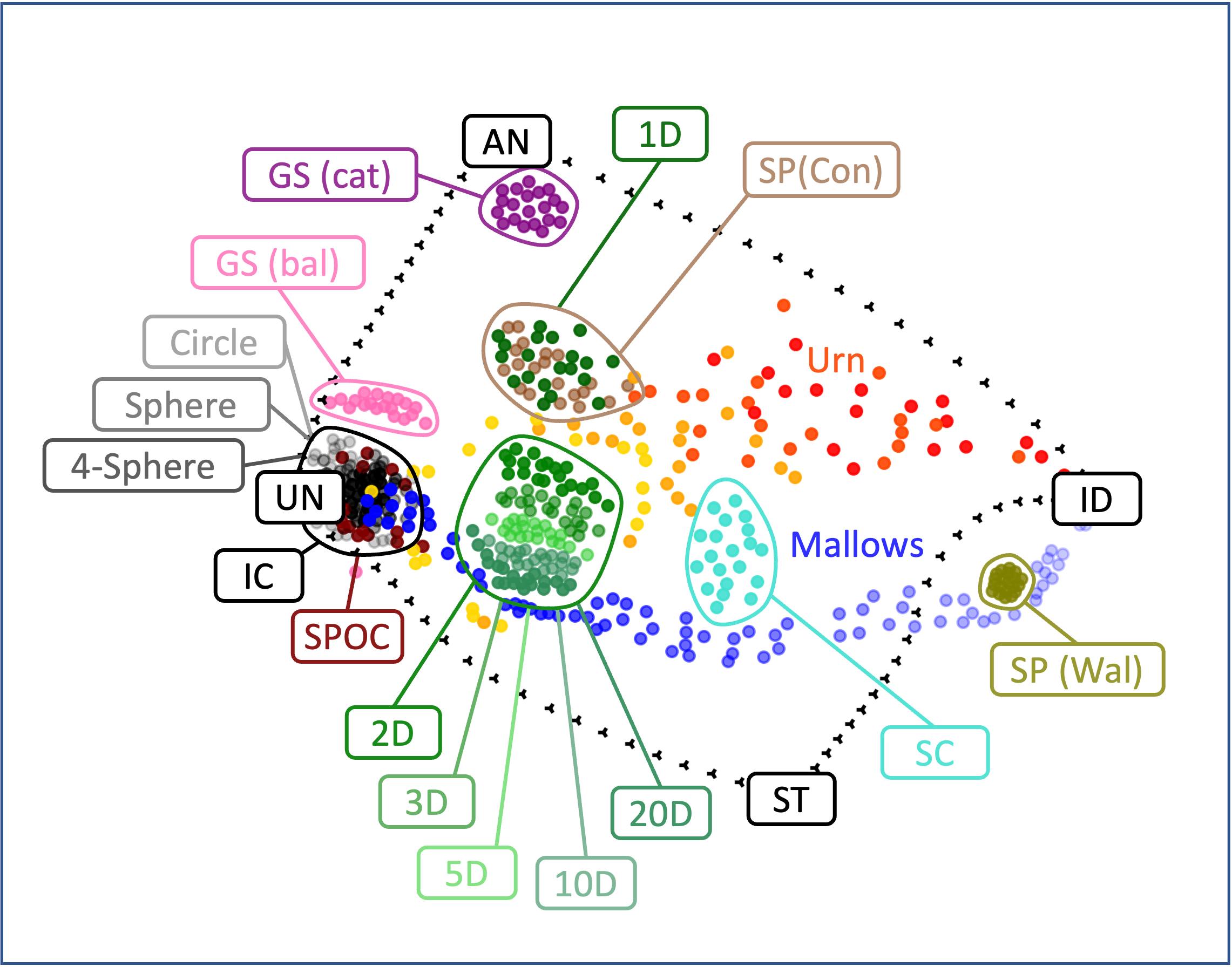}
        \caption{KK}
    \end{subfigure}
    \begin{subfigure}[b]{0.49\textwidth}
        \centering
        \includegraphics[width=6cm, trim={0.2cm 0.2cm 0.2cm 0.2cm}, clip]{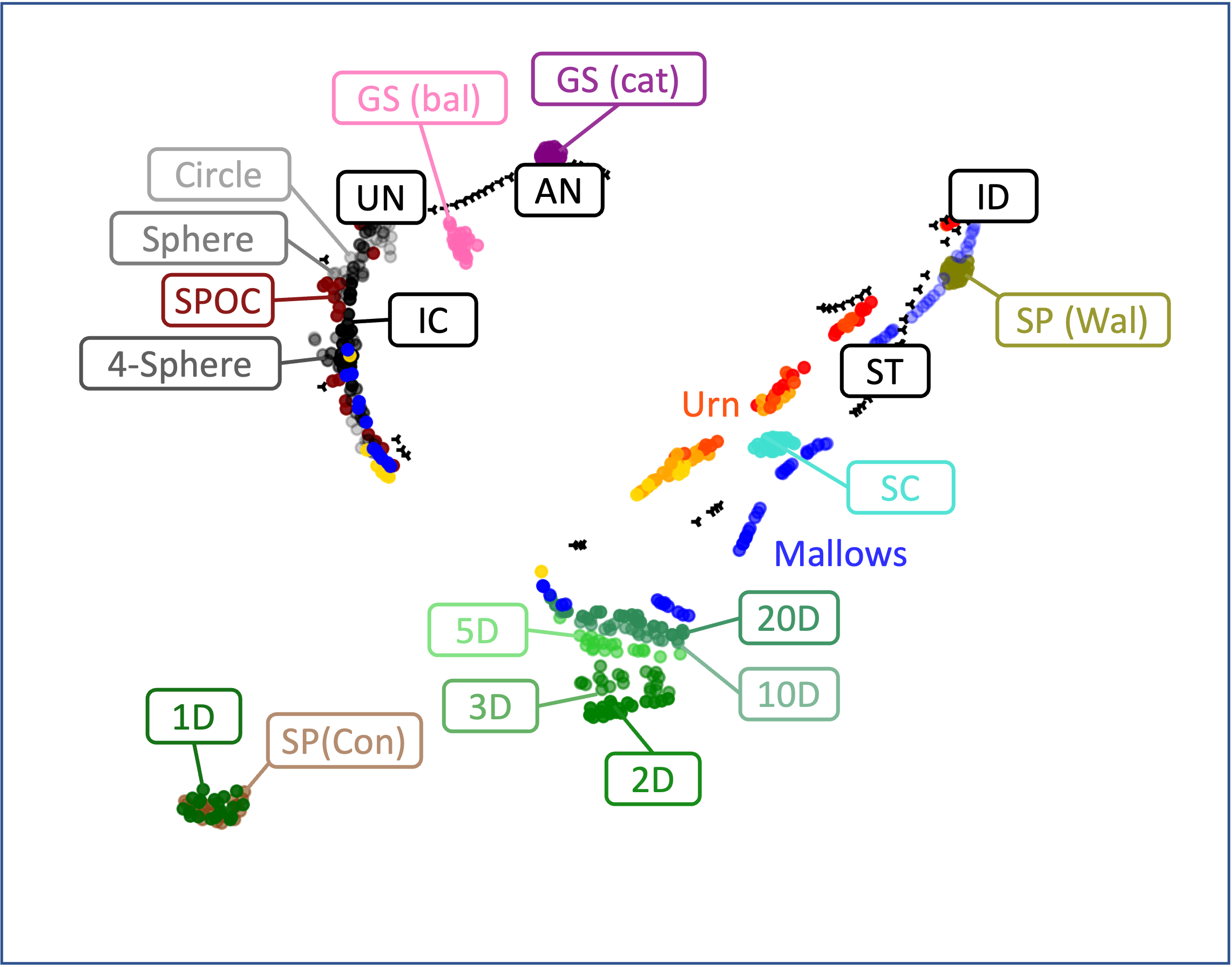}
        \caption{t-SNE}
    \end{subfigure}

     \vspace{1em}
    \begin{subfigure}[b]{0.49\textwidth}
        \centering
        \includegraphics[width=6cm, trim={0.2cm 0.2cm 0.2cm 0.2cm}, clip]{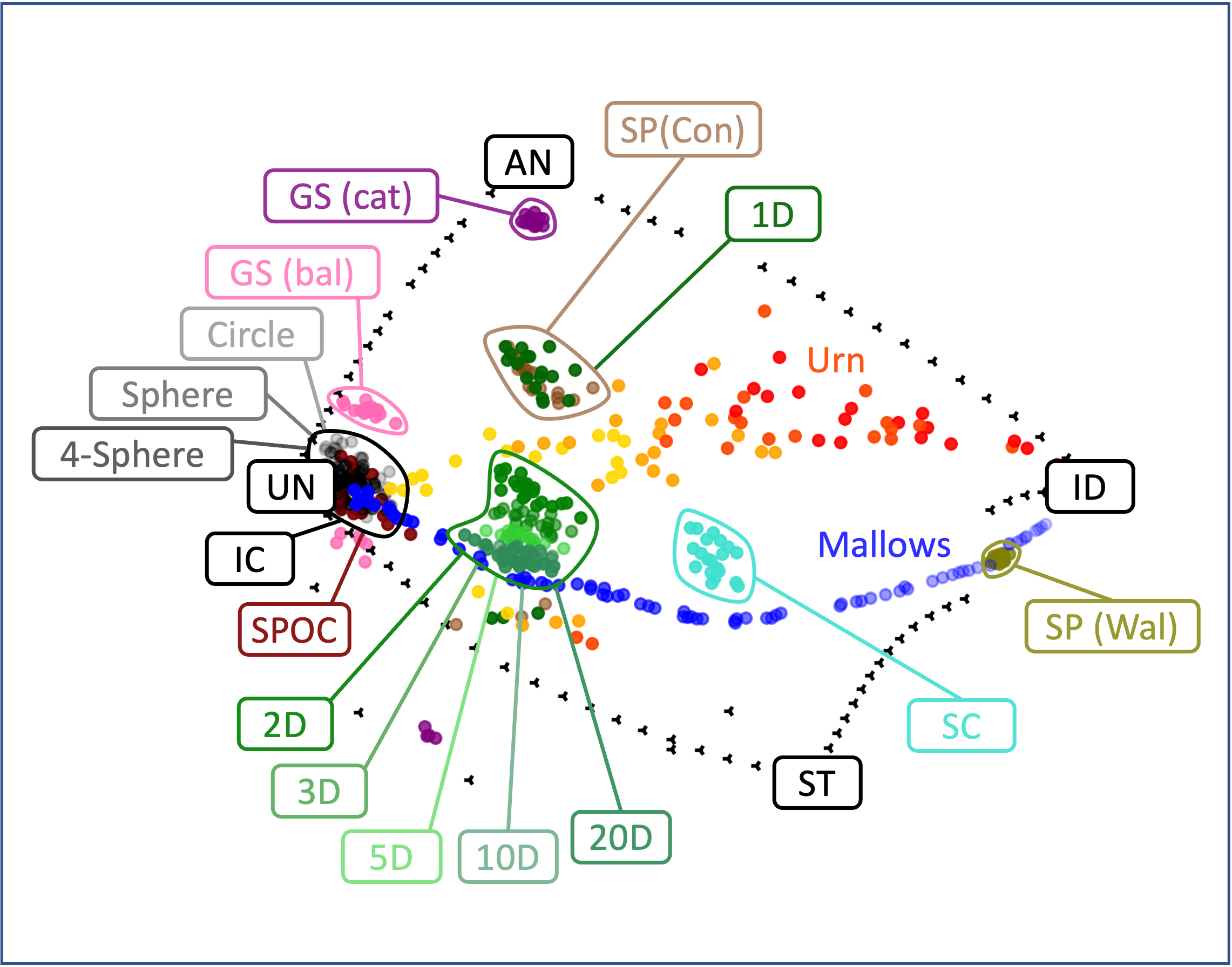}
        \caption{MDS}
    \end{subfigure}
    \begin{subfigure}[b]{0.49\textwidth}
        \centering
        \includegraphics[width=6cm, trim={0.2cm 0.2cm 0.2cm 0.2cm}, clip]{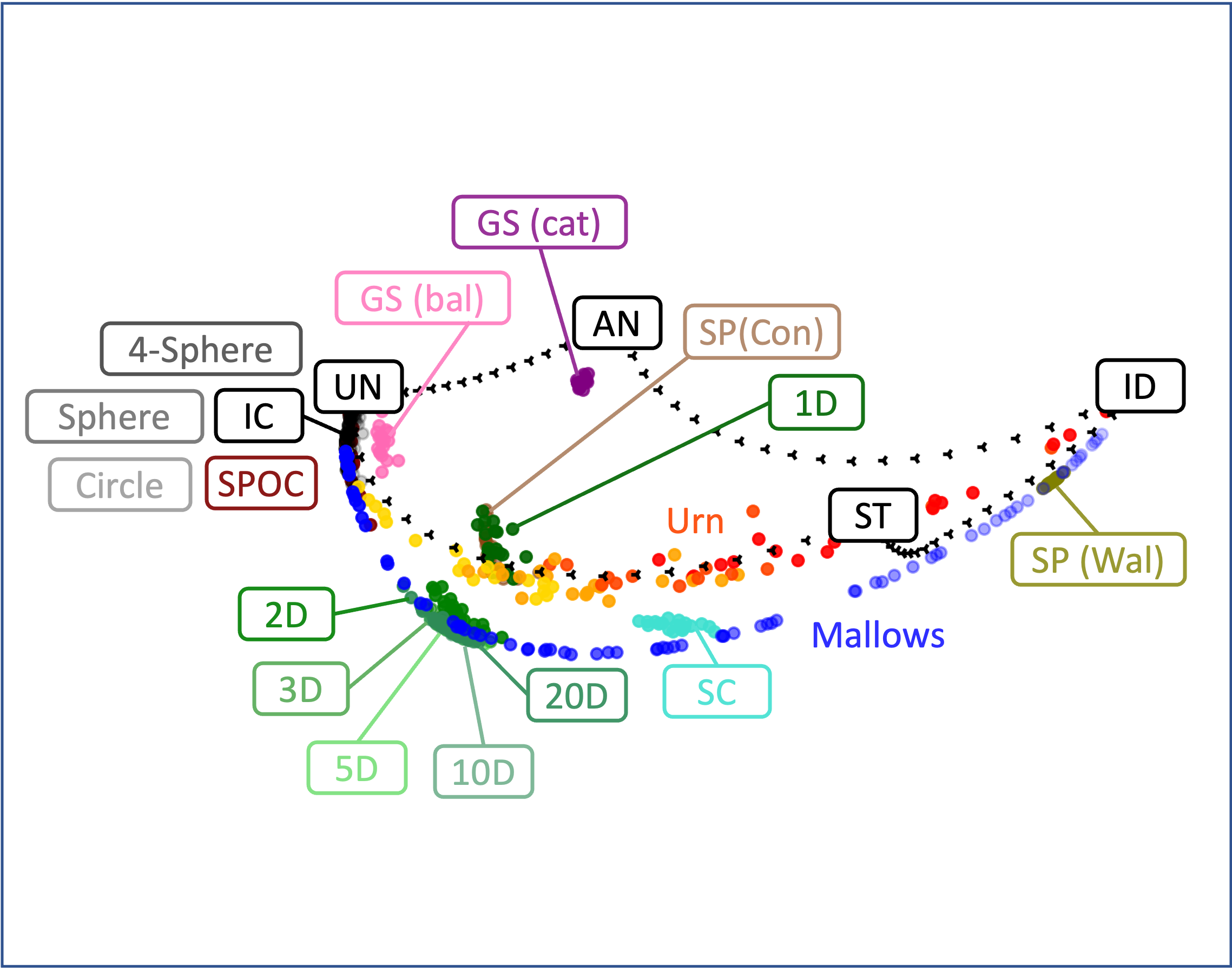}
        \caption{PCA}
    \end{subfigure}

    \caption{Comparison of embeddings algorithms.}
    \label{fig:embed}
\end{figure}

How to compare two different embeddings? To answer this question, first we have to explain what the main purpose of the map is. We want to make it easier and more intuitive to see certain features and properties of elections. Although we have the table with the original distances, it is hard to analyze the data solely by looking at the values in the table. 
If two elections are similar, we would like them to lie next to each other on the map. However, in most cases we are interested especially in the local correctness of the map (i.e., if two elections are far away on the map, we do not give that much attention to distinguishing whether they are far or very far). Nonetheless, it is important to know which embeddings are focusing on local correctness, and which are trying to correctly embed all the distances.

In \Cref{fig:embed} we present the results for several embeddings. As we can see, the maps are quite diverse. For t-SNE and PCA, we clearly see that they will not be very useful for us. The LLE embedding is slightly better, and we can see the main shape, yet it is still far from what we want (we would like most elections to lie between the four paths, and, if possible, to be more spread over the space). As to the MDS, KK and FR, they produce more or less the same shape but with different levels of compactness---with the MDS being the most compact and the FR being the least. From now on, we focus only on these three embeddings and discuss them in more detail.

For the MDS map we witness a flaw, that is, some elections are questionably placed. For example, several elections from the group-separable caterpillar model are far away from the rest---which is not the case when we look closely at the original distances. We usually observe such flaws in maps with a high number of candidates (for example~$100$). For a moment, let us forget about this flaw because for numerous maps with smaller numbers of candidates it is not occurring.

Which map is the best? One approach would be to verify the correlation between the embedded distances and the  original ones. Here is what we get: For the KK method the PCC is the highest and is equal to~$0.9805$, for the MDS method it is equal to~$0.9748$ and for the FR method it is equal~$0.9364$.

If our goal were to localize where a given election precisely lies, we would recommend KK---we give arguments for this in the following sections about monotonicity and distortion of the embeddings. However, if we would later color the map according to certain features, for example, the highest Borda score in each election or the time needed to compute the winning committee under a particular rule, it is useful to have a less compact map---as long as it maintains the proper shape and is not giving us misleading impressions. Therefore, for the maps colored by features, we recommend using the FR embedding. Nevertheless, we still find the MDS algorithm useful for some other tasks, such as, for example, the maps of preferences (for ordinal preferences recall~\Cref{ordinal_map_pref}, and for approval preferences see~\Cref{approval_map_pref}).
    
Whenever we write~$m \times n$ elections, we refer to elections with~$m$ candidates and~$n$ voters. To simplify the discussion of the concepts of \textit{monotonicity} and \textit{distortion} we introduce the notion of an experiment.

\subsubsection*{Experiment}
By an experiment~$Q=(\mathcal{E}, d_\mathcal{M}, d_\Euc)$ we refer to a triple that consists of a set of elections~$\mathcal{E}$, original distances~$d_\mathcal{M}$ between these elections according to metric~$\mathcal{M}$, and Euclidean distances~$d_\Euc$ between these elections after the embedding. In our case, for~$d_\mathcal{M}$ we select the EMD-positionwise distance.

\begin{figure}[t]
    \centering
    
    \begin{subfigure}[b]{0.49\textwidth}
        \centering
        \includegraphics[width=6cm, trim={0.2cm 0.2cm 0.2cm 0.2cm}, clip]{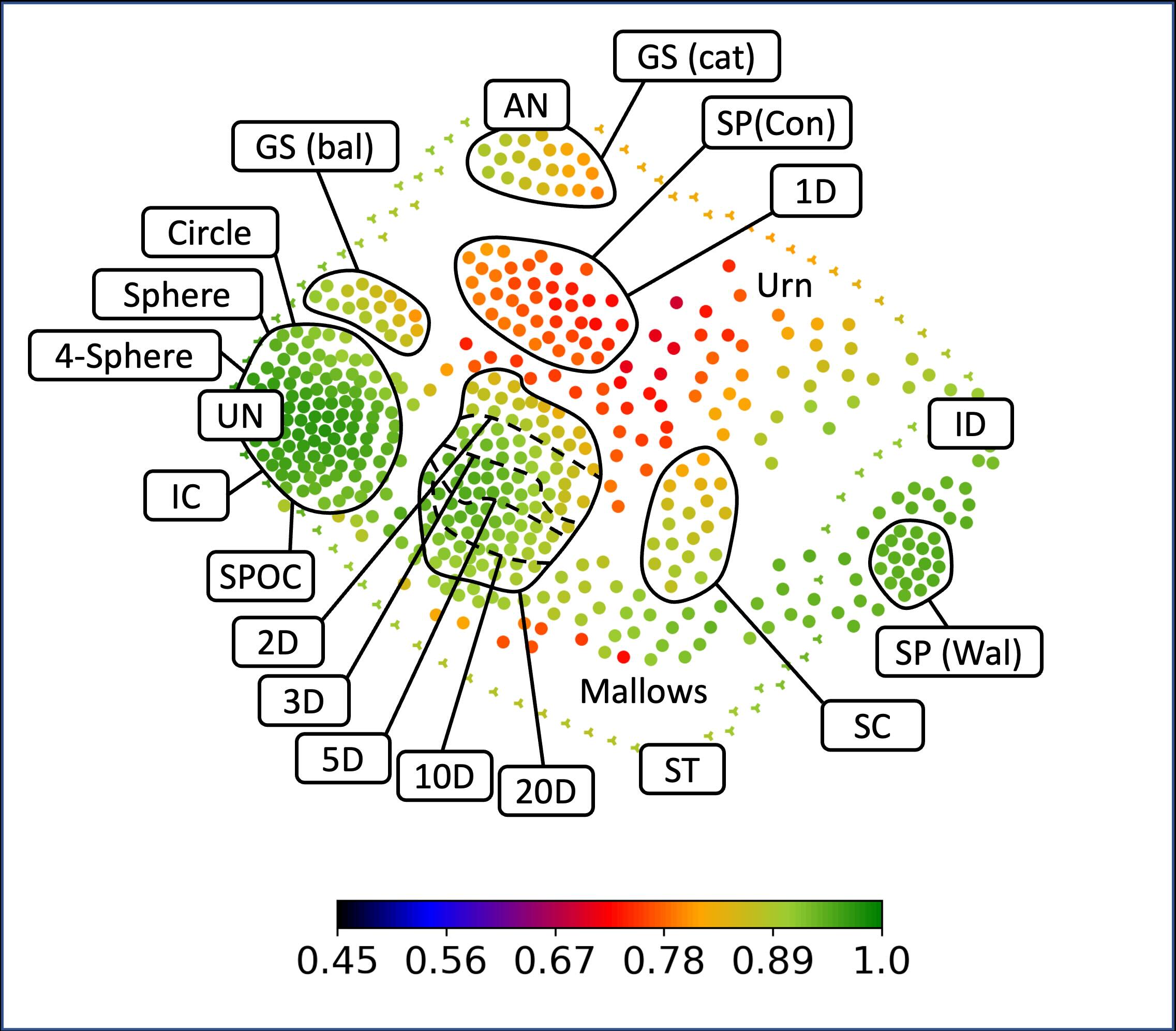}
        \caption{FR}
    \end{subfigure}%
    \begin{subfigure}[b]{0.49\textwidth}
        \centering
        \includegraphics[width=6cm, trim={0.2cm 0.2cm 0.2cm 0.2cm}, clip]{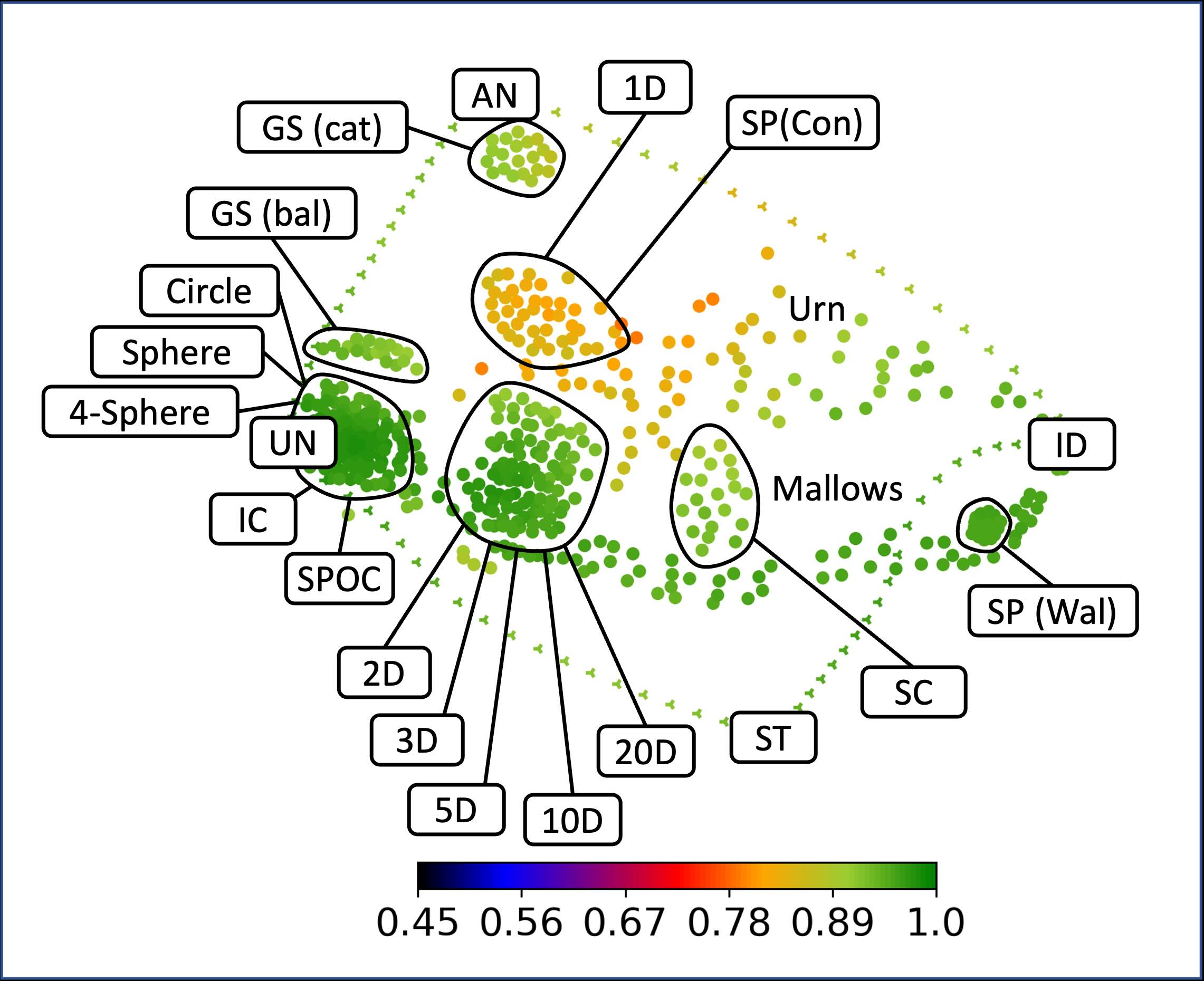}
        \caption{KK}
    \end{subfigure}
    
    \begin{subfigure}[b]{0.49\textwidth}
        \centering
        \includegraphics[width=6cm, trim={0.2cm 0.2cm 0.2cm 0.2cm}, clip]{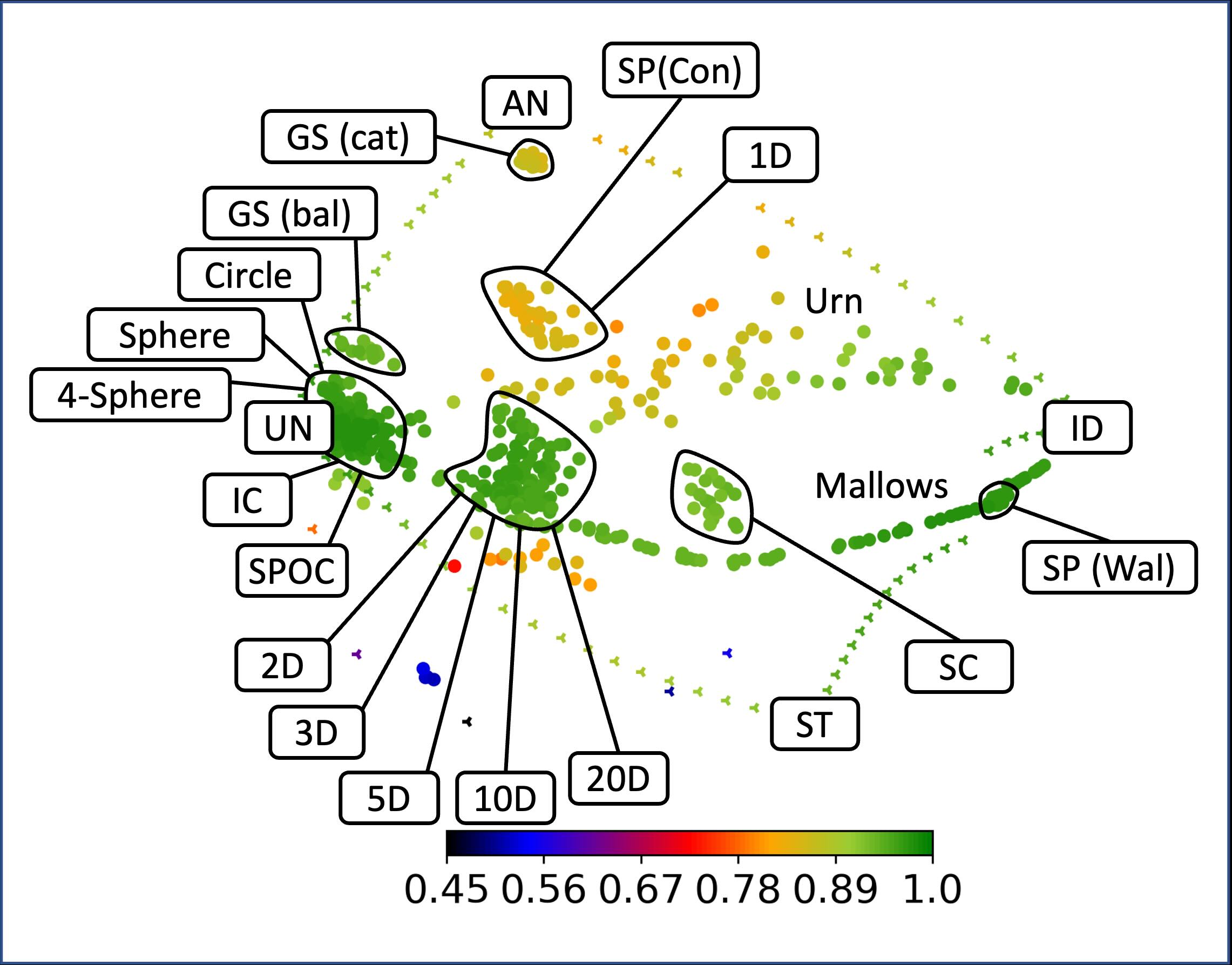}
        \caption{MDS}
    \end{subfigure}
    
    \caption{Monotonicity with~$\epsilon=0$ ($100\times 100$).}
    \label{fig:monotonicity}
\end{figure}

\subsection{Monotonicity}\label{ch:applications:sec:monotonicity}

One of the tools that we use to evaluate the quality of different embeddings is what we call \textit{monotonicity}. The intuition is that if the original distance between elections~$A$ and~$B$ is larger than the original distance between elections~$A$ and~$C$, then we expect that the same will hold for the embedded distances, that is, the embedded distance between elections~$A$ and~$B$ will be larger than the embedded distance between elections~$A$ and~$C$. Now, we move to formal definition.

For a given experiment~$Q=(\mathcal{E}, d_\mathcal{M}, d_\Euc)$ and a given election~$X\in \mathcal{E}$, we define the total monotonicity of this election in this experiment to be:
\[
\mu_{Q}(X) = \sum_{Y,Z \in \mathcal{E}} \Delta_{X}(Y,Z),
\]
where~$\Delta_{X}(Y,Z)$ is equal to 1, if 
\[
\sgn(d_\Euc(X,Y)-d_\Euc(X,Z)) = \sgn(d_\mathcal{M}(X,Y)-d_\mathcal{M}(X,Z)),
\]
and is equal to~$0$ otherwise. Positive (negative) signs mean that both the original and the embedded distances between~$X$ and~$Y$ were larger (smaller) than the distances between~$X$ and~$Z$. In principle, the larger the total monotonicity the better. We also consider a relaxed variant of the monotonicity notion, where in the case of different signs we allow for a small error. Formally, for a given~$\epsilon \in \R$, maximal error~$\Delta_{X}^{\epsilon}(Y,Z)$ is equal to~$1$ if
\[
\sgn(d_\Euc(X,Y)-d_\Euc(X,Z)) = \sgn(d_\mathcal{M}(X,Y)-d_\mathcal{M}(X,Z)), \\
\]
or
\[
|d_\Euc(X,Y)-d_\Euc(X,Z)| \leq \epsilon \cdot \min{(d_\Euc(X,Y), d_\Euc(X,Z))}.
\]
This means that, given target point~$A$ and two other points~$B$ and~$C$, if originally point~$B$ was closer to~$A$ than point~$C$, and after the embedding point~$B$ is further from~$A$ than point~$C$, but the difference between embedded distances between points~$A$ and~$B$, and~$A$ and~$C$ is relatively small, than we can argue that the embedding of~$A$ in relation to~$B$,~$C$ is not perfect but still useful, because the error is small. 

\newcommand{\numberbarMonotonicity}[1]{\tikz{
    \fill[blue!17] (0,0) rectangle (#1*15mm,10pt);
    \node[inner sep=0pt, anchor=south west] at (0,0) {#1};}
}

In \Cref{fig:monotonicity} we present the maps (created using the FR, KK, and MDS embeddings), where each point (election) is colored accordingly to its monotonicity (with~$\epsilon=0$). The larger (the closer to green) the value, the better, and the lower (the closer to black) the value, the worse. Monotonicity equal to~$1$ means that all inequalities are maintained after the embedding. For all three maps, the main message is the same, elections from the IC, SPOC, Mallows, Walsh, and multidimensional Euclidean models are nicely embedded. Then, elections from the single-crossing and group-separable models are still fine, but on average worse than the previously mentioned models. Finally, we have elections from the Interval, Conitzer, and urn models --- which are the worst embedded (not counting some elections from group-separable caterpillar group for MDS embedding, which are obviously wrong). 

\begin{table*}[]
\centering
\begin{tabular}{l|l|l|l}
    \toprule
    Algorithm &~$\epsilon=0$ &~$\epsilon=0.05$ &~$\epsilon=0.1$ \\
    \midrule
    FR & \numberbarMonotonicity{0.887} & \numberbarMonotonicity{0.912} & \numberbarMonotonicity{0.929} \\
    KK & \numberbarMonotonicity{0.928} & \numberbarMonotonicity{0.951} & \numberbarMonotonicity{0.964} \\
    MDS & \numberbarMonotonicity{0.925} & \numberbarMonotonicity{0.947} & \numberbarMonotonicity{0.947} \\ 
    \bottomrule
  \end{tabular}
  \caption{\label{tab:monotonicity}Monotonicity  ($100\times 100$).}
\end{table*}

Moreover, in \Cref{tab:monotonicity} we enclose the average monotonicity for the presented maps, also for two other~$\epsilon$ values,~$0.05$, and~$0.1$. We see that with respect to monotonicity, the KK embedding performs best, with MDS right behind it, and followed by FR.

\subsection{Distortion}\label{ch:applications:sec:distortion}
\begin{figure}[t]
    \centering
    
     \begin{subfigure}[b]{0.49\textwidth}
        \centering
        \includegraphics[width=6.cm, trim={0.2cm 0.2cm 0.2cm 0.2cm}, clip]{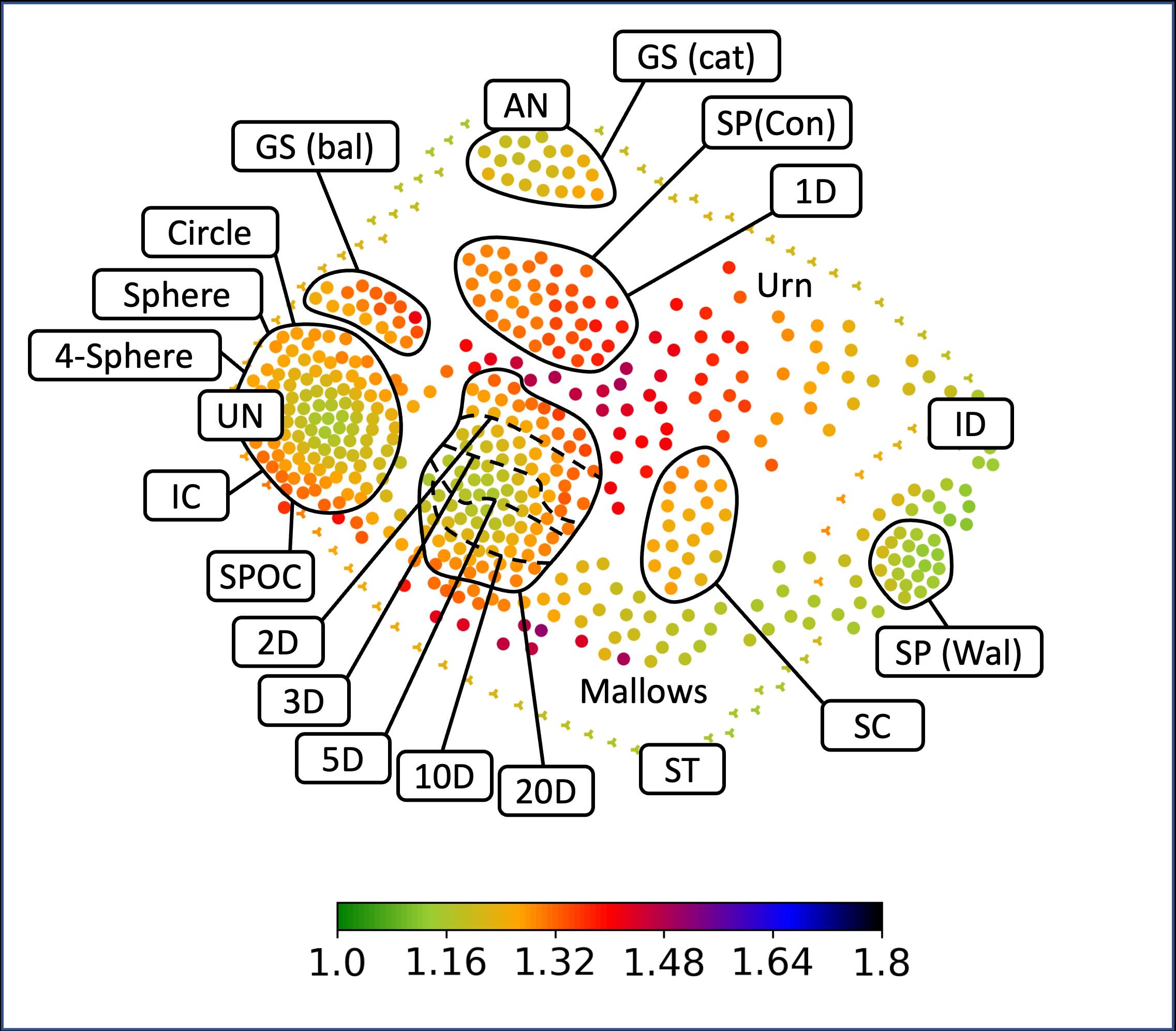}
        \caption{FR}
    \end{subfigure}
     \begin{subfigure}[b]{0.49\textwidth}
        \centering
        \includegraphics[width=6.cm, trim={0.2cm 0.2cm 0.2cm 0.2cm}, clip]{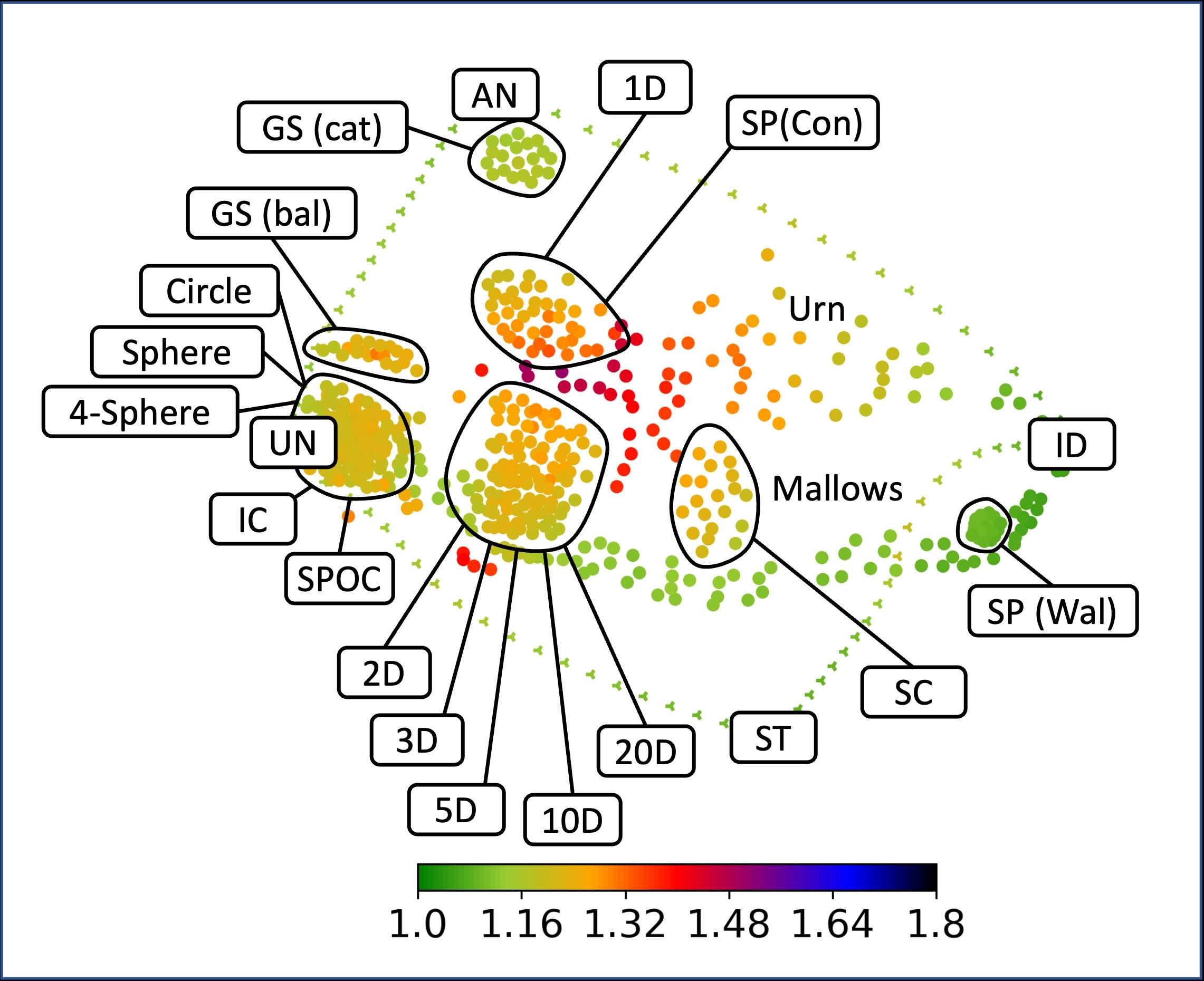}
        \caption{KK}
    \end{subfigure}     
    \begin{subfigure}[b]{0.49\textwidth}
        \centering
        \includegraphics[width=6.cm, trim={0.2cm 0.2cm 0.2cm 0.2cm}, clip]{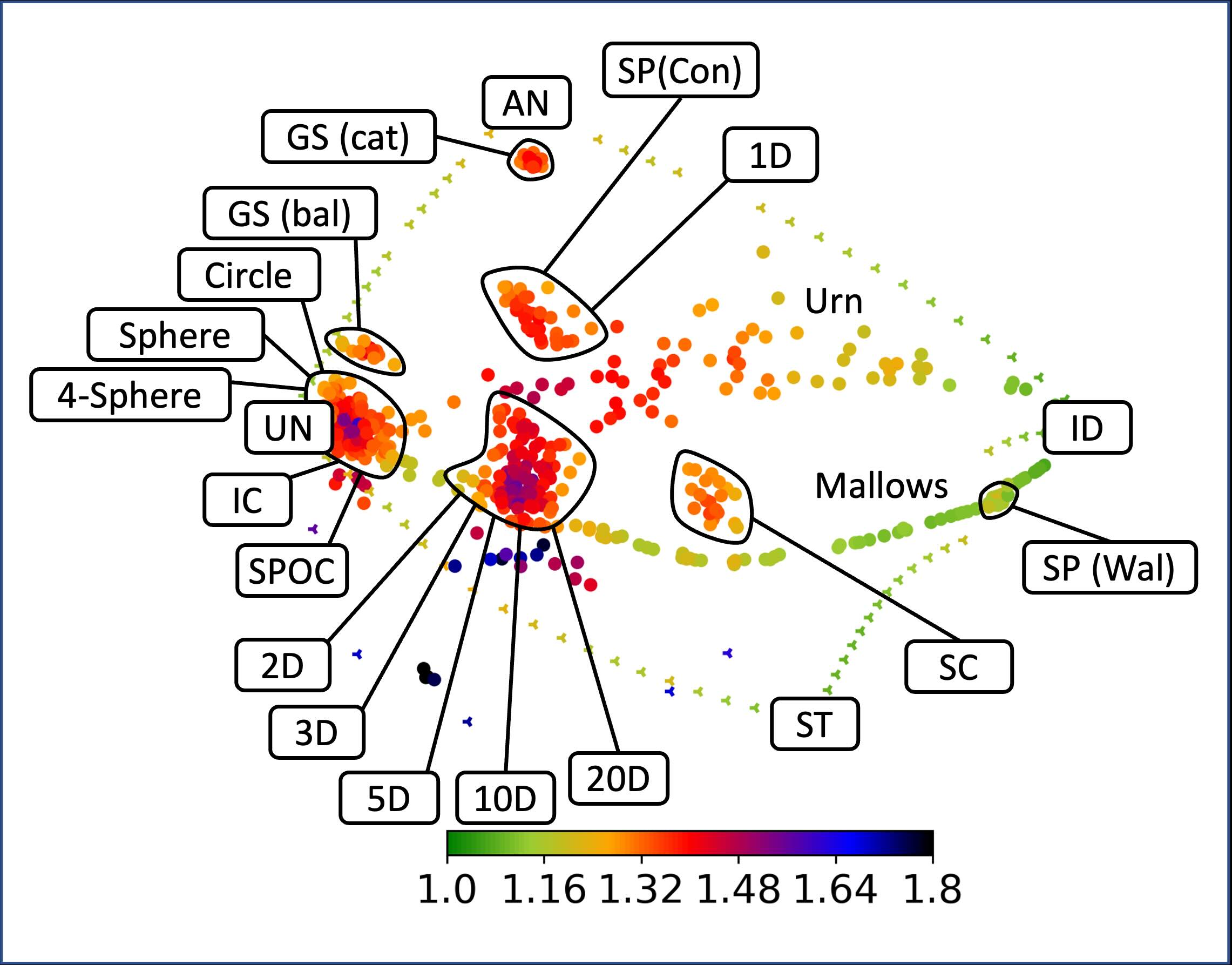}
        \caption{MDS}
    \end{subfigure}

    \caption{\label{fig:distortion} Distortion ($100\times 100$).}
\end{figure}

In addition to monotonicity, we also consider \textit{distortion}. In spirit, it is similar to monotonicity but instead of triplets, it analyzes pairs. The intuition is that the normalized embedded distance should be similar to the normalized original distance. Formally, for a given pair of elections~$X$ and~$Y$ the 
\emph{distortion} is 
defined as:
\[
\MR(X,Y) = \frac{\max(\bar{d}_\Euc(X,Y), \bar{d}_\mathcal{M}(X,Y))}{\min(\bar{d}_\Euc(X,Y), \bar{d}_\mathcal{M}(X,Y))},  
\]
where~$\bar{d}(X,Y)$ means that the
distance between~$X$ and~$Y$ is normalized by the distance between~$\ID$ and~$\UN$. 
For a given experiment~$Q$ and a given election~$X$, we define the total 
\emph{distortion}
of this election in this experiment to be:
\[
\TMR_{Q}(X) = \sum_{Y \in Q} \MR(X,Y).
\]

The closer is the~$\mathrm{TMR}$ value to one, the better---this means that the embedded distanced are proportional to the original ones.

In \Cref{fig:distortion}, we present the maps colored according to their distortion. The best distortion is witnessed by the KK embedding, with elections from the urn model having the worst distortion. For the FR embedding, the situation is very similar, however on average we have slightly worse distortion. For the MDS embedding the situation differs. Besides the misplaced group-separable caterpillar elections and some urn elections in the lower part of the map, we can see dark points in the middle---in the impartial culture cluster, and among highly dimensional Euclidean cluster---which was not the case for the previous two embeddings.

Moreover, in \Cref{tab:distortion} we present the average values for the discussed methods and different numbers of candidates. 
We observe two patterns. The first one is related to the embedding methods: KK is always the best, followed by FR, with MDS being the worst. The second pattern is related to the number of candidates: For all three embedding algorithms the higher the number of candidates, the lower the distortion.
Nonetheless, for each method the distances in the embedding are, on average, off by $20-30\%$. This means that we can get intuitions from the maps, but we always need to carefully verify them.

\newcommand{\numberbarDistortion}[1]{\tikz{
    \fill[red!80!black!16] (0,0) rectangle (#1*12mm,10pt);
    \node[inner sep=0pt, anchor=south west] at (0,0) {#1};}
}

\begin{table*}
\centering
\begin{tabular}{l|l|l|l|l}
    \toprule
    Algorithm &~$4 \times 100$ &~$10 \times 100$ &~$20 \times 100$ &~$100 \times 100$ \\
    \midrule
    FR & \numberbarDistortion{1.322} & \numberbarDistortion{1.282} & 
    \numberbarDistortion{1.272} & \numberbarDistortion{1.255} \\
    KK & \numberbarDistortion{1.258} & \numberbarDistortion{1.248} & 
    \numberbarDistortion{1.236} & \numberbarDistortion{1.200} \\
    MDS & \numberbarDistortion{1.333} & \numberbarDistortion{1.331} & 
    \numberbarDistortion{1.320} & \numberbarDistortion{1.315} \\ 
    \bottomrule
  \end{tabular}
  \caption{\label{tab:distortion}Distortion}
\end{table*}

\subsection{Scalability}

\begin{figure}[ht]
    \centering
    
    \begin{subfigure}[b]{0.49\textwidth}
        \centering
        \includegraphics[width=6.cm, trim={0.2cm 0.2cm 0.2cm 0.2cm}, clip]
        {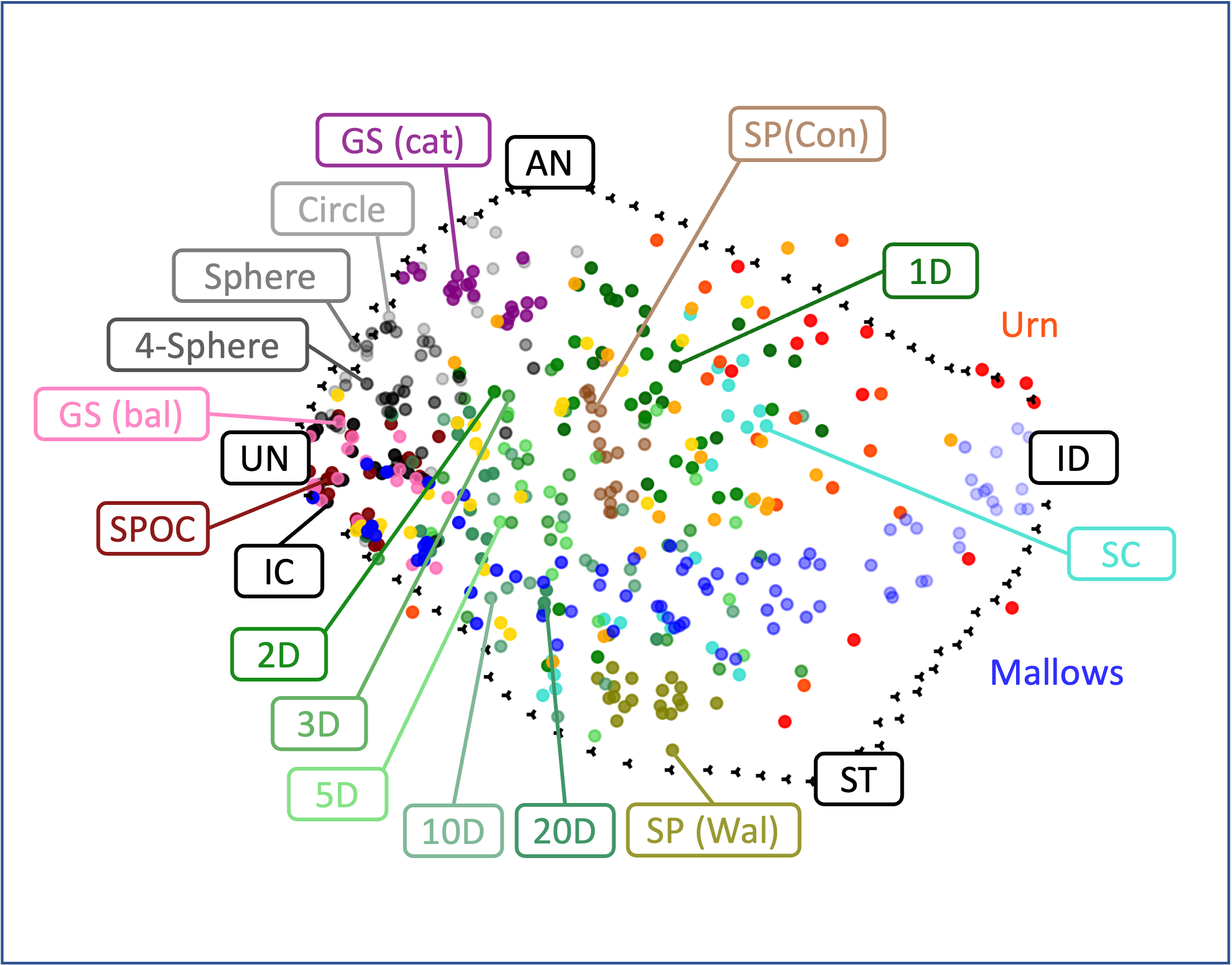}
        \caption{4 candidates \& 100 voters}
    \end{subfigure}
    \begin{subfigure}[b]{0.49\textwidth}
        \centering
        \includegraphics[width=6.cm, trim={0.2cm 0.2cm 0.2cm 0.2cm}, clip]
        {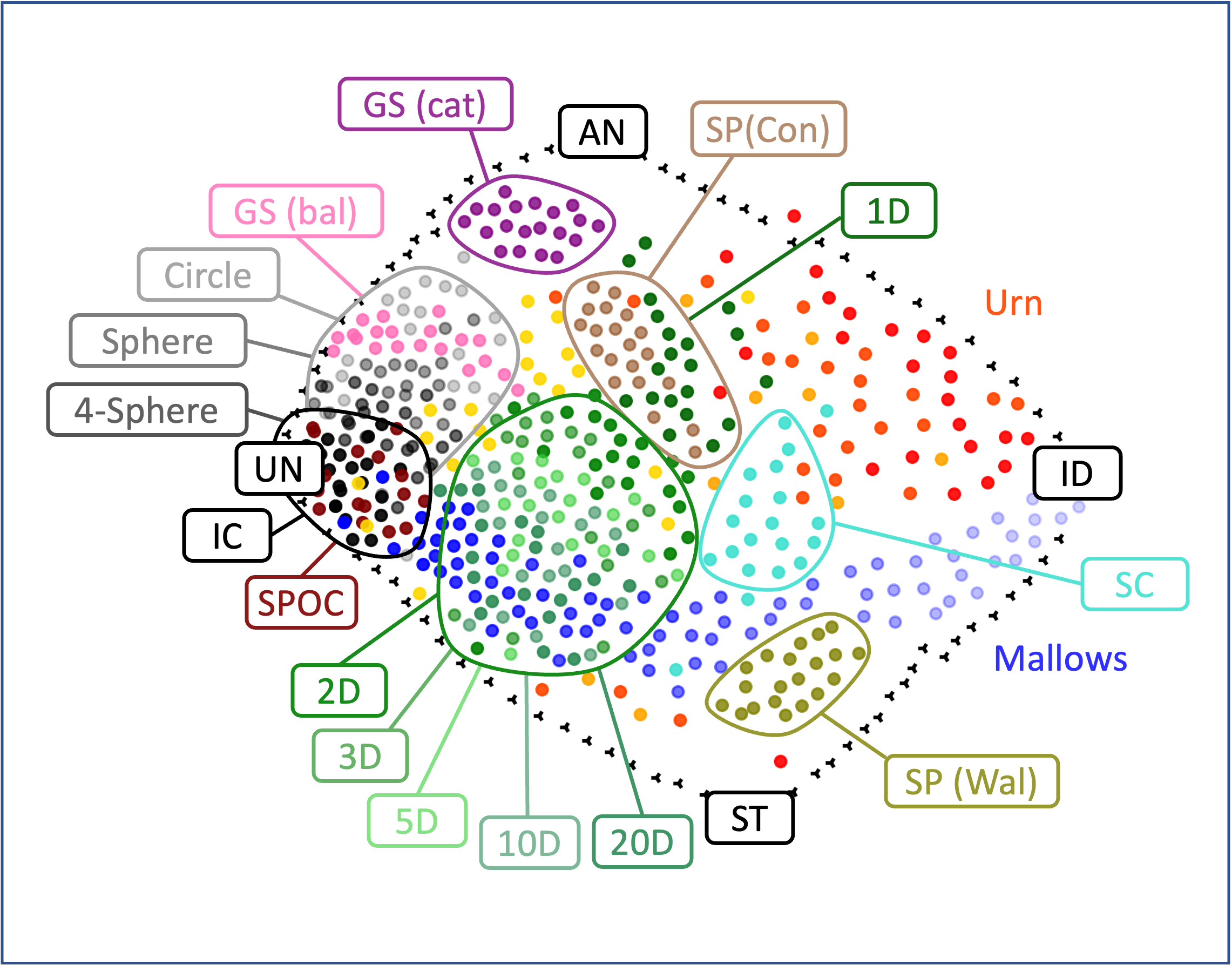}
        \caption{10 candidates \& 100 voters}
    \end{subfigure}

    \vspace{1em}

    \begin{subfigure}[b]{0.49\textwidth}
        \centering
        \includegraphics[width=6.cm, trim={0.2cm 0.2cm 0.2cm 0.2cm}, clip]
        {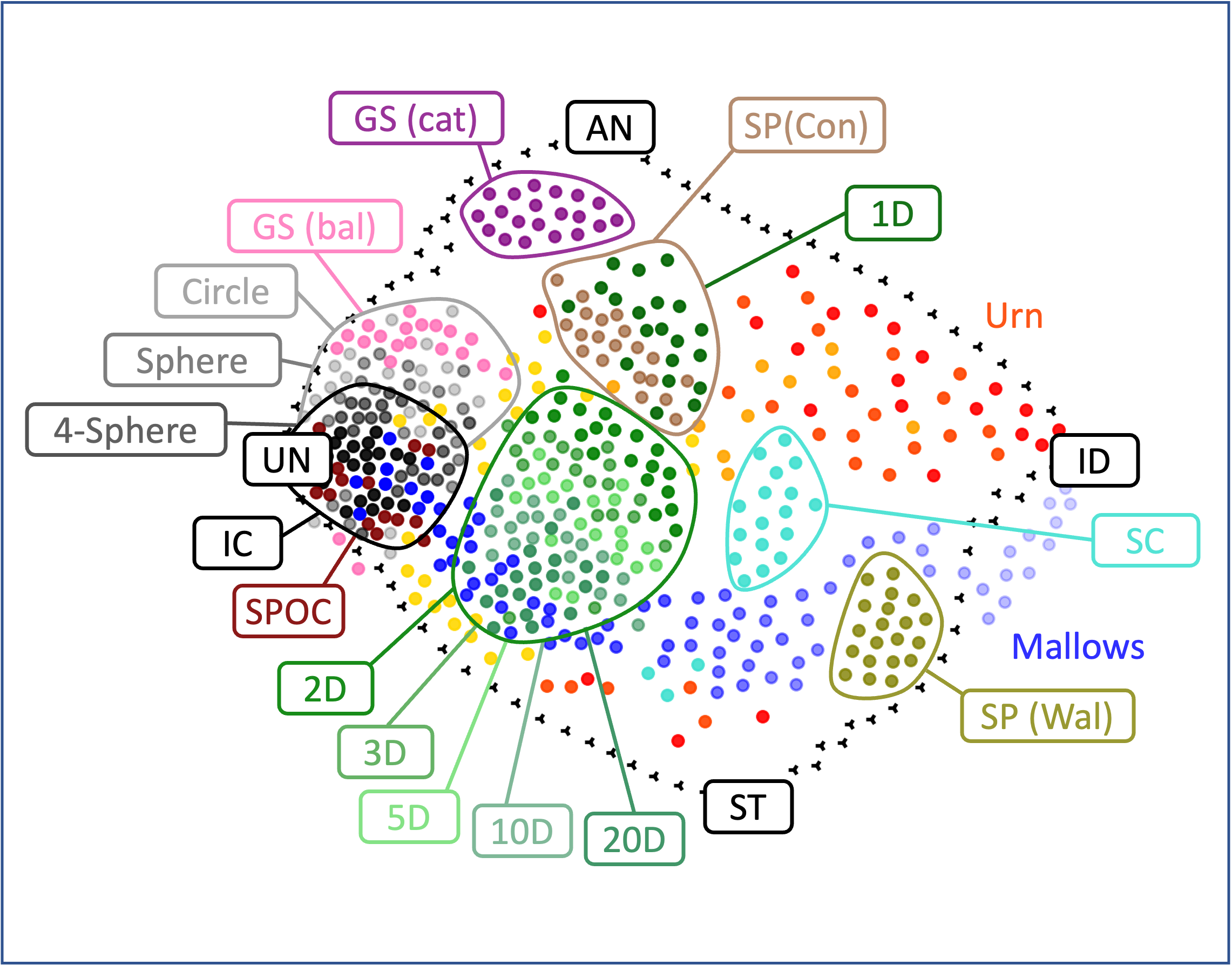}
        \caption{20 candidates \& 100 voters}
    \end{subfigure}
    \begin{subfigure}[b]{0.49\textwidth}
        \centering
        \includegraphics[width=6.cm, trim={0.2cm 0.2cm 0.2cm 0.2cm}, clip]
        {img/embed/spring_emd-positionwise.png}
        \caption{100 candidates \& 100 voters}
    \end{subfigure}
    
    \caption{Maps of elections with different number of candidates.}
    \label{fig:scale}
\end{figure}

In this section we compare the embedding results for elections with different numbers of candidates. In \Cref{fig:scale}, we present four maps with~$4$,~$10$,~$20$, and~$100$ candidates, created using the FR embedding. As we can see, the maps for~$10$,~$20$, and~$100$ candidates are surprisingly similar. The largest difference we can observe is for the balanced group-separable elections---for the case of~$100$ candidates they are clearly separated, while for~$10$ and~$20$ candidates they are mingling with the Circle elections. In addition, multidimensional Euclidean elections are better separated for~$100$ candidates than for~$10$ or~$20$ candidates.
Another interesting thing is that the Walsh elections are shifting toward the right side of the embedding as we increase the number of candidates. At the same time, the caterpillar group-separable elections are shifting toward~$\AN$.

Only the map with four candidates is different. Yet, the main shape is still maintained. Note that for~$4$ candidates, there are only~$24$ possible different votes, which likely explains why the map is not as meaningful. In principle, elections with only four candidates are very similar to each other (we confirm this in~\Cref{ch:sub:sec:experiments}).

\subsubsection{Parametrized Models}
Some of the statistical cultures that we study are parametrized. We are especially interested in the urn and Mallows models, because for both these models, depending on the values of their parameters we can generate elections that are either similar to~$\ID$ or to~$\UN$ or lie somewhere in between. What we would like to verify is whether the urn or Mallows elections with a given parameter occupy the same part of the map regardless of the number of candidates.

\paragraph{Urn Model}
\begin{figure}[t]
    \centering
    
     \begin{subfigure}[b]{0.32\textwidth}
        \centering
        \includegraphics[width=4.cm, trim={0.2cm 0.2cm 0.2cm 0.2cm}, clip]{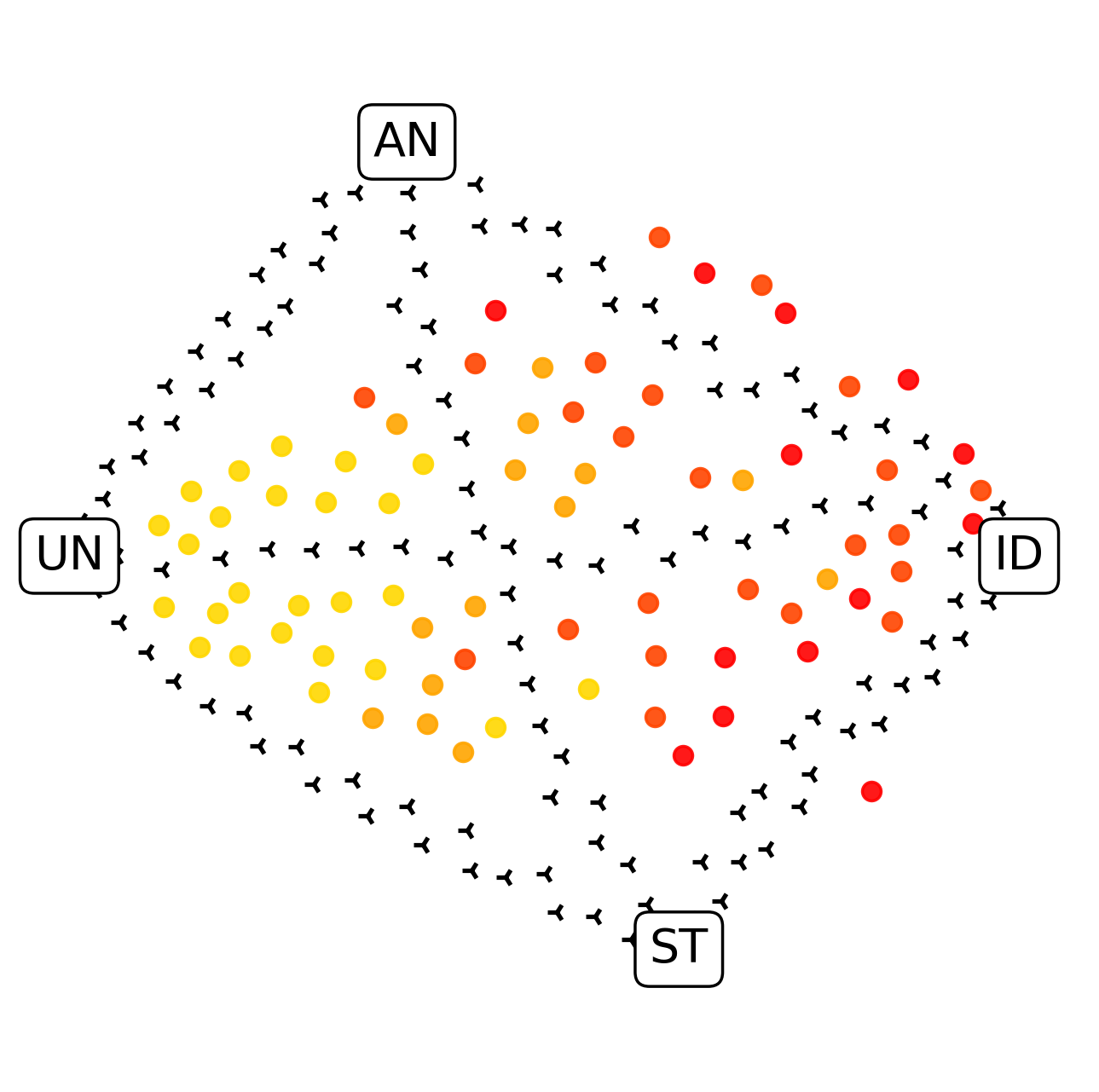}
        \caption{$10 \times 100$}
    \end{subfigure}
     \begin{subfigure}[b]{0.32\textwidth}
        \centering
        \includegraphics[width=4.cm, trim={0.2cm 0.2cm 0.2cm 0.2cm}, clip]{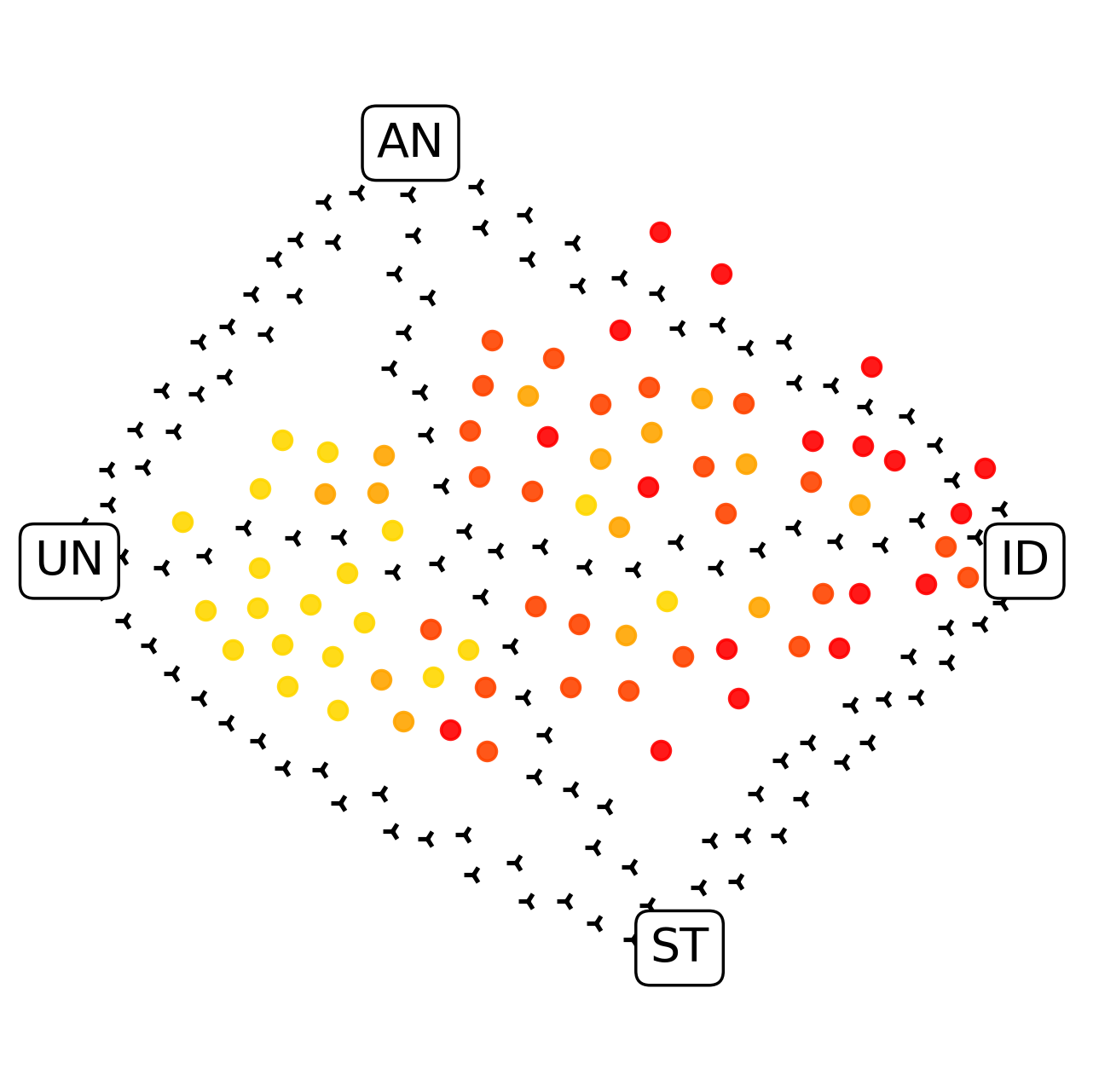}
        \caption{$20 \times 100$}
    \end{subfigure}     
    \begin{subfigure}[b]{0.32\textwidth}
        \centering
        \includegraphics[width=4.cm, trim={0.2cm 0.2cm 0.2cm 0.2cm}, clip]{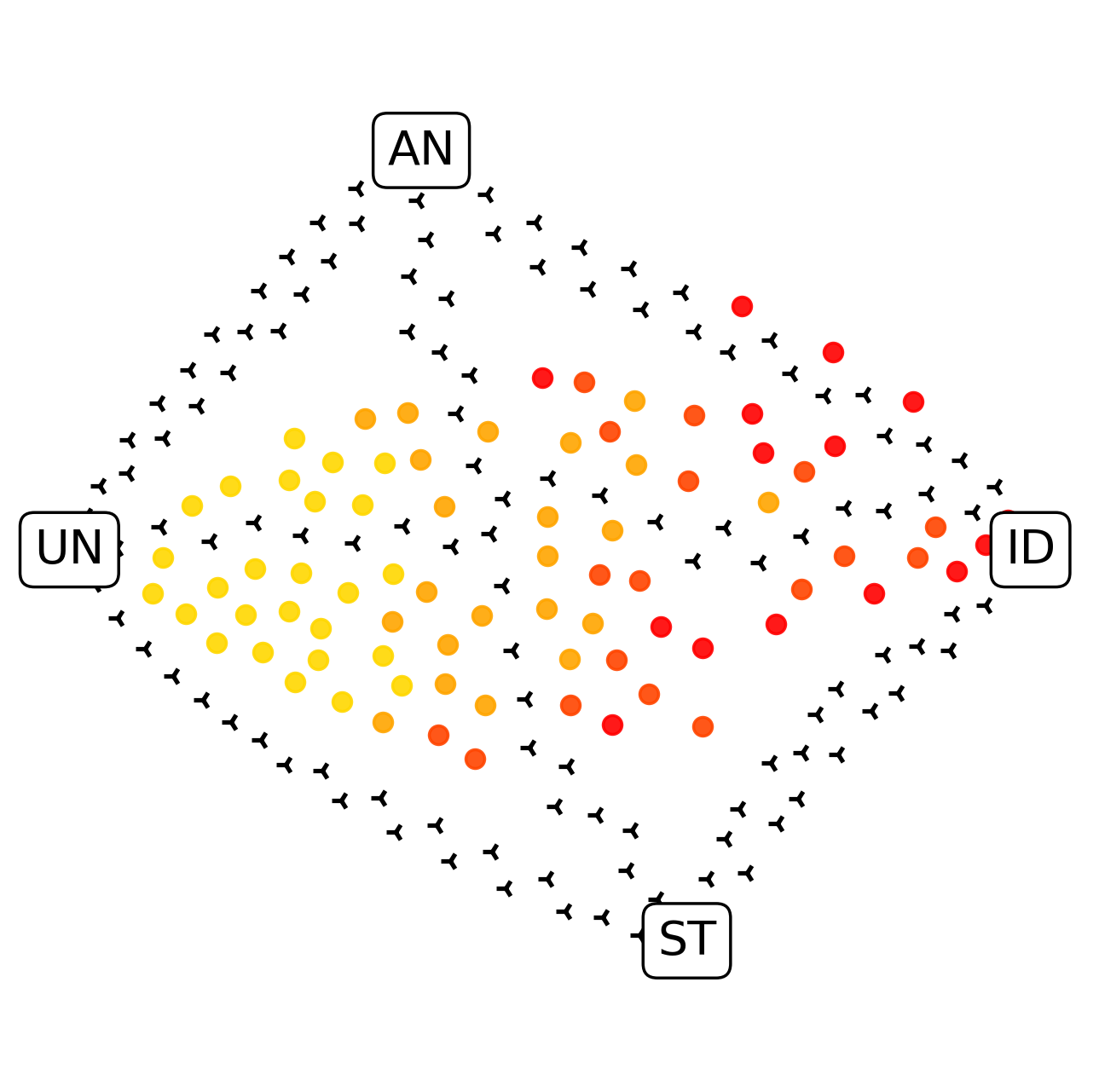}
        \caption{$40 \times 100$}
    \end{subfigure}

    \caption{Scalability of the urn model;~$\alpha$~follows the Gamma distribution.}
    \label{fig:urn_scalability}
\end{figure}

Interestingly, given certain parameters~$\alpha$, the urn model behaves the same no matter how many candidates we have. In \Cref{fig:urn_scalability} we present the comparison of the urn model elections with~$100$ voters and~$10$,~$20$, and~$40$ candidates. We generated six paths between the compass points,~$20$ points each, and~$80$ elections from the urn model, where~$\alpha$ was sampled from the~$Gamma(0.8,1)$ distribution.\footnote{It is far from obvious how to sensibly sample the~$\alpha$ parameter. Firstly, its domain is unbounded on one side. Secondly, the larger are the~$\alpha$ values, the smaller are the differences between elections generated using them. For example, changing~$\alpha$ from~$0.1$ to~$0.3$ is influencing the result far more than changing it from~$2.1$ to~$2.3$. This means that we would like to have a decreasing probability density function. Taking the above into consideration, we suggest using the \emph{Gamma} distribution. A particular selection of parameters (that is, \emph{shape} equal to 0.8, and \emph{scale} equal to 1) produces the outcome where more or less half of the urn elections lie closer to~$\ID$ and the other half lie closer to~$\UN$.} 
As we can see, all three maps are similar to each other. 


\paragraph{Mallows Model}
\begin{figure}[ht]
    \centering

     \begin{subfigure}[b]{0.32\textwidth}
        \centering
        \includegraphics[width=4.cm, trim={0.2cm 0.2cm 0.2cm 0.2cm}, clip]{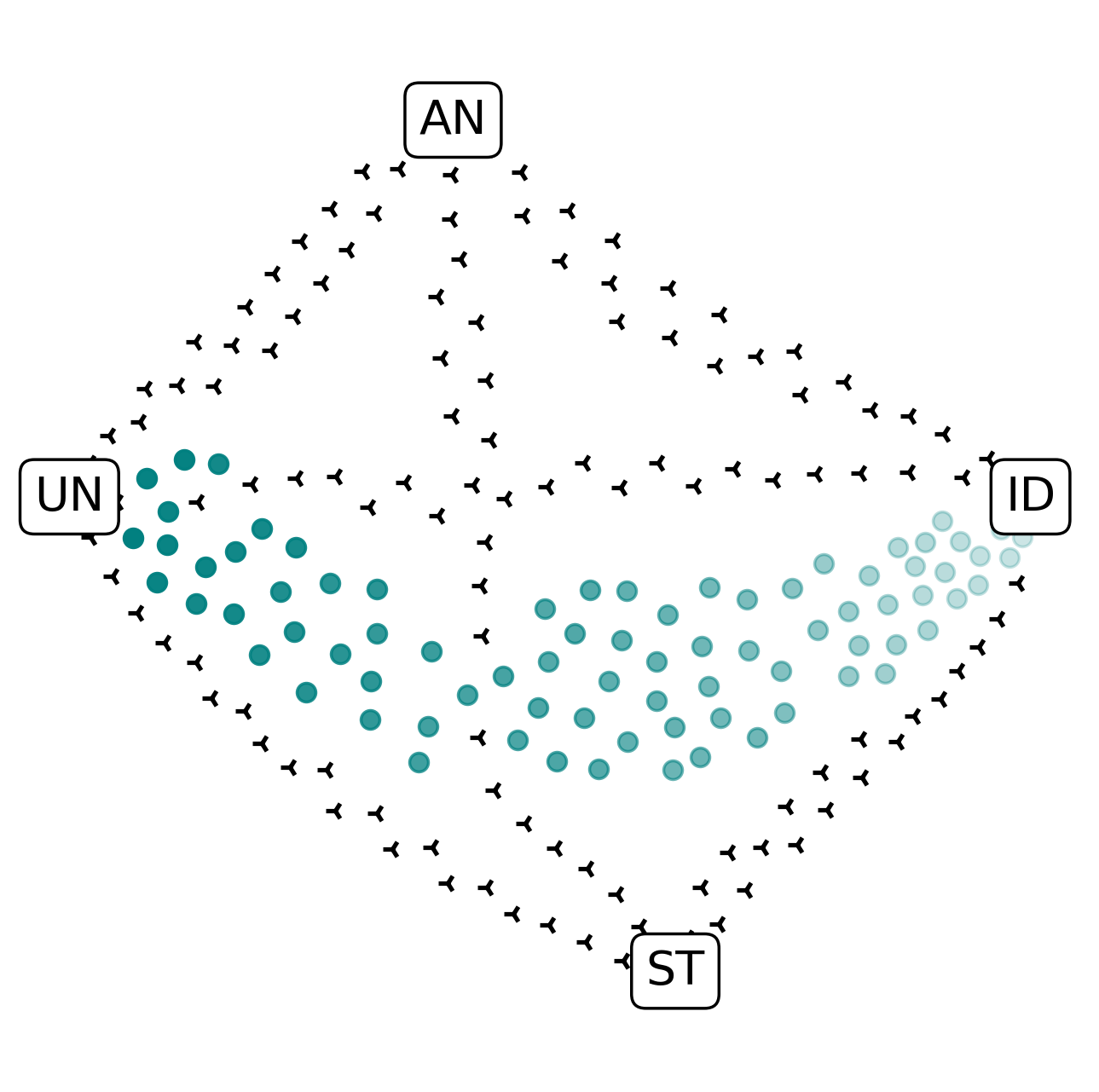}
    \end{subfigure}
     \begin{subfigure}[b]{0.32\textwidth}
        \centering
        \includegraphics[width=4.cm, trim={0.2cm 0.2cm 0.2cm 0.2cm}, clip]{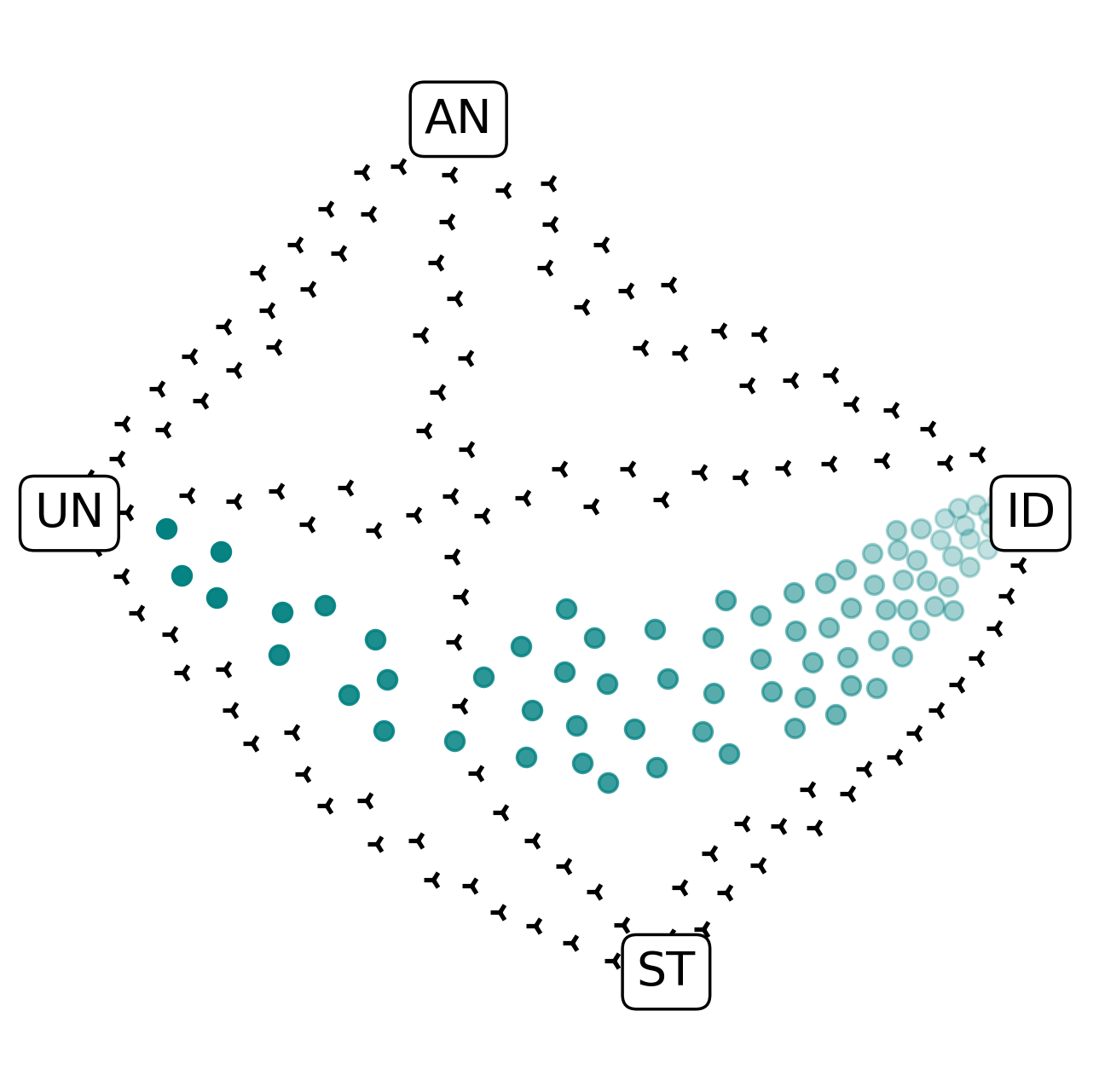}
    \end{subfigure}     
    \begin{subfigure}[b]{0.32\textwidth}
        \centering
        \includegraphics[width=4.cm, trim={0.2cm 0.2cm 0.2cm 0.2cm}, clip]{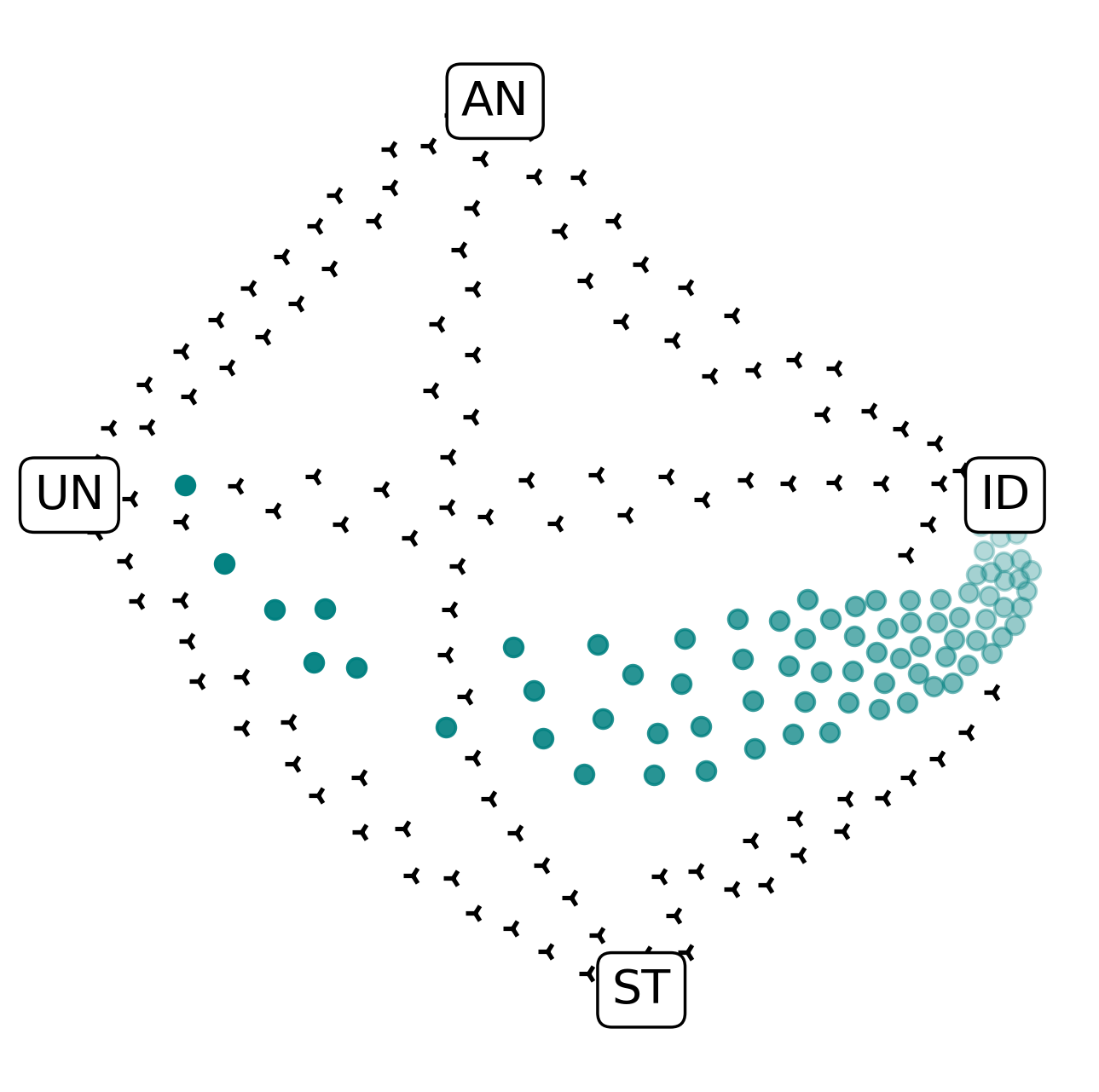}
    \end{subfigure}
    
    \begin{subfigure}[b]{0.32\textwidth}
        \centering
        \includegraphics[width=4.cm, trim={0.2cm 0.2cm 0.2cm 0.2cm}, clip]{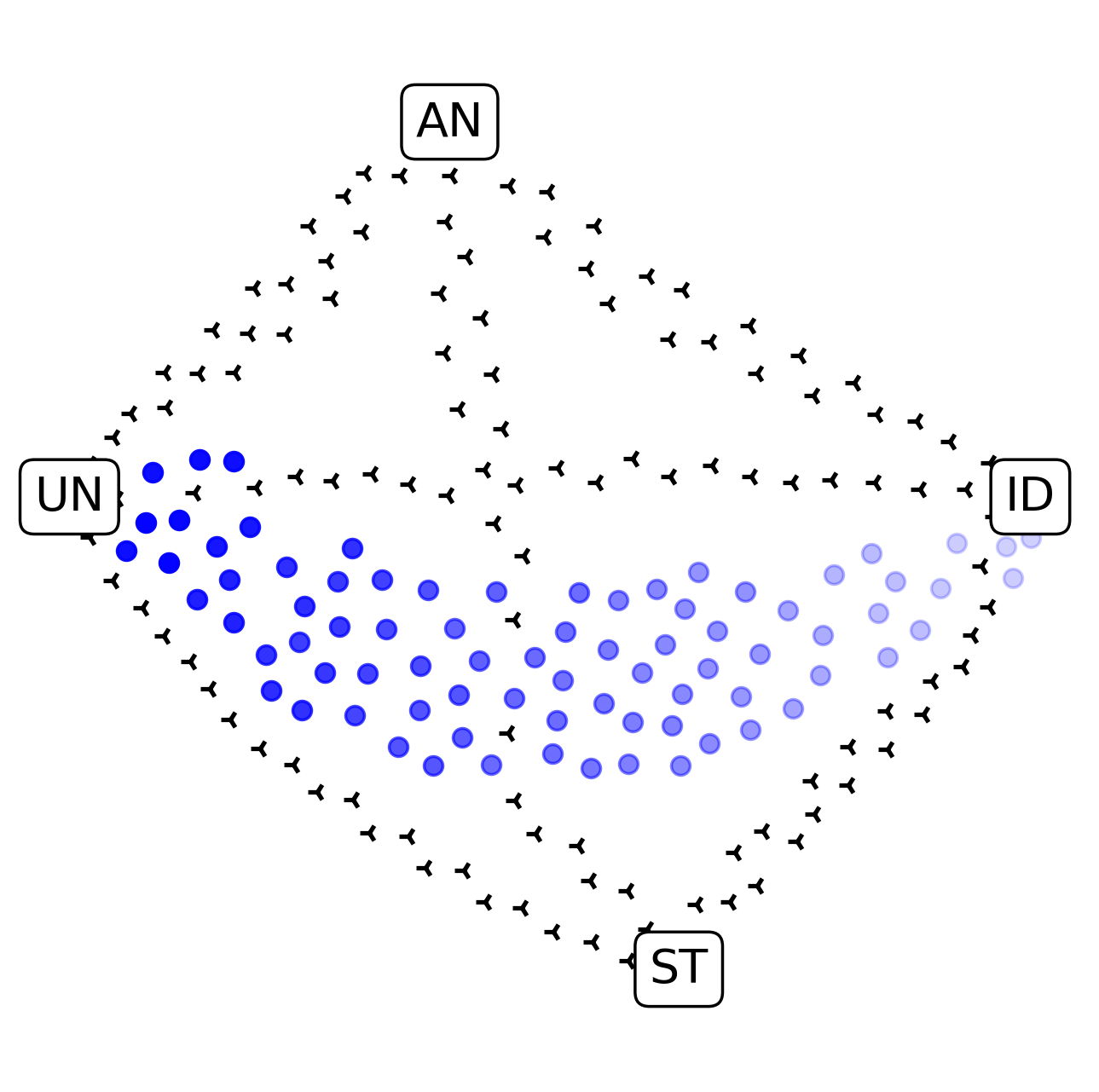}
        \caption{$10 \times 100$}
    \end{subfigure}
     \begin{subfigure}[b]{0.32\textwidth}
        \centering
        \includegraphics[width=4.cm, trim={0.2cm 0.2cm 0.2cm 0.2cm}, clip]{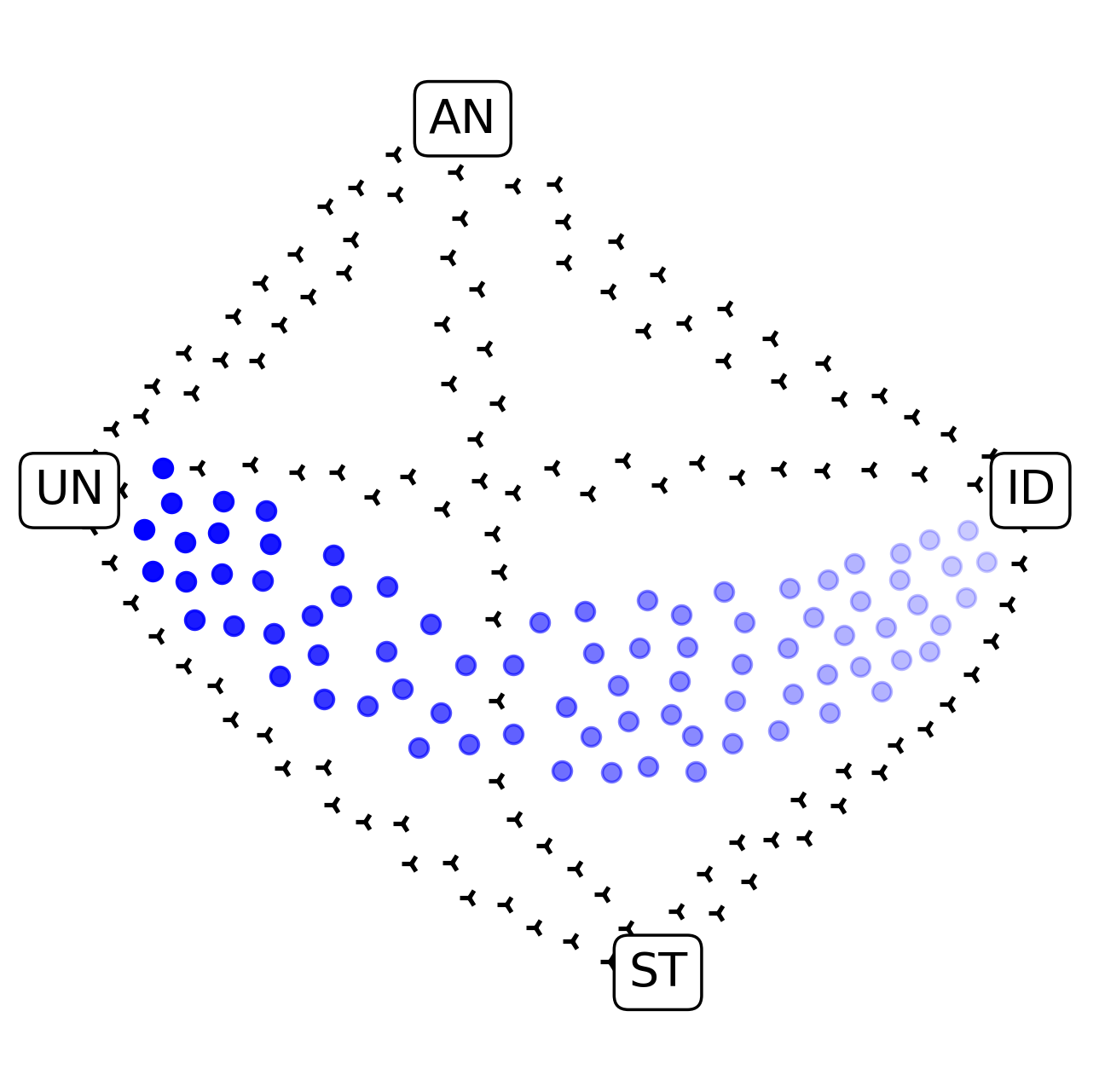}
        \caption{$20 \times 100$}
    \end{subfigure}     
    \begin{subfigure}[b]{0.32\textwidth}
        \centering
        \includegraphics[width=4.cm, trim={0.2cm 0.2cm 0.2cm 0.2cm}, clip]{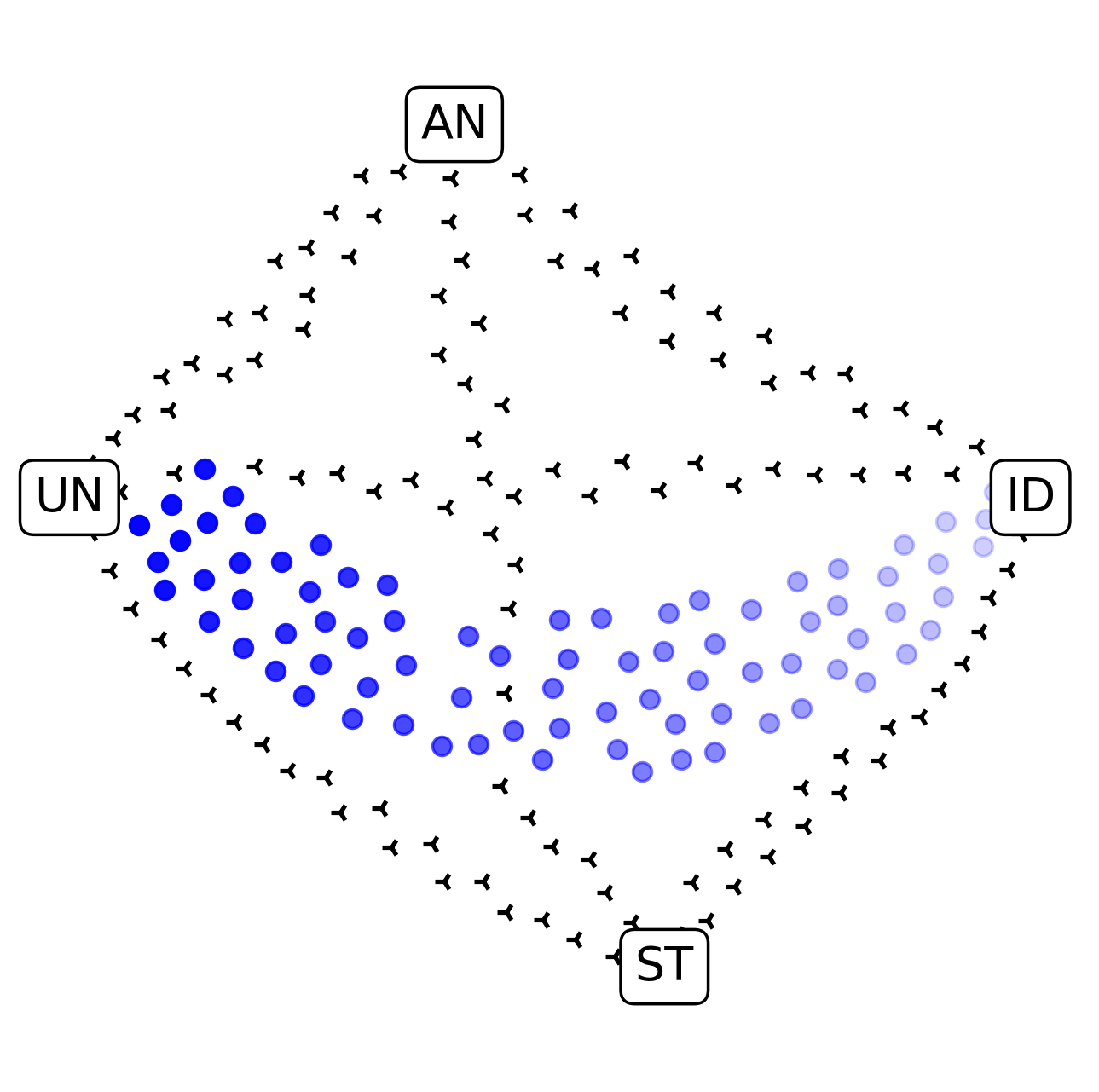}
        \caption{$40 \times 100$}
    \end{subfigure}

    \caption{Scalability of the Mallows model. Teal points in the upper row represent elections from the Mallows model, while blue points in the lower row represent elections from the Normalized Mallows model;~$\phi$ and~$\normphi$ follow the uniform distribution.}
    \label{fig:mallows_scalability}
\end{figure}

Unfortunately, for the standard Mallows model (i.e., not normalized one), given a certain~$\phi$ parameter, the more candidates we have, the closer our elections are to the identity ones. It means that when comparing the results of some experiment with a fixed~$\phi$ and different numbers of candidates, we might get a false impression of some phenomena. Fortunately, \cite{boe-bre-fal-nie-szu:c:compass} propose a useful way of normalizing the~$\phi$ parameter by the number of swaps---whereby, if we use the normalized parameter, we maintain the same position between uniformity and identity even if we change the number of candidates.
In \Cref{fig:mallows_scalability} we present a comparison of the standard Mallows model and its normalized variant, for elections with~$100$ voters and~$10$,~$20$, and~$40$ candidates. As before, we generated six paths between the compass points,~$20$ points each, and~$80$ elections from the Mallows (Norm-Mallows) model, where~$\phi$ ($\normphi$) was sampled from the uniform distribution. In the upper row, we show the results for the Mallows model, and in the lower row for the Norm-Mallows model. For the Norm-Mallows model, all three pictures look more or less the same, as was the case for the urn model. However, for elections from the standard Mallows model, we see that the more candidates we have, the more are the points shifted toward the identity.

\vspace{0.2cm}
\begin{conclusionbox}
\begin{itemize}
    \item Depending on the context, all three embedding algorithms (i.e., FR, KK, and MDS) might be a reasonable choice, with FR generating the most scattered maps, and MDS the least. Nonetheless, with regard to monotonicity and distortion criteria, KK performs best.
    \item Map of elections framework is not recommended for maps with few candidates.
    \item One has to be careful when comparing elections of different sizes (in particular, when comparing elections with different numbers of candidates). 
\end{itemize}
\end{conclusionbox}

\section{Voting Rules}\label{ch:applications:sec:rules}
To demonstrate the usefulness of our map framework, we show several practical applications related to analysis of voting rules. We use the same elections as described in~\Cref{tab:embed_setup} with~$100$ candidates and~$100$ voters. First we focus on scores obtained by winning candidates for single-winner rules, and winning committees for multiwinner rules. Then we look more closely at the running time of ILP-based algorithms for selected NP-hard rules. Finally, we focus on several approximation methods used to approximate the highest score for the Chamberlin--Courant and Harmonic-Borda voting rules.

\subsection{Score}\label{ch:applications:sec:score}
In this section we present the behavior of various voting rules on the map. We start with single-winner voting rules such as the Plurality, Borda, Copeland, and Dodgson. Next, we consider the following two multiwinner voting rules: Chamberlin--Courant and Harmonic-Borda. 

\subsubsection{Single-winner Voting Rules}
All rules that we discuss work similarly, that is, we compute a score for each candidate and then the candidate with the highest (or lowest) score wins the election:

\begin{description}
    \item[Plurality.] Each voter assigns one point to his or her favorite candidate. The candidate with the highest score wins.
    \item[Borda.] Each voter assigns~$m-1$ points to his or her favorite candidate,~$m-2$ points to his or her second favorite candidate, and so on.  The candidate with the highest score wins.
    \item[Copeland.] We examine all pairs of candidates. In each pair, the candidate who is preferred by more than half of the voters gets a point. In case of a draw, both candidates receive half a point. The candidate with the highest score wins.
    \item[Dodgson.] Condorcet winner is a candidate who, when compared one-to-one with every other candidate, is preferred by more than half of the voters. For each candidate, we check what is the minimum number of swaps of adjacent candidates in the votes needed to make him or her the Condorcet winner. The candidate for whom the value is the lowest wins.
\end{description}

Note that if a Condorcet winner exists, he or she will always be selected by the Copeland and Dodgson rules. A rule that always elects the Condorcet winner when one exists is called a Condorcet Extension.

\begin{example}
  Consider an election~$E = (C,V)$, where~$C = \{a,b,c,d,e\}$,~$V = (v_1,v_2,v_3,v_4,v_5)$, and the votes are as follows:
  \begin{align*}
    \small
    v_1\colon&  a \pref c \pref b \pref d \pref e, \\
    v_2\colon&  a \pref c \pref b \pref e \pref d, \\
    v_3\colon&  d \pref e \pref b \pref c \pref a, \\
    v_4\colon&  b \pref e \pref d \pref c \pref a, \\
    v_5\colon&  c \pref b \pref e \pref d \pref a.
 \end{align*}
  
According to the plurality rule,~$a$ is the winner having score of~$2$, while all the other candidates have score of~$0$ or~$1$. According to the Borda rule,~$b$ is the winner having~$13$ points, followed by~$c$ with~$12$ points. Then we have~$e$ with~$9$ points, and in the end there are~$a$ and~$d$ having~$8$ points each. Under the Copeland rule,~$a$ has~$0$ points (losing all duels),~$c$ has~$4$ points (winning all duels),~$b$ has~$3$ points,~$d$ has~$1$ point, and~$e$ has~$2$ points. Clearly,~$c$ is selected as the winner. Moreover, note that~$c$ is a Condorcet winner as well. Hence, he or she will also win under the Dodgson rule, having the lowest score of 0 (i.e., no swap is needed to make him or her a Condorcet winner).
\end{example}

For each of the elections in our map, we computed a winning candidate and his or her score. For Plurality, Borda, and Copeland, the higher, the better, while for Dodgson the lower, the better. Computing the Plurality, Borda, and Copeland scores is straightforward. To compute the Dodgson score, we used the ILP proposed by~\cite{bar-tov-tri:j:who-won}.
We present the results in \Cref{fig:score_single_winner}, where the color of each point corresponds to the highest (lowest for Dodgson) score obtained by the winning candidate.

\begin{figure}[]
    \centering
    
    \begin{subfigure}[b]{0.49\textwidth}
        \centering
        \includegraphics[width=6.5cm, trim={0.2cm 0.2cm 0.2cm 0.2cm}, clip]
        {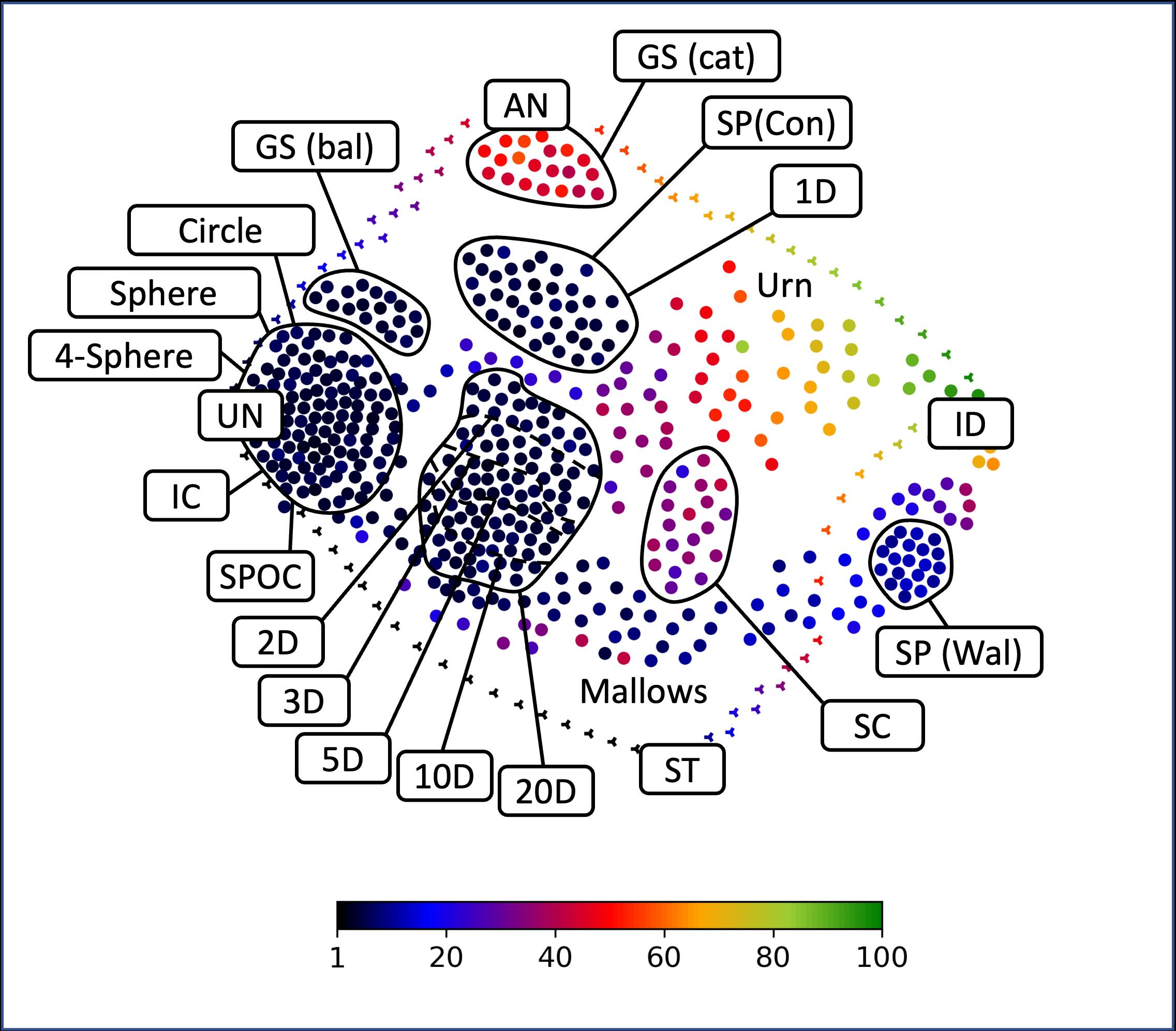}
        \caption{Highest Plurality score}
    \end{subfigure}
    \begin{subfigure}[b]{0.49\textwidth}
        \centering
        \includegraphics[width=6.5cm, trim={0.2cm 0.2cm 0.2cm 0.2cm}, clip]
        {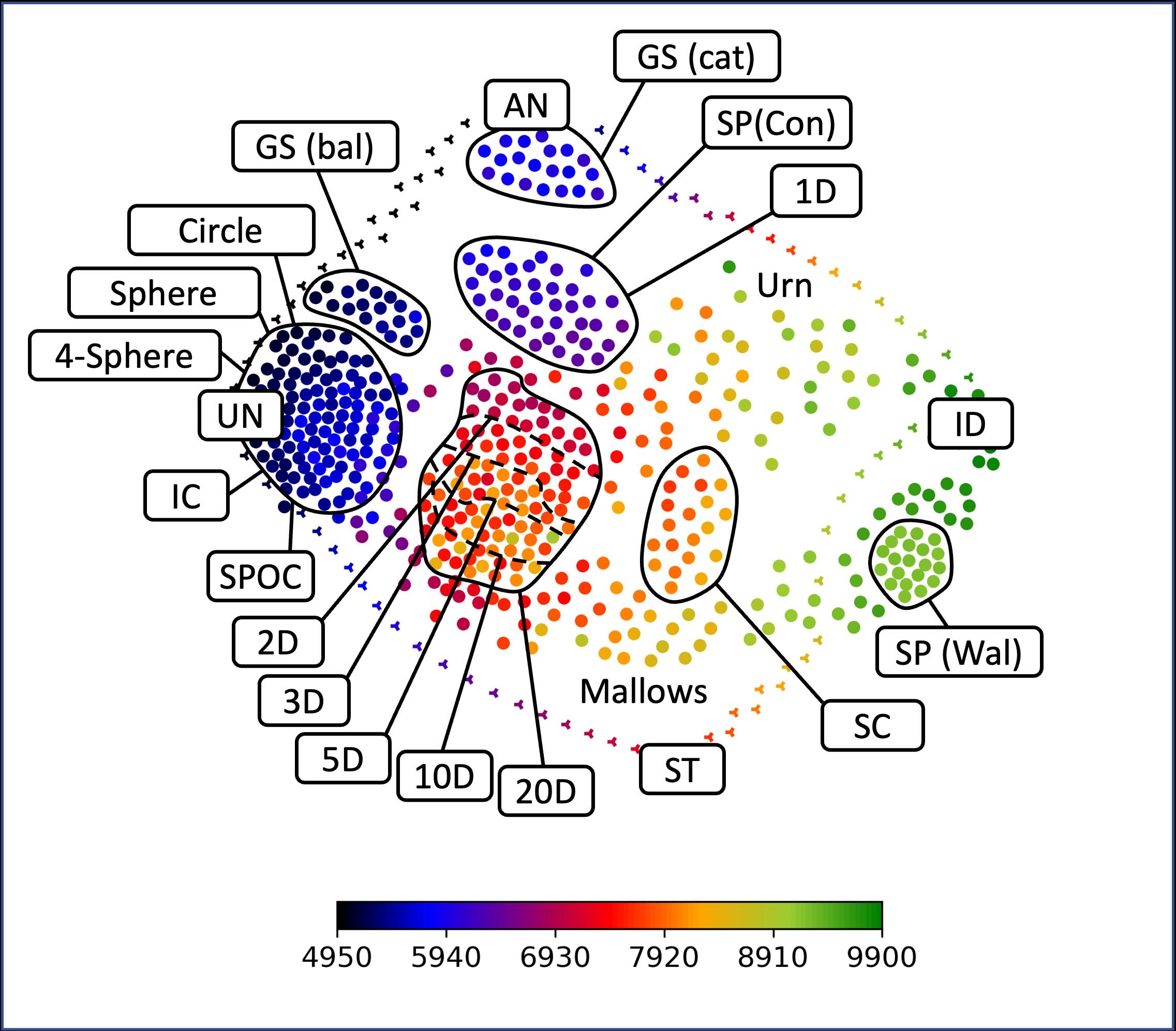}
        \caption{Highest Borda score}
    \end{subfigure}

    \vspace{1em}
    \begin{subfigure}[b]{0.49\textwidth}
        \centering
        \includegraphics[width=6.5cm, trim={0.2cm 0.2cm 0.2cm 0.2cm}, clip]
        {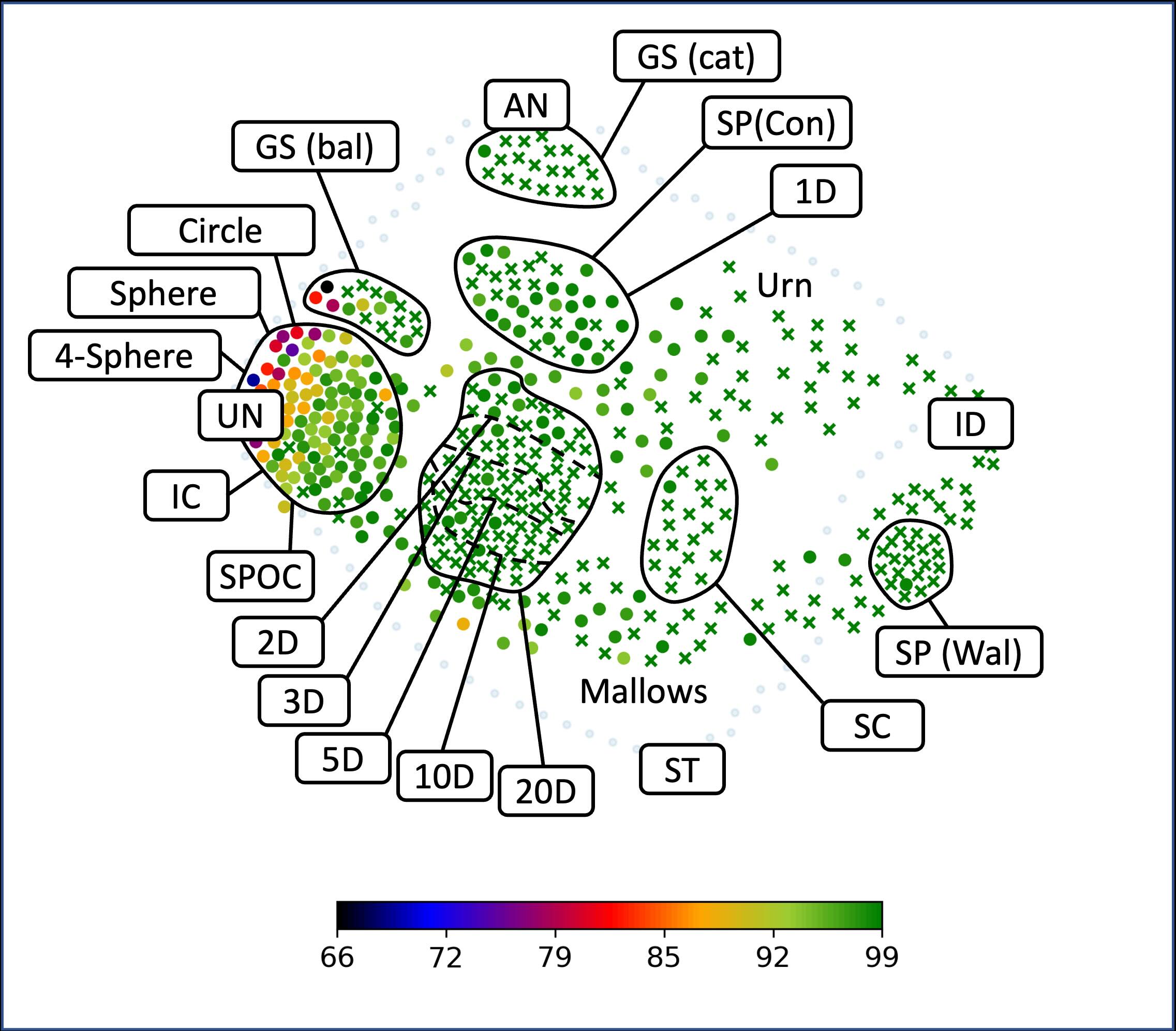}
        \caption{Highest Copeland score}
    \end{subfigure}
    \begin{subfigure}[b]{0.49\textwidth}
        \centering
        \includegraphics[width=6.5cm, trim={0.2cm 0.2cm 0.2cm 0.2cm}, clip]
        {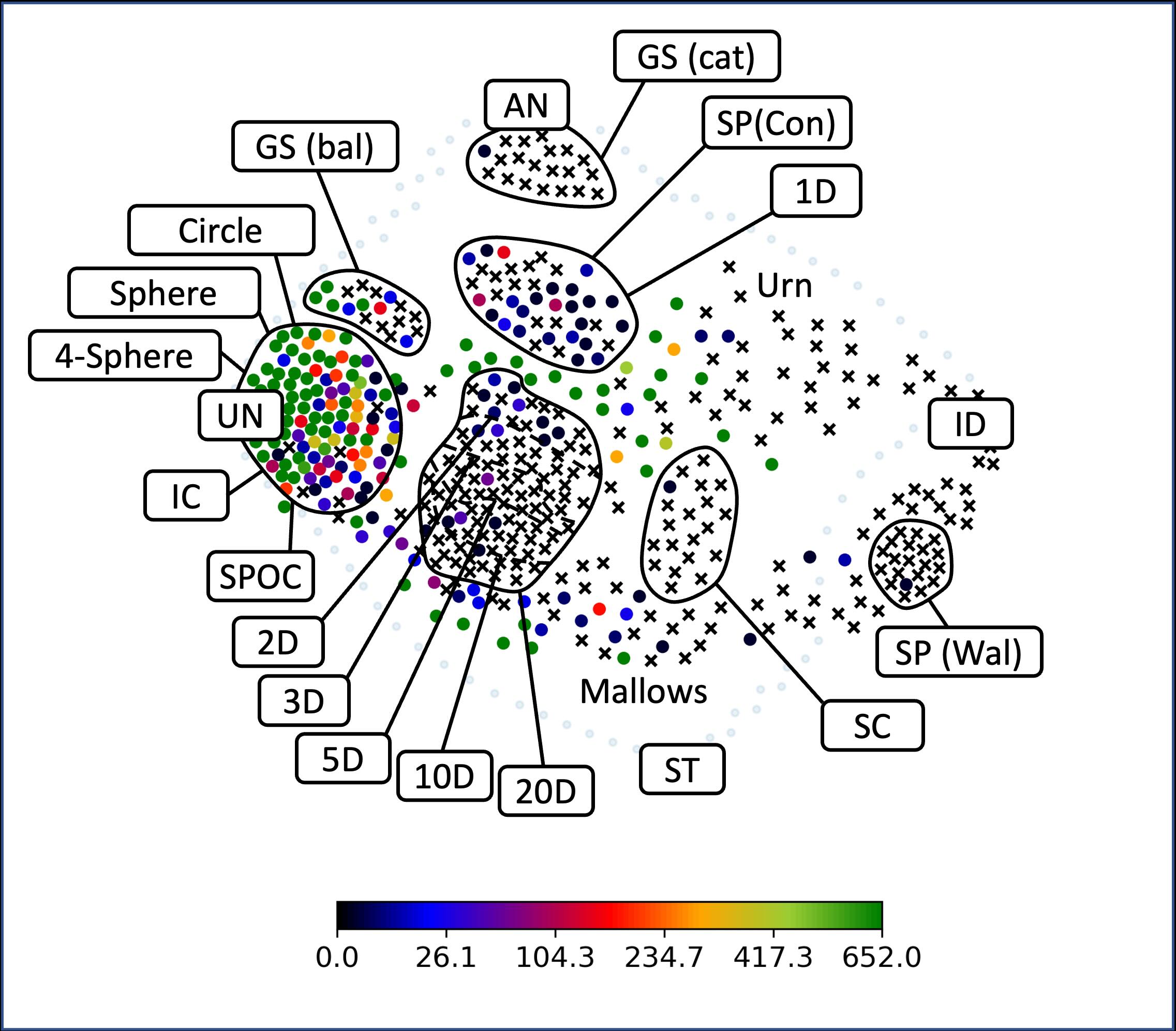}
        \caption{Lowest Dodgson score}
    \end{subfigure}
    
    \caption{Maps colored according to the score obtained by the winner.}
    \label{fig:score_single_winner}
\end{figure}


For Plurality, half of the elections have the highest score of six or less---which we regard as very low. Not many models witness higher values; nonetheless, let us have a closer look at them. For caterpillar group-separable elections, the highest score is around~$48$ (with more voters, or more elections, the average should converge to~$50$); for elections from the Walsh model, the highest score is around~$12$; for single-crossing elections it is around~$35$. For the Norm-Mallows model it is strongly correlated with the~$\normphi$ parameter, hence we see shading from UN to ID, but still most of the points are dark. Only for the urn elections we witness numerous elections with high score values, which is not surprising, because for elections from the urn model the highest plurality score will be similar to the size of the largest group of (identical) voters. Sometimes it might be slightly higher if two or more groups have the same favorite candidate.

As to the Borda coloring, we observe much smoother shading than for Plurality, because, for Plurality we ignore all the voters' preferences, but first choices, while for Borda we care about each position in the votes. Moreover, Borda score nicely correlates with position on the map. If we move closer to ID the highest score is increasing, and if we move towards the UN the score is decreasing.

Note that for Plurality and Borda maps, the paths are colored as well, because to compute Plurality or Borda score it suffices to have the position matrix (or frequency matrix and the number of voters). As to the Copeland and Dodgson scores, we have left the paths' points uncolored, because the position matrix is not sufficient to determine the score.

Last observation regarding the Borda score is about the \emph{Borda balance} area (mentioned previously in~\Cref{sec:iso_maps_of_elections}). We observe that in the left upper part of the map most elections have very low highest Borda score (close to the lowest possible). However, if the highest Borda score is close to the lowest possible Borda score, it means that most candidates need to have very similar scores. To emphasize this, we present one additional map in~\Cref{fig:borda_spread} where the color of each point corresponds to the difference between the highest and the lowest Borda score in a given election. As expected, the closer we are to the upper left part of the map, the smaller the difference.

\begin{figure}[]
    \centering
    \includegraphics[width=6.5cm, trim={0.2cm 0.2cm 0.2cm 0.2cm}, clip]{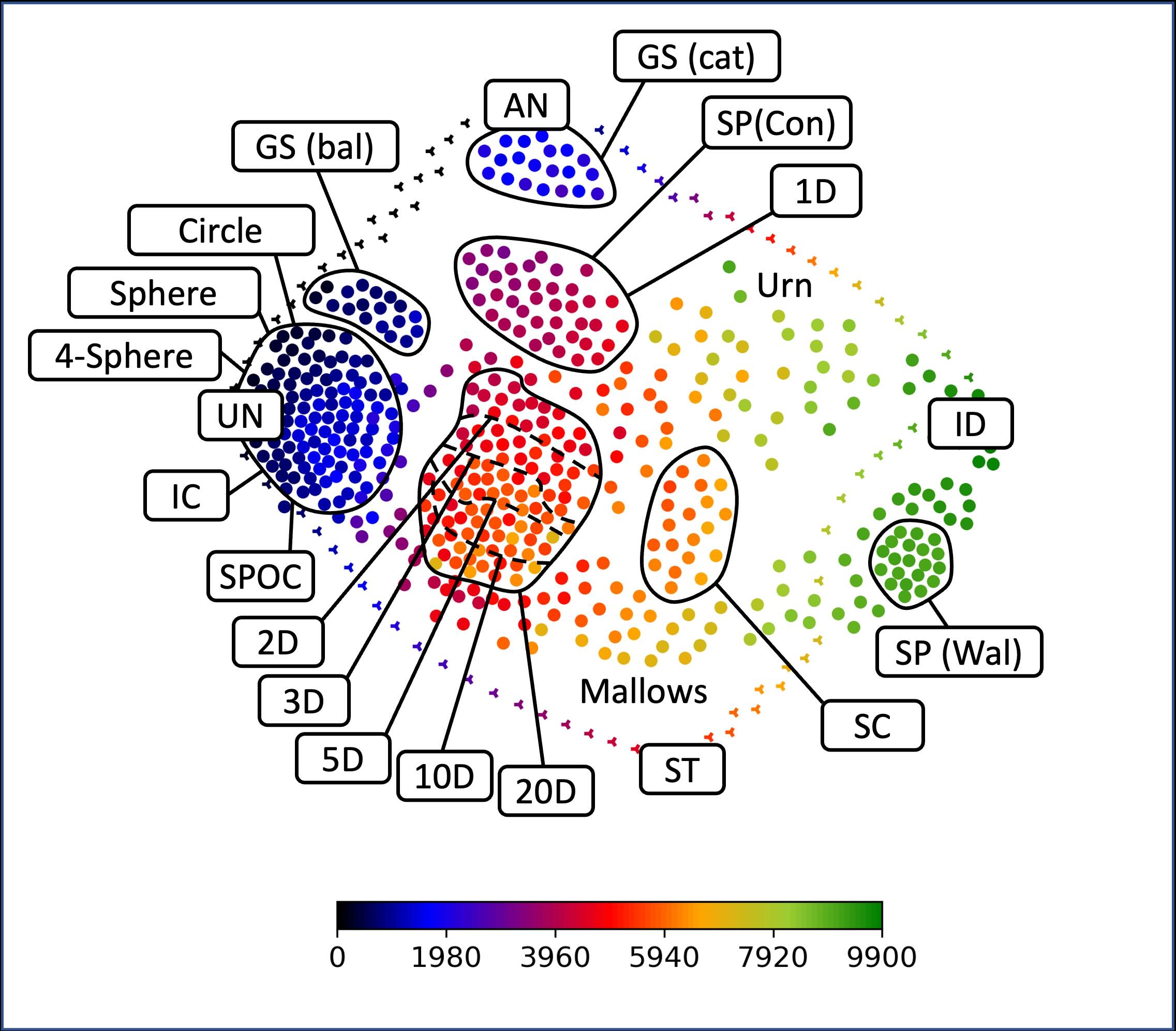}
    \caption{Difference between the highest and the lowest Borda score.}
    \label{fig:borda_spread}
\end{figure}

For the Copeland and Dodgson maps, with crosses, we mark elections that have a Condorcet winner. In our map~$55\%$ (264 out of 480) of generated elections have a Condorcet winner.\footnote{Elections that are single-peaked, single-crossing, or group-separable, for an uneven number of voters, always have a Condorcet winner, and for an even number of voters, always have a weak Condorcet winner (or winners).} In particular, almost all the single-crossing, Walsh, caterpillar group-separable, 3-Cube, 5-Cube, 10-Cube, 20-Cube, and around half of Interval, Square, Conitzer, balanced group-separable elections have a Condorcet winner. Even one IC elections has a Condorcet winner. In terms of the urn and Norm-Mallows elections, having a Condorcet winner is strongly related to their parameters.

The lowest Copeland score is witnessed by a balanced group-separable election. In general, the vast majority of lowest values are obtained by Circle, Sphere and~$4$-Sphere and some SPOC and balanced group-separable elections. Note that all values are much larger than the lowest possible value, which is~$49.5$ with everyone having exactly the same Copeland score. Such an election can be achieved by, for example, taking the identity election and reversing half of the votes.


The highest (i.e., the worst) Dodgson score is witnessed by an urn election. In general, the highest scores are obtained by the urn elections and some balanced group-separable, Circle, Sphere and~$4$-Sphere, SPOC elections.


\begin{figure}[]
    \centering
    
    \begin{subfigure}[b]{0.49\textwidth}
        \centering
        \includegraphics[width=6.5cm, trim={0.2cm 0.2cm 0.2cm 0.2cm}, clip]{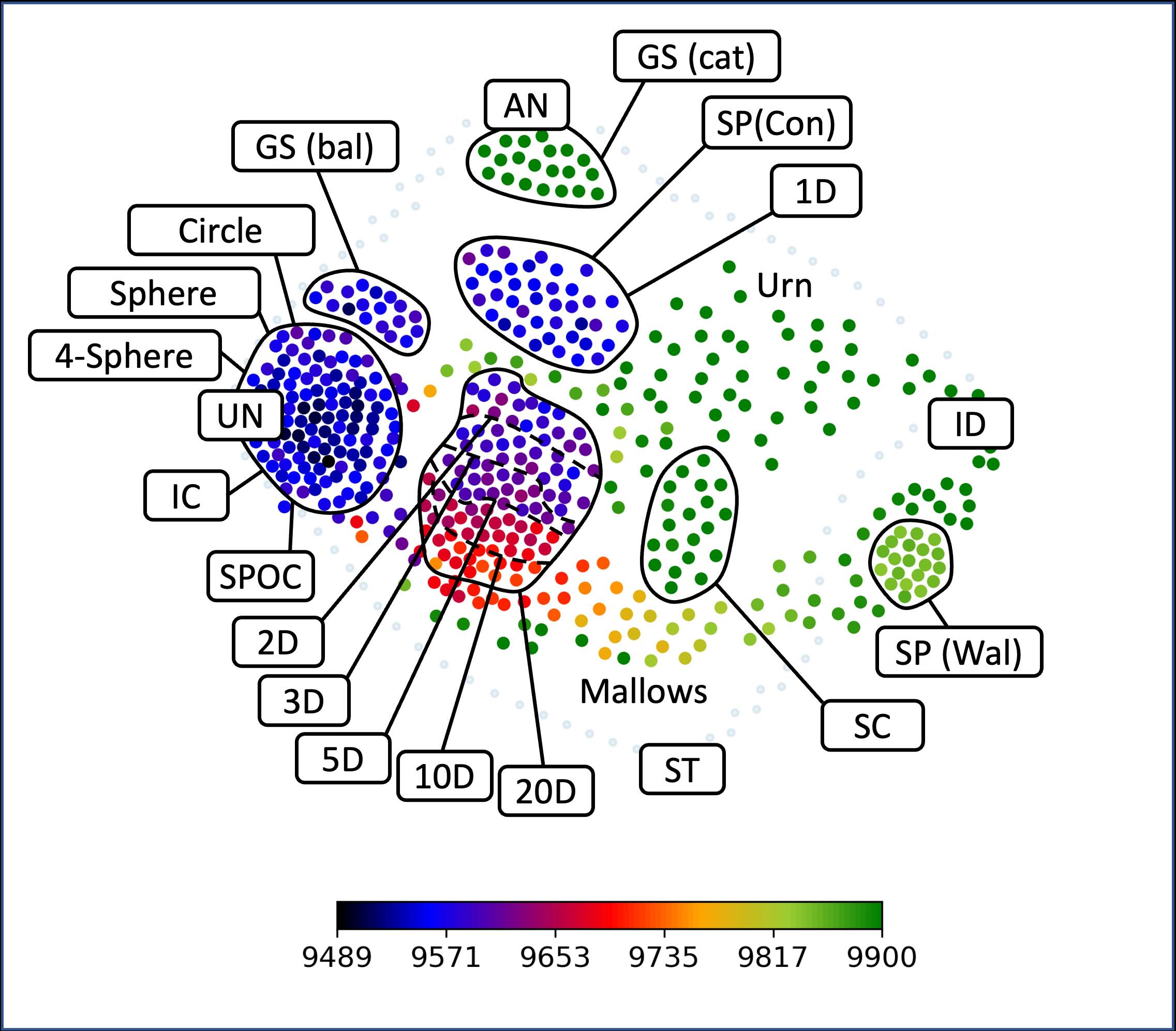}
        \caption{Highest CC score}
    \end{subfigure}
    \begin{subfigure}[b]{0.49\textwidth}
        \centering
        \includegraphics[width=6.5cm, trim={0.2cm 0.2cm 0.2cm 0.2cm}, clip]
        {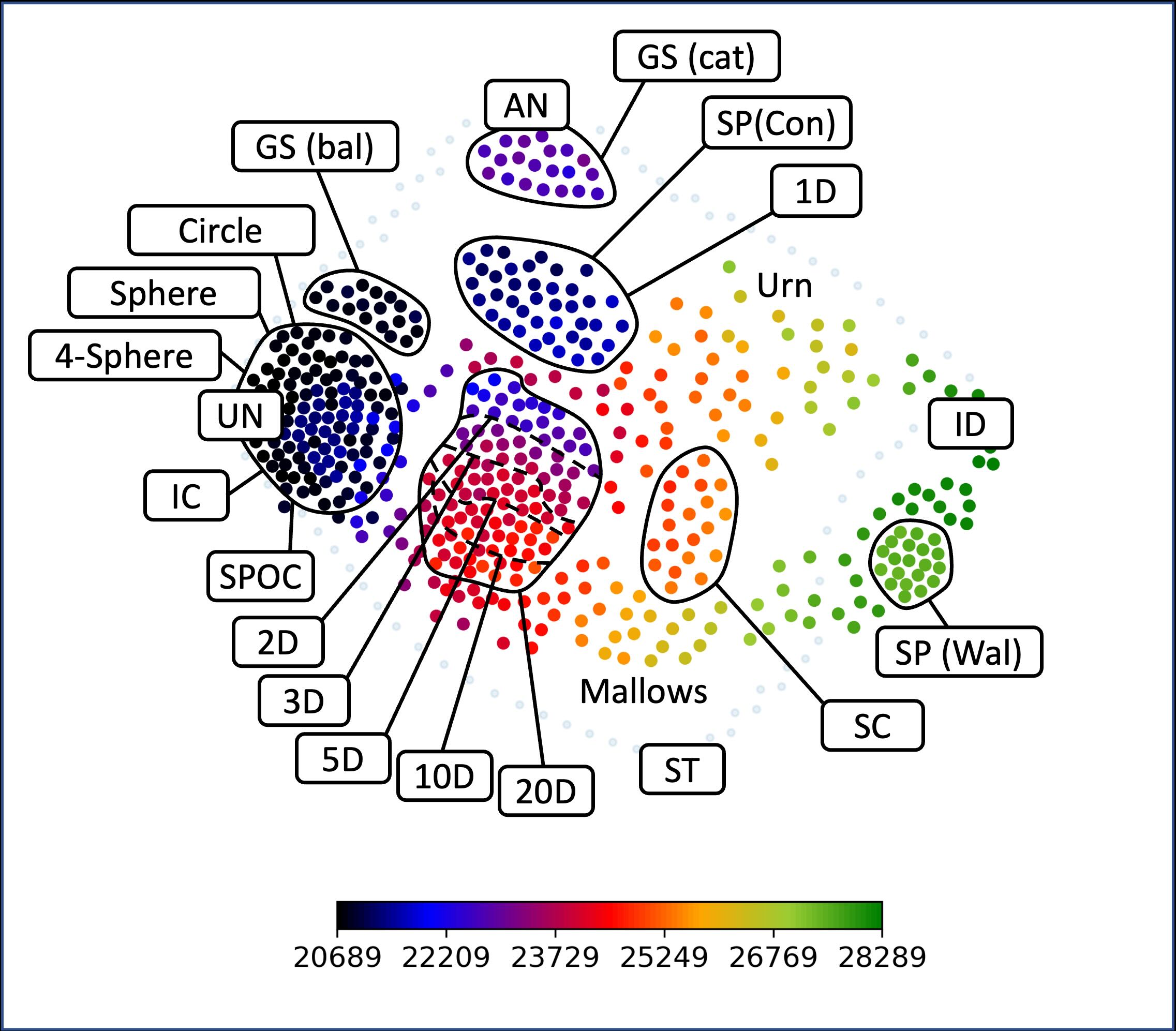}
        \caption{Highest HB score}
    \end{subfigure}
    
    \caption{Maps colored according to the score of the winning committee.}
    \label{fig:score_multi_winner}
\end{figure}

\subsubsection{Multiwinner Rules}

We consider two multiwinner voting rules, Harmonic-Borda (HB) and Chamberlin--Courant (CC).
We start by defining HB.
Given an election~$E=(C,V)$ and committee
size~$k$, the rule outputs a set of~$k$ candidates, referred to as the
\emph{winning committee}. It chooses this committee as follows:
Consider a committee~$S$, a voter~$v$, and denote by~$p_1, \ldots, p_k$ the positions 
of the members of~$S$, sorted from the smallest (most preferred) to the largest
(e.g., for a vote~$v \colon c_2 \pref c_3 \pref c_1$ and committee~$S =
\{c_1, c_3\}$, we would have~$p_1 = 2$,~$p_2 = 3$).  
Then the
\emph{satisfaction} of~$v$ is~$\sum_{i \in [k]}
\nicefrac{(m-p_i)}{i}$; it captures the notion of how a given voter is satisfied with a given committee. HB selects a committee~$S$ that maximizes the
sum of the voters' satisfaction values. 
CC is similar to HB, but simpler. Under CC, each voter gives points only to his or her favorite candidate. Formally, the \emph{satisfaction} of~$v$ is~$m-p_1$.

\begin{example}
  Consider an election~$E = (C,V)$, where~$C = \{a,b,c,d\}$, $V = (v_1,v_2,v_3,v_4)$, and the votes are:
  \begin{align*}
    \small
    v_1\colon&  a \pref b \pref c \pref d \\
    v_2\colon&  a \pref b \pref c \pref d, \\
    v_3\colon&  a \pref b \pref c \pref d, \\
    v_4\colon&  d \pref c \pref b \pref a.
 \end{align*}
According to the CC rule,~$\{a,d\}$ is the winning committee with score of~$12$ (which is the largest possible score for elections with~$4$ candidates and~$4$ voters).
Under the HB rule,~$\{a, b\}$ is the winning committee with score of~$13$, with first three voters giving~$4$ points each ($3\cdot1$ for~$a$ and~$2\cdot\frac{1}{2}$ for~$b$), so~$12$ in total, and the last voter giving~$1$ point  ($1\cdot1$ for~$b$ and~$0\cdot\frac{1}{2}$ for~$a$)

\end{example}

Rules such as CC or HB have received quite some attention from the research
community (for more details, see, e.g., the chapter of~\cite{fal-sko-sli-tal:b:multiwinner-voting}). Both CC and HB are OWA-based~\citep{lang2018multi} committee scoring rules~\citep{fal-sko-sli-tal-tal:c:hierarchy-committee}).

Unfortunately, identifying a winning committee under CC is NP-hard~\citep{procaccia2012maximum, lu2011budgeted, bet-sli-uhl:j:mon-cc} and the same is true for HB~\citep{fal-sko-sli-tal:c:paths}, but we can try to overcome
this issue, for example, by formulating the problem as an integer linear program (ILP) and solving it with an
off-the-shelf ILP solver, or by designing efficient polynomial-time approximation algorithms.
%
We show how our map can be helpful in establishing 
how feasible the ILP approach is (i.e., how quickly can we compute
winning committees).

For each of the elections on our map, we computed a winning
committee of size~$10$ using an ILP solver (CPLEX; we used the ILP
formulation for OWA-based rules of~\cite{sko-fal-lan:j:collective}, applied to the case of HB and CC).

In \Cref{fig:score_multi_winner}, we present the scores obtained by the winning committee under CC (left) and HB (right). In both pictures, we can see nice shading from ID to UN for the Norm-Mallows elections. However, when we look at urn elections, we see a huge difference; for HB the shading of urn elections is similar to that of Norm-Mallows ones, but for CC, the vast majority of urn elections witness a very high score. It is because in the urn model, we have groups of identical voters, so a committee selected by the CC rule will usually satisfy~$k$ largest group of voters, by selecting their top candidates. Another striking difference is the behavior of group-separable caterpillar elections. Under CC, all of them get the highest possible score, while under HB the scores are relatively small.

It is also interesting to consider elections where our rules have the lowest scores of the winning committees.
For CC, most of them are witnessed by the IC elections, while for HB the lowest scores are witnessed by the Circle, Sphere,~$4$-Sphere, and SPOC models, with the IC elections having noticeably higher scores.

\subsection{Running Time}
\begin{figure}[]
    \centering

    \begin{subfigure}[b]{0.49\textwidth}
        \centering
        \includegraphics[width=6.5cm, trim={0.2cm 0.2cm 0.2cm 0.2cm}, clip]{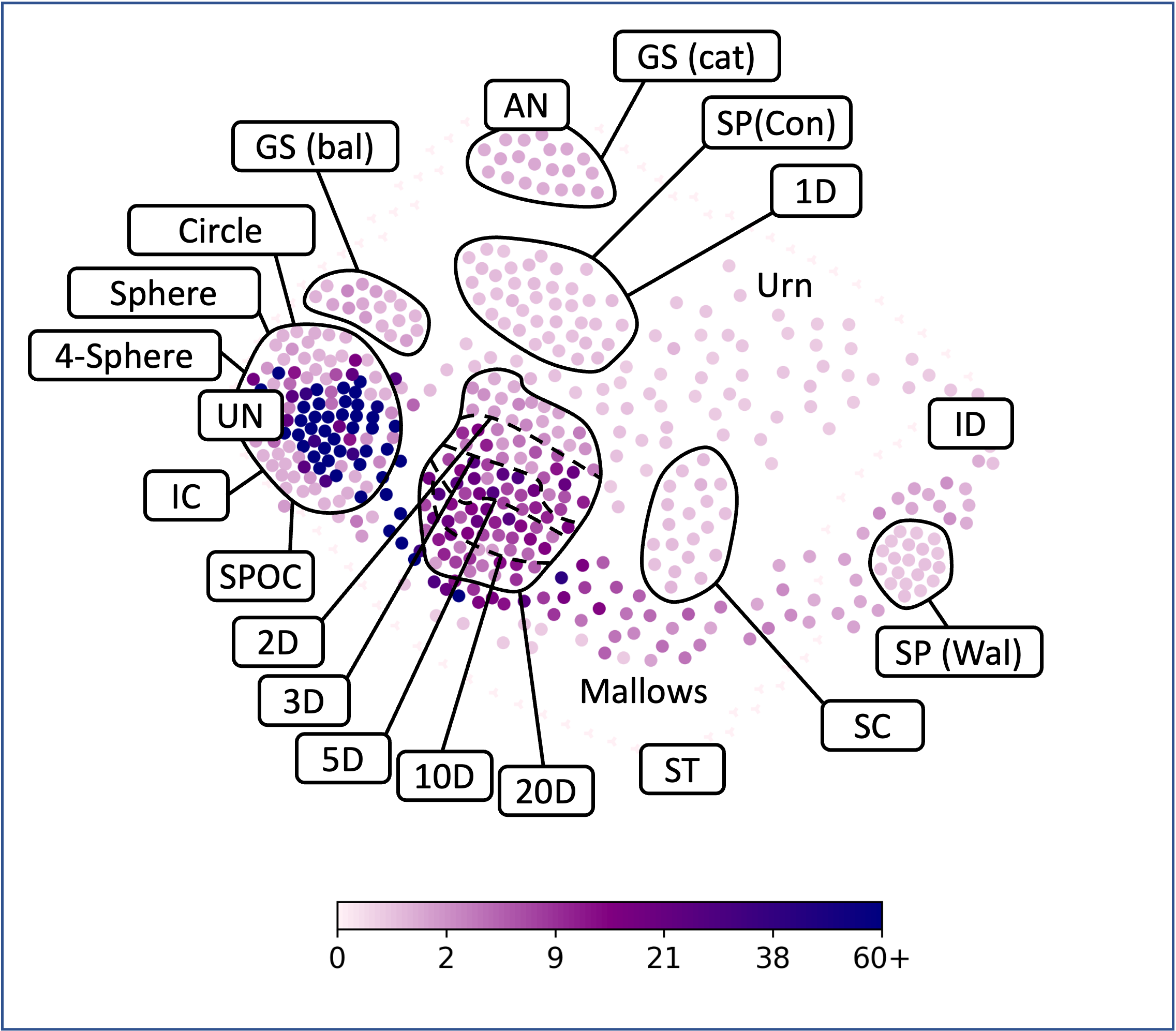}
        \caption{CC running time}
    \end{subfigure}
    \begin{subfigure}[b]{0.49\textwidth}
        \centering
        \includegraphics[width=6.5cm, trim={0.2cm 0.2cm 0.2cm 0.2cm}, clip]{img/time/highest_hb_score_time.png}
        \caption{HB running time}
    \end{subfigure}
    
    \begin{subfigure}[b]{0.49\textwidth}
        \centering
        \includegraphics[width=6.5cm, trim={0.2cm 0.2cm 0.2cm 0.2cm}, clip]{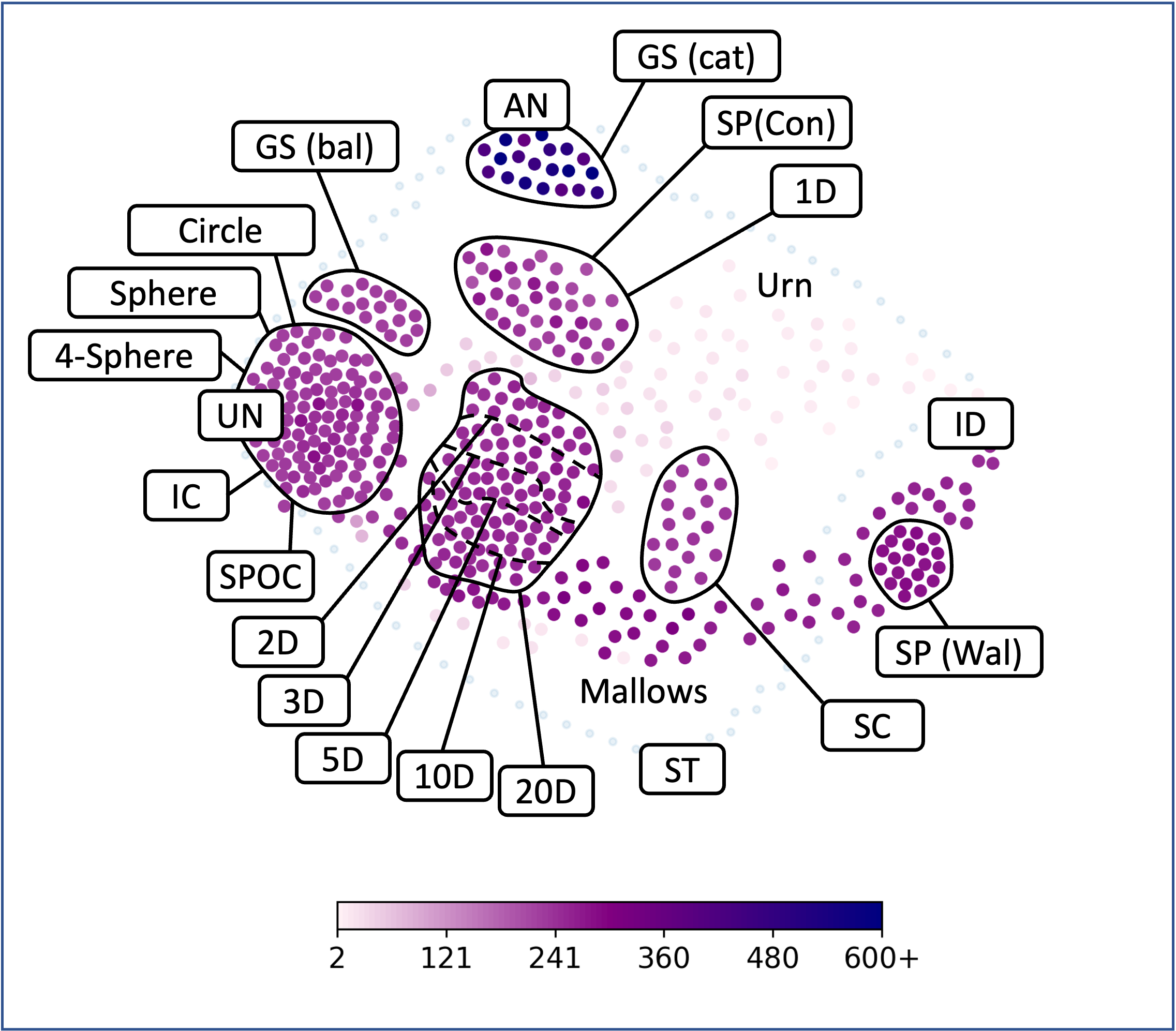}
        \caption{Dodgson running time}
    \end{subfigure}
    
    \caption{ILPs running time (in seconds) for CC, HB, and Dodgson.}
    \label{fig:voting_rules_time}
\end{figure}

\begin{figure}[t]
    \centering
  \begin{subfigure}[b]{0.49\textwidth}
        \centering
        \includegraphics[width=6cm, trim={0.2cm 0.2cm 0.2cm 0.2cm}, clip]{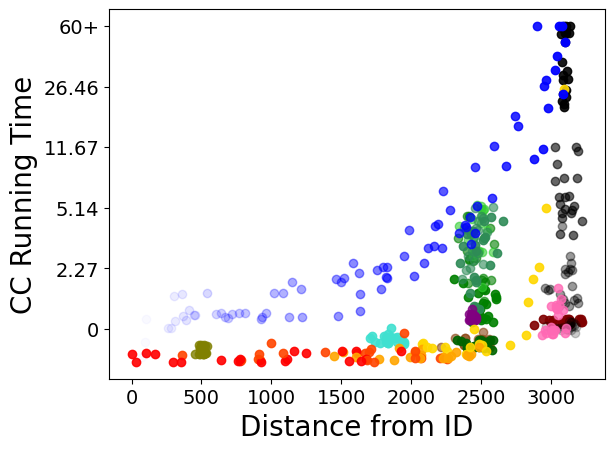}
        \caption{CC}
    \end{subfigure}
    \begin{subfigure}[b]{0.49\textwidth}
        \centering
        \includegraphics[width=6cm, trim={0.2cm 0.2cm 0.2cm 0.2cm}, clip]{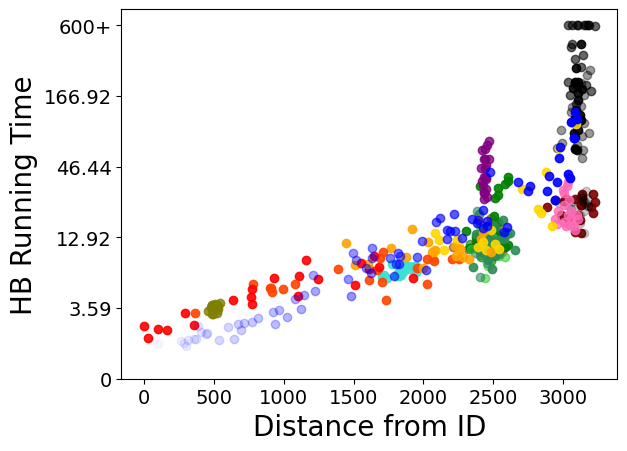}
        \caption{HB}
    \end{subfigure}
    
    \begin{subfigure}[b]{0.49\textwidth}
        \centering
        \includegraphics[width=6cm, trim={0.2cm 0.2cm 0.2cm 0.2cm}, clip]{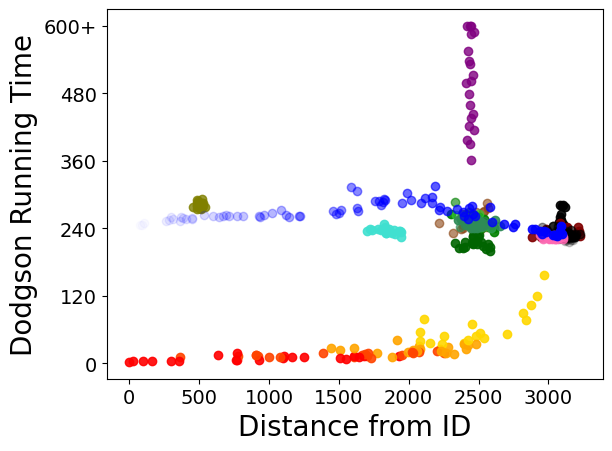}
        \caption{Dodgson}
    \end{subfigure}
    
    \caption{ILPs running time (in seconds) vs distance from ID. Note that, for the Dodgson picture (right) we use the linear scale, while for the HB picture (left) we use logarithmic one.}
    \label{fig:voting_rules_time_vs_id}
\end{figure}

In this part, we analyze the time that is needed to compute the outcomes of the voting rules described above. Four of them, that is, Plurality, Borda, Copeland, the running time relies only on the input size (i.e., numbers of candidates and voters); hence, it is hard to conclude anything interesting. We focus only on Dodgson (single-winner rule), and CC and HB (multiwinner rules).

We report the achieved running times in Figure~\ref{fig:voting_rules_time}, where the colors give the running times (the darker the color, the longer the computation time; for the CC picture we set a limit\footnote{All instances that took longer to compute than the limit are colored with the same color.} at $60$ seconds, and for HB and Dodgson ones we set a limit at~$600$ seconds), and all instances that need longer time have the same color. Moreover, the scale for the HB picture is quadratic.\footnote{We would prefer to have the same scale for all three pictures, however if all three pictures would have linear scale it would have been hard to see anything interesting for CC and HB, and if all three pictures would have quadratic scale it would have been hard to see anything interesting for Dodgson.}



%

We start our analysis with the CC rule by looking at the running time of particular instances. The worst case (i.e., the longest running time) took two minutes to solve, while the simplest one (i.e., the shortest running time) took less than one second. The twenty worst cases were due to impartial culture and Norm-Mallows (with $\normphi$ parameter having close to $1$ value) instances. On the other hand, the simplest ones were those from the urn model.

As to the HB rule, the worst case took five hours to solve, while the simplest one took less than two seconds. Ten worst cases were witnessed by~$4$-Sphere elections---which suggests that it is particularly hard to find optimal winning committee under HB rule for elections from the~$4$-Sphere model.

Perhaps the most visible phenomenon is that the ILP solver needs most
time on the elections similar to those from the impartial culture, and the farther elections we
consider, the less time is needed.


    

\begin{table}
    \centering
    \footnotesize
    \begin{tabular}{c | c | c | c | c | c | c }
     \toprule
        Culture &  \multicolumn{2}{c|}{CC} & \multicolumn{2}{c|}{HB} & \multicolumn{2}{c}{Dodgson} \\
         & avg. time & std. dev. & avg. time & std. dev.  & avg. time & std. dev. \\
    \midrule
    Impartial Culture & 43.8s & 22.8 & 155.7s & 85.5 & 251.4s & 17.2 \\
    \midrule
    Conitzer SP & 0.8s & 0.0 & 12.2s & 2.2 & 252.7s & 14.9 \\
    Walsh SP & 0.7s & 0.0 & 3.6s & 0.2 & 281.1s & 5.1 \\
    SPOC & 1.1s & 0.0 & 22.3s & 4.1 & 230.0s & 8.0 \\
    Single-Crossing & 0.9s & 0.0 & 7.1s & 0.6 & 235.6s & 4.4 \\
    \midrule
    Interval & 0.8s & 0.0 & 11.8s & 1.4 & 213.2s & 7.8 \\
    Square & 1.4s & 0.3 & 20.6s & 9.1 & 249.9s & 7.3 \\
    Cube & 2.7s & 1.0 & 11.4s & 1.8 & 254.1s & 10.6 \\
    5-Cube & 3.7s & 1.0 & 10.8s & 2.0 & 251.0s & 6.7 \\
    10-Cube & 3.5s & 0.9 & 12.0s & 3.4 & 249.3s & 5.2 \\
    20-Cube & 2.8s & 1.0 & 12.1s & 3.0 & 249.7s & 6.0 \\
    \midrule
    Circle & 1.0s & 0.0 & 22.9s & 3.7 & 222.3s & 2.9 \\
    Sphere & 1.9s & 0.6 & 126.8s & 64.9 & 230.5s & 5.0 \\
    4-Sphere & 6.5s & 2.5 & 1614.1s & 3912.1 & 232.3s & 6.0 \\
    \midrule
    Balanced GS & 1.2s & 0.2 & 22.2s & 6.9 & 223.0s & 2.1 \\
    Caterpillar GS & 1.2s & 0.0 & 44.7s & 14.3 & 513.0s & 106.1 \\
    \midrule
    Urn & 1.1s & 2.8 & 11.1s & 13.1 & 30.0s & 35.1 \\
    Norm-Mallows & 8.6s & 17.9 & 17.3s & 24.1 & 264.2s & 19.7 \\
    	\bottomrule

        \end{tabular}
    \caption{\label{table:time_hb_dodgson} Analysis of ILPs running time for CC, HB and Dodgson rules.}
    
\end{table}

For the Dodgson rule the situation is quite different. Most instances need the same amount of time (i.e., around four minutes on average). Two exceptions are the urn elections (which needed half a minute per election on average) and caterpillar group-separable elections (which needed eight and a half minutes per election on average).
In~\Cref{table:time_hb_dodgson} we present average values and standard deviation for each statistical culture that we used. Moreover, in~\Cref{fig:voting_rules_time_vs_id} we show the correlation between the running time and distance from $\ID$. As we can see, for the CC and HB rules all the hardest instances were at almost the largest possible distance from $\ID$, while for the Dodgson rule it is not the case. Although, usually the less structure in the election the longer it takes to compute a given voting rule, sometimes (like, for example, in the case of the Dodgson rule) adding structure to the election might increase the running time of a particular algorithm.


\subsection{Approximation}

We compare four approximation algorithms for the CC multiwinner voting rule.
We refer to them as \textit{SeqCC} (sequential, sometimes also referred to as the greedy variant), \textit{RemovalCC}, \textit{RangingCC}, and \textit{BanzhafCC}. At the second part of this section, we also compare two approximation algorithms for the HB multiwinner voting rule.

Let~$k\in[m]$ be the size of the committee we want to select, where~$m$ is the total number of candidates.
Below we describe these four algorithms.

SeqCC starts with an empty committee 
and works in~$k$ iterations, where in each of them it adds to the
committee a single candidate, so that the resulting committee has as
large total satisfaction as possible. RemovalCC proceeds similarly,
but it starts with a committee containing all candidates and works in~$m-k$ iterations, 
in each of them removing a single candidate, so the
resulting committee has as large total satisfaction as possible.
Both algorithms are
well-known in the literature 
and are used for computing approximate
winning committees under various
voting rules~\citep{sko-fal-lan:j:collective,sko-lac-bri-pet-elk:c:proportional-rankings,fal-lac-pet-tal:c:csr-heuristics}, for example, CC or HB.

\begin{figure}[]
    \centering
    
    \begin{subfigure}[b]{0.49\textwidth}
        \centering
        \includegraphics[width=6.5cm, trim={0.2cm 0.2cm 0.2cm 0.2cm}, clip]{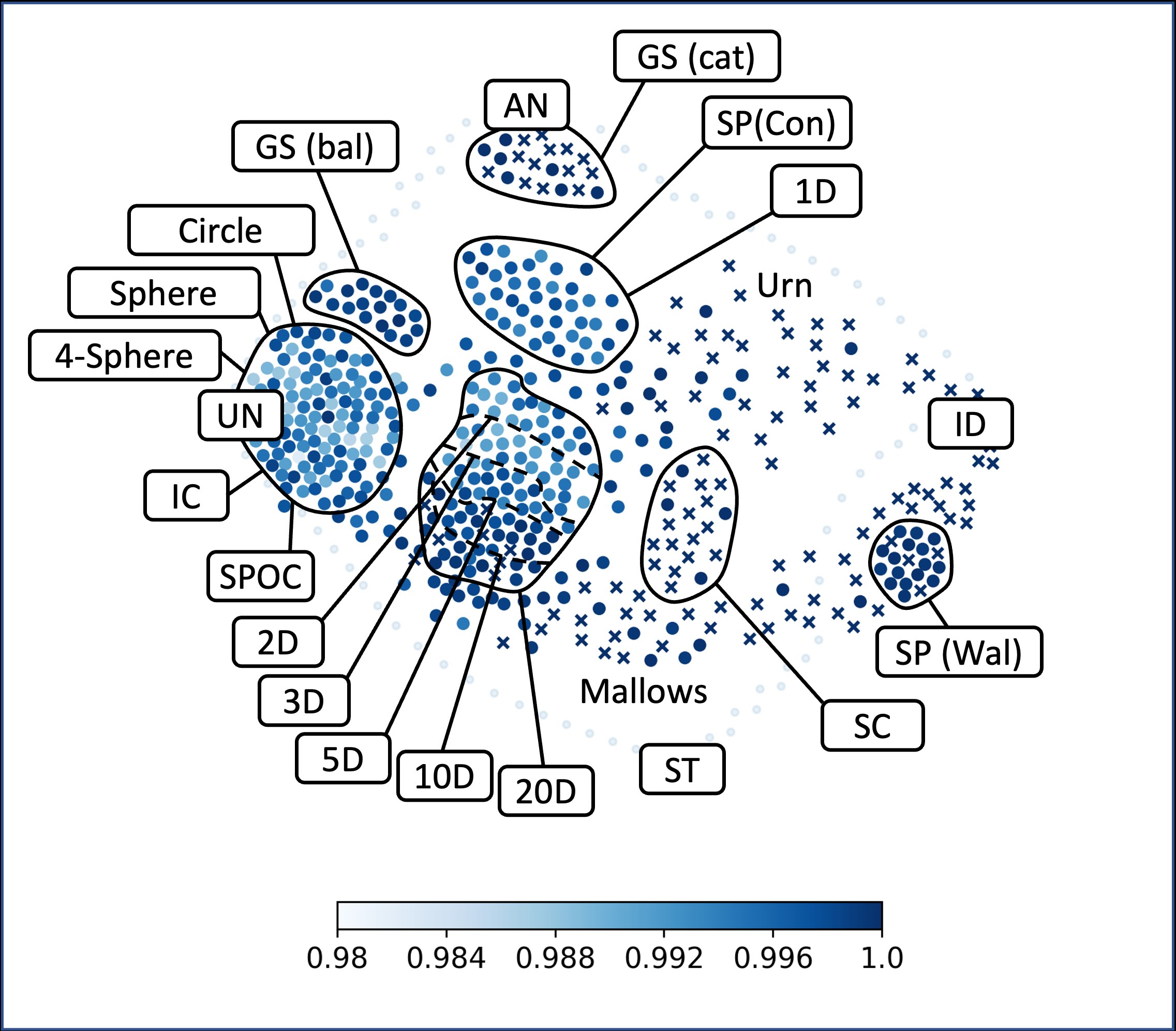}
        \caption{Sequential CC}
    \end{subfigure}
    \begin{subfigure}[b]{0.49\textwidth}
        \centering
        \includegraphics[width=6.5cm, trim={0.2cm 0.2cm 0.2cm 0.2cm}, clip]{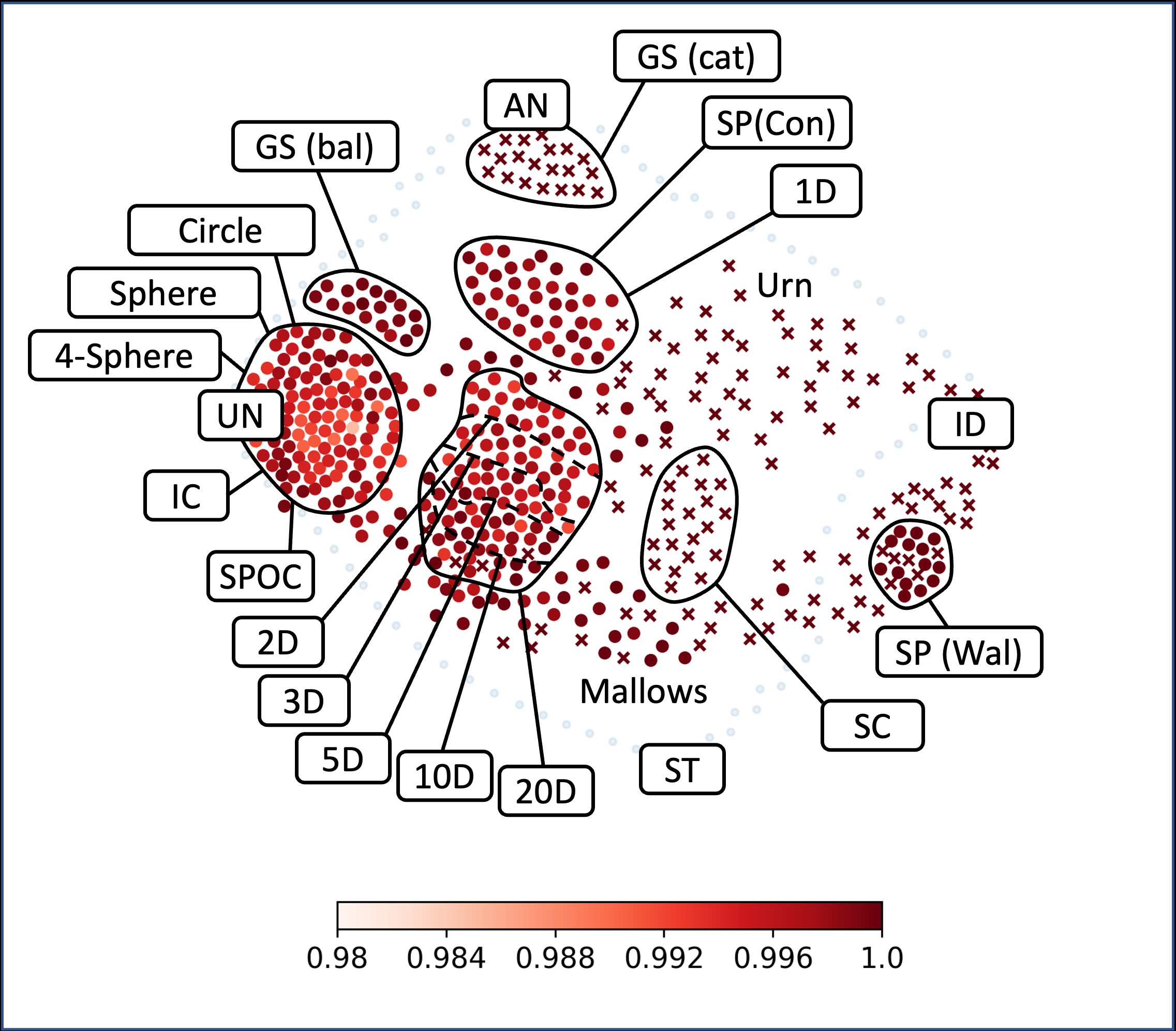}
        \caption{Removal CC}
    \end{subfigure}
    
    \vspace{1em}
    
    \begin{subfigure}[b]{0.49\textwidth}
        \centering
        \includegraphics[width=6.5cm, trim={0.2cm 0.2cm 0.2cm 0.2cm}, clip]{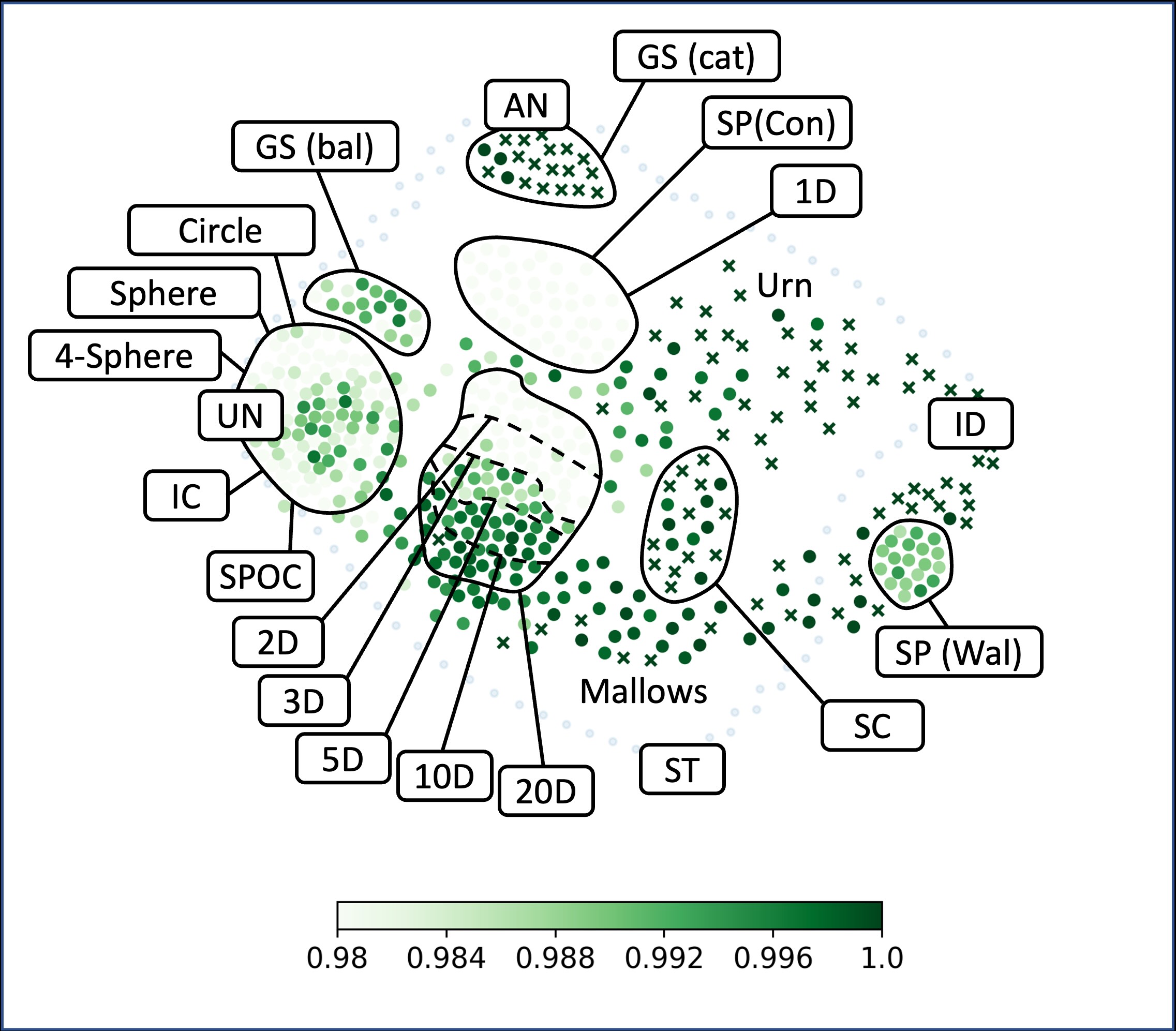}
        \caption{Banzhaf CC}
    \end{subfigure}
    \begin{subfigure}[b]{0.49\textwidth}
        \centering
        \includegraphics[width=6.5cm, trim={0.2cm 0.2cm 0.2cm 0.2cm}, clip]{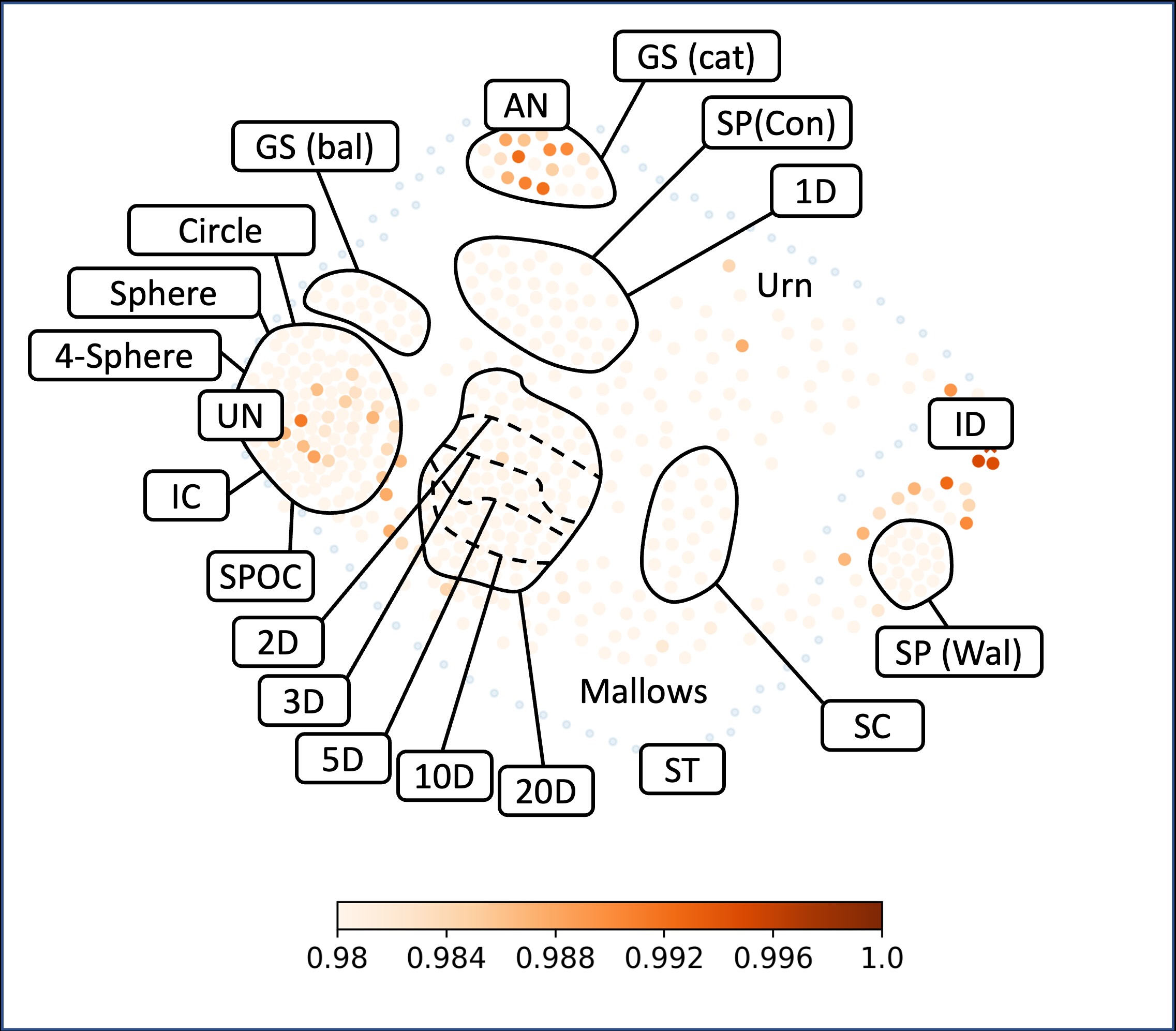}
        \caption{Ranging CC}
    \end{subfigure}
    
    \caption{Comparison of approximation algorithms for CC.}
    \label{fig:approx_4}
\end{figure}

We also consider RangingCC algorithm \citep{sko-fal-sli:j:multiwinner, elk-fal-las-sko-sli-tal:c:2d-multiwinner}, which was designed especially for the CC voting rule.

The last algorithm we discuss is BanzhafCC \citep{fal-lac-pet-tal:c:csr-heuristics}, which is a special variant of SeqCC method. When deciding which candidate should be added to the committee in each step, it is using the concept of the Banzhaf index.
Briefly put, SeqCC always adds the candidate that currently leads to the highest score. BanzhafCC, on the other hand, adds the candidate that is expected to maximize the committee score if all but one missing committee members were chosen randomly. For details, we point the readers to the work of \cite{fal-lac-pet-tal:c:csr-heuristics} and \cite{munagala2021optimal}. (Both RangingCC and BanzhafCC can serve as a base for polynomial-time approximation schemes for CC).



We evaluate the approximation algorithms by computing the approximation ratio, that is, the score of the winning committee selected by the approximation algorithms divided by the score of the winning committee selected by the optimal method. \new{The larger the approximation ratio, the better.}

In \Cref{fig:approx_4}, we present the results. For all four approximation algorithms, the closer we are to~$\UN$, the worse is the approximation ratio. Moreover, the results for SeqCC and RemovalCC are significantly better than for BanzhafCC and RangingCC. Finally, RangingCC is clearly the worst one. Nonetheless, all four algorithms on average have very high approximation ratios (close to one).

\begin{figure}[t]
    \centering
    
    \includegraphics[width=11cm, trim={0.2cm 0.2cm 0.2cm 0.2cm}, clip]{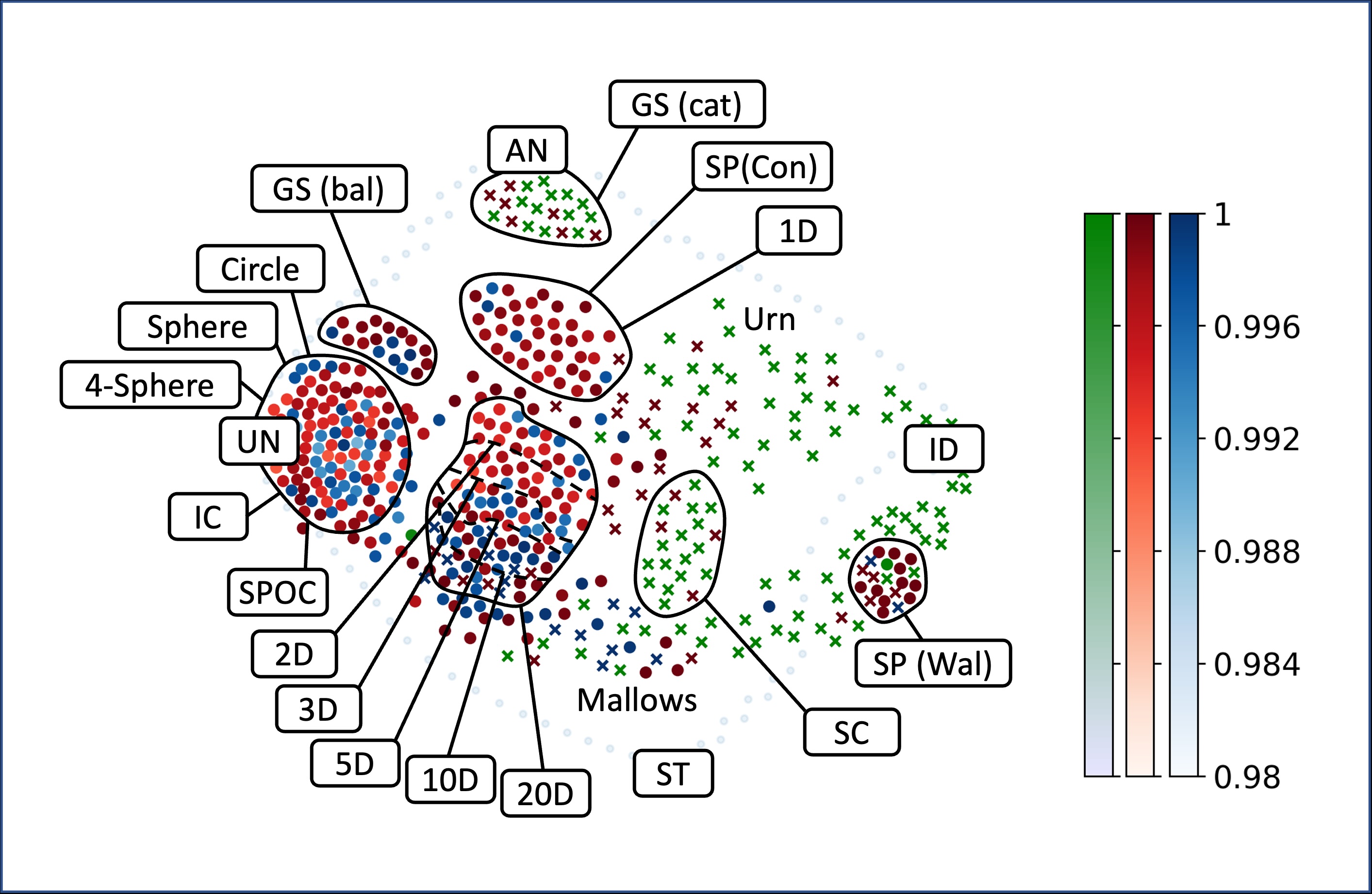}
    
    \caption{SeqCC versus RemovalCC.}
    \label{fig:approx_2in1}
\end{figure}

What is interesting for Banzhaf, is the fact that even though Walsh elections are relatively close to~$\ID$, the approximation ratio is quite bad in comparison to the elections lying next to them (e.g., the Norm-Mallow elections).

In \Cref{fig:approx_2in1}, we present a comparison of SeqCC and RemovalCC, where the blue points refer to places where SeqCC is better at approximating CC, the red points mark the elections where RemovalCC is better, and finally the green points depict elections where there is a draw.

On most single-crossing, around half of the urn and Norm-Mallows elections (those with low~$\normphi$), some caterpillar group-separable and three Walsh elections there is a draw. For all the rest, one or the other of the algorithms was better. 


When there was no tie, for the majority of the instances, RemovalCC was better. The only three models for which for more than half of the instances SeqCC was better are
the IC,~$10$-Cube, and Norm-Mallows. This means that for low dimensional Euclidean elections RemovalCC was much better, and for high dimensional Euclidean elections, like~$10$-Cube elections, SeqCC was better.



\begin{figure}[t]
    \centering
    
    \begin{subfigure}[b]{0.49\textwidth}
        \centering
        \includegraphics[width=6.5cm, trim={0.2cm 0.2cm 0.2cm 0.2cm}, clip]{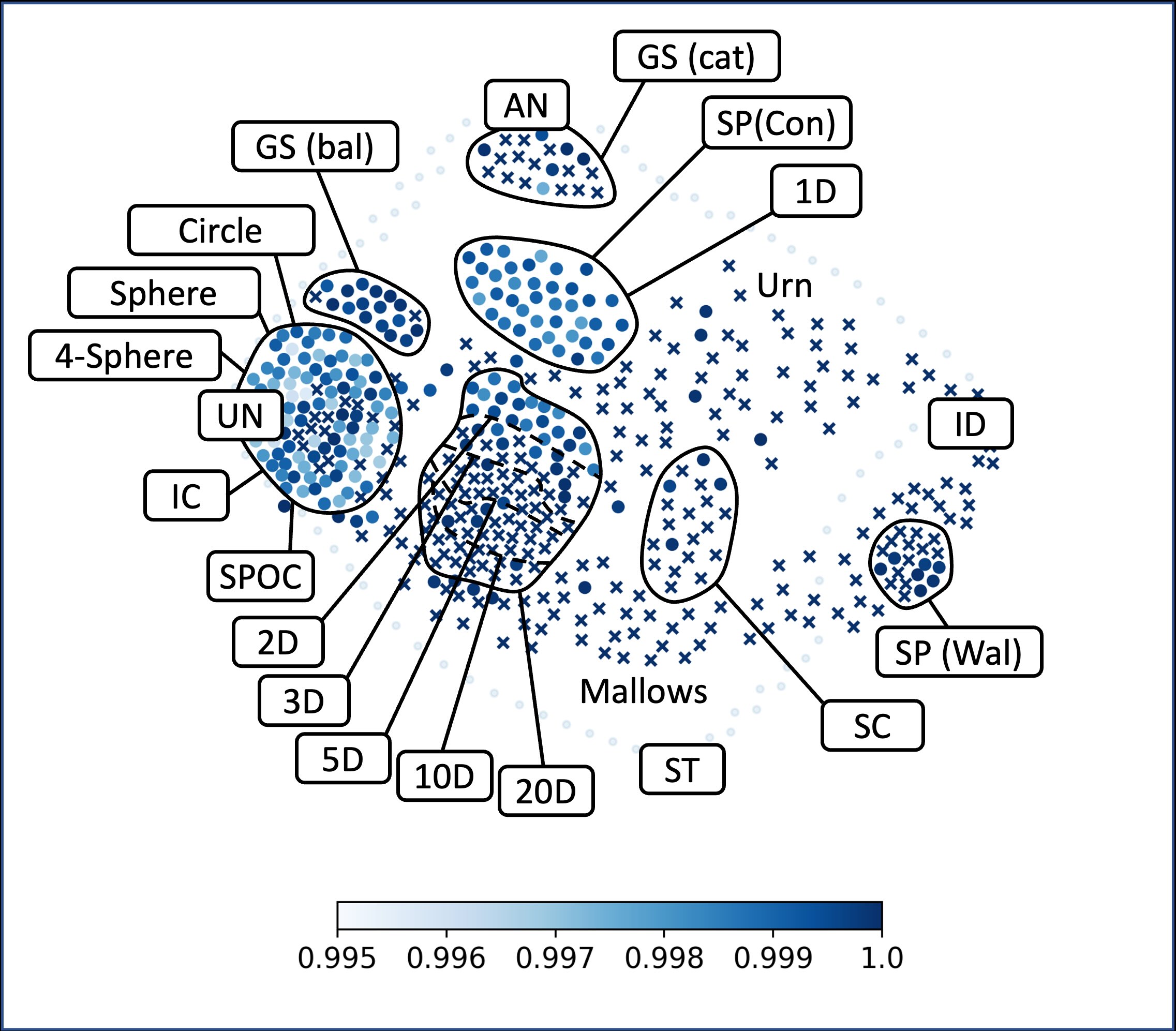}
        \caption{Sequential HB}
    \end{subfigure}
    \begin{subfigure}[b]{0.49\textwidth}
        \centering
        \includegraphics[width=6.5cm, trim={0.2cm 0.2cm 0.2cm 0.2cm}, clip]{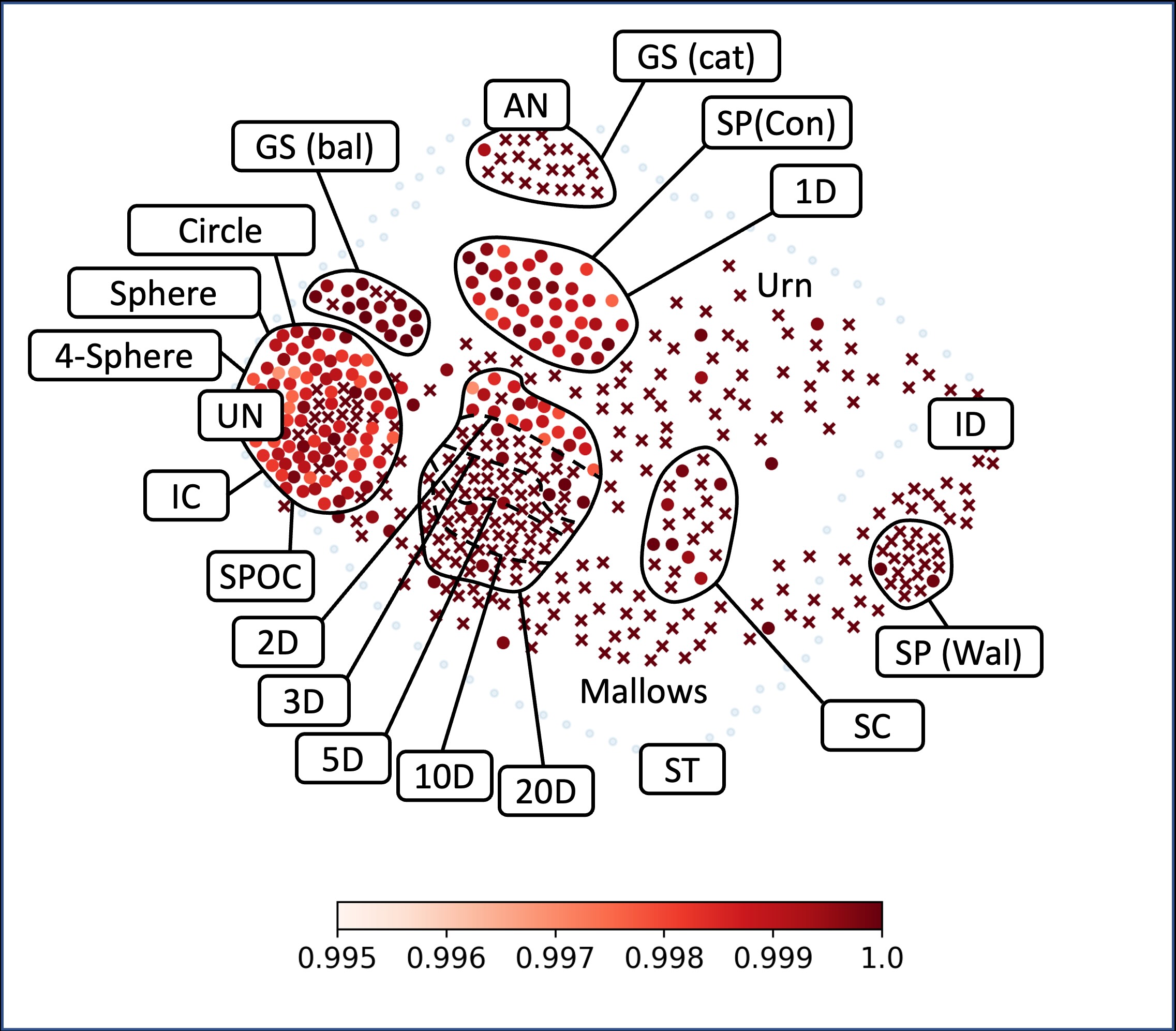}
        \caption{Removal HB}
    \end{subfigure}
    
    \caption{Comparison of approximation algorithms for HB.}
    \label{fig:hb_2}
\end{figure}

\begin{figure}[t]
    \centering
    
    \includegraphics[width=11cm, trim={0.2cm 0.2cm 0.2cm 0.2cm}, clip]{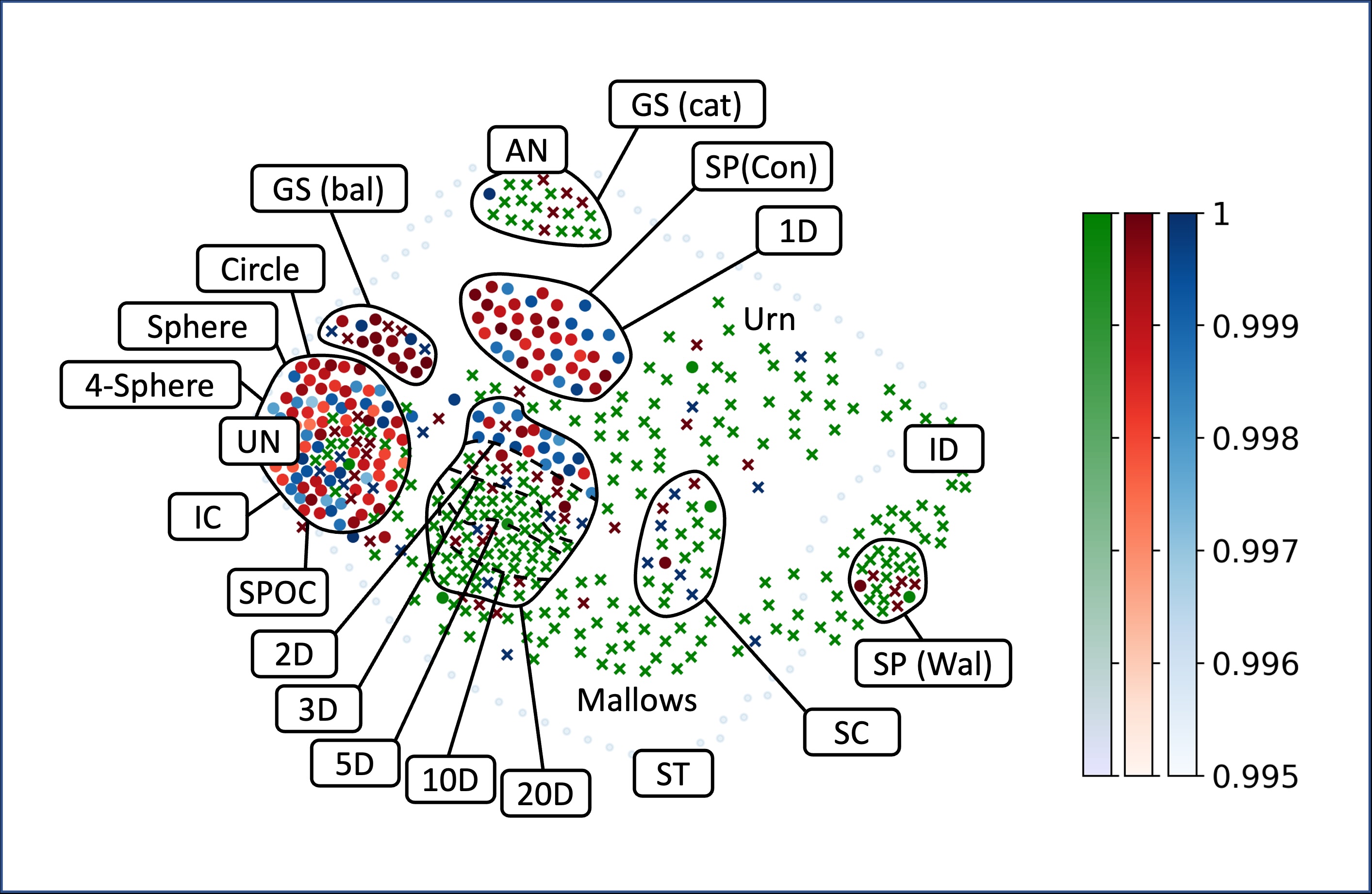}
    
    \caption{SeqCC versus RemovalCC.}
    \label{fig:approx_hb2in1}
\end{figure}

Next, we move on to a similar experiment, but for the HB rule. We study two approximation algorithms, to which we refer to as SeqHB and RemovalHB. These two methods are defined analogously to the SeqCC and RemovalCC. In~\Cref{fig:hb_2} we present results individually for each of the algorithms, and in~\Cref{fig:approx_hb2in1} we present them jointly. The performance of both algorithms is almost excellent. In approximately half of the instances ($239$ out of~$480$) an optimal solution was found by both methods. In~$153$ instances solution found by RemovalHB was better, and in~$82$ instances the one found by SeqHB was better. 

For the Walsh and caterpillar group separable elections almost always an optimal solution was found, however, sometimes it was only found by RemovalHB, while the solution found by SeqHB was suboptimal. In most elections from the Interval, Conitzer, and balanced group-separable, RemovalHB was better than SeqHB.

While SeqCC and SeqHB have approximation guarantees of $1-\frac{1}{e}$, RemovalCC and RemovalHB do not have any such guarantees.

\vspace{0.2cm}
\begin{conclusionbox}
The conclusions from~\Cref{ch:applications:sec:rules} are as follows.
\begin{itemize}
    \item Elections that are similar (i.e., are at small distance from one another) behave similarly under different voting rules. For instance, similar elections obtain similar scores under scoring voting rules and require similar amount of time to compute winners under a given voting rule. For instance, usually the closer elections are to IC, the longer it takes to compute their outcomes. However, there are some exceptions. For example, for the Dodgson voting rule, the longest time was witness by elections from group-separable statistical culture, which are far away from IC.
    \item While SeqCC and SeqHB have approximation guarantees of $1-\frac{1}{e}$ (and even a stronger one for SeqCC), in practice we observe that they are outperformed by RemovalCC and RemovalHB, which do not have such guarantees.
\end{itemize}
\end{conclusionbox}

\section{Real-Life Instances} \label{preflib_info}

Next, we focus on real-life elections. We have several sources of our datasets, where the most prominent two are PrefLib \citep{mat-wal:c:preflib} and the work of \cite{boe-sch:t:datasets}. There are two main problems with real-life elections. One is the fact that usually there are few candidates participating. The second one is that in many cases, numerous votes are incomplete. In the beginning, we will describe how to preprocess the data in general, for example, to have complete preference orders. Then we will describe our datasets one by one, and how we preprocessed each of them specifically---some of the datasets needed some special treatment. Then, we present where these real-life elections land on our maps of elections. Finally, we will try approximating real life elections with the Norm-Mallows model, i.e., we will search for such~$\normphi$ parameters so that elections sampled from the Norm-Mallows model with this parameter are as close to the real-life elections as possible.

Whenever we speak of real-life elections in this section, we mean
elections from our datasets.

\subsubsection{Selection of Datasets}
In \Cref{tab:preflib_selected}, we present a detailed description of the selected datasets. All of them are available at PrefLib.
We chose eleven real-life datasets of different types, and we divided them into three categories. 
The first group contains \textit{political} elections: city council
elections in Glasgow and Aspen~\citep{openstv}, elections from
Dublin North and Meath constituencies (Irish), and elections held by
non-profit organizations, trade unions, and professional organizations
(ERS). The second group consists of \textit{sport} elections: 
Tour de France (TDF)~\citep{boe-sch:t:datasets}, 
Giro d'Italia (GDI)~\citep{boe-sch:t:datasets}, 
speed skating~\citep{boe-bre-fal-nie-szu:c:compass}, 
and figure skating. The last
group consists of \textit{surveys}: preferences over Sushi~\citep{kam:c:sushi}, T-Shirt
designs, and costs of living and population in different cities~\citep{caragiannis2019optimizing}. 
For TDF and GDI, each race is a vote, 
and each season is an election. For speed skating, each lap is a
vote, and each competition is an election. For figure skating, each
judge's opinion is a vote, and each competition is an election.

We are only interested in elections that have at least~$10$ candidates\footnote{The more candidates, the more interesting is the data. On the other hand, if we had required too many candidates, we would end up having very few instances. In that sense,~$10$ candidates is a tradeoff between the number of candidates, and the number of instances that have at least~$10$ candidates.}.
As our model only allows us to consider complete votes without ties, we are interested in instances where votes are as complete as possible and contain only a few ties. In some datasets, only parts of the data meet our criteria (i.e., complete votes without ties over at least~$10$ candidates). For example, in the dataset containing Irish elections, we have three different elections, but one of them (an election from the Dublin West constituency) contains only nine candidates. We delete all such elections. After doing so, we finally arrive at eleven real-life datasets containing elections meeting our criteria.

As we cannot include all elections from each dataset on the map of elections\footnote{If there are too many elections embedded jointly on a single map, the picture is becoming unclear and difficult to interpret.}, we further reduce the number of elections by considering only selected elections. In \Cref{tab:preflib_selected}, we include in the column \textit{\# Valid Elections} the number of elections we selected from each dataset in the end. We based our decision on the number of voters and candidates. That is, for ERS, we only take elections with at least~$500$ voters, for Speed Skating with at least~$80$ voters, for TDF with at least~$20$ voters, and for Figure Skating with at least~$9$. In addition to that, for TDF, we only select elections with no more than~$75$ candidates.

\begin{table*}[t]
\centering
\resizebox{\textwidth}{!}{\begin{tabular}{c  c  c  c  c  c}
			\toprule
			Category & Name & \# Valid Elections & Avg.~$m$ & Avg.~$n$ & Description\\	
			\midrule
			Political & Irish &~$2$ &~$13$ &~$\sim$~$54011$ & Elections from Dublin North and Meath\\
			Political & Glasgow	&~$13$ &~$\sim$~$11$ &~$\sim$~$8758$ & City council elections \\
            Political & Aspen &~$1$ &~$11$ &~$2459$	& City council elections\\
			Political & ERS	&~$13$ &~$\sim$~$12$ &~$\sim$~$988$ & Various elections held by non-profit organizations,\\ 
			        & & & & & trade unions, and professional organizations  \\
            \midrule
			Sport & Figure Skating &~$40$ &~$\sim$~$23$ &~$9$ & Figure skating  \\
			Sport & Speed Skating &~$13$ &~$\sim$~$14$ &~$196$ & Speed skating  \\
            Sport & TDF &~$12$ &~$\sim$~$55$ &~$\sim$~$22$ & Tour de France\\	
            Sport & GDI &~$23$ &~$\sim$~$152$ &~$20$  & Giro d’Italia	\\	
            \midrule
            Survey & T-Shirt &~$1$ &~$11$ &~$30$ & Preferences over T-Shirt logo \\	
            Survey & Sushi &~$1$ &~$10$ &~$5000$ & Preferences over Sushi\\	
            Survey & Cities &~$2$ &~$42$ &~$392$ & Preferences over cities	\\			
			\bottomrule
	\end{tabular}}
	\caption{\label{tab:preflib_selected} Each row contains a description of one of the real-life datasets we consider. In the column \textit{\# Selected Elections}, we denote the number of elections we finally select from the respective dataset.}
\end{table*}

\subsubsection{Preprocessing of Datasets}
There are two types of problems that we encounter in selected datasets. First, ties (i.e., pairs or larger sets of
candidates that are reported as equally good). Any ties that appear we break randomly.
Second, incomplete votes (i.e., votes where some of the top candidates are ranked and the remaining candidates are not). Sometimes both problems happen at the same time.

For all elections from our selected datasets that contain incomplete
votes, we need to fill-in all the missing data. For the
decision how to complete each vote, we use the other votes as
references, assuming that voters that rank the same candidates on top
also continue to rank candidates similarly toward the bottom. 

\new{For each incomplete vote~$v$, we proceed as follows. Let us assume that vote~$v$ is over $m'$ candidates. Let~$V_P$ be the set of all original votes of which~$v$ is a prefix. We uniformly at random select
one vote~$v_p$ from~$V_P$ and then at the end of vote~$v$ we add candidate which is at position $m'+1$ in vote $v_p$. We repeat the procedure until vote~$v$ is
complete. If the set~$V_P$ is empty, then we choose~$c$ uniformly at
random (from those candidates that are not part of~$v$ yet).}

After applying these preprocessing steps, we arrive at a collection of datasets containing elections with ten or more candidates and complete votes without ties. As we focus on ten candidates, we need to select a subset of ten candidates for each election. In a given election we compute the Borda score of each candidate, and select ten candidates with the highest ones. In case there is a tie, we break it randomly.

 We refer to the resulting datasets as \emph{intermediate} datasets.

\subsubsection{Sampling Elections from the Intermediate Datasets}
We treat each of our intermediate datasets as a separate election model from which we sample~$15$ elections to create the final datasets that we use. For each intermediate dataset, we sample elections as follows. First, we randomly select one of the elections present internally in it (for example, the election held in Dublin North constituency from the Irish dataset). Second, we sample~$100$ votes from this election uniformly at random (this implies that for elections with less than~$100$ votes, we select some votes multiple times, and for elections with more than~$100$ votes, we do not select some votes at all). We do so to make full use of elections with far more than~$100$ votes. For instance, our Sushi intermediate dataset contains only one election consisting of~$5000$ votes. Sampling an election from the Sushi intermediate dataset thus corresponds to drawing~$100$ votes uniformly at random from the set of~$5000$ votes. On the other hand, for intermediate datasets containing a higher number of elections, e.g., the Tour de France intermediate dataset, most of the sampled elections come from different original elections.

After executing this procedure, we arrive at eleven sets, each containing~$15$ elections consisting of~$100$ complete and strict votes over~$10$ candidates, which we use for our experiments.

\subsection{Real-Life Elections on the Map} \label{sub:maprel}

In Figure~\ref{fig:main_preflib_map}, we show a map of our real-life elections
along with the compass, Mallows, and urn elections. For readability, we
present the Mallows and urn elections as large, pale-colored areas. Not
all real-life elections form clear clusters, hence the labels refer to
the largest compact groupings.

While the map is not a perfect representation of distances among
elections,
nevertheless, analyzing it leads to many conclusions.
Most strikingly, real-life elections occupy a very
limited area of the map; this is especially true for political
elections and surveys. Except for several sport elections, all
elections are closer to~$\UN$ than to~$\ID$, and none of the
real-life elections falls in the top-right part of the map. Another
observation is that Mallows elections go right through the real-life
elections, while urn elections are on average further away. This means
that for most real-life elections there exists a parameter~$\phi$ such
that elections generated according to the Mallows model with that
parameter are relatively close
(see the next section for specific recommendations).
    
\begin{figure}[t]
    \centering
    \includegraphics[width=8.5cm]{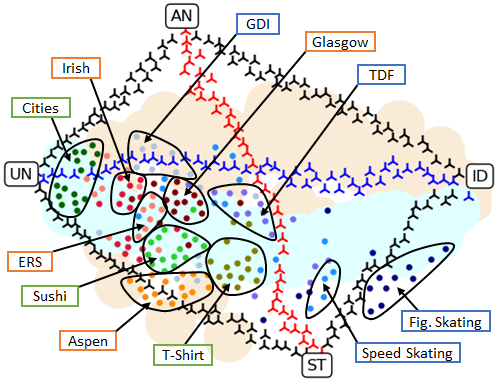}

    \caption{Map of real-life instances.}
    \label{fig:main_preflib_map}
\end{figure}

Most of the political elections lie close to each other and are
located next to the Mallows elections and high-dimensional hypercube
ones. At the same time, sport elections are spread over a larger part
of the map and, with the exception of GDI, are shifted toward~$\ID$. Regarding the surveys, the Cities survey is very similar to a sample from IC.\footnote{In the survey people were casting votes in the form of truncated ballots, ranking only their six favorite options. This is partly the reason why it is so similar to IC. Nevertheless, we have not observed any particular structure within these votes, hence its similarity to IC is not accidental.} The Sushi survey is surprisingly similar to political
elections. The T-Shirt survey is shifted toward stratification (apparently, people often agree which designs are better and which are
worse).


\subsection{Capturing Real-Life Elections} \label{sub:recom}
Let us now analyze how to choose the~$\normphi$ parameter so
that elections generated using the Mallows model with our normalization resemble the real-life
ones. We consider four different datasets, each consisting of elections with~$10$ candidates and~$100$ voters (created as described in \Cref{sub:maprel}): the set of all political
elections, the set of all sport elections, the set of all survey
elections, and the combined set, i.e., the union of the three preceding ones. For
each of these four datasets, to find the value of~$\normphi$ that
produces elections that are as similar as possible to the respective
real-life elections, we conducted the following experiment. 
For each~$\normphi\in \{0,0.001,0.002,...,0.999,1\}$, we generate~$100$
elections with~$10$ candidates and~$100$ voters from the Mallows model
with the given~$\normphi$ parameter. Subsequently, we compute the
average distance between these elections and the elections from the
respective dataset. Finally, we select the value of~$\normphi$ that minimizes this
distance. We present the results of this experiment in
\Cref{ta:realphiRL}.

\begin{table*}
\centering
\small
\begin{tabular}{l|c|c|c|c}
    \toprule
    Type of elections& Value of & Avg. Norm. & Norm. Std.& Num. of \\
   &~$\normphi$ & Distance & Dev. & Elections \\
    \midrule
    Political elections &~$0.750$ &~$0.15$ &~$0.036$ &~$60$ \\
    Sport elections &~$0.534~$ &~$0.27$ &~$0.080$ &~$60$ \\
    Survey elections &~$0.730$ &~$0.20$ &~$0.034$ &~$45$ \\ 
    All real-life elections &~$0.700$ & ~$0.22$ &~$0.106$ &~$165$ \\
    \bottomrule
  \end{tabular}
  \caption{\label{ta:realphiRL}Values of~$\normphi$ such that elections
    generated with the Mallows model for~$m=10$ are, on average, as close
    as possible to elections from the respective dataset. We include
    the average distance of the elections generated with Mallows model for
    this parameter~$\normphi$ from the elections from the dataset as
    well as the standard deviation, both normalized by distance between the uniformity and identity.
    The last column gives the number of elections in the respective
    real-life dataset.}
\end{table*}

Recall that in the previous section we have observed that a majority
of real-life elections are close to some elections generated from the
Mallows model with a certain dispersion parameter. However, we have
also seen that the real-life datasets consist of elections that differ
to a certain extent from one another (in particular, this is very
visible for the sports elections). Thus, it is to be expected that
elections drawn from the Mallows model for a fixed dispersion
parameter are at some nonzero (average) distance from the real-life
ones. Indeed, this is the case here. However, the more
homogeneous political elections and survey elections can be captured quite well using the Mallows model with parameter~$\normphi=0.750$ and~$\normphi=0.730$, respectively.  Generally speaking, if one wants to
generate elections that should be particularly close to elections from
the real world, then choosing a~$\normphi$ value between~$0.7$ and~$0.78$ seems like a good strategy. If, however, one wants to capture the full
spectrum of real-life elections, then we recommend using the Mallows
model with different values of~$\normphi$ from the interval~$[0.5,0.8]$.


\vspace{0.2cm}
\begin{conclusionbox}
The conclusions from~\Cref{preflib_info} are as follows.
\begin{itemize}
     \item Real-life elections that we study witness surprisingly little antagonism.
     \item Most of the political elections are quite similar to each other, and can be approximated by the Norm-Mallows model with~$\normphi=0.75$.
\end{itemize}
\end{conclusionbox}


\section{Skeleton Map}
Our final goal in this chapter is to form what we call a
\emph{skeleton map of vote distributions} (skeleton map, for short),
evaluate its quality and robustness, 
and compare it to the previous maps, thus getting some insights
regarding the quality and credibility of the latter.  
A trivial version of a skeleton map was presented in~\Cref{fig:paths}, where we showed a map with only compass matrices and paths between them.

\cite{boehmer2022expected} proved that for some statistical cultures, it is possible to create frequency matrices of distributions. 
In other words, instead of doing it empirically (i.e., sampling numerous elections from a given model, computing frequency matrices for each of them, and then creating a final matrix as the average over these matrices), we can analytically calculate the expected frequency matrix of a given distribution---i.e., if we had sampled infinitely many elections, and took the average over their frequency matrices, we would have obtained such an expected frequency matrix. 

We start with the Mallows model. Let~$\Phi = \{0, 0.05, 0.1, \ldots, 1\}$ be a set of normalized
dispersion parameters that we will use for Mallows-based
distributions.

For a given number of candidates, we consider the four
compass matrices ($\UN$, $\ID$, $\AN$, $\ST$) and paths between each matrix pair consisting of their convex combinations (denoted by gray dots on the map, which will be shown later), the frequency matrices
of the Norm-Mallows distribution with its normalized dispersion parameters from~$\Phi$ (denoted by blue triangles), and the frequency matrices of Conitzer (CON), Walsh (WAL), and group-separable caterpillar (CAT). Moreover, we add the frequency matrices of the following 
vote distributions (we again use the dispersion parameters from~$\Phi$):
\begin{enumerate}
\item 
The distribution for Norm-Mallows with~$\omega=0.5$ and~$\omega=0.25$   (denoted by red and green triangles, respectively),
\item 
The~$\phi$-Conitzer and~$\phi$-Walsh distributions 
where first we sample a vote~$v$ from the Walsh and Conitzer
 distributions, and then we sample the
 final vote from the Mallows distribution with dispersion parameter $\phi$ and with~$v$ as the central
 vote (denoted by magenta and orange crosses, respectively).
\end{enumerate}

In Figure~\ref{fig:skeleton_map} we show our map for the
case of~$10$ candidates. The lines between some points/matrices show
their positionwise distances (to maintain clarity, we provide only some of them). The map was created using the MDS embedding.

\begin{figure*}[t]
	\centering
 	\begin{minipage}{1\textwidth}
 		\centering
 		\scalebox{1.4}{\input{img/tex/skeleton_10}}
        \caption{The skeleton map with 10 candidates. We have~$\textrm{MID} = \nicefrac{1}{2}\AN + \nicefrac{1}{2}\ID$.}
      \label{fig:skeleton_map}
 	\end{minipage}\hfill
\end{figure*}

\begin{figure*}[t]
	\centering
 	\begin{minipage}[t]{.49\textwidth}
 		\centering
        \includegraphics[width=0.83\textwidth]{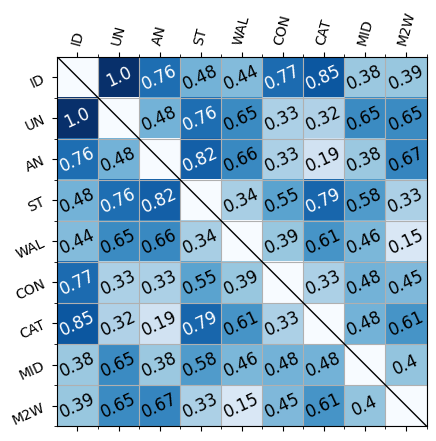}
        \caption{Normalized positionwise distances between selected matrices.}
      \label{fig:skeleton:distance}
 	\end{minipage}\hfill
 	\begin{minipage}[t]{.49\textwidth}
 		\centering
        \includegraphics[width=0.83\textwidth]{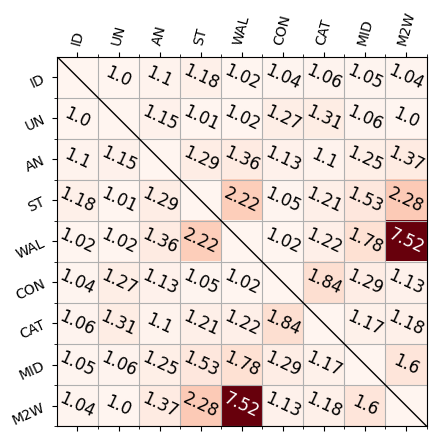}
        \caption{Distortion of the MDS embedding of the skeleton map.}
        \label{fig:skeleton:distortion}
 	\end{minipage}\hfill
\end{figure*}


We now verify the credibility of the skeleton map.
As the map 
does not have many points,  
we expect its
embedding to truly reflect the positionwise distances between the
matrices. This, indeed, seems to be the case, although 
some distances are represented (much) more accurately than others.

In Figure~\ref{fig:skeleton:distance} 
we provide the positionwise distances between the several selected matrices (for~$m=10$; matrix M2W is the Mallows matrix in our data set that is closest to the Walsh matrix), and
in Figure~\ref{fig:skeleton:distortion} for each pair of selected
  matrices we report the distortion (the smaller the value, the more accurate the embedding, recall~\Cref{ch:applications:sec:distortion}).
Most of the distortions are below $1.5$, with the majority being below $1.2$, and all but one are below $2.5$\footnote{Note that the average distortion for MDS map presented in~\Cref{fig:embed} was $1.315$}.
Thus, in most cases, the map is quite accurate and offers good intuition
about the relations between the matrices. Yet, some distances are
particularly badly represented.  As an extreme example, 
the Euclidean distance between the Walsh matrix and
the closest Mallows matrix, M2W, is off by almost a factor of~$8$ (these matrices are close, but not as close as the map suggests).
While one always has to verify claims
suggested by the skeleton map, we view it as quite credible.
This conclusion is particularly valuable when we compare the skeleton
map and the previous maps. The two maps are similar, and analogous
points (mostly) appear in analogous positions. Perhaps the biggest
difference is the location of the Conitzer matrix on the skeleton
map and Conitzer elections in the previous maps, but even this difference is
not huge.
We remark that the Conitzer matrix is closer to~$\UN$
and~$\AN$ than to~$\ID$ and~$\ST$, whereas for the Walsh matrix the opposite is true.
We made a similar observation before;
our results allow us to make this claim formal.


\subsubsection{Skeleton Map for Different Numbers of Candidates}

A natural question to raise is whether the map would look similarly if we had taken different numbers of candidates.
In \Cref{fig:skeletonMapVaryingM} we present 
three additional skeleton maps with~$5$,~$25$, and~$50$ candidates.
We see that the maps in general are quite similar, however some of the models are “moving”. In particular, we observe that when we increase the number of candidates, Walsh model is shifting towards~$\ID$, and group-separable caterpillar model is shifting towards~$\AN$. At the same time, the Conitzer model is staying almost in the same place.

\vspace{0.2cm}
\begin{conclusionbox}
The "skeleton" map witnesses lover distortion than the standard map of elections. It confirms general intuition about the positions of particular statistical cultures on the map.
\end{conclusionbox}

\begin{figure*}[t]
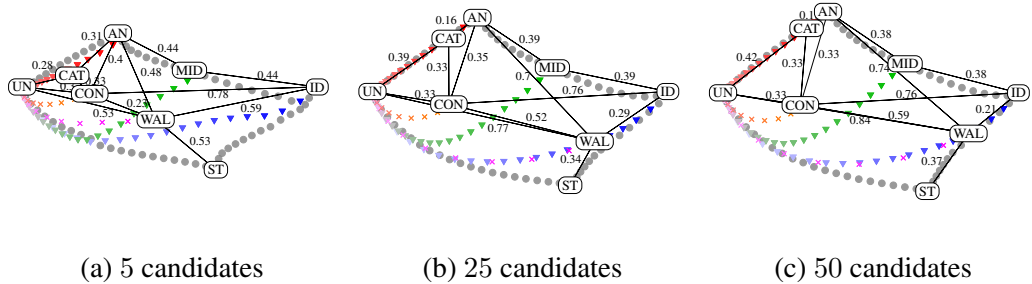

\centering
\begin{subfigure}{.32\textwidth}
  \centering
  \resizebox{\textwidth}{!}{\input{img/tex/skeleton_5}} 
  \caption{$5$ candidates}
  \label{fig:m5}
\end{subfigure}\hfill
\begin{subfigure}{.32\textwidth}
  \centering
  \resizebox{\textwidth}{!}{\input{img/tex/skeleton_25}} 
  \caption{$25$ candidates}
  \label{fig:m20}
\end{subfigure}\hfill
\begin{subfigure}{.32\textwidth}
  \centering
  \resizebox{\textwidth}{!}{\input{img/tex/skeleton_50}} 
  \caption{$50$ candidates}
  \label{fig:m50}
\end{subfigure}
\caption{Skeleton map for different number of candidates.}
\label{fig:skeletonMapVaryingM}
\end{figure*}

\section{Summary}
The main reason behind the analysis given in this chapter was to justify the use of particular parameters, embeddings, etc., and showing how selecting one embedding or another can influence the outcome. The second reason was to prove the practicality of the map, as well as its usefulness. We presented how the map of elections framework can be used to study the behavior of voting rules. We also showed the relationship between synthetic elections and real-life data.

Some of the most interesting findings are the following. For many voting rules (e.g., Borda) there is a correlation between the position of a given election on the map and the result (e.g., the highest Borda score in that election). While Urn model is easily scalable, with Mallows model we should be more cautious and depending on what the goal is, we should decide on using the normalized version or not. Most of political real-life elections lie in the particular (lower-left) part of the map and can be approximated by the Mallows model.

\vspace{0.2cm}
\begin{contributionsbox}
\begin{itemize}
    \item Evaluation of the robustness of the map of elections framework.
    \item Examples of usefulness of the map of elections framework:
    \begin{itemize}
        \item Analysis of voting rules, i.e., showing that different voting rules behave differently on different types of elections (different parts of the map). 
        \item Analysis of real-life elections, i.e., showing where they land on the map, and how to approximate them with the Norm-Mallows model.
    \end{itemize}
    \item Introduction of the "skeleton" map approach, which confirms general intuition about the positions of particular statistical cultures on the map.
\end{itemize}
\end{contributionsbox}

\chapter{Subelections}
\label{ch:subelections}

\section{Introduction}
In this chapter we study the computational complexity of several extensions of the \textsc{Election Isomorphism}
problem, which was introduced in \Cref{ch:distances} as an analogue of \textsc{Graph Isomorphism}. While in the latter we are given two graphs and we ask if they can be made identical by renaming the
vertices, in the former we are given two ordinal elections and we ask if they can be made identical by renaming
the candidates and reordering the voters. 
As we mentioned in \Cref{ch:distances}, even though the exact complexity of \textsc{Graph Isomorphism}, as well as of many related problems, remains elusive, \textsc{Election Isomorphism} has a simple polynomial-time
algorithm. Yet, in many
practical settings, perfect isomorphism is too stringent and
approximate variants are necessary. For the case of \textsc{Graph Isomorphism},
researchers considered two types of relaxation:
Either they focused on making a small number of modifications to the
input graphs that make them isomorphic (see, e.g.,
the works of~\citet{arv-koe-kuh-vas:c:approximate-graph-isomorphism}
and~\citet{gro-rat-woe:c:approximate-isomorphism}), or
they sought (maximum) isomorphic subgraphs of the input ones (see,
e.g., the classic paper of~\citet{coo:c:theorem-proving} and the
textbook of~\citet{gar-joh:b:int}; for an overview
focused on applications in cheminformatics we point to
the work of~\citet{ray-wil:j:max-common-subgraph-isomorphism}).

While in~\Cref{ch:distances} we focused on comparing different elections by measuring distances between them, which is analogous to the first type of relaxation of \textsc{Graph Isomorphism}, in this chapter we consider the second type.
In particular, we consider the \textsc{Subelection Isomorphism}
and \textsc{Maximum Common Subelection} families of problems. In the
former, we are given two elections, a smaller and a larger one, and we
ask if it is possible to remove some candidates and voters from the
larger election so that it becomes isomorphic to the smaller one. Put
differently, we ask if the smaller election occurs as a minor in the
larger one. One reason why this problem is interesting is its
connection to restricted preference domains.
For example, single-peaked and single-crossing 
elections are characterized as those that do not have certain forbidden
minors~\citep{bal-har:j:characterization-single-peaked,bre-che-woe:j:single-crossing}.
We show that \textsc{Subelection Isomorphism} is~$\np$-complete and
$\wone$-hard to parameterize by the size of the smaller
election, which suggests that there are no fast algorithms for the
problem. Fortunately, the characterizations of single-peaked and
single-crossing elections use minors of constant size and
such elections can be recognized efficiently; indeed, there are
very fast algorithms for these
tasks~\citep{bar-tri:j:stable-matching-from-psychological-model,esc-lan-ozt:c:single-peaked-consistency,elk-fal-sli:c:decloning}.
Our results show that characterizations with
nonconstant minors might lead to~$\np$-hard recognition problems.

In our second problem, \textsc{Maximum Common Subelection}, we
ask for the largest isomorphic subelections of the two input ones.
Although we find that
many of our problems are~$\np$-hard, we
also find polynomial algorithms, also for practically useful cases.

For both of our problems, we consider their \textit{candidate} and \textit{voter} variants.
For example, in \textsc{Candidate Subelection Isomorphism} we ask if it is possible
to remove candidates from the larger election (but without deleting
any voters) so that it becomes isomorphic with the smaller
one. Similarly, in \textsc{Maximum Common Voter-Subelection} we ask if
we can ensure the isomorphism of the two input elections by only deleting
voters (so that at least a given number of voters remains).
In Section~\ref{ch:sub:sec:experiments} we use this latter problem to
evaluate the similarity between elections generated from various
statistical cultures and some real-life elections.
These results confirm some findings observed in previous chapters and provide a new perspective on some of these statistical cultures and real-life elections.

In the most general variants of our problems, we assume that both
input elections are over different candidate sets and include
different voters. Yet, sometimes it is natural to assume that the
candidates or the voters are the same (for example, in a presidential
election votes collected in two different districts would have the
same candidate sets, but different voters, whereas two consecutive
presidential elections would largely involve the same voters, but not
necessarily the same candidates). We model such scenarios by variants
of our problems in which either the matchings between the candidates or
the voters of the input elections are given. Although one would expect
that having such matchings would make our problems easier, there are
cases where they remain~$\np$-hard even with both matchings. This
contrasts sharply with the results
from \Cref{tab:iso_complexity} from \Cref{ch:distances}. For a summary of our
results, see Table~\ref{tab:sub_complexity_results}.

The approach taken in this chapter is significantly different from the one presented in~\Cref{ch:distances} but, as before, the main aim is to get a better understanding of the nature of statistical culture models as well as of real-life elections. Results presented in this chapter are complementary to the previous ones, and give us a better understand of the map of elections.

\begin{table*}[t]
\centering
\resizebox{\textwidth}{!}{\begin{tabular}{ c | c | c | c | c }
		\toprule
		 & no & voter & candidate & both\\
		Problem & matching & matching & matching & matchings\\
		\midrule
		Election Isomorphism &~$\p$ &~$\p$  &~$\p$  &~$\p$  \\
		\midrule
		Subelection Isomorphism &~$\wone$-hard &~$\np$-com.&~$\p$&~$\p$  \\

		Cand.-Subelection Isomorphism &~$\np$-com. &~$\np$-com. &~$\p$ &~$\p$ \\

		Voter-Subelection Isomorphism &~$\p$&~$\p$ &~$\p$ &~$\p$ \\
		\midrule
		Max. Common Subelection &~$\wone$-hard &~$\np$-com.&~$\np$-com. &$\np$-com. \\

		Max. Common Cand.-Subelection &~$\np$-com. &~$\np$-com. &~$\np$-com. &~$\wone$-com. \\

		Max. Common Voter-Subelection &~$\p$ &~$\p$ &~$\p$ &~$\p$\\
		\bottomrule
\end{tabular}}
\caption{\label{tab:sub_complexity_results} An overview of our results;
    those for \textsc{Election Isomorphism} are taken from \Cref{ch:distances}.
    $\wone$-hardness holds with respect to the size of the smaller election or a common subelection.
    The~$\wone$-hard problems are also~$\np$-hard.
  }
\end{table*}

\vspace{-0.3cm}
\section{Variants of the Isomorphism Problem}
\vspace{-0.2cm}

Given elections~$E = (C,V)$ and~$E' = (C',V')$, we say that~$E'$ is a
\emph{subelection} of~$E$ if~$C'$ is a subset of~$C$ and~$V'$ can be
obtained from~$V$ by deleting some voters and restricting the remaining
ones to the candidates from~$C'$.
We say that~$E'$ is a \emph{voter subelection} of~$E$
if we can obtain it by only deleting voters from~$E$,
and that~$E'$ is a \emph{candidate subelection} of~$E$
if we can obtain it from~$E$ by only deleting candidates. 
By the \textit{size} of an election, we mean the number of candidates multiplied by the number of voters.

As a reminder, two elections are \emph{isomorphic} if it is possible to rename their
candidates and reorder their voters so that they become
identical. Formally,
elections~$(C_1,V_1)$ and~$(C_2,V_2)$, are isomorphic if
$|C_1| = |C_2|$,~$|V_1| = |V_2|$, and there is a bijection
$\sigma \colon C_1 \rightarrow C_2$ and a permutation
$\pi \in S_{|V_1|}$ such that~$(\sigma(C_1),\sigma(\pi(V_1))) = (C_2,V_2)$.
We refer to~$\sigma$ as the candidate matching and to~$\pi$ as the
voter matching. 

Given a graph~$G$, we write~$V(G)$ to refer to its set of
vertices and~$E(G)$ to refer to its set of edges. Most of our intractability proofs follow by reductions from the
\textsc{Clique} problem. An instance of \textsc{Clique} consists of a
graph~$G$ and a nonnegative integer~$k$, and we ask if~$G$ contains
$k$ vertices that are all connected to each other. \textsc{Clique} is
well-known to be both~$\np$-complete and~$\wone$-complete, for the
parameterization by~$k$~\citep{downey1995fixed}. As all the problems that we study can easily be seen to belong to~$\np$,
in our~$\np$-completeness proofs we only give hardness arguments.

Now we are ready to introduce two extensions of the \textsc{Election Isomorphism} problem,
called \textsc{Subelection Isomorphism} and
\textsc{Maximum Common Subelection}. In the former, we are given two
elections, and we ask if the smaller one is isomorphic to a subelection
of the larger one. That is, we ask if we can remove some
candidates and voters from the larger election to make the two
elections isomorphic.

\begin{definition}
  An instance of \textsc{Subelection Isomorphism} consists of two
  elections,~$E_1 = (C_1, V_1)$ and~$E_2 = (C_2,V_2)$, such that~$|C_1| \leq |C_2|$ 
  and~$|V_1| \leq |V_2|$. We ask if there is a
  subelection~$E'$ of~$E_2$ isomorphic to~$E_1$.
\end{definition}
The \textsc{Voter-Subelection Isomorphism} problem is defined in the
same way, except that we require~$E'$ to be a voter subelection of
$E_2$. Similarly, in \textsc{Candidate-Subelection Isomorphism} we
require~$E'$ being a candidate subelection. We often abbreviate the
name of the latter problem to \textsc{Cand.-Subelection Isomorphism}.

\begin{example}\label{ex:1}
  Consider elections~$E = (C,V)$ and~$F = (D, U)$, where~$C=\{a,b,c\}$,
  ~$D = \{x,y,z,w\}$,
  ~$V=(v_1,v_2,v_3)$ and~$U=(u_1,u_2,u_3)$, with 
  preference orders:
  \begin{align*}
    &v_1 \colon a \pref b \pref c,
    &u_1 \colon w \pref x \pref y \pref z,& \\
    &v_2 \colon b \pref a \pref c,
    &u_2 \colon y \pref w \pref x \pref z,& \\
    &v_3 \colon c \pref b \pref a,
    &u_3 \colon z \pref w \pref y \pref x.
  \end{align*}
  If we remove candidate~$w$ from~$(D,U)$, then we find that the
  resulting elections are isomorphic (to see this, it suffices to match
  voters~$v_1,v_2,v_3$ with~$u_1,u_2,u_3$, respectively, and
  candidates~$a,b,c$ with~$x,y,z$). Thus~$E$ is isomorphic to a
  (candidate) subelection of~$F$ and, so,~$(E,F)$ is a
  \emph{yes}-instance of \textsc{(Cand.-)Subelection Isomorphism}.
\end{example}

In the \textsc{Maximum Common Subelection} problem,
we seek the largest isomorphic subelections of two given ones.
We often abbreviate \textsc{Maximum} as \textsc{Max.}

\begin{definition}
  An instance of \textsc{Max. Common Subelection} consists of two
  elections,~$E_1 = (C_1,V_1)$ and~$E_2 = (C_2,V_2)$, and a
  positive integer~$t$. We ask if there is a subelection~$E'_1$
  of~$E_1$ and a subelection~$E'_2$ of~$E_2$ such that~$E'_1$ 
  and~$E'_2$ are isomorphic and the size of~$E'_1$ (or equivalently, the
  size of~$E'_2$) is at least~$t$.
\end{definition}

Analogously to the case of \textsc{Subelection Isomorphism},
we also consider the \textsc{Max. Common Cand.-Subelection} and
\textsc{Max. Common Voter-Subelection} problems. In the former,~$E'_1$
and~$E'_2$ must be candidate subelections
and in the latter they need to be voter subelections
(thus, in the former problem~$E_1$ and~$E_2$ must have the
same numbers of voters, and in the latter~$E_1$ and~$E_2$ must have
the same numbers of candidates).

For each of the above-defined problems, we consider its variant
with or without the candidate or voter matching. Specifically, the
variants defined above are \emph{with no matchings}. Variants
\emph{with candidate matching} include a bijection~$\sigma$ that
matches (some of) the candidates in one election to (some of) those
in the other (in the case of \textsc{Subelection Isomorphism}
and its variants, all candidates in the smaller election must be
matched to those in the larger one;
in case of \textsc{Max. Common Subelection} there are no such requirements).
Then we ask for an isomorphism between respective subelections that agrees with
$\sigma$. In particular, this means that none of the unmatched
candidates remains in the considered subelections (another
interpretation is to assume that both input elections have the
same candidate sets).
\begin{example}
  Consider elections~$(C,V)$ and~$(D, U)$ from Example~\ref{ex:1}, and
  a matching~$\sigma$ such that~$\sigma(a)=x, \sigma(b)=w$, where~$c$,~$y$, and~$z$ are unmatched.
  After applying it and dropping the unmatched candidates,
  the votes in the first election become
  \begin{align*}
      v_1 \colon x \pref w, \ \ v_2 \colon w \pref x, \ \ v_3 \colon w \pref x,
  \end{align*}
  whereas all the voters in the second election
  have preference order~$w \pref x$. Thus, this instance of
  \textsc{Max. Common Subelection with Candidate Matching} has
  isomorphic subelections, respecting the matching~$\sigma$, of size~$2 \cdot 2 = 4$.
\end{example}
The variants \emph{with voter matching} are defined similarly: We are
given a matching between (some of) the voters from one election and
(some of) the voters from the other (and, again, for
\textsc{Subelection Isomorphism} and its variants, each voter in the
smaller election is matched to some voter in the larger one). The
sought-after isomorphism must respect this matching (again, this means
that we can disregard the unmatched voters).

The variants \emph{with both matchings} include both the matching between
the candidates and the matching between the voters (note that these
variants are not trivial because we still need to decide who to
remove). By writing \emph{all four matching cases} we mean the four
just described variants of a given problem.

Finally, we note that each variant of \textsc{Max. Common Subelection}
is at least as computationally difficult as its corresponding variant
of \textsc{Subelection Isomorphism}.

\begin{repproposition}{M reduces to S}
  \label{prop:reduction}
  Let~$M$ be a variant of \textsc{Max. Common Subelection} and let~$S$
  be a corresponding variant of \textsc{Subelection Isomorphism}. We
  have that~$S$ 
  reduces to~$M$ in polynomial time.
\end{repproposition}

\begin{proof}
  We are given a problem~$M$ which is a variant of \textsc{Max. Common
  Subelection}, and problem~$S$, which is an analogous variant of
  \textsc{Subelection Isomorphism} (so if the former only allows
  deleting candidates, then so does the latter, etc.). We want to show that~$S$ 
  reduces to~$M$ in
  polynomial time. Let~$I_S = (E_1,E_2)$ be an instance of~$S$, where~$E_1$ 
  is the smaller election. We form an instance~$I_M = (E_1,E_2,t)$ of~$M$, 
  which uses the very same elections and
  where~$t$ is set to be the size of~$E_1$. This means that in~$I_M$
  we cannot perform any operations on election~$E_1$, because that
  would decrease its size below~$t$. Therefore, we can only perform
  operations on~$E_2$, so the situation is the same as in the~$S$
  problem and in the~$I_S$ instance.
\end{proof}

\section{Computational Complexity Analysis}\label{sec:comp-complexity}

In this section, we present our complexity results. Although in most cases
we obtain intractability (recall Table~\ref{tab:sub_complexity_results} for a
summary of our results), we find that all our problems focused on
voter subelections are solvable in polynomial time, and having
candidate matchings leads to the polynomial-time algorithm for all
variants of \textsc{Subelection Isomorphism}.

All our polynomial-time results are based on the trick
used for the
case of \textsc{Election Isomorphism} in \Cref{ch:distances}. The idea is to guess a pair
of (matched) voters and use them to derive the candidate matching.

\begin{reptheorem}{Polynomial-time algorithm results}
\label{thm:poly-algo}
  \textsc{Voter-Subelection Isomorphism} and
  \textsc{Max. Common Voter-Subelection} are in~$\p$ for all four matching
  cases. \textsc{Subelection Isomorphism},
  \textsc{Cand.-Subelection Isomorphism} are in~$\p$ for cases
  with candidate matchings.
\begin{proof}
  We first give an algorithm for \textsc{Max. Common Voter-Subelection}.
  Let~$E = (C, V)$ and~$F = (D,U)$ be our input elections and let~$t$
  be the desired size of their isomorphic subelections. Since we are
  looking for a voter subelection, without loss of generality we may
  assume that~$|C| = |D|$ (and we write~$m$ to denote the number of candidates
  in each set). For each voter~$v \in V$ and each voter~$u \in U$ 
  we perform the following algorithm:
  \begin{enumerate}
  \item Denoting the preference orders of~$v$ and~$u$ as~$c_1 \pref_v c_2 \pref_v \cdots \pref_v c_m$
  and~$d_1 \pref_u d_2 \pref_u \cdots \pref_u d_m$, respectively, we
    form a bijection~$\sigma \colon C \rightarrow D$ such that for
    each~$c_i \in C$ we have~$\sigma(c_i) = d_i$.
  \item We form a bipartite graph where the voters from~$V$ form one
    set of vertices, the voters from~$U$ form the other set of
    vertices, and there is an edge between voters~$v' \in V$ 
    and~$u'\in U$ if~$\sigma(v') = u'$.
  \item We compute the maximum cardinality matching in this graph and
    form subelections that consist of the matched voters. We accept if
    their size is at least~$t$.
  \end{enumerate}
  If the algorithm does not accept for any choice of~$v$ and~$u$, we
  reject.

Very similar algorithms also work for the variants of \textsc{Max. Common
Voter-Subelection} with either one or both of the matchings: If we are
given the candidate matching, then we can omit the first step in the
enumerated algorithm above, and if we are given a voter matching then
instead of trying all pairs of voters~$v$ and~$u$ it suffices to try
all voters from the first election and obtain the other one via the
matching. Analogous algorithms also work for
\textsc{Voter-Subelection Isomorphism} (for all four matching cases)
and for all the other variants of \textsc{Subelection Isomorphism},
provided that the candidate matching is given.
\end{proof}
\end{reptheorem}

\subsection{Intractability of Subelection Isomorphism}

Next, we show the computational hardness of all the remaining variants of our problems.
In this section we consider \textsc{Subelection Isomorphism}.

\begin{theorem}\label{thm:sub-elec-np-hard}
  \textsc{Subelection Isomorphism} is~$\np$-complete 
  and~$\wone$-hard with respect to the size of the smaller election.
\end{theorem}

\begin{proof}
  Before we describe our reduction, we first provide a method for
  transforming a graph into an election: For a graph~$H$, we let~$E_H$
  be an election whose candidate set consists of the vertices of~$H$
  and two special candidates,~$\alpha_H$ and~$\beta_H$, and whose
  voters correspond to the edges of~$H$. Specifically, for each edge~$e = \{x,y\} \in E(H)$ 
  we have four voters,~$v^1_e, \ldots v^4_e$,
  with preference orders:
  \begin{align*}
    v^1_e \colon& x \pref y \pref \alpha_H \pref \beta_H \pref V(H) \setminus \{x,y\}, \\
    v^2_e \colon& x \pref y \pref \beta_H \pref \alpha_H \pref V(H) \setminus \{x,y\}, \\
    v^3_e \colon& y \pref x \pref \alpha_H \pref \beta_H \pref V(H) \setminus \{x,y\}, \\
    v^4_e \colon& y \pref x \pref \beta_H \pref \alpha_H \pref V(H) \setminus \{x,y\}.
  \end{align*}
  Note that elements from the set~$V(H) \setminus \{x,y\}$ are always in the same order.
  We give a reduction from \textsc{Clique}. Given an instance~$(G,k)$
  of \textsc{Clique}, where~$G$ has at least~$k$ vertices 
  and~$\binom{k}{2}$ edges,
  we let~$K$ be a size-$k$ complete graph and we form 
  an instance~$(E_K, E_G)$ of \textsc{Subelection Isomorphism}. The reduction
  runs in polynomial time and it remains to show its correctness.

  First, let us assume that~$G$ has a size-$k$ clique. Let~$X$ be the
  set of its vertices and let~$Y$ be the set of its edges. We form a
  subelection~$E'$ of~$E_G$ by removing all the candidates that are
  not in~$X \cup \{\alpha_G, \beta_G\}$ and removing all the voters that
  do not correspond to the edges from~$Y$. One can verify that~$E'$ 
  and~$E_K$ are, indeed, isomorphic.

  Second, let us assume that~$E_K$ is isomorphic to some 
  subelection~$E'$ of~$E_G$. We will show that this implies that~$G$ has a
  size-$k$ clique. First, we claim that~$E'$ includes 
  both~$\alpha_G$ and~$\beta_G$. To see why this is so, consider the
  following two cases:
  \begin{enumerate}
  \item If~$E'$ contained exactly one of~$\alpha_G, \beta_G$, then this
    candidate would appear in every vote in~$E'$ among the top three
    positions. Yet, in~$E_K$ there is no candidate with this property,
    so~$E'$ and~$E_K$ would not be isomorphic.
  \item If~$E'$ contained neither~$\alpha_G$ nor~$\beta_G$ then every
    vote in~$E'$ would rank some vertex candidates~$z$ and~$w$ on
    positions three and four
    (to be able to match~$\alpha_K$ and~$\beta_K$ to them).
    However, by the construction of~$E_G$,
    either in every vote from~$E'$ we would have~$z \pref w$ or in
    every vote from~$E'$ we would have~$w \pref z$. Since in~$E_K$
    half of the voters rank the candidates from positions three and
    four in the opposite way,~$E'$ and~$E_K$ would not be isomorphic.
  \end{enumerate}
  Thus~$\alpha_G$ and~$\beta_G$ are included in~$E'$.
  Moreover~$\alpha_G$ and~$\beta_G$ are matched 
  with~$\alpha_K$ and~$\beta_K$ because they are the only
  candidates from~$E_G$ that can appear on positions three and four
  in every vote in~$E'$ but possibly in different order.
  As a consequence, for each vote~$v$ from~$E_G$ that appears in~$E'$,
  the candidate set of~$E'$ must include
  the two candidates from~$V(G)$ that~$v$ ranks
  on top (if it were not the case, then~$E'$ would contain a
  candidate---either~$\alpha_G$ or~$\beta_G$---that appeared in all
  the votes within the top four positions and in some vote within top
  two positions; yet~$E_K$ does not have such a candidate).
  This means that for each edge~$e \in E(G)$, if~$E'$ contains some
  voter~$v^i_e$ for~$i \in [4]$, then it also contains the other
  voters corresponding to~$e$ (otherwise,~$E'$ and~$E_K$ would not be
  isomorphic).
  The number of voters in~$E_K$ is~$4\binom{k}{2}$,
  and the number of distinct corresponding edges 
  from~$G$ is~$\binom{k}{2}$.
  As said before, for each such chosen edge, we also choose two
  corresponding vertices as candidates.
  It means that the number of chosen candidates
  (except~$\alpha_G$ and~$\beta_G$) is between~$k$ and~$2\binom{k}{2}$.
  However, the number of candidates in~$E_K$ except~$\alpha_K$ 
  and~$\beta_K$ is~$k$, therefore we conclude that chosen vertex-candidates
  form a size-$k$ clique in~$G$.
  This completes the proof of~$\np$-hardness.

  To show~$\wone$-hardness, note that the number of candidates and
  voters in the smaller election equals~$k+2$ and~$4\binom{k}{2}$
  respectively, hence the size of the smaller election is a function of
  parameter~$k$ for which \textsc{Clique} is~$\wone$-hard.
\end{proof}

Next, we consider \textsc{Cand.-Subelection Isomorphism}.
In this problem both elections have the same number of voters, and
we ask if we can delete candidates from the one that has more,
so that they become isomorphic.
We first show that this problem is~$\np$-complete for the
case where the voter matching is given (which also proves the same
result for \textsc{Subelection Isomorphism with Voter Matching}) and
next we describe how this proof can be adapted to the variant without
any matchings (the variant with candidate matching is in~$\p$ and was
considered in the preceding section).

\begin{reptheorem}{Cand-Subelection with voter matching hardness}\label{thm:subelection-voter-candidate}
  \textsc{Subelection Isomorphism with Voter Matching} and
  \textsc{Cand.-Subelection Isomorphism with Voter Matching} are~$\np$-complete.
\end{reptheorem}

\begin{proof}
  It suffices to consider \textsc{Cand.-Subelection Isomorphism with
  Voter Matching}. We give a reduction from the \textsc{Exact Cover
  by 3-Sets} problem \textsc{(X3C)}. An instance of \textsc{X3C}
  consists of a set~$X = \{x_1, \ldots, x_m\}$ of elements and a
  family~$\calS = \{S_1, \ldots, S_n\}$ of three-element subsets of~$X$. We ask if~$\calS$ contains a subfamily~$\calS'$ such that each
  element from~$X$ belongs to exactly one set from~$S'$. Given such
  an instance, we form two elections,~$E_1$ and~$E_2$, as follows.

  Election~$E_1$ will be our smaller election and~$E_2$ will be the
  larger one. We let~$X$ be the candidate set for election~$E_1$,
  whereas to form the candidate set of~$E_2$ we proceed as follows.
  For each set~$S_t = \{x_i, x_j, x_k\} \in \calS$, we introduce
  candidates~$s_{i,t}$,~$s_{j,t}$, and~$s_{k,t}$. Intuitively, if some
  candidate~$s_{u,t}$ remains in a subelection isomorphic to~$E_1$, 
  then we will interpret this fact as saying that element~$x_u$
  is covered by set~$S_t$; this will, of course, require introducing
  appropriate consistency gadgets to ensure that~$S_t$ also covers its
  other members. For each~$i \in [m]$, we let~$P_i$ be the set of all
  the candidates of the form~$s_{i,t}$, where~$t$ belongs to~$[n]$ (in
  other words, each~$P_i$ contains the candidates from~$E_2$ that
  are associated with element~$x_i$).

  \begin{example}
    Let~$X=\{x_1,x_2,x_3,x_4,x_5,x_6\}$ 
    and~$\calS=\{S_1, \ldots, S_4\}$, 
    where
    \begin{align*}
    S_1 = \{x_1,x_2,x_5\}, S_2 = \{x_1,x_3,x_6\}, S_3 = \{x_2,x_3,x_5\}, S_4 = \{x_3, x_4,x_6\} \}. 
    \end{align*}
    Then
    \begin{align*}
    &P_1 = \{s_{1,1}, s_{1,2}\}, &P_2 = \{s_{2,1}, s_{2,3}\}, \ \ \ \ \ \ \ \ \ \  &P_3 = \{s_{3,2}, s_{3,3}, s_{3,4}\}, \\
    &P_4 = \{s_{4,4}\}, &P_5 = \{s_{5,1}, s_{5,3}\},  \ \ \ \ \ \ \ \ \ \  &P_6 = \{s_{6,2}, s_{6,4}\}.
    \end{align*}
  \end{example}

  We denote the voter collection of election~$E_1$ 
  as~$V = (v, v', v_1, v'_1, \ldots, v_{n}, v'_n)$ and the voter
  collection of~$E_2$ 
  as~$U = (u, u', u_1, u'_1, \ldots, u_{n}, u'_n)$. Voters~$v$ and~$v'$
  are matched to voters~$u$ and~$u'$, respectively, and for 
  each~$i \in [2n]$,~$v_i$ is matched to~$u_i$ and~$v'_i$ is matched 
  to~$u'_i$. The preference orders of the first two pairs of voters are:
  \begin{align*}
    v  &\colon x_1 \succ x_2 \succ \dots \succ x_m, &
    u  &\colon P_1 \succ P_2 \succ \dots \succ P_m, \\
    v' &\colon x_1 \succ x_2 \succ \dots \succ x_m, &
    u' &\colon
      \overleftarrow{P}_1 \succ \overleftarrow{P}_2 \succ \dots \succ
      \overleftarrow{P}_m.
  \end{align*}
  Next, for each~$t \in [n]$ such that~$S_t = \{x_i,x_j,x_k\}$,
    ~$i < j < k$, we let the preference orders of~$v_{t}$,~$v'_{t}$ and
  their counterparts from~$U$ be as follows (by writing ``$\cdots$'' in
  the votes from~$E_1$ we mean listing the candidates 
  from~$X \setminus \{x_i,x_j,x_k\}$ in the order of increasing indices, and for
  the voters from~$E_2$ by ``$\cdots$'' we mean 
  order~$P_1 \pref P_2 \pref \cdots \pref P_m$ with~$P_i$,~$P_j$, and~$P_k$
  removed):
  \begin{align*}
    v_{t} & \colon x_{i} \pref x_{j} \pref x_{k} \pref \cdots,\\
    u_{t} & \colon s_{i,t} \pref s_{j,t} \pref s_{k,t}
        \pref P_{i} \setminus \{s_{i,t}\} \pref P_{j} \setminus \{s_{j,t}\} \pref P_{k} \setminus \{s_{k,t}\}
        \pref \cdots, \\
    v'_{t} & \colon x_{k} \pref x_{j} \pref x_{i} \pref \cdots,\\
    u'_{t} & \colon s_{k,t} \pref s_{j,t} \pref s_{i,t}
        \pref
        P_{k} \setminus \{s_{k,t}\} \pref P_{j} \setminus \{s_{j,t}\} \pref
        P_{i} \setminus \{s_{i,t}\} \pref \cdots.
  \end{align*}
  This finishes our construction. It is clear that it is
  polynomial-time computable and it remains to show that it is
  correct.

  Let us assume that we have a \emph{yes}-instance of \textsc{X3C},
  that is, there is a family~$\calS'$ of sets from~$\calS$ such that
  each element from~$X$ belongs to exactly one set from~$\calS'$. We
  form a subelection~$E'$ of~$E_2$ by deleting all the 
  candidates~$s_{i,t}$ except for those for whom set~$S_t$ belongs to~$\calS'$.
  Then, let~$\sigma$ be a function such that for each~$x_i \in X$ we
  have~$\sigma(x_i) = s_{i,t}$, where~$S_t$ is a set from~$\calS'$
  that contains~$x_i$. Together with our voter matching,~$\sigma$
  witnesses that~$E_1$ and~$E'$ are isomorphic.

  For the other direction, let us assume that~$E_2$ has 
  a subelection~$E'$ that is isomorphic to~$E_1$ and let~$C'$ be its candidate set.
  First, we claim that for each~$i \in [m]$ exactly one candidate 
  from~$P_i$ is included in~$C'$. Indeed, if it were not the case, then~$u$
  and~$u'$ would not have identical preference orders, as required by
  the fact that they are matched to~$v$ and~$v'$. Second, we note
  that for each~$i \in [m]$ the candidate matching that witnesses our
  isomorphism must match~$x_i$ with the single candidate 
  in~$P_i \cap C'$.
  Finally, we claim that if some candidate~$s_{i,t}$ is included 
  in~$C'$, where~$S_t = \{x_i,x_j,x_k\}$,~$i < j < k$, then 
  candidates~$s_{j,t}$ and~$s_{k,t}$ are included in~$C'$ as well. Indeed, 
  if~$s_{j,t}$ were not included in~$C'$, then~$u_t$ and~$u'_t$ would
  rank the members of~$P_i \cap C'$ and the members of~$P_j \cap C'$ in
  the same order, whereas~$v_j$ and~$v'_j$ rank~$x_i$ and~$x_j$ in
  opposite orders (and, by the second observation,~$x_i$ and~$x_j$ are
  matched to the member of~$P_i \cap C'$ and~$P_j \cap C'$,
  respectively). The same argument applies to~$x_{k,t}$.
  We say that a set~$S_t = \{x_i, x_j, x_k\} \in \calS$ is selected 
  by~$C'$ if the candidates~$s_{i,t}$,~$s_{j,t}$, and~$s_{k,t}$ belong to~$C'$. 
  By the above reasoning, we see that exactly~$n/3$ sets are
  selected and that they form an exact cover of~$X$.
\end{proof}

\textsc{Cand.-Subelection Isomorphism}
remains~$\np$-complete also without the voter matching.
By doubling the voters and using a few extra candidates, we ensure that
only the intended voter matching is possible.

\begin{repproposition}{Cand.-Subelection Isomorphism problem}
\label{prop:can-subelection}
\textsc{Cand.-Subelection Isomorphism} is~$\np$-complete.
\end{repproposition}

\begin{proof}
  We give a reduction from \textsc{Cand.~Subelection Isomorphism with
  Voter Matching}. Let the input instance be~$(E_1,E_2)$, where the
  smaller election,~$E_1$, has voter collection~$(v_1, \ldots, v_n)$
  and the larger election,~$E_2$, has voter 
  collection~$(u_1, \ldots, u_n)$. Additionally, for each~$i \in [n]$ voter~$v_i$ is
  matched with voter~$u_i$. We form elections~$E'_1$ and~$E'_2$ in
  the following way. The candidate set of~$E'_1$ is the same as that
  of~$E_1$ except that it also includes candidates from the 
  set~$D = \{d_1, \ldots, d_{2n}\}$. Similarly,~$E'_2$ contains the same
  candidates as~$E_2$ plus the candidates from the 
  set~$F = \{f_1, \ldots, f_{2n}\}$. The voter collections of~$E'_1$ 
  and~$E'_2$ are, respectively,~$(v'_1, v''_1, \ldots, v'_n,v''_n)$ 
  and~$(u'_1, u''_1, \ldots, u'_n, u''_n)$. For each~$i \in [n]$ these
  voters have the following preference orders (by writing~$[v_i]$ 
  or~$[u_i]$ in a preference order we mean inserting the preference order
  of voter~$v_i$ or~$u_i$ in a given place:
  \begin{align*}
    v'_{i} \colon & d_{2i-1} \pref d_1 \pref \cdots \pref d_{2n} \pref [v_i], \\
    u'_{i} \colon & f_{2i-1} \pref f_1 \pref \cdots \pref f_{2n} \pref [u_i], \\
    v''_{i} \colon & d_{2i} \pref d_1 \pref \cdots \pref d_{2n} \pref [v_i], \\
    u''_{i} \colon & f_{2i} \pref f_1 \pref \cdots \pref f_{2n} \pref [u_i].
  \end{align*}
  We claim that~$E'_1$ is isomorphic to a candidate subelection 
  of~$E'_2$ if and only if~$E_1$ is isomorphic to a candidate subelection
  of~$E_2$ with the given voter matching. In one direction this is
  clear: If~$E_1$ is isomorphic to a subelection of~$E_2$ with a given
  voter matching, then it suffices to use the same voter matching (extended in the obvious way) for
  the case of~$E'_1$ and~$E'_2$, and the same candidate matching,
  extended with matching each candidate~$d_i$ to~$f_i$.

  Next, let us assume that~$E'_1$ is isomorphic to some 
  subelection~$E'$ of~$E'_2$. By a simple counting argument, we note that~$E'$
  must contain some candidates not in~$F$. Further, we also note that
  it must contain all members of~$F$. Indeed, each voter in~$E'_1$ has
  a different candidate on top, and this would not be the case in~$E'$
  if it did not include all members of~$F$ (if~$E'$ did not include
  any members of~$F$ then this would hold for each two votes~$u'_i$
  and~$u''_i$, and if~$E'$ contained some members of~$F$ but not all
  of them, then this would hold because each voter in~$E'$
  would rank some member of~$F$ on top, but there would be fewer
  members of~$F$ than voters in the election).

  As a consequence, every candidate matching~$\sigma$ that witnesses
  isomorphism between~$E'_1$ and~$E'$ matches some member of~$D$ to
  some member of~$F$. Furthermore, we claim that for each~$i \in [2n]$,
    ~$\sigma(d_i) = f_i$. For the sake of contradiction, let us assume
  that this is not the case and consider some~$i \in [2n-1]$ for which
  there are~$j$ and~$k$ such that~$\sigma(d_i) = f_j$,
    ~$\sigma(d_{i+1}) = f_k$ and~$j > k$ (such~$i$,~$j$,~$k$ must exist
  under our assumption). However, in~$E'_1$, all but one voter 
  rank~$d_i$ ahead of~$d_{i+1}$, while in~$E'$ all but one voter 
  rank~$\sigma(d_{i+1})$ ahead of~$\sigma(d_i)$. Thus~$\sigma$ cannot
  witness isomorphism between~$E'_1$ and~$E'$.

  Finally, since for each~$i \in [2n]$ we have that~$d_i$ is matched
  to~$f_i$, it also must be the case that for each~$j \in [n]$ 
  voters~$v'_j$ and~$v''_j$ are matched to~$u'_j$ and~$u''_j$, respectively
  (indeed,~$v'_j$ is the only voter who ranks~$d_{2j-1}$ on top, 
  and~$u'_j$ is the only voter who ranks~$f_{2j-1}$ on top; the same
  argument works for the other pair of voters). As a consequence, we
  have that~$E_1$ is isomorphic to a subelection of~$E_2$ under the
  voter matching that for each~$i \in [n]$ matches~$v_i$ to~$u_i$.
\end{proof}

\subsection{Intractability of Max. Common Subelection}

Perhaps the most surprising result regarding
\textsc{Max. Common Subelection} is that it is
$\np$-complete even when both matchings are given.
The surprise stems from the fact that \textsc{Isomorphism
  Distance} problems (i.e., computing Spearman or Swap distances between elections) are solvable in
polynomial-time given both matchings.
We first show this result for
candidate subelections.

\begin{reptheorem}{Common-Cand with both matching hardness}\label{thm:common-cand-both-np-hard}
\textsc{Max. Common Cand.-Subelection with Both Matchings} is~$\np$-complete and~$\wone$-complete with respect to the candidate set size of isomorphic candidate subelections.
\end{reptheorem}

\begin{proof}
  We give a reduction from the \textsc{Clique} problem,
  where the idea is to encode the adjacency matrix of a given graph
  by a pair of elections with both matchings defined.
  Missing edges in the graph we encode as a conflict on candidate
  ordering within matched voters.

  Formally, given an instance~$(G,k)$ of \textsc{Clique}, we form
  two elections,~$E_1 = (C,V_1)$ and~$E_2 = (C,V_2)$,
  where~$C = V(G)$.
  Since we need to provide an instance with candidate matching,
  we simply specify both elections over the same candidate set.
  Without loss of generality, we assume that~$V(G) = \{1, \ldots, n\}$.
  For each~$x \in V(G)$ we define the neighborhood of~$x$ in~$G$
  as~$N(x) = \{y \in V(G): \{x,y\} \in E(G)\}$ and the set of
  non-neighbors as~$M(x) = V(G) \setminus \{N(x) \cup \{x\}\}$.

  For each vertex~$x \in V(G)$ we define two matched voters,
  ~$v^1_x$ in~$E_1$ and~$v^2_x$ in~$E_2$, defined as follows:
  \begin{align*}
    &v^1_x \colon M(x) \pref x \pref N(x),& \\
    &v^2_x \colon x \pref M(x) \pref N(x).&
  \end{align*}
  We ask if~$E_1$ and~$E_2$ have isomorphic candidate subelections
  that contain at least~$k$ candidates each.
  Intuitively, in a solution to the problem, for each vertex $x$ one has to remove
  either~$x$ or all vertices from~$M(x)$.
  It is a direct definition of a clique: Either~$x$ is not in
  a clique or all its nonneighbors are not in a clique.
  It is clear that the reduction can be computed in polynomial time
  and it remains to show its correctness.

  First, let us assume that~$G$ has a size-$k$ clique. Let~$K$ be the
  set of this clique's vertices. We form elections~$E'_1$ and~$E'_2$
  by restricting~$E_1$ and~$E_2$ to the candidates from~$K$.
  To verify that~$E'_1$ and~$E'_2$ are isomorphic via the given
  matchings, let us consider an arbitrary pair of matched 
  voters~$v^1_x$ and~$v^2_x$.
  If~$x$ is not included in~$K$ then preference orders of~$v^1_x$
  and~$v^2_x$ restricted to~$K$ are identical.
  Indeed, removing even only~$x$ from the set of candidates
  makes~$v^1_x$ and~$v^2_x$ identical.
  Otherwise, if~$x$ is in~$K$ then~$K \cap M(x) = \emptyset$ 
  as~$K$ is a clique. Therefore, removing~$M(x)$ from the set of
  candidates makes~$v^1_x$ and~$v^2_x$ identical.

  For the other direction, let us assume that there are subelections~$E'_1$ 
  and~$E'_2$ of~$E_1$ and~$E_2$, respectively, each with
  candidate set~$K$, such that~$|K| \geq k$ and~$E'_1$ and~$E'_2$ are
  isomorphic via the given matchings. It must be the case that the
  vertices from~$K$ form a clique because if~$K$ contained two
  vertices~$x$ and~$y$ that were not connected by an edge, then 
  votes~$v^1_x$ and~$v^2_x$ would not be identical.
  Indeed, we would have~$y \pref x$ in~$v^1_x$ 
  and~$x \pref y$ in~$v^2_x$, respectively, when restricted to candidates
  from~$K$. This completes the proof of~$\np$-hardness.
    To show~$\wone$-hardness, note that the required number of
  candidates in isomorphic candidate subelections is equal to the
  parameter~$k$ for which \textsc{Clique} is~$\wone$-hard.



    Let us now give a reduction in the opposite direction. Let~$E_1 = (C,V_1)$ and~$E_2 = (C,V_2)$ be our input elections and let~$k$ be the number of candidates in maximum isomorphic candidate subelections (since we are in the ``with candidate matching'' regime, we take the candidate sets to be equal).
  Let~$m = |C|$,~$n = |V_1| = |V_2|$ (since we cannot remove the voters).

  We create an instance~$(G,k)$ of \textsc{Clique} as follows.
  We define~$G$ as having vertices corresponding to candidates, i.e.,~$V(G) = C$.
  We construct the set of edges by starting from a complete graph and removing some of them as follows.
  For every two matched voters~$v$ and~$u$ and every two candidates~$x$ and~$y$ such that~$x \pref_v y$ and~$y \pref_u x$, we remove edge~$\{x,y\}$ from the graph.
  It is clear that the reduction can be computed in polynomial time and both parameters have the same value.
  It remains to show its correctness.

  First, let us assume that there are subelections~$E'_1$ and~$E'_2$ of~$E_1$ and~$E_2$, respectively, each with
  candidate set~$K$, such that~$|K| \geq k$ and~$E'_1$ and~$E'_2$ are
  isomorphic via the given matchings.
  It must be the case that the vertices from~$K$ form a clique.
  Indeed, if~$K$ contained two vertices~$x$ and~$y$ that were not connected by an edge,
  then edge~$\{x,y\}$ had to be removed by some two matched voters~$v$ and~$u$ such that~$x \pref_v y$ and~$y \pref_u x$.
  Since both voters belong to subelections~$E'_1$ and~$E'_2$, we obtain a contradiction that they are isomorphic via the given matchings.

  For the other direction, let us assume that~$G$ has a size-$k$ clique.
  Let~$K$ be the set of this clique's vertices.
  We form elections~$E'_1$ and~$E'_2$ by restricting~$E_1$ and~$E_2$ to the candidates from~$K$.
  To verify that~$E'_1$ and~$E'_2$ are isomorphic via the given matchings, let us consider an arbitrary pair of matched voters~$v$ and~$u$ and arbitrary pair of candidates~$x,y \in K$.
  It follows that~$x \pref_v y$ and~$x \pref_u y$.
  Otherwise edge~$\{x,y\}$ would have been removed during the reduction, hence~$K$ would not be a clique.
  A contradiction.
\end{proof}

All the remaining variants of \textsc{Max. Common Cand.-Subelection}
also are~$\np$-complete. The proofs follow either by applying
Proposition~\ref{prop:reduction} or by introducing candidates that
implement a required voter matching. In the latter case,
$\wone$-hardness does not follow from this reduction as we introduce
dummy candidates that have to be included in a solution, but their
number is not a function of the \textsc{Clique} parameter (clique
size).

\begin{repproposition}{Max. Common Cand.-Subelection problem with
    candidate matching}
    \label{prop:mccs-cm-vm}
    \textsc{Max. Common Cand.-Subelection} is~$\np$-complete and so are its
    variants with a given candidate matching and with a given voter matching.
\end{repproposition}

\begin{proof}
  Below we give the reductions for all the three cases, i.e., the case
  with a given candidate matching, with a given voter matching, and
  without any matchings.

  \paragraph{The case with a given candidate matching.}
  We give a reduction from \textsc{Max. Common Cand.-Subelection with
  both Matchings}. Let~$E_1 = (C,V)$ and~$E_2 = (C,U)$ be our input
  elections, where~$V = (v_1, \ldots, v_n)$ and~$U = (u_1, \ldots, u_n)$, and let~$t$ be the desired size of the
  isomorphic subelection (since we are in the setting with both
  matchings, we can assume that both elections are over the same
  candidate set). We assume that for each~$i \in [n]$ voter~$v_i$ is
  matched to~$u_i$. Let~$m = |C|$ and let~$k = t/n$. We note that~$k \leq m$.

  Our construction proceeds as follows. First, we form~$m+1$ sets,~$A$,~$D_1, \ldots, D_m$, each containing~$m+1$ new candidates. Let~$\calD = A \cup D_1 \cup \cdots \cup D_m$. Note that~$|\calD| = (m+1)^2$. We form elections~$E'_1 = (C \cup \calD, V')$
  and~$E'_2 = (C \cup \calD, U')$, where~$V' = (v'_1, \ldots, v'_n)$
  and~$U' = (u'_1, \ldots, u'_n)$. For each~$i \in [n]$, we set their
  preference orders as follows (by writing~$[v_i]$ or~$[u_i]$ we mean
  copying the preference order of the respective voter):
  \begin{align*}
    v'_i \colon& D_1 \pref \cdots \pref D_{i-1} \pref A \pref D_{i} \pref \cdots \pref D_{m} \pref [v_i],\\
    u'_i \colon& D_1 \pref \cdots \pref D_{i-1} \pref A \pref D_{i} \pref \cdots \pref D_{m} \pref [u_i].
  \end{align*}
  Finally, we set the desired size of the isomorphic subelections
  to be~$t' = n \cdot (k + (m+1)^2) = t + n(m+1)^2$.

  We claim that~$E'_1$ and~$E'_2$ have isomorphic candidate
  subelections of size~$t'$ for the given candidate matching if and
  only if~$E_1$ and~$E_2$ have isomorphic candidate subelections of
  size~$t$ for given candidate and voter matchings.

  Let us assume that~$E'_1$ and~$E'_2$ have the desired candidate
  subelections,~$E''_1$ and~$E''_2$. We claim that their isomorphism
  is witnessed by such a matching that for each~$i$ voter~$v'_i$ is
  matched to~$u'_i$. If it were not the case, then to maintain the
  isomorphism these subelections would have to lose at least~$m-1$
  candidates from~$\calD$ (e.g., the candidates from~$A$) and their
  sizes would be at most~$n(m+m(m+1)) = n((m+1)^2-1) < t'$. Thus the
  isomorphism of~$E''_1$ and~$E''_2$ is witnessed by the same voter
  matching as the one required by our input instance. A simple
  counting argument shows that after dropping candidates from~$\calD$
  from subelections~$E''_1$ and~$E''_2$, we obtain elections
  witnessing that~$(E_1,E_2)$ is a \emph{yes}-instance of
  \textsc{Max. Common Cand.-Subelection with both Matchings}. The
  reverse direction is immediate.

  \paragraph{The case with a given voter matching.}
  This case follows by Proposition~\ref{prop:reduction} and the fact that
  \textsc{Cand.-Subelection with Voter Matching} is~$\np$-complete.

  \paragraph{The case without any given matchings.}
  This case follows by Proposition~\ref{prop:reduction} and the fact
  that \textsc{Cand.-Subelection Isomorphism} problem is~$\np$-complete.
\end{proof}

Similarly to all four matching cases of
the \textsc{Max. Common Cand.-Subelection}, all four matching cases of the
\textsc{Max. Common Subelection} also are~$\np$-complete.

\begin{repproposition}{All 4 MCS}
\label{prop:mcs}
All four matching cases of \textsc{Max. Common Subelection} are~$\np$-complete.
\end{repproposition}

\begin{proof}
  For the case without any matchings and the case with the voter
  matching, we use Proposition~\ref{prop:reduction} to reduce from the
  corresponding variant of \textsc{Subelection Isomorphism}.
  For the variants that include the candidate matching (for which
  \textsc{Subelection Isomorphism} is in~$\p$), we reduce from the
  corresponding variants of \textsc{Max. Common Cand.-Subelection}.
  Let~$E_1 = (C,V_1)$ and~$E_2 = (C,V_2)$ be our input elections and
  let~$t$ be the desired size of their isomorphic candidate
  subelections (since we are in the ``with candidate matching''
  regime, we take the candidate sets to be equal). Without loss of
  generality, we can assume that~$|V_1| = |V_2|$; 
  our~$\np$-completeness proofs for \textsc{Max. Common Cand.-Subelection}
  give such instances.

  Let~$m = |C|$,~$n = |V_1| = |V_2|$, and let~$D$ be a set of~$(n-1)m$
  dummy candidates. We form elections~$E'_1$ and~$E'_2$ to be
  identical to~$E_1$ and~$E_2$, respectively, except that they also
  include the candidates from~$D$, who are always ranked on the
  bottom, in the same order.
  Therefore, the number of candidates in~$E'_1$ and~$E'_2$ equals~$nm$.
  We ask if~$E'_1$ and~$E'_2$ have
  isomorphic subelections of size~$t' = t + n(n-1)m$.

  If~$E_1$ and~$E_2$ have isomorphic candidate subelections of 
  size~$t$, then certainly~$E'_1$ and~$E'_2$ have isomorphic subelections
  of size~$t'$ (it suffices to take the same subelections as 
  for~$E_1$ and~$E_2$ and include the candidates from~$D$).

  On the other hand,
  if~$E'_1$ and~$E'_2$ have isomorphic subelections of size
~$t'$, then~$E_1$ and~$E_2$ have size-$t$ isomorphic candidate
  subelections. In fact, the subelections of~$E'_1$ and~$E'_2$ must
  include all the~$n$ voters. Otherwise their sizes would
  be at most~$(n-1)mn < t+(n-1)mn \leq t'$.
  Thus the subelections of~$E'_1$ and~$E'_2$ are candidate
  subelections. As we can also assume that the subelections of~$E'_1$ 
  and~$E'_2$ include all the candidates from~$D$, by omitting
  these candidates we get the desired candidate subelections of~$E_1$
  and~$E_2$.
\end{proof}

\vspace{0.2cm}
\begin{conclusionsbox}
\begin{itemize}
    \item All problems related to \textsc{Voter-Subelection} are in P. On the other hand, general Subelection problems and \textsc{Cand-Subelection} ones tend to be NP-hard. The \emph{isomorphic} variants, given the candidate matching, become significantly simpler (shifting to P). However, the most interesting (or surprising) is the fact that the \textsc{Max. Common Subelection} problem with both matchings remains NP-complete.
\end{itemize}
\end{conclusionsbox}

\section{Experiments}\label{ch:sub:sec:experiments}

Next we use the \textsc{Max. Common Voter-Subelection} problem to analyze
similarity between elections generated from various
statistical models. While \textsc{Max. Common Voter-Subelection} has
a polynomial-time algorithm, it is too slow for our purposes.
Thus we have expressed it as an integer linear program (ILP) and
we were solving it using the CPLEX ILP solver.
A formal ILP formulation is as follows.
\begin{enumerate}
\item For each pair of voters~$v \in V$ and~$u \in U$, we have a binary
  variable~$N_{v,u}$. If it is set to~$1$, then we interpret it as
  saying that voter~$v$ is included in the subelection of~$E$,
  voter~$u$ is included in the subelection of~$F$, and the two
  voters are matched. Value~$0$ means that the preceding statement
  does not hold.
\item For each pair of candidates~$c \in C$ and~$d \in D$, we have a
  binary variable~$M_{c,d}$. If it is set to~$1$ then we interpret it
  as saying that~$c$ is matched to~$d$ in isomorphic subelections
  (note that, since we are looking for voter subelections, every
  candidate from~$C$ has to be matched to some candidate from~$D$, and
  the other way round).
\end{enumerate}
To ensure that variables~$N_{v,u}$ and~$M_{c,d}$ describe the
respective matchings, we have the following basic constraints:
\begin{align*}
  &\textstyle \sum_{u \in U}N_{v,u} \leq 1, \;\; \forall v \in V,&
  &\textstyle \sum_{d \in D}M_{c,d} = 1, \;\; \forall c \in C, \\
  &\textstyle \sum_{v \in V}N_{v,u} \leq 1, \;\; \forall u \in U,&
  &\textstyle \sum_{c \in C}M_{c,d} = 1, \;\; \forall d \in D.
\end{align*}
For each pair of voters~$v \in V$,~$u \in U$ and each pair of
candidates~$c \in C$ and~\mbox{$d \in D$}, we introduce constant
$w_{v,u,c,d}$ which is set to~$1$ if~$v$ ranks~$c$ on the same
position as~$u$ ranks~$d$, and which is set to~$0$ otherwise. We use
these constants to ensure that the matchings specified by variables
$N_{v,u}$ and~$M_{c,d}$ indeed describe isomorphic
subelections. Specifically, we have the following constraints (let
$m = |C| = |D|$):
\begin{align*}
  \textstyle \sum_{c \in C}\sum_{d \in D} w_{v,u,c,d}\cdot M_{c,d} \geq m \cdot N_{v,u}, \quad \forall v \in V, u \in U.
\end{align*}
For each~$v \in V$ and~$u \in U$, they ensure that if~$v$ is matched
to~$u$ then each candidate~$c$ appears in~$v$ on the same position as
the candidate matched to~$c$ appears in~$u$.

We stress that we could have used other problems from the \textsc{Max. Common Subelection} family in this section.
We chose \textsc{Max. Common Voter-Subelection} because its outcomes are particularly easy to interpret, which is not always the case for \textsc{Max. Common Subelection}. For example, in \textsc{Max. Common Subelection} problem if the resulting value is~$k$, then we do not know if it is due to an election with one vote over~$k$ candidates or an election with~$k$ voters voting for a single candidate, or (if~$k$ is not a prime number) something in between.

Our findings are similar to those presented in the previous chapters, but our claims of similarity
between statistical cultures are stronger, whereas our dissimilarity
claims are weaker. Further, our results are most appealing for very small numbers of candidates, whereas 
in the preceding chapters we focused on larger candidates sets. 

\subsection{Results and Analysis}
We study the following nine models: IC, 1D-Interval, Conitzer model, Walsh model, urn (with $\alpha \in \{0.1, 0.5\}$), Norm-Mallows (with $\mathrm{norm}$-$\phi \in \{\nicefrac{1}{3}, \nicefrac{2}{3}\}$), and identity.
We consider elections with~$4$,~$6$,~$8$, and~$10$ candidates and with
$50$ voters.
For each scenario and
each two of the selected models, we have generated~$1000$ pairs
of elections. For each pair of
models, we recorded the average number of voters in the maximum common
voter subelections (normalized by fifty, i.e., the number of voters in
the original elections), as well as the standard deviation of this
value. 

We show our numerical results in Figure~\ref{fig:sub_experiment}
(each cell corresponds to a pair of models; the number in the top-left
corner is the average, and the one in the bottom-right corner is the
standard deviation). Note that the matrices in
Figure~\ref{fig:sub_experiment} are symmetric (the results for models~$A$ and
$B$ are the same as for models~$B$ and~$A$).

\begin{figure}[]
    \centering
    \begin{subfigure}[b]{0.49\textwidth}
        \centering
        \includegraphics[width=6cm]{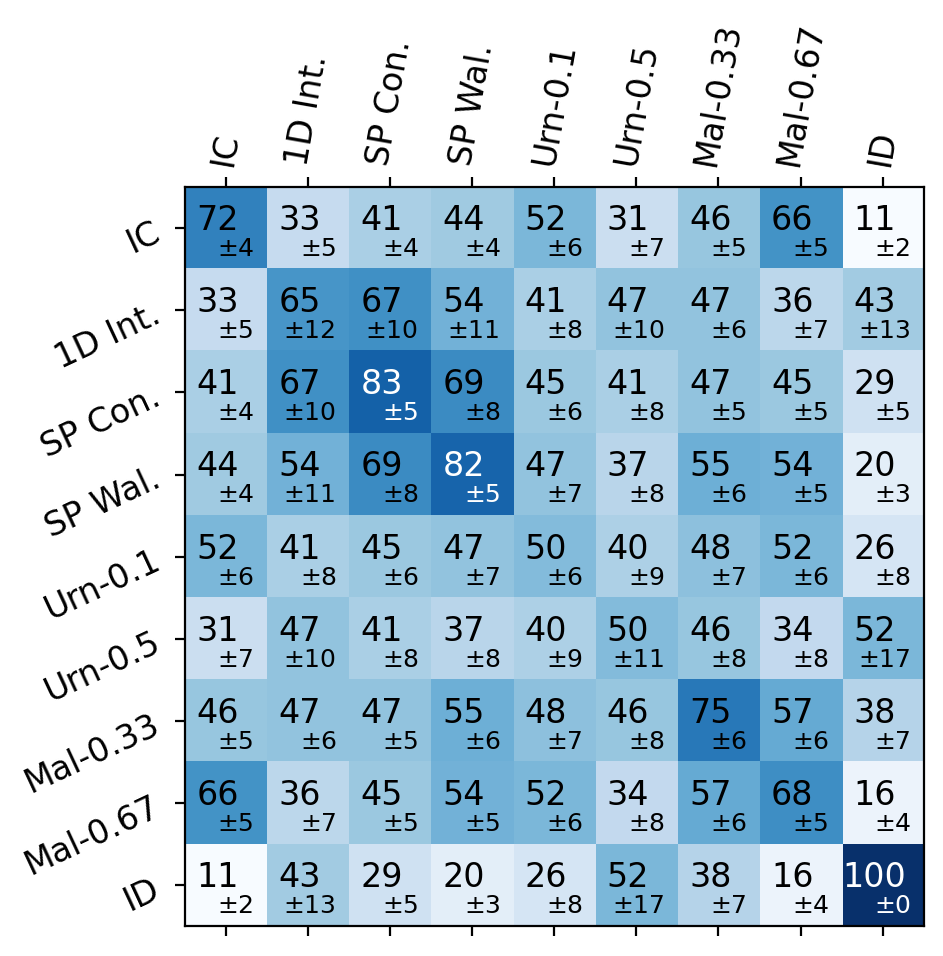}
        \caption{4 candidates \& 50 voters}
    \end{subfigure}
    \begin{subfigure}[b]{0.49\textwidth}
        \centering
        \includegraphics[width=6cm]{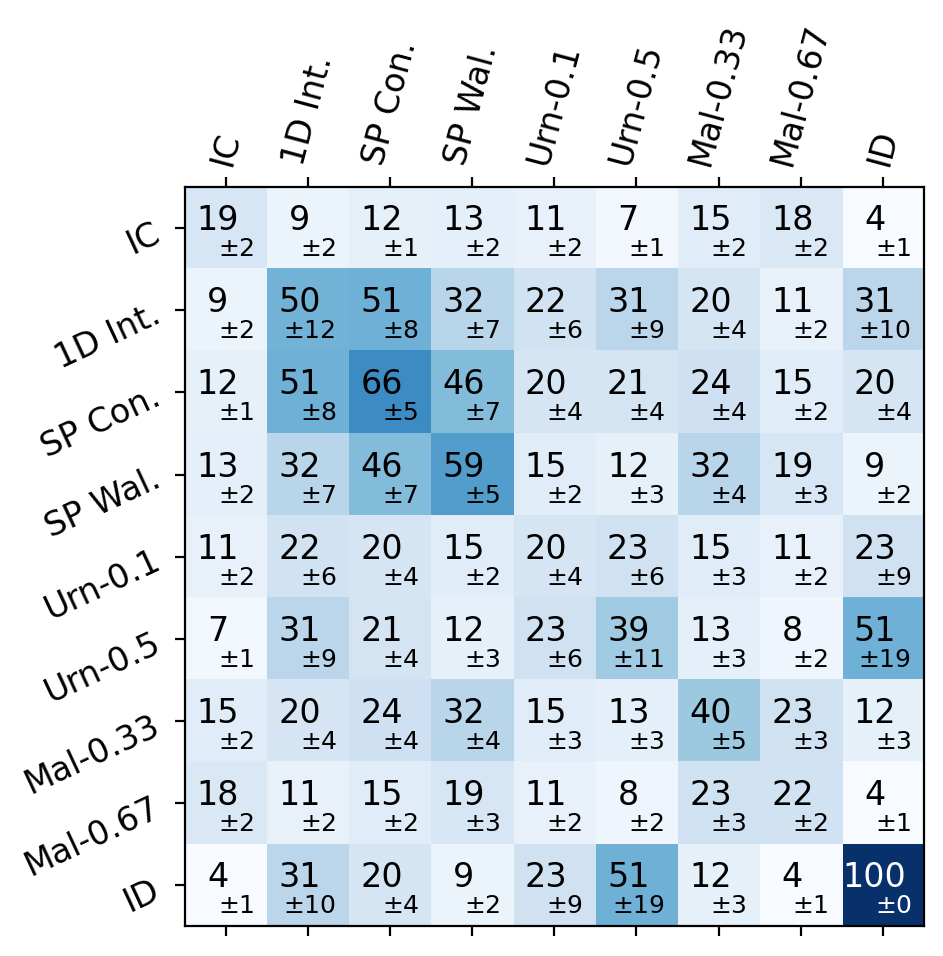}
        \caption{6 candidates \& 50 voters}
    \end{subfigure}
    
    \vspace{1em}
    
    \begin{subfigure}[b]{0.49\textwidth}
        \centering
        \includegraphics[width=6cm]{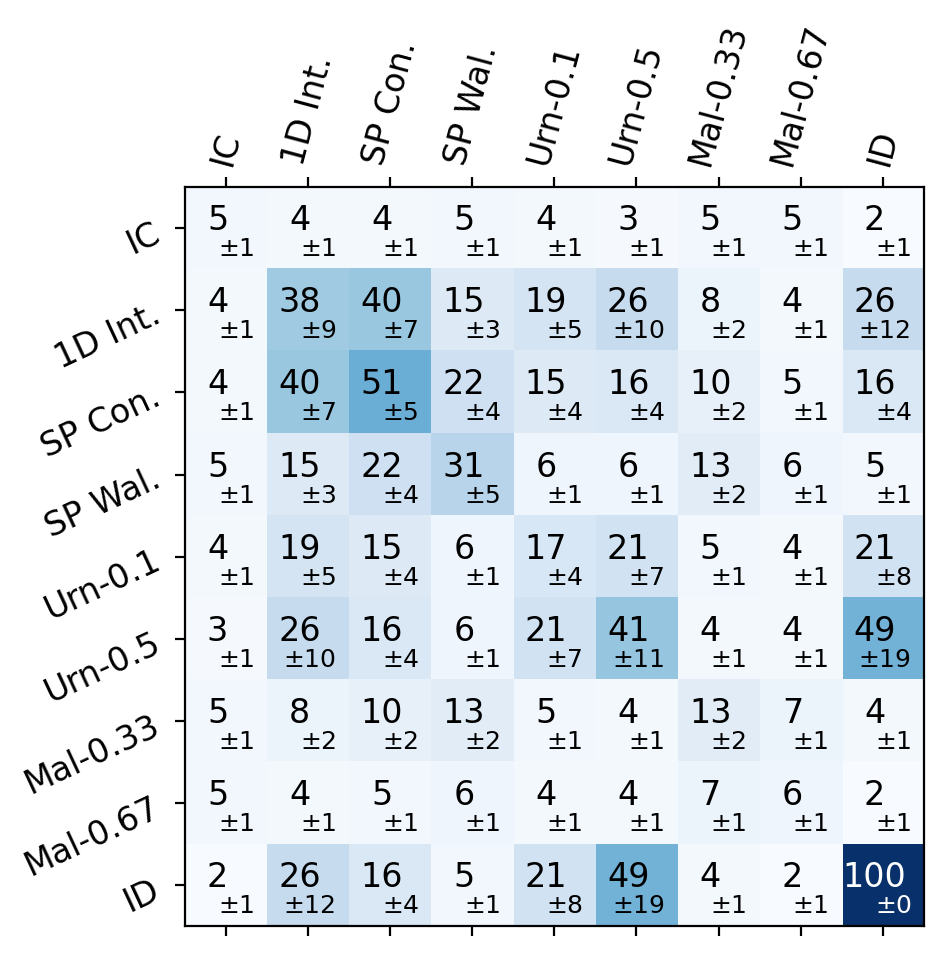}
        \caption{8 candidates \& 50 voters}
    \end{subfigure}
    \begin{subfigure}[b]{0.49\textwidth}
        \centering
        \includegraphics[width=6cm]{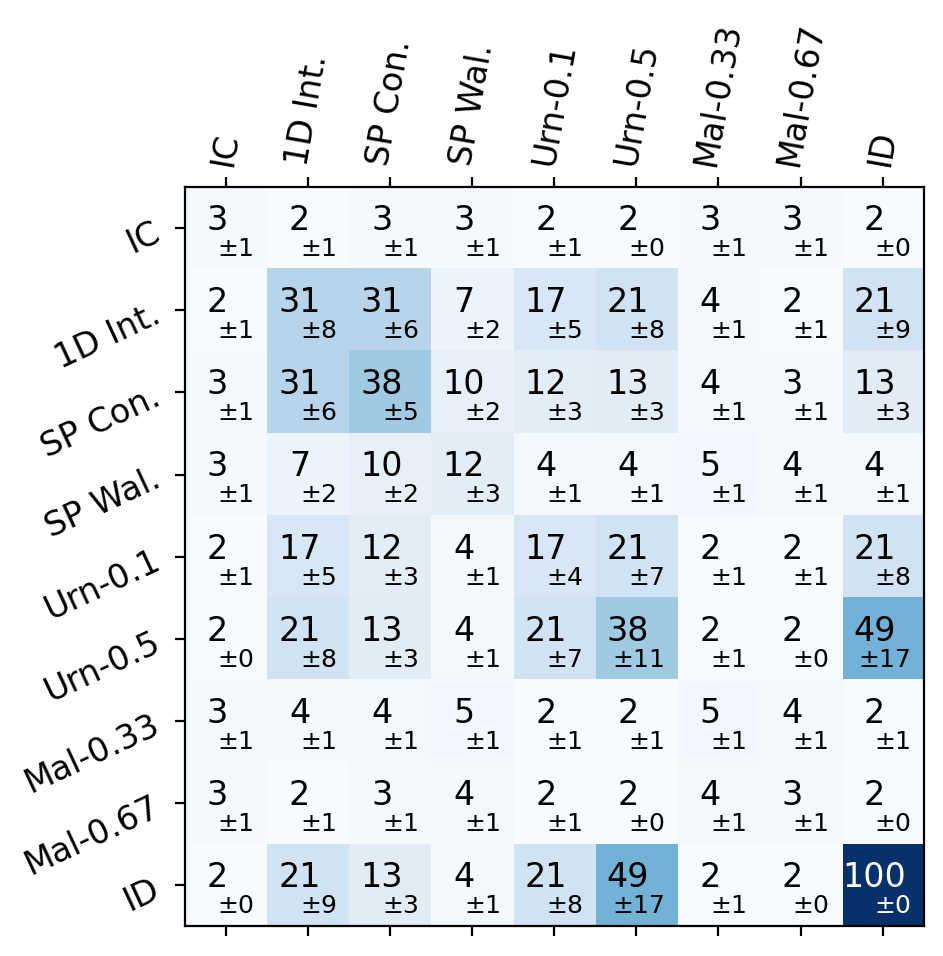}
        \caption{10 candidates \& 50 voters}
    \end{subfigure}

    \caption{\label{fig:sub_experiment} The numbers typeset in large font denote the rounded \% of matched votes for \textsc{Max. Common Voter-Subelection}. The numbers typeset in small font denote the rounded standard deviation. There are results for elections with~$4$,~$6$,~$8$, and~$10$ candidates and~$50$ voters.}
\end{figure}

For the case with four candidates,
we see that the level of similarity between elections from
various models is quite high and drops sharply as the number of candidates increases.
This shows that for experiments with very few candidates it is not as
relevant to consider very different election models, but for
more candidates using diverse models is justified.

Despite the above, some models remain
similar even for~$6$,~$8$, and sometimes even~$10$ candidates. This is particularly
visible for the case of single-peaked elections. The 1D-Interval model
remains very similar to the Conitzer model, and the Walsh model is
quite similar to these two for up to~$6$ candidates, but for~$8$ and~$10$
candidates it starts to stand out.

We also note that the urn models remain relatively similar to each
other (and to the 1D-Interval and Conitzer models) for all numbers of
candidates, but this is not the case for the Norm-Mallows models. One
explanation for this is that the urn model proceeds by copying some of
the votes already present in the election, whereas the Norm-Mallows model
generates votes by perturbing the central one. The former leads to
more identical votes in an election. Indeed, to verify this, it
suffices to consider the ``ID'' column (or row) of the matrix: The
similarity to the identity elections simply shows how often the most
frequent vote appears in elections from a given model. For 10
candidates, urn elections with~$\alpha \in \{0.1, 0.5\}$ have, on
average,~$21\pm8\%$ and~$49\pm17\%$ identical votes, respectively. For
Norm-Mallows elections, this value drops to around~$2\%$ (in our
setting, this means 1 or 2 voters, on average).

Finally, we consider the diagonals of the matrices in
Figure~\ref{fig:sub_experiment}, which show the self-similarity of our models.
Intuitively, the larger these values, the fewer elections of a
given type one needs in an experiment.
Single-peaked elections stand out here for all numbers of candidates,
whereas urn models become more prominent for larger candidate sets.

\begin{figure}[]
    \centering
    \includegraphics[width=9cm]{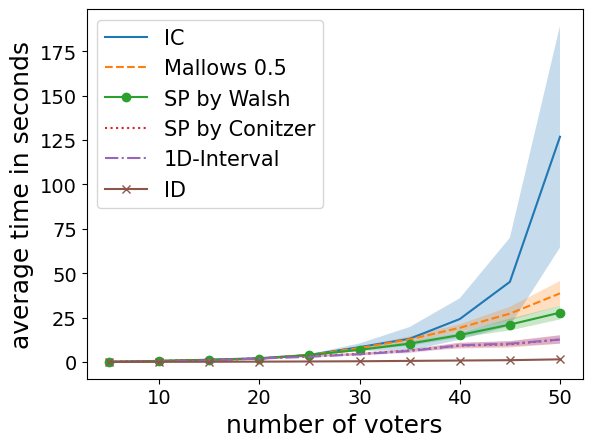}
    \includegraphics[width=9cm]{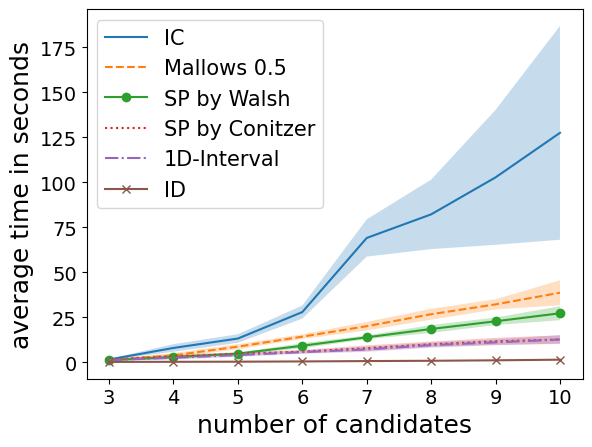}
    \caption{\label{fig:sub_experiment_time} Average time needed to find the maximum common voter subelections with the fixed number of candidates (upper), and fixed number of voters (lower). The shaded parts depict the standard deviation.}
\end{figure}

We have also analyzed the average running time that CPLEX needed to
find the maximum common voter subelections.
We focus on IC, identity, Walsh model,
Conitzer model, Norm-Mallows model with norm-$\phi=0.5$, and
1D-Interval. First, we generate~$1000$ pairs of elections from each
model with~$10$ candidates and~$5,10,\dots,45,50$ voters (in each pair both elections are from the same model),  and
calculate the average time needed to find the maximum common voter
subelections. Second, we fixed the number of voters to~$50$ and
generated elections with~$3,4,\dots,9,10$ candidates, and, like before,
calculate the average time needed to find the maximum common voter
subelections.

The results are
presented in Figure~\ref{fig:sub_experiment_time}. As we increase the
number of voters, the time seems to increase exponentially. We observe
large differences between the models, with the IC being by far the
slowest. Conitzer model and Walsh model are significantly different
from each other, even though both generate single-peaked
elections. Moreover, the fact that the 1D-Interval and Conitzer models
need on average the same amount of time confirms their similarity.

\vspace{0.2cm}
\begin{conclusionbox}
    Most elections with few candidates are very similar to each other -- this explains why the maps with few candidates are not that informative.
\end{conclusionbox}

\subsection{Real-life Subelections}\label{real-life-sub-exp}

We also conducted analogous experiments, but instead of using statistical cultures, we used real-life data.
We used the same~$11$ models as in \Cref{preflib_info} and we also added impartial culture as a reference point.
We selected one election from each model, and then treating each election as a \emph{distribution},\footnote{\new{To treat election $E$ as a distribution means that each vote is sampled with probability $\frac{a}{n}$, where $a$ is the number of copies of a given vote, and $n$ is the number of all votes in election $E$.}} we sampled~$10$ instances from it, in total having~$120$ elections.

We consider elections with~$4$,~$6$,~$8$, and~$10$ candidates, and~$50$ voters. 
In \Cref{tab:sub_preflib_params} we present the total number of votes and the number of distinct votes for each distribution that we use. Note that for several models, the number of votes is smaller than~$50$, and the smaller is the number of votes, the larger is the probability that some votes will be selected multiple times.

\begin{table}[]
\centering
\small
\begin{tabular}{ c  c  c  c }
			\toprule
			Category & Name & \# Votes & \# Distinct Votes \\	
			\midrule
			Political & Irish & 43942 & 29908 \\
			Political & Glasgow	& 10376 & 5790 \\
            Political & Aspen & 2459 & 2018 \\
			Political & ERS	& 380 & 336 \\
            \midrule
			Sport & Figure Skating & 9 & 9 \\
			Sport & Speed Skating & 12 & 12 \\
            Sport & TDF & 15 & 15 \\	
            Sport & GDI & 17 & 17 \\	
            \midrule
            Survey & T-Shirt & 30 & 30 \\	
            Survey & Sushi & 5000 & 4926 \\	
            Survey & Cities & 392 & 392 \\			
			\bottomrule
	\end{tabular}
	\caption{\label{tab:sub_preflib_params} Number of votes in the real-life elections used as distributions for sampling.}
\end{table}

In \Cref{fig:sub_preflib_experiment} we show the results. For the experiment with only four candidates (upper left matrix), we observe that there is a correlation between the similarity with the impartial culture elections and their position in \Cref{fig:main_preflib_map}. The same is true (but on a smaller scale) for the cases of~$6$,~$8$ and~$10$ candidates. Without surprise, we observe that the smaller is the number of distinct votes in a given distribution, the more similar are the elections from that distribution to each other (values displayed on the diameter). For elections with~$8$ and~$10$ candidates, all values, except those on the diameter, are very small. The only part of the matrix with slightly larger values is the lower right corner, but the similarity between these models is due to a smaller number of different votes in the distribution from which these elections were sampled. 
  
Interestingly, for elections with only four candidates, sport elections are less similar to each other than the rest of elections, even though they have fewer distinct votes. The exception is the similarity between Tour de France (TDF) and Giro d'Italia (GDI) -- two cycling competitions, which (among sport elections) seem to be very similar. Moreover, note that for GDI and TDF there is a very large difference between the case for~$4$ and~$6$ candidates, when compared with the differences for other pairs of sport instances. 

Regarding impartial culture, political elections, and surveys with four candidates, when comparing any two elections, usually we can match two thirds of the votes, which is around 33 votes out of 50. This number seems to be quite large and implies a high level of chaos in these instances.

Another thing worth pointing out is the relative similarity between two political elections: Irish and Glasgow. For elections with four candidates, Irish elections are the most similar to Glasgow ones; however, the opposite is not true. But for elections with six candidates, both Glasgow and Irish are each other's closest instances and more similar to each other than any other pair (except for the Speed Skating and Figure Skating).


\begin{figure}[]
    \centering
    \begin{subfigure}[b]{0.49\textwidth}
        \centering
        \includegraphics[width=6.6cm]{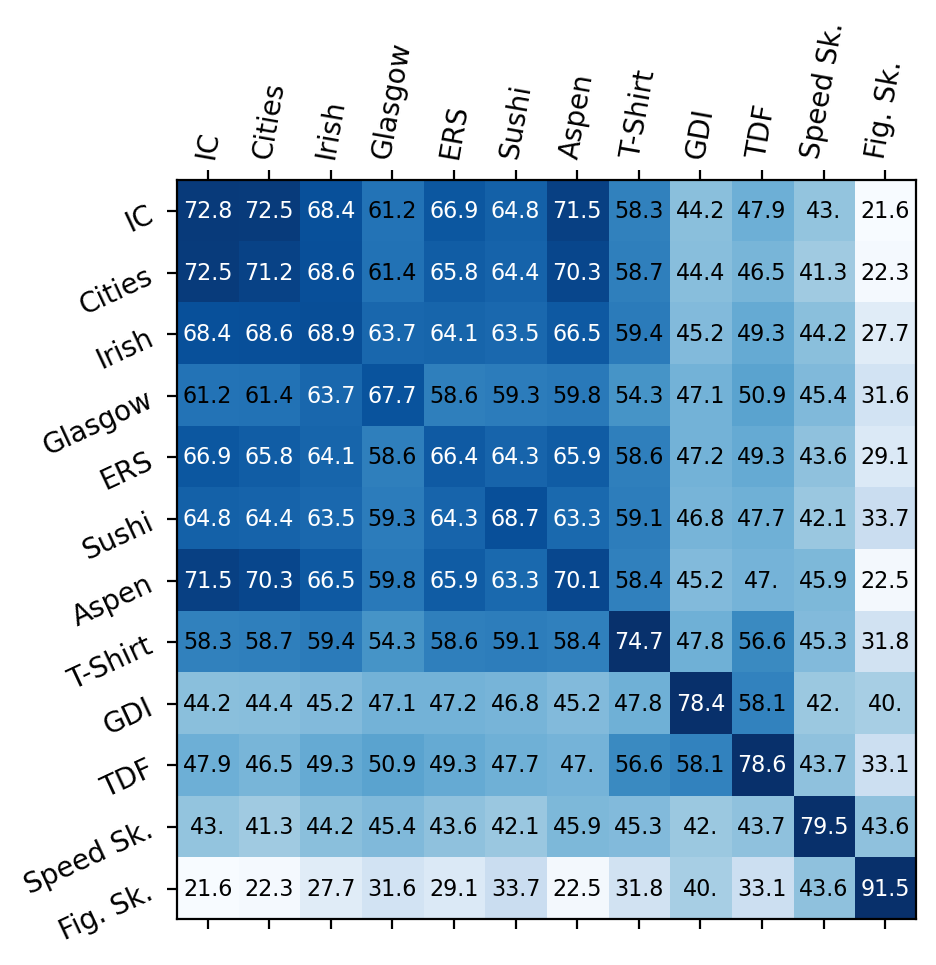}
        \caption{4 candidates \& 50 voters}
    \end{subfigure}
    \begin{subfigure}[b]{0.49\textwidth}
        \centering
        \includegraphics[width=6.6cm]{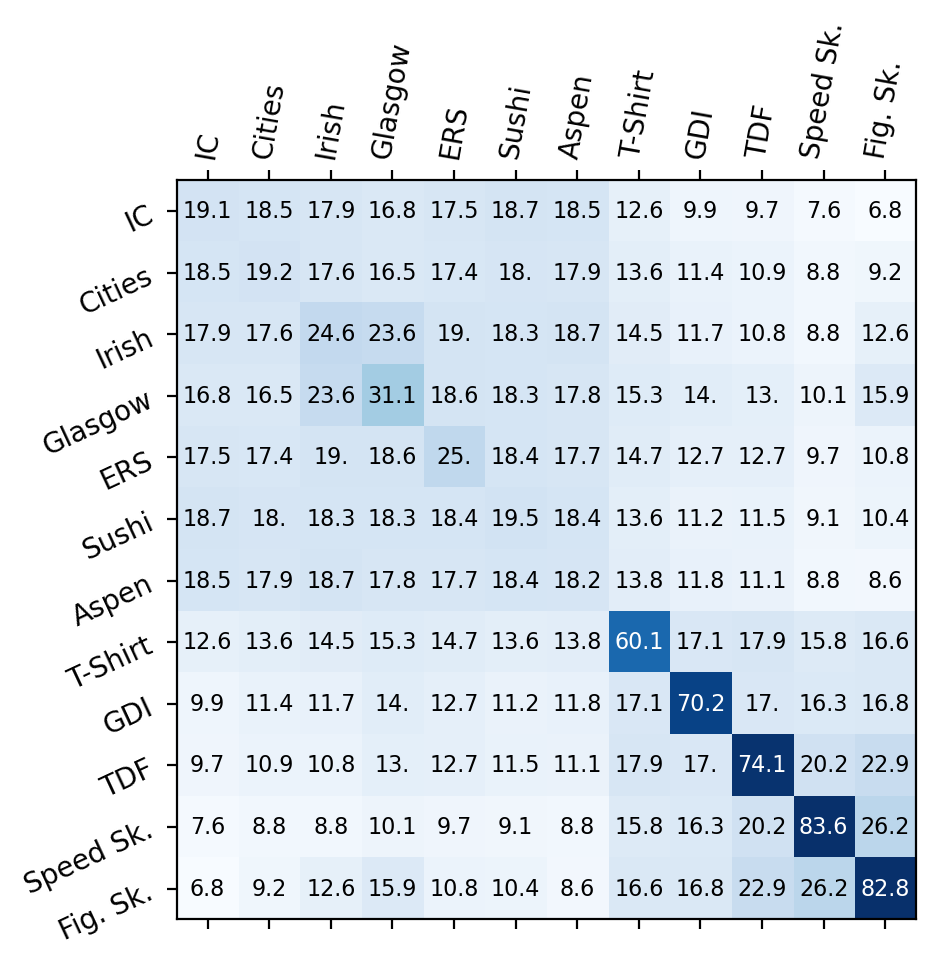}
        \caption{6 candidates \& 50 voters}
    \end{subfigure}
    
    \vspace{1em}
    
    \begin{subfigure}[b]{0.49\textwidth}
        \centering
        \includegraphics[width=6.6cm]{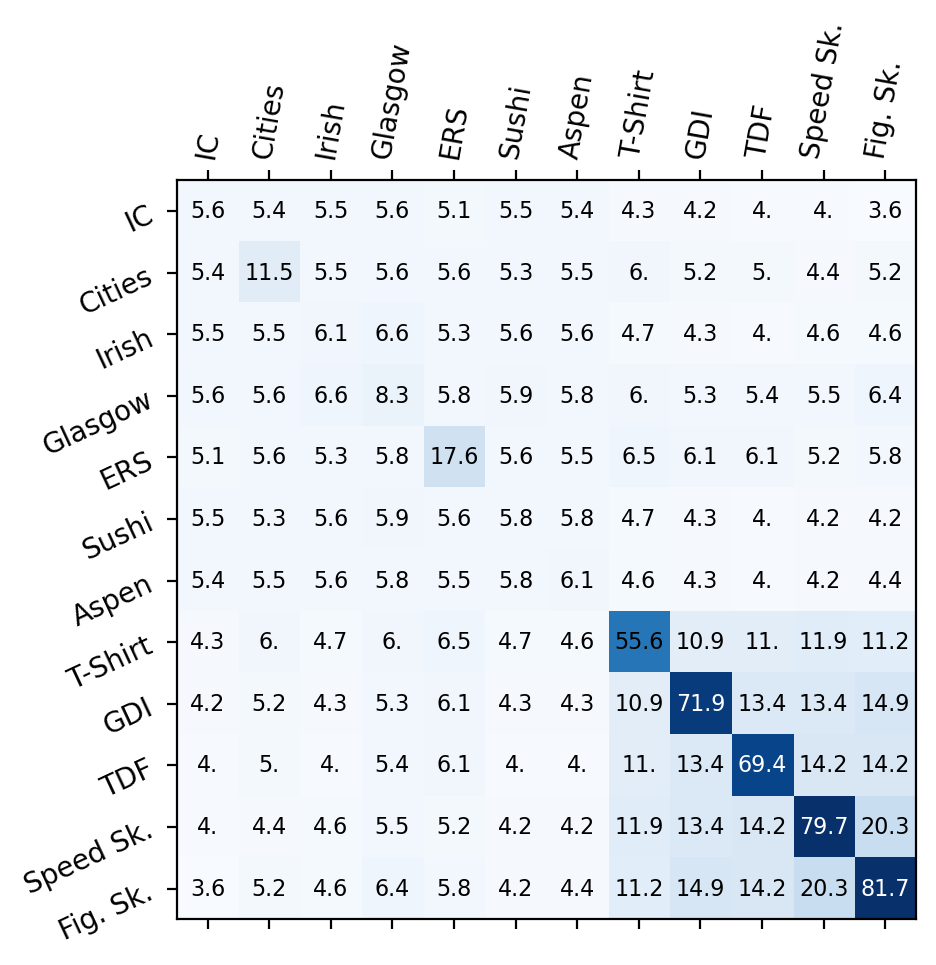}
        \caption{8 candidates \& 50 voters}
    \end{subfigure}
    \begin{subfigure}[b]{0.49\textwidth}
        \centering
        \includegraphics[width=6.6cm]{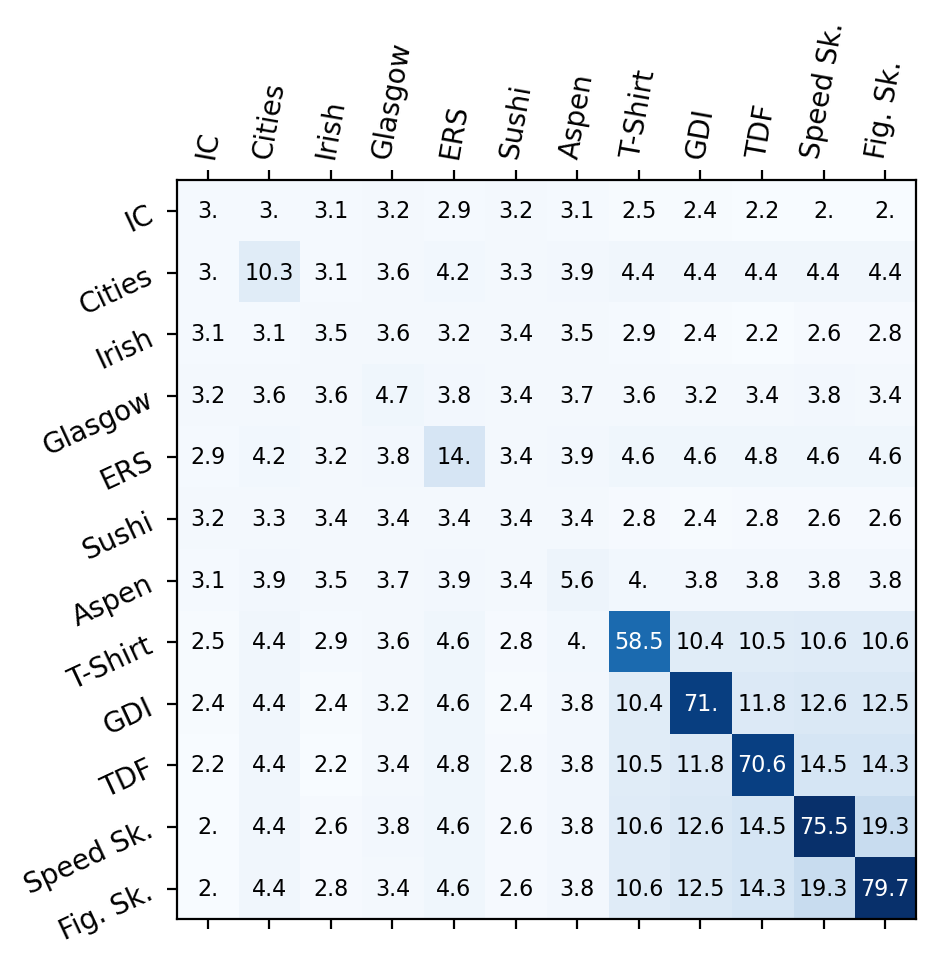}
        \caption{10 candidates \& 50 voters}
    \end{subfigure}

    \caption{\label{fig:sub_preflib_experiment} The numbers denote the rounded \% of matched votes for \textsc{Max. Common Voter-Subelection}. There are results for elections with~$4$,~$6$,~$8$, and~$10$ candidates and~$50$ voters.}
\end{figure}


\vspace{0.2cm}
\begin{conclusionbox}
The experimental results based on real-life elections from Preflib confirm our previous observations (drawn from the map of real-life elections). For example, we see that sport elections are similar to each other and quite similar to identity. Also, as was the case for statistical cultures, analyzing elections with very few candidates is not particularly meaningful because most such elections are very similar to each other.
\end{conclusionbox}

\newpage
\section{Summary}
We have shown that variants of \textsc{Election Isomorphism} that are
based on considering subelections are largely intractable but,
nevertheless, some of them can be solved in polynomial-time.
In fact, we have used the polynomial-time solvable
\textsc{Max. Common Voter-Subelection} problem to analyze the similarity
between various different models of generating random elections.

In Section~\ref{sec:comp-complexity} we classified variants of the problem as either
belonging to~$\p$ or being~$\np$-complete (and some being~$\wone$-hard).

Finally, in \Cref{ch:sub:sec:experiments} we presented some more experimental results based on synthetic and real-life data, showing that computing the \textsc{Max. Common Voter-Subelection} can serve as a measure of similarity between elections. For example, it helped noticing the difference between Walsh and Conitzer model, with elections from Conitzer model being more similar to each other, than those from Walsh model.
Experimental results for four candidates confirms observations from~\Cref{ch:applications}, i.e., maps for elections with small number of candidates are a bit chaotic because most of the elections are quite similar to each other.

\vspace{0.2cm}
\begin{contributionbox}
The main contribution of this chapter is theoretical analysis of a family of subelection problems. Moreover, we provide experimental results both on synthetic and real-life data, showing that majority of small elections are very similar to each other. 
\end{contributionbox}

\chapter{Approval Elections}
\label{ch:approval}

\section{Introduction}
So far we were focusing only on ordinal elections, where each voter ranks all the candidates from the most to the least appreciated one. In this chapter we consider approval elections~\citep{bra-fis:b:approval-voting}.
In an approval election, each voter
indicates which candidates he or she finds acceptable for a certain
task (e.g., to be a president, to join the parliament, or to enter the
final round of a competition), and a voting rule is used to aggregate these preferences and determine
the winner or the winning committee.
In the single-winner setting (e.g., when choosing the president), the
most popular rule is to pick the candidate with the highest number of
approvals.  In the multiwinner setting (e.g., in parliamentary
elections or when choosing finalists in a competition), there is a rich
spectrum of rules to select from, each with different properties and
advantages (see, e.g.,  the overview of~\cite{lackner2023approval}.
Approval voting is particularly attractive due to its simplicity and
low cognitive load imposed on the voters.
%
%
In fact,
its practical applicability 
has already been tested in several field experiments, including
those in
France~\citep{las-str:j:approval-experiment,bau-ige:b:french-approval-voting,baujard2014s,bou-bla-bau-dur-ige-lan-lar-las-leb-mer:t:french-approval-voting-2017}
and Germany~\citep{alo-gra:j:german-approval-voting}.  Over the recent
years, there was also a tremendous progress regarding its theoretical
properties (see, e.g., the overview of
\cite{las-san:chapter:approval-multiwinner}.


In spite of all these achievements, numerical experiments regarding
approval voting are still challenging to design. One of the main
difficulties is caused by the lack of consensus about which statistical
cultures to use. \new{To answer this problem, in particular, we introduced various new statistical cultures. Moreoever, we evaluated them using experiments -- showing their usefullness, and at the same time showing drawbacks of previously used models.}
Below we list a few cultures that were recently used:


\begin{enumerate}
\item In the impartial culture setting, we assume that each vote is
  equally likely. Taken literally, this means that each voter approves
  each candidate with
  probability~$\nicefrac{1}{2}$ \citep{bar-lan-yok:c:hamming-approval-manipulation}.
  As this is quite unrealistic, several authors treat the approval
  probability as a
  parameter~\citep{bre-fal-kac-nie2019:experimental_ejr,fal-sli-tal:c:vnw}
  or require that all voters approve the same (small) number of
  candidates~\citep{lac-sko:j:av-vs-cc}. A further refinement is to
  choose an individual approval probability for each
  candidate~\citep{lac-mal:c:vnw-shortlisting}.

\item In Euclidean models, each candidate and voter is a point 
 in~$\mathbb{R}^d$, where~$d$ is a parameter, and a voter approves a
  candidate if they are sufficiently near. Such models are used, e.g.,
  by \cite{bre-fal-kac-nie2019:experimental_ejr} and
  \cite{god-bat-sko-fal:c:2d}. Naturally, the distribution of the
  candidate and voter points strongly affects the outcomes.

\item Some authors consider statistical cultures designed for the
  ordinal setting (where the voters rank the candidates from the most
  to the least desirable one) and let the voters approve some
  top-ranked candidates (e.g., a fixed number of them). This approach
  is taken, e.g., by \cite{lac-sko:j:av-vs-cc} on top of the ordinal
  Mallows model (later on, \cite{allouche2022truth} and \cite{caragiannis2022evaluating} provided
  approval-based analogues of the Mallows model).
\end{enumerate}
Furthermore, even if two papers use the same model, they often choose its
parameters differently. Since it is not clear how the parameters
affect the models, comparing the results from different papers is not
easy.

Our goal is to initiate a systematic study of approval-based
statistical cultures and to attempt to rectify at least some of the above
issues. We do so by applying our map of elections framework.
 
To create a map for approval elections,
we start by identifying two metrics between approval elections, 
the {\em isomorphic Hamming distance} and the
{\em approvalwise distance}. 
The first one is accurate, but difficult to compute,
whereas the second one is less precise, but easily computable.
Fortunately, 
in our election datasets
the two
metrics are strongly correlated; thus, we use mostly the latter one.

Next, we analyze the space of approval elections with a given number
of candidates and voters.  For each~$p \in [0,1]$, by~$p$-identity
($p$-ID) elections we mean those
where all the votes are identical and approve the same~$p$-fraction of
candidates.
By~$p$-impartial culture ($p$-IC) elections we mean those where each voter chooses to approve each
candidate with probability~$p$.  We view~$p$-ID and
$p$-IC elections as two extremes on the spectrum of agreement between
the voters and, intuitively, we expect that every election (where each
voter approves on average a~$p$ fraction of candidates) is located
somewhere between these two.  In particular, for~$p, \phi \in [0,1]$,
we introduce the~$(p,\phi)$-resampling model, which generates
elections whose expected approvalwise distance from~$p$-ID is exactly
the~$\phi$ fraction of the distance between~$p$-ID and~$p$-IC (and the
expected distance from~$p$-IC is the~$1-\phi$ fraction). 

Armed with these tools, we proceed to draw maps of
elections. First, we consider~$p$-ID,~$p$-IC, and
$(p,\phi)$-resampling elections, where the~$p$ and~$\phi$ values are
chosen to form a grid, and compute the approvalwise distances between
them. 
We find that, for a fixed value
of~$p$, the~$(p,\phi)$-resampling elections indeed form lines between
the~$p$-ID and~$p$-IC ones, whereas for fixed~$\phi$ values they form
lines between~$0$-ID and~$1$-ID ones (which we refer to as the
\emph{empty} and \emph{full} elections).  We obtain more maps by
adding elections generated according to other statistical
cultures; the presence of the \emph{$(p,\phi)$-resampling grid} helps
in understanding the locations of these new elections.
For each of our elections we compute several parameters, such as, e.g,
the highest number of approvals that a candidate receives, the time
required to compute the results of a certain multiwinner voting rule,
or the cohesiveness level (see Section~\ref{sec:prelims} for a
definition). For each of the statistical cultures, we present maps
where we color the elections according to these values. This gives
further insight into the nature of the elections they generate.
Finally, we compare the results for randomly generated elections with
those appearing in real-life, in the context of participatory
budgeting.

We also provide maps of approval preferences (similar to those for ordinal preferences shown in~\Cref{ordinal_map_pref}). We present maps from both voters' and candidates' perspectives (in the former ones each point depicts a voter, while in the latter ones each point depicts a candidate). To create these maps we use the Hamming distance and the Jaccard distance (which is a normalized variant of the Hamming distance that is putting more emphasis on approvals than on disapprovals).

The structure of this chapter is different from that of the previous ones. Since we move from the ordinal to the approval world of elections, we need \textit{new} preliminaries where we define several things, such as, for instance, a \textit{vote} or an \textit{election}. In a way, within this chapter we repeat the work that for ordinal elections was divided into several parts. That is, this chapter is an application of all the contributions from the previous ones and shows how the map framework can be applied to new types of objects (see also our work on the maps of stable roommates instances \citep{boehmer2023map})

\section{Preliminaries}\label{sec:prelims}

    
\paragraph{Elections.}

A (simple) approval election~$E=(C,V)$ consists of a set of candidates
$C=\{c_1,\dots,c_m\}$ and a collection of voters
$V=(v_1,\dots,v_n)$. Each voter~$v \in V$ casts an approval ballot,
i.e., he or she selects a subset of candidates that he or she
approves.  Given a voter~$v$, we denote this subset by
$A(v)$. Occasionally, we refer to the voters or their approval
ballots as votes; the exact meaning will always be clear from the
context.
%
%
%
An approval-based committee election (an ABC election) is a triple
$(C,V,k)$, where~$(C,V)$ is a simple approval election and~$k$ is the
size of the desired committee.  We use simple elections when the goal
is to choose a single individual and ABC elections when we seek a
committee.

%
%

Given an approval election~$E$ (be it a simple election or an ABC
one) and a candidate~$c$, we write~$\score_\av(c)$ to denote the
number of voters that approve~$c$. We refer to this value as the
\emph{approval score} of~$c$. The single-winner approval rule (called AV)
returns the candidate with the highest approval score (or the set of
such candidates, in case of a tie).

\paragraph{Distances Between Votes.}

For two voters~$v$ and~$u$, their Hamming distance is \linebreak
$\ham(v, u) = |A(v) \triangle A(u)| = |A(v) \setminus A(u)| + |A(u) \setminus A(v)|$, i.e.,  the number of
candidates approved by exactly one of them. Other distances include,
e.g., the Jaccard one, defined as
$\mathrm{jac}(v, u) = \frac{\ham(v,u)}{|A(v) \cup A(u)|}$.
For other examples of such distances, 
we point to the work
of \cite{caragiannis2022evaluating}.



\paragraph{Approval-Based Committee Voting Rules.}

An \emph{approval-based committee voting rule} (an ABC rule) is a function that maps an ABC election~$(C,V, k)$ to a nonempty set of committees of size~$k$.
If an ABC rule returns more than one committee, then we consider them tied.
%

We introduce two prominent ABC rules. \emph{Multiwinner Approval  Voting (AV)} selects the~$k$ candidates with the highest approval scores.
Given a committee~$W$, its approval score is the sum of the scores of its members;~$\score_\av(W) = \sum_{w \in W}\score_\av(w)$.
If there is more than one committee that achieves a maximum score, AV returns all tied committees.
The second rule is \emph{Proportional Approval Voting (PAV)}. PAV outputs all committees with the maximum PAV-score:
\[ \textstyle
  \score_{\text{pav}}(W)=\sum_{v\in V} h(|A(v)\cap W|),
\]
where~$h(x)=\sum_{j=1}^x \nicefrac{1}{j}$ is the harmonic function.
Intuitively, AV selects committees that contain the ``best'' candidates (in the sense of having the most approvals) and PAV selects committees that are in a strong sense proportional \citep{justifiedRepresentation,brill2018multiwinner}.
In contrast to AV, which is polynomial-time computable, PAV is NP-hard to compute \citep{azi-gas-gud-mac-mat-wal:c:approval-multiwinner,sko-fal-lan:j:collective}.
In practice, PAV can be computed  by solving an integer linear program \citep{pet-lac:j:spoc} or by an approximation algorithm~\citep{DudyczMMS20-tight-pav-apx}.

\paragraph{Cohesive Groups.}

Intuitively, a proportional committee should represent all groups of voters in a way that (roughly) corresponds to their size.
To speak of proportional committees in ABC elections, \cite{justifiedRepresentation} introduced the concept of \emph{cohesive groups}.
\begin{definition}
Consider an ABC election~$(C,V,k)$ with~$n$ voters and some non-negative integer~$\ell$.
A group of voters~$V'\subseteq V$ is \emph{$\ell$-cohesive} if
(i)~$|V'| \geq \ell\cdot \frac{n}{k}$
and
(ii)~$\left|\bigcap_{v \in V'} A(v) \right| \geq  \ell$.
\end{definition}
An~$\ell$-cohesive group is large enough to deserve~$\ell$ representatives in the committee and is cohesive in the sense that there are~$\ell$ candidates that can represent it.
A number of proportionality notions have been proposed based on cohesive groups,
such as 
(extended) justified representation~\citep{justifiedRepresentation}, proportional justified representation~\citep{Sanchez-Fernandez2017Proportional}, proportionality degree~\citep{sko:c:prop-degree}, and others.
For our purposes, it is sufficient to note that all these concepts guarantee cohesive groups different types and levels of representations
(see also the  survey of \cite{lackner2023approval} for a comprehensive overview).


\section{Statistical Cultures for Approval Elections}
\label{sec:statistical-cultures}

In the following, we present several statistical cultures (probabilistic models) for
generating approval elections.  Our input consists of the desired
number of voters~$n$ and a set of candidates~$C=\{c_1,\dots,c_m\}$.
For models that already exist in the literature, we provide examples of
papers that use them.


\paragraph{Resampling, IC, and ID Models.}
Let~$p$ and~$\phi$ be two numbers in~$[0,1]$. In the
\emph{$(p,\phi)$-resampling} model, we first draw a central
ballot~$u$, by choosing~$\lfloor p \cdot m\rfloor$ approved candidates
uniformly at random. Then, we generate each new vote~$v$ by initially
setting~$A(v) = A(u)$ and executing the following procedure for every
candidate~$c_i \in C$: With probability~$1-\phi$, we leave~$c_i$'s
approval intact and with probability~$\phi$ we resample its value
(i.e., we let~$c_i$ be approved with probability~$p$). The resampling
model is our contribution and is one of our basic tools for analyzing
approval elections.  By fixing~$\phi = 1$, we get the
\emph{$p$-impartial culture} model ($p$-IC) where each
candidate in each vote is approved with probability~$p$; it was used,
e.g., by \cite{bre-fal-kac-nie2019:experimental_ejr} and
\cite{fal-sli-tal:c:vnw}.  By fixing~$\phi = 0$, we ensure that all
votes in an election are identical (i.e., approve the same~$p$ fraction of the candidates). We refer to this model as \emph{$p$-identity} ($p$-ID).

\paragraph{Moving Model.} The \emph{$(p,\phi)$-moving} model is a variant
of the~$(p,\phi)$-resampling one, where each time a new vote is
generated, the new vote replaces the central one. Occasionally, we also consider~$(p,\phi, g)$-moving model where we add one more parameter~$g$, which denotes the number of groups. It works as follows. After each~$\lfloor\frac{n}{g}\rfloor$ votes are generated, we set the central vote back to the original central ballot instead of setting it to the last vote. Note that if the value of $g$ is equal to the number of voters then $(p,\phi, g)$-moving model and $(p,\phi)$-resampling model are equivalent, because after sampling each vote we are setting central ballot back to the original one.

\paragraph{Disjoint Model.} The \emph{$(p,\phi, g)$-disjoint} model,
where~$p$ and~$\phi$ are numbers in~$[0,1]$ and~$g$ is a non-negative
integer, works as follows: We draw a random partition of~$C$ into $\lfloor p\cdot m \rfloor$-sized~$g$
sets,~$C_1, \ldots, C_g$ (note that, if $p\cdot g < 1$ then some candidates will not be members of any group, and if $p\cdot g > 1$ the model is not well-defined), and, to generate a vote, we choose~$i \in [g]$ 
uniformly at random and sample the vote from a~$(p,\phi)$-resampling model with the central vote that approves exactly the candidates from~$C_i$.

\paragraph{Noise Models.} Let~$p$ and~$\phi$ be two numbers from~$[0,1]$ and let~$d$ be a distance between approval votes (such as the Hamming or Jaccard ones). We require
that~$d$ is polynomial-time computable and, for each two approval votes~$u$ and~$v$,~$d(u,v)$ depends only on~$|A(u)|$,
$|A(v)|$, and~$|A(u) \cap A(v)|$; both Hamming and Jaccard distances have this property. In the~$(p,\phi,d)$-noise model we first generate a central vote~$u$ 
as in the resampling model and, then, each new vote~$v$ is generated with probability
proportional to~$\phi^{d(u,v)}$. Such noise models are analogous to
the Mallows model for ordinal elections and were
studied, e.g., by \cite{allouche2022truth} and \cite{caragiannis2022evaluating}. In particular,
\cite{caragiannis2022evaluating} gave a sampling procedure for the Hamming
distance. We extend it to arbitrary distances.

\begin{proposition}
  There is a polynomial-time sampling procedure for the
 ~$(p,\phi,d)$-noise models (as defined above).
\end{proposition}
\begin{proof}
  Let~$u$ be the central vote and let~$z = |A(u)|$. Consider
  non-negative integers~$x$ and~$y$ such that~$x \leq z$ 
  and~$y \leq m - z$. The probability of generating a vote~$v$ that
  contains~$x$ candidates from~$A(u)$ and~$y$ candidates 
  from~$C \setminus A(u)$ is proportional to the following value (abusing notation, we write~$d(x,y,z)$ to mean the value~$d(u,v)$;
  indeed,~$d(u,v)$ depends only on~$x$,~$y$, and~$z$):
  \[
    f(x,y) = \textstyle \binom{z}{x} \binom{m-z}{y} \phi^{d(x,y,z) }.
  \]
  Next, let~$Z = \sum_{x \in [z]_0, y \in [m-z]_0} f(x,y)$.
%
  To sample a vote, we draw values~$x \in [z]$ and~$y \in [m-z]$ with
  probability~$\frac{f(x,y)}{Z}$ and form the vote as approving~$x$
  random members of~$A(u)$ and~$y$ random members 
  of~$C \setminus A(u)$.
\end{proof}
\noindent
In the reminder, we only use the noise model with the Hamming distance
and we refer to it as the~$(p,\phi)$-noise model.  Note that the roles of~$p$
and~$\phi$ in this model are similar but not the same as in 
the~$(p,\phi)$-resampling model (for example, for~$\phi = 0$ we get the~$p$-ID model, 
but for~$\phi=1$ we get the~$0.5$-IC one).

\vspace{-0.125cm}
\paragraph{Euclidean Models.}
In the~$t$-dimensional Euclidean model, 
each candidate and each voter is a point 
from~$\reals^t$ and a voter~$v$ approves candidate~$c$ if the distance
between their points is at most~$r$ (this value is called the radius);
such models were discussed, e.g., in the classical works
of Enelow and Hinich~[\citeyear{enelow1984spatial},\citeyear{enelow1990advances}], and more recently
by~\citet{elk-lac:c:ci-vi--approval},
\citet{elk-fal-las-sko-sli-tal:c:2d-multiwinner},
\citet{bre-fal-kac-nie2019:experimental_ejr}, and
\citet{god-bat-sko-fal:c:2d}.
We consider~$t$-dimensional models for~$t \in \{1,2\}$, where the
agents' points are distributed uniformly at random on~$[0,1]^t$. We
refer to them as Interval and Square models (note that to fully
specify each of them, we also need to indicate the radius value).

\vspace{-0.125cm}
\paragraph{Truncated Urn Models.}
Let~$p$ be a number in~$[0,1]$ and let~$\alpha$ be a non-negative real
number (the parameter of contagion).  Truncated urn model is based on Pólya-Eggenberger urn model (see \Cref{ch:stat_cult}), however, after sampling an ordinal vote we convert it to an approval one.
We start with an
urn that contains all~$m!$ possible linear orders over the candidate
set. To generate a vote, we (1) draw a random order~$r$ from the urn,
(2)~produce an approval vote that consists of~$\lceil p\cdot m \rceil$
top candidates according to~$r$ (this is the generated vote), and
(3)~return~$\alpha m!$ copies of~$r$ to the urn.
%
For~$\alpha = 0$, all votes with~$\lceil p\cdot m \rceil$ approved
candidates are equally likely, whereas for large values of~$\alpha$ 
all votes are likely to be identical (so the model becomes 
similar to~$p$-ID).

\vspace{0.2cm}
\begin{conclusionbox}
    Introduction of new statistical culture models, including: resampling, disjoint, and moving models.
\end{conclusionbox}

\section{Maps of Approval Preferences}\label{approval_map_pref}
Now we will have a closer look at instances generated according to the statistical cultures described above. We will conduct an analogous experiment to the one described in \Cref{ordinal_map_pref}, however, this time we focus on approval elections.

\begin{figure}
    \centering
    \begin{minipage}{.49\textwidth}
        \centering
        \text{ \ \ \ Hamming}\\
        \text{}
 	    \includegraphics[width=7cm]{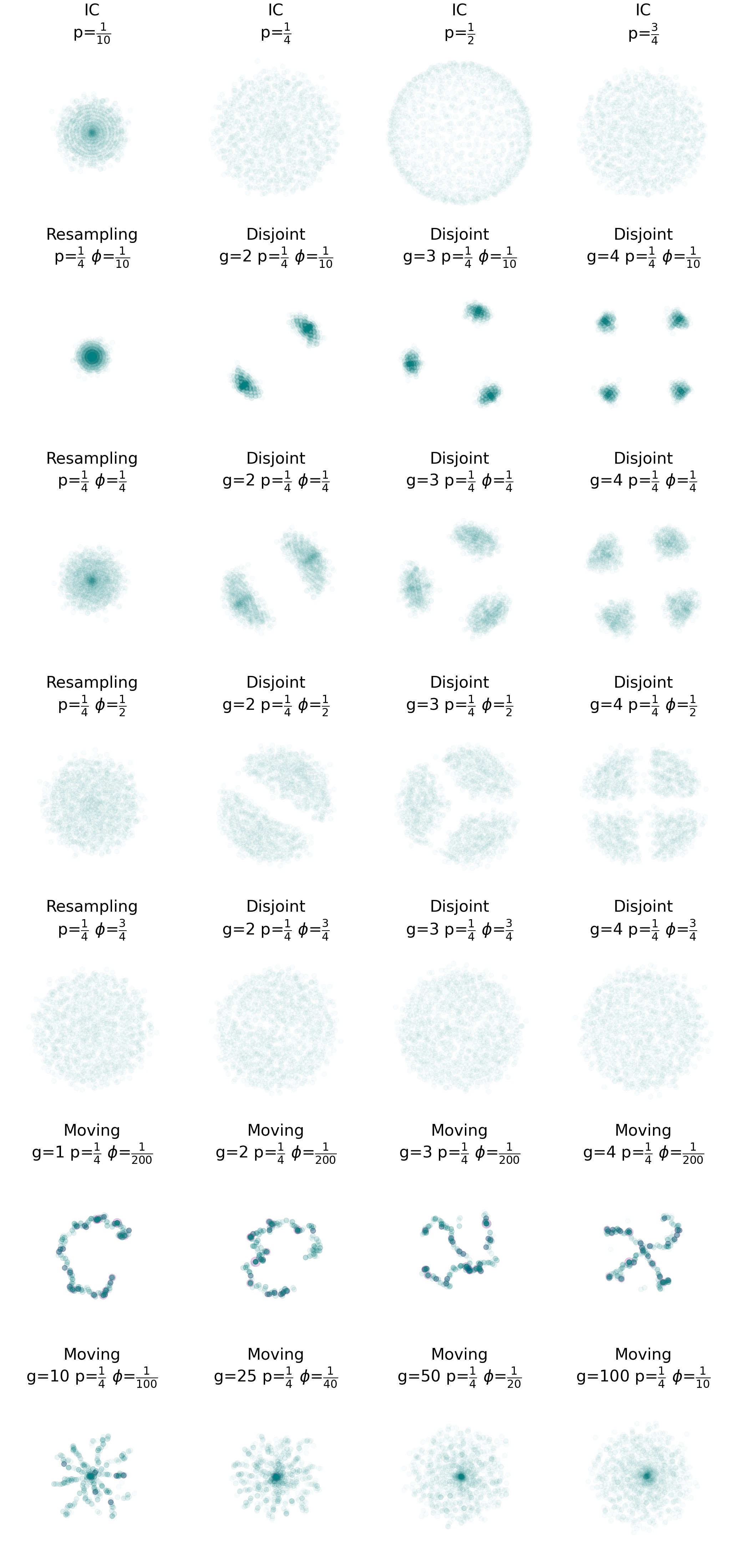}
 	\end{minipage}\hfill
 	\begin{minipage}{.49\textwidth}
 		\centering
        \text{ \ \ Jaccard}\\
        \text{}
        \includegraphics[width=7cm]{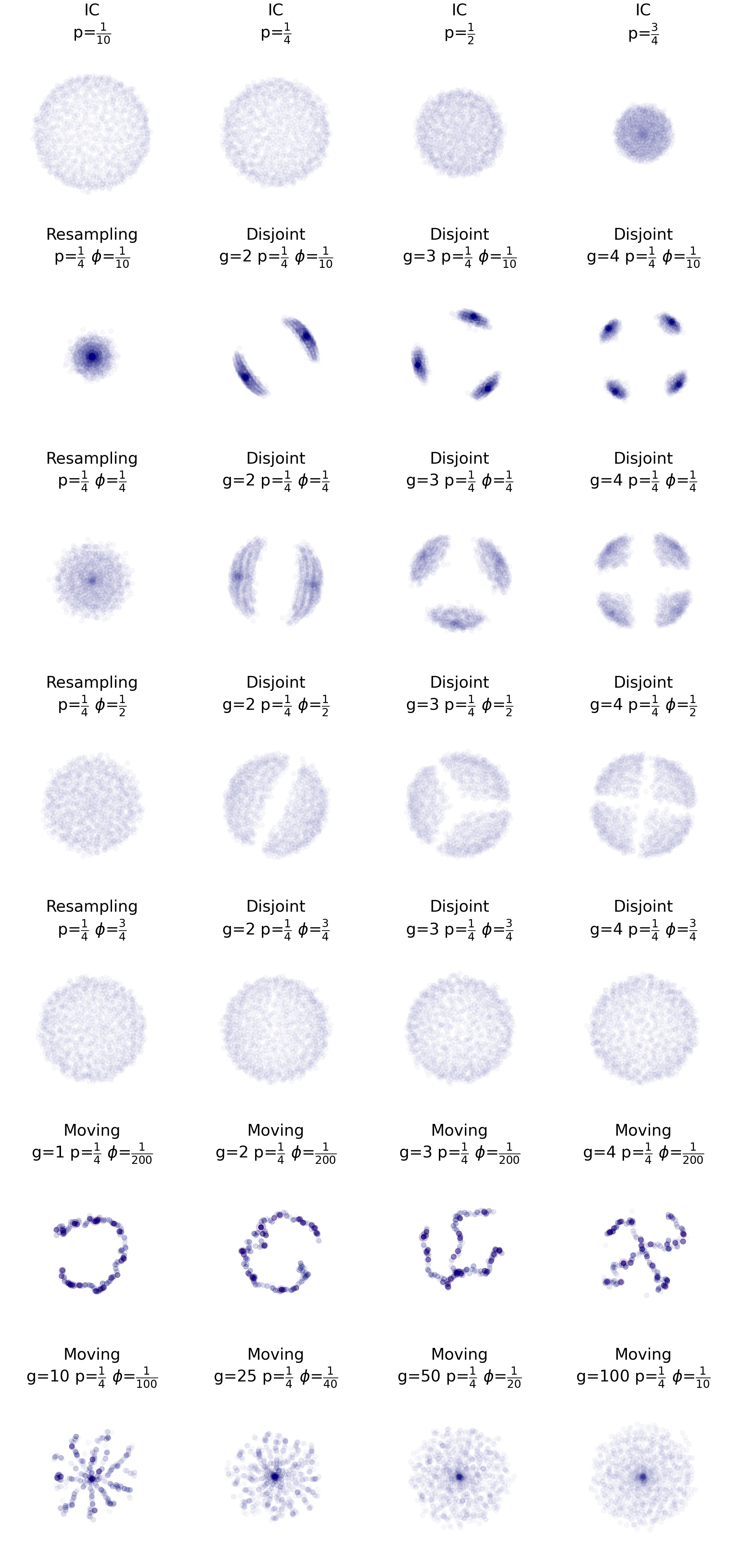}
 	\end{minipage}\hfill

    \caption{Maps of (Approval) Preferences ($100$ candidates,~$1000$ voters). On the left (teal) based on the Hamming distance, and on the right (navy) based on the Jaccard distance.}
    \label{fig:app_vote_part_1}
\end{figure}

\begin{figure}
    \centering
     	\begin{minipage}{.49\textwidth}
     	        \centering
        \text{ \ \ \ Hamming}\\
        \text{}
 		    \includegraphics[width=7cm]{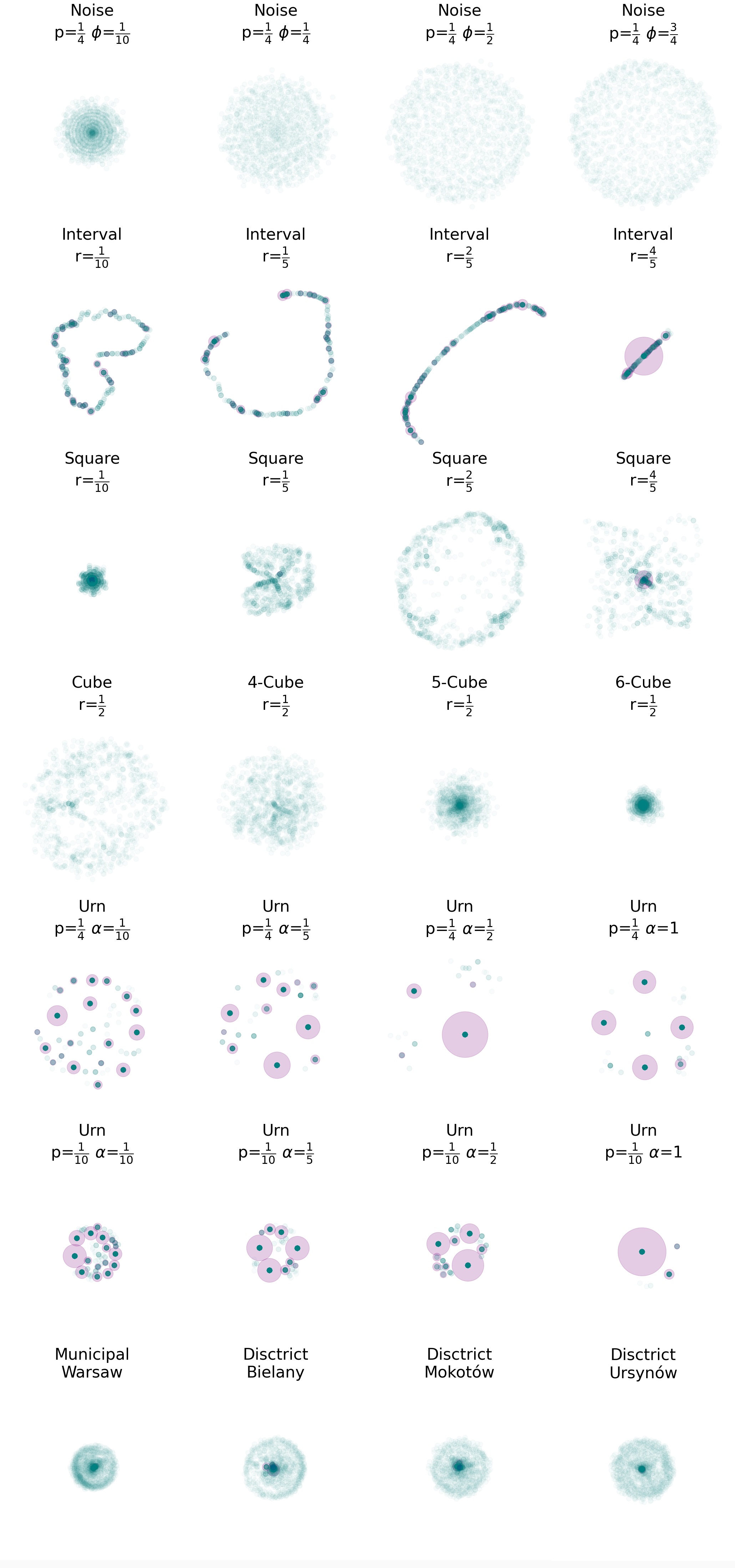}
 	\end{minipage}\hfill
 	\begin{minipage}{.49\textwidth}
 		\centering
        \text{ \ \ Jaccard}\\
        \text{}
        \includegraphics[width=7cm]{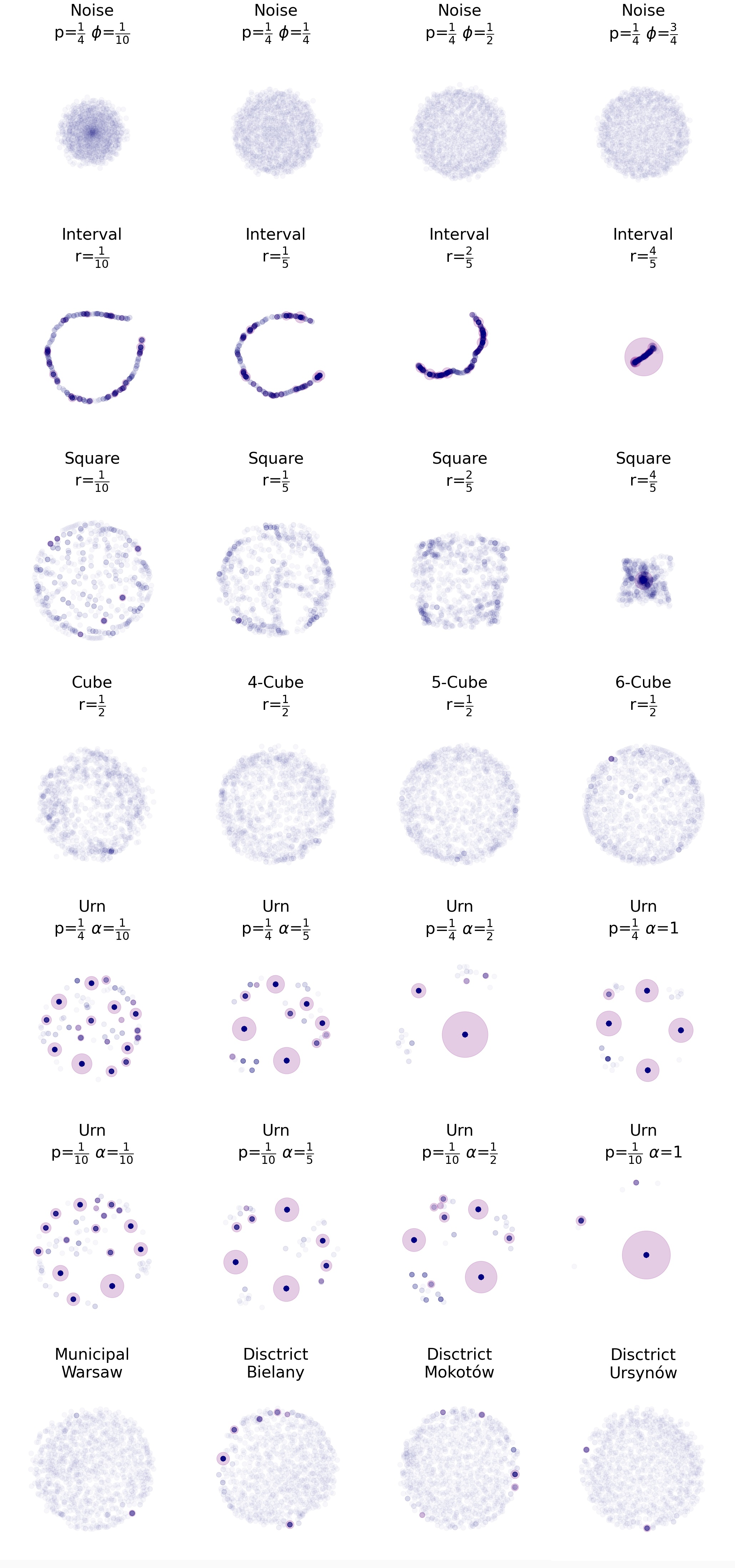}
 	\end{minipage}\hfill

    \caption{Maps of (Approval) Preferences ($100$ candidates,~$1000$ voters). On the left (teal) based on the Hamming distance, and on the right (navy) based on the Jaccard distance.}
    \label{fig:app_vote_part_2}
\end{figure}

As a metric between two approval votes, we use the Hamming and Jaccard distances. All generated instances consist of~$100$ candidates and~$1000$ voters. As for maps of ordinal preferences, we use the MDS embedding and with purple discs we depict the cases where more than~$10$ votes were identical (the larger the circle, the more votes were identical).
We present the results in \Cref{fig:app_vote_part_1,fig:app_vote_part_2}. The names and the parameters of each instance are presented above each picture. Now we discuss the results (more or less moving from the upper rows toward the bottom ones).

In the first row, we have the impartial culture elections. We see that for the Hamming distance the closer we are to~$p=0.5$, the larger is the circle. For small values of~$p$ the whole map has smaller diameter (i.e., the largest distance between any two votes in a~$p$-IC election is smaller) because all votes are similar to one another just by not accepting numerous candidates. (For instance, if we have~$100$ candidates and two disjoint votes each approving~$10$ candidates, then their Hamming distance is~$20$, while if we have two disjoint votes each approving~$50$ candidates, then their Hamming distances is~$100$; hence, we observe a large difference between these values, even though in both cases the sets of approved candidates are disjoint).
The Hamming distance is symmetric with regard to approvals and disapprovals, so if we replace~$p$ with~$1-p$ in the IC model, we should have the same result. For example, we observe that the picture for~$p=0.25$ is alike to the one for~$p=\frac{3}{4}$. However, this is far from true for the Jaccard distance, which is not symmetric with regard to approvals and disapprovals, and, in some sense, favors the approvals.
For the Jaccard distance when we increase the average number of approvals, the votes from the impartial culture will be on average at smaller distances from each other.

In the next four rows, we show results for the disjoint model (note that the resampling model is equivalent to the disjoint model with only one group). In the first three rows (i.e., for~$\phi \in \{\frac{1}{100},\frac{1}{40},\frac{1}{20}\}$ we observe clear division into groups. For the last row (i.e.,~$\phi=0.75$) the boundaries between the groups are fading away. 

In the last two rows, we show the moving model. In the first row, we fix~$\phi$ value to~$\frac{1}{200}$ and increase the number of groups from one to four. Note that having two groups is equivalent to having one group. In the second row, we consider larger numbers of groups, i.e.,~$10,25,50,100$ and proportionally increased~$\phi$ values, i.e.,~$\frac{1}{100},\frac{1}{40},\frac{1}{20},\frac{1}{10}$, respectively. (We increase the $\phi$ value because otherwise, for large numbers of groups, like $50$ or $100$, we would end up having many votes extremely similar to one another). Moreover, note that, the resampling model can be seen as an extreme case of the moving model, where the number of groups is equal to the number of votes.

In the second set of maps (\Cref{fig:app_vote_part_2}), we start with the noise model. When we increase the~$\phi$ value, we move closer toward IC (in particular, closer to~$0.5$-IC). Unlike for the resampling model, for noise model when we increase the noise we also increase the average number of approvals (or decrease if the initial~$p$ value was above~$0.5$).
Then we have three rows of the Euclidean elections. Keep in mind that the larger the radius, the more approvals we have on average.

Next, we have two rows for the urn elections. With~$\alpha$ parameter increasing from~$\frac{1}{10}$ up to~$1$, and~$p$ equal~$\frac{1}{4}$ in the upper row, and~$\frac{1}{10}$ in the lower row. Again we can see that under the Hamming distance, for smaller~$p$ the diameter of the whole map is smaller. At the same time the Jaccard distance is proportionally stretching the maps in the lower row so for both values of~$p$ they look similar.

Finally, we have four real-life instances based on the participatory budgeting elections held in Warsaw in 2022. One Municipal, where citizens could approve up to~$10$ projects, and three district ones, where citizens were allowed to approve up to~$15$ projects. 
For the Hamming distance, we observe dense centers in all four instances. These centers depict the voters that selected only one project, hence, are at most at distance $2$ from each other. The further a given point is from the center, the more projects were approved by the voter which that point represents. For the Jaccard distance, we see that the votes were very diverse; however, some of them had some copies. In principle, we do not observe any particular structure.

\new{
We also conducted a very similar experiment, but from the candidates' perspective. Due to the fact, that results for candidates, in essence, were not significantly different from those for voters, we decided to shift detailed description of these results to~\Cref{apdx:map_app_cand}. 
}

\vspace{0.2cm}
\begin{conclusionbox}
    The maps of approval preferences (similarly to the maps of ordinal preferences) confirm our intuition about the behavior of statistical cultures. They also help us get a better understanding of how particular parameters influence the models.
\end{conclusionbox}

\section{Metrics}
Next, we describe two (pseudo)metrics used to measure
distances between approval elections. Since we are interested in
distances between randomly generated elections, our metrics are
independent of renaming the candidates and 
voters.

Consider two equally-sized candidate sets~$C$ and~$D$, and a voter~$v$ with a ballot over~$C$.
For a bijection~$\sigma \colon C \rightarrow D$, by~$\sigma(v)$ we mean a voter with
an approval ballot~$A(\sigma(v)) = \{ \sigma(c) \mid c \in C\}$. In
other words,~$\sigma(v)$ is the same as~$v$, but with the candidates
renamed by~$\sigma$. 
Next, we define the isomorphic Hamming distance (inspired by the isomorphic swap and Spearman distances \Cref{swapSpearDef}).


\begin{definition}
  Let~$E = (C,V)$ and~$F = (D,U)$ be two elections, where~$|C| = |D|$,
 ~$V = (v_1, \ldots, v_n)$ and~$U = (u_1, \ldots, u_n)$.  The
  \emph{isomorphic Hamming distance} between~$E$ and~$F$, denoted~$d_{\hamming}(E,F)$, is defined as:
  \begin{align*}
   \textstyle\min_{\sigma \in \Pi(C,D)}\min_{\rho \in S_n} \left(\sum_{i=1}^n \ham(\sigma(v_i), u_{\rho(i)}) \right).
  \end{align*}
\end{definition}
\noindent Intuitively, under the isomorphic Hamming distance we unify
the names of the candidates in both elections and match their voters
to minimize the sum of the resulting Hamming distances. We call this
distance \emph{isomorphic} because its value is zero exactly if the
two elections are identical, up to renaming the candidates and voters.
Computing this distance is~$\np$-hard (see also the
related results for approximate graph
isomorphism~\citep{arv-koe-kuh-vas:c:approximate-graph-isomorphism,gro-rat-woe:c:approximate-isomorphism}).

\begin{proposition}\label{hamming-hard}[\cite{SFJLSS22}]
  Computing the isomorphic Hamming distance between two approval elections is $\np$-hard.
\end{proposition}

Consequently, we compute this distance using a brute-force
algorithm (which is faster than using, e.g., ILP formulations).
%
%
Since this limits the size of elections that we can deal with, we also
introduce a simple, polynomial-time computable metric.

\newcommand{\fracc}[2]{{#1}/{#2}}
\begin{definition}
  Let~$E$ 
  be an 
  election 
  with candidate set~$\{c_1, \ldots, c_m\}$ and
$n$ voters. 
  Its approvalwise vector, denoted~$\av(E)$, is
  obtained by sorting the vector~$(\fracc{\score_\av(c_1)}{n}$,~$\ldots, \fracc{\score_\av(c_m)}{n})$
  in the non-increasing order.
  Then, the approvalwise distance between elections~$E$ and~$F$ with approvalwise
  vectors~$\av(E) = (x_1, \ldots, x_m)$ and~$\av(F) = (y_1, \ldots, y_m)$
  is defined as:~$$d_\app(E,F) = |x_1-y_1| + \cdots + |x_m-y_m|.$$
\end{definition}

\noindent In other words, the approvalwise vector of an election 
is a sorted vector of the normalized approval scores of its candidates, and an approvalwise distance between two elections is the~$\ell_1$ distance between their approvalwise vectors.
We sort the vectors 
to avoid the explicit use of candidate matching,
as is needed in the Hamming distance.  Occasionally we will speak of
approvalwise distances between approvalwise vectors, without
referring to the elections that provide them.

It is easy to see that the approvalwise distance is computable in
polynomial time. In fact, its definition is so simplistic that it is
natural to even question its usefulness. In its spirit, the approvalwise distance is very similar to the Bordawise distance (used for the ordinal elections); both distances convert elections to vectors of length $m$ and compare them. While the Bordawise distance seems not to be very useful, surprisingly, the approvalwise distance is quite effective. 

In~\Cref{sec:correlation} we
will see that in our election datasets the approvalwise distance is strongly correlated with the
Hamming distance. Thus, in the following discussion, we focus on
approvalwise distances.

\section{A Grid of Approval Elections}\label{sec:grid}

To better understand the approvalwise metric space of elections, next we analyze expected distances between elections generated according
to the~$(p,\phi)$-resampling model.

Fix some number~$m$ of candidates and parameters~$p, \phi \in [0,1]$, such that~$pm$ is an
integer, and consider the process of generating votes from
the~$(p,\phi)$-resampling model. 
In the limit, the approvalwise vector of the resulting election is:
\[
  (\underbrace{(1-\phi) + (\phi \cdot p), \ldots, (1-\phi) +
    (\phi\cdot p)}_{p\cdot m}, \underbrace{\phi\cdot p, \ldots,
    \phi\cdot p}_{(1-p)\cdot m}).
\]
Indeed, each of the~$p\cdot m$ candidates approved in the central
ballot either stays approved (with probability~$1-\phi$) or is
resampled (with probability~$\phi$, and then gets an approval with
probability~$p$). Analogous reasoning applies to the remaining~$(1-p)\cdot m$
candidates.
With a slight abuse of notation, we call the above vector~$\av(p,\phi)$. Furthermore, we refer to~$\av(p,0)$ as the~$p$-ID vector,
to~$\av(p,1)$ as the~$p$-IC vector, and to~$0$-ID and \linebreak~$1$-ID vectors
as the \emph{empty} and \emph{full} ones, respectively (note that~$0$-ID~$=$~$0$-IC and \linebreak~$1$-ID~$=$~$1$-IC).

Now, consider two additional numbers,~$p', \phi' \in [0,1]$, such that
$p' m$ is an integer. Simple
calculations show that:
\begin{align*}
  d_\app(\mathit{empty},\mathit{full}) &= m, \\
  d_\app(p\hbox{-}\mathrm{IC}, p\hbox{-}\mathrm{ID}) &= 2mp(1-p),\\
  d_\app(\av(p,\phi), \av(p',\phi) ) &= m \cdot|p-p'|,\\
  d_\app(\av(p,\phi), \av(p,\phi') ) &= 2mp(1-p)\cdot|\phi-\phi'|.
\end{align*}
Thus~$d_\app(\av(p,\phi), \mathit{empty})=mp$ is a~$p$~fraction of the
distance between \emph{empty} and \emph{full}, and~$d_\app(\av(p,\phi), p\hbox{-}\mathrm{ID}) = 2mp(1-p)\phi$ is a~$\phi$
fraction of the distance between~$p$-IC and~$p$-ID (see also
Figure~\ref{compass}).  Furthermore,
$d_\app(\mathit{empty},\mathit{full}) = m$ is the largest possible
approvalwise distance.

\begin{figure}[t]
\centering
    \scriptsize
   \begin{center}
     \begin{tikzpicture}
       \node[inner sep=3pt, anchor=south] (up)  at (1,4) {full};
       \node[inner sep=3pt, anchor=east] (left)  at (0,2) {0.5-IC};
       \node[inner sep=3pt, anchor=west] (right)  at (2,2) {0.5-ID};
       \node[inner sep=3pt, anchor=north] (down)  at (1,0) {empty};
       
        \draw[draw=black, line width=0.4mm, -]     (left) edge node  {} (up);
        \draw[draw=black, line width=0.4mm, -]     (left) edge node  {} (right);
        \draw[draw=black, line width=0.4mm, -]     (left) edge node  {} (down);
        \draw[draw=black, line width=0.4mm, -]     (up) edge node  {} (right);
        \draw[draw=black, line width=0.4mm, -]     (up) edge node  {} (down);
        \draw[draw=black, line width=0.4mm, -]     (right) edge node  {} (down);
        
        \normalsize
        \node[] at (0,3.4) {$\nicefrac{m}{2}$};
        \node[] at (2,3.4) {$\nicefrac{m}{2}$};
        \node[] at (0,0.6) {$\nicefrac{m}{2}$};
        \node[] at (2,0.6) {$\nicefrac{m}{2}$}; 
        
        \node[] at (1.5,2+0.25) {$\nicefrac{m}{2}$}; 
        \scriptsize
        \node[] at (1+0.25,3) {$m$}; 
        
        \scriptsize
        \node[] at (0.05-0.4,1.25) {$p$-IC};
        \node[] (pic) at (0.05,1.25) [circle,fill,inner sep=1.5pt]{};
        \node[] at (1.95+0.4,1.25) {$p$-ID}; 
        \node[] (pid) at (1.95,1.25) [circle,fill,inner sep=1.5pt]{};
        \draw[draw=black, line width=0.2mm, -]     (pic) edge node  {} (pid);
        \node[] at (0.66,1.05) {$(p,\phi)$}; 
        \node[] (p) at (0.66,1.25) [circle,fill,inner sep=1.5pt]{};
        
        \node[] at (1.4,1.35) {$a$};
        \node[] at (0.4,1.38) {$b$}; 
        \node[] at (0.67,0.62) {$c$}; 
        \node[] at (0.74,3.0) {$d$};  
        \node[anchor=west] at (4.1,2.5) {$a=2mp(1-p)\phi$};
        \node[anchor=west] at (4.1,2.2) {$b=2mp(1-p)(1-\phi)$};
        \node[anchor=west] at (4.1,1.9) {$c=mp$};  
        \node[anchor=west] at (4.1,1.6) {$d=m(1-p)$};  
        
        \draw[draw=black, line width=0.2mm, -]     (p) edge node  {} (up);
        \draw[draw=black, line width=0.2mm, -]     (p) edge node  {} (down);
    
      \end{tikzpicture}
    \end{center}
    \caption{\label{compass} Distances between resampling elections.}
\end{figure}
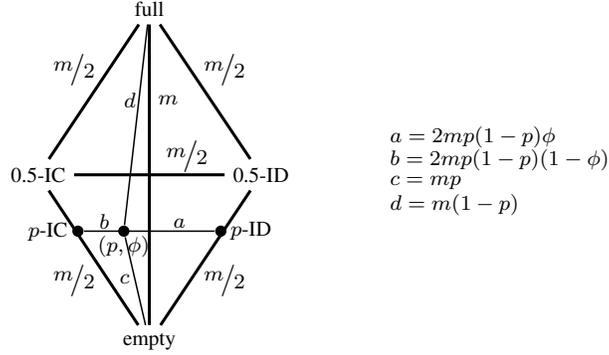

Intuitively,~$(p,\phi)$-resampling elections form a grid that spans
the space between the extreme points of our election space; the larger
the~$\phi$ parameter, the more ``chaotic'' an election becomes
(formally, the closer it is to the~$p$-IC elections), and the larger
the~$p$~parameter, the more approvals it contains (the closer it is to the
\emph{full} election). We use~$(p,\phi)$-resampling elections as
a \emph{background} dataset, which consists of~$241$ elections with~$100$
candidates and~$1000$ voters each, with the following~$p$ and~$\phi$
parameters:
\begin{enumerate}
\item~$p$ is chosen from~$\{0, 0.1, 0.2, \dots, 0.9,1\}$ and~$\phi$ is
  chosen from the interval~$(0,1)$,\footnote{By generating~$t$
    elections with a parameter from interval~$(a,b)$, we mean
    generating one election for each value~$a+i\frac{b-a}{t+1}$, 
    for~$i \in [t]$.}  
\item~$\phi$ is chosen from~$\{0, 0.25, 0.5, 0.75, 1\}$ and~$p$ is
  chosen from the interval~$(0,1)$.
\end{enumerate}
For each of these elections, we compute a point in~$\reals^2$, 
so that the Euclidean distances between these points are as similar to
the approvalwise distances between the respective elections as
possible.  For this purpose, we use the Fruchterman-Reingold
force-directed algorithm~(see \Cref{desc:embed}). For
the resulting map, we see the clear grid-like shape on the left side of Figure~\ref{app-main-results-1}.\footnote{While our visualizations fit nicely into the two-dimensional
embedding, our election space has a much higher dimension.}
Whenever we present maps of elections later in the paper, we compute
them in the same way as described above (but for datasets that include
other elections in addition to the background ones).

\newcommand\width{3.3cm}
\newcommand\smallwidth{2.2cm}
\newcommand\finalwidth{5.2cm}

\begin{figure}[]
    \centering
    \includegraphics[width=\finalwidth, trim={3 0 3 0}, clip]{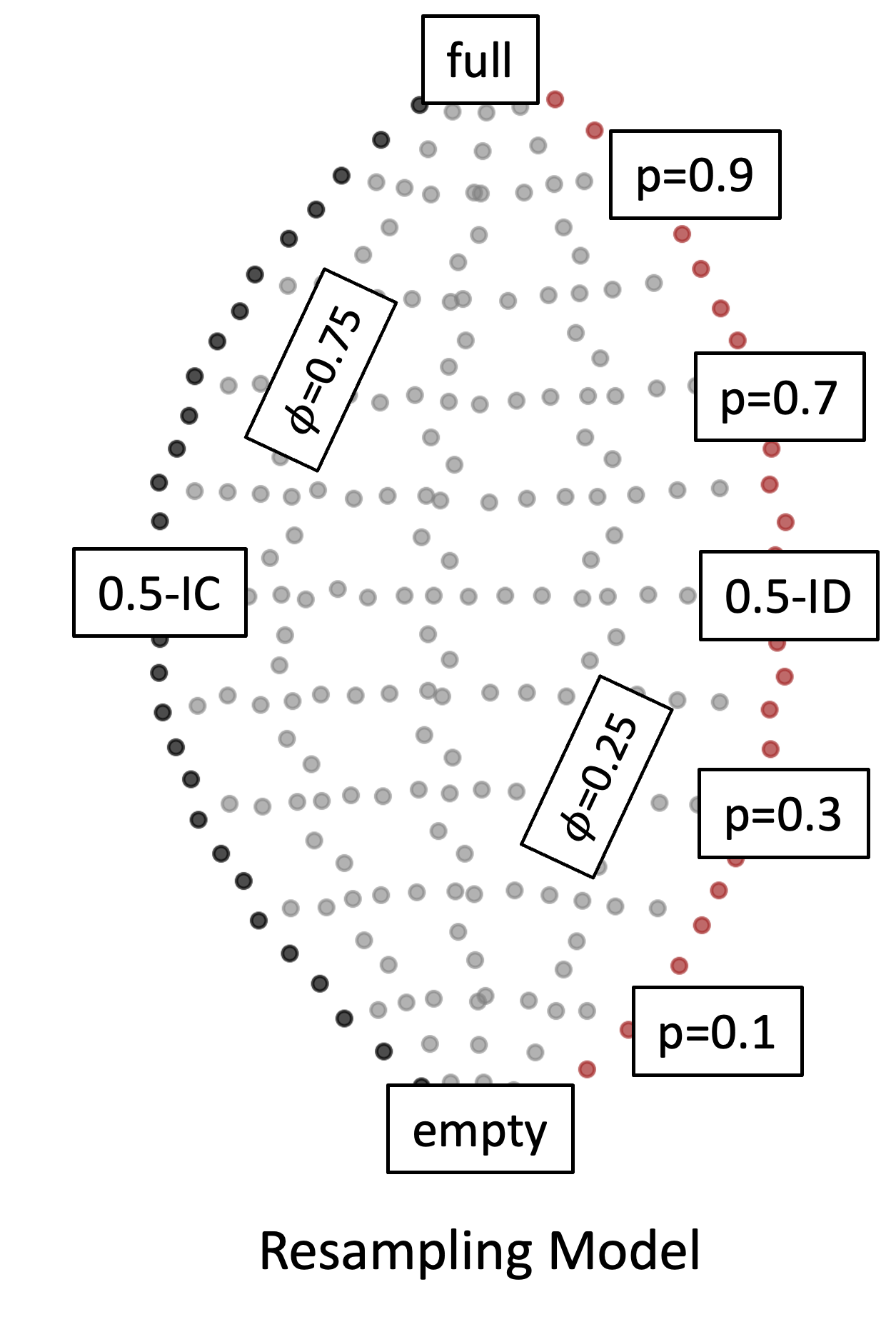}%
    \includegraphics[width=\finalwidth]{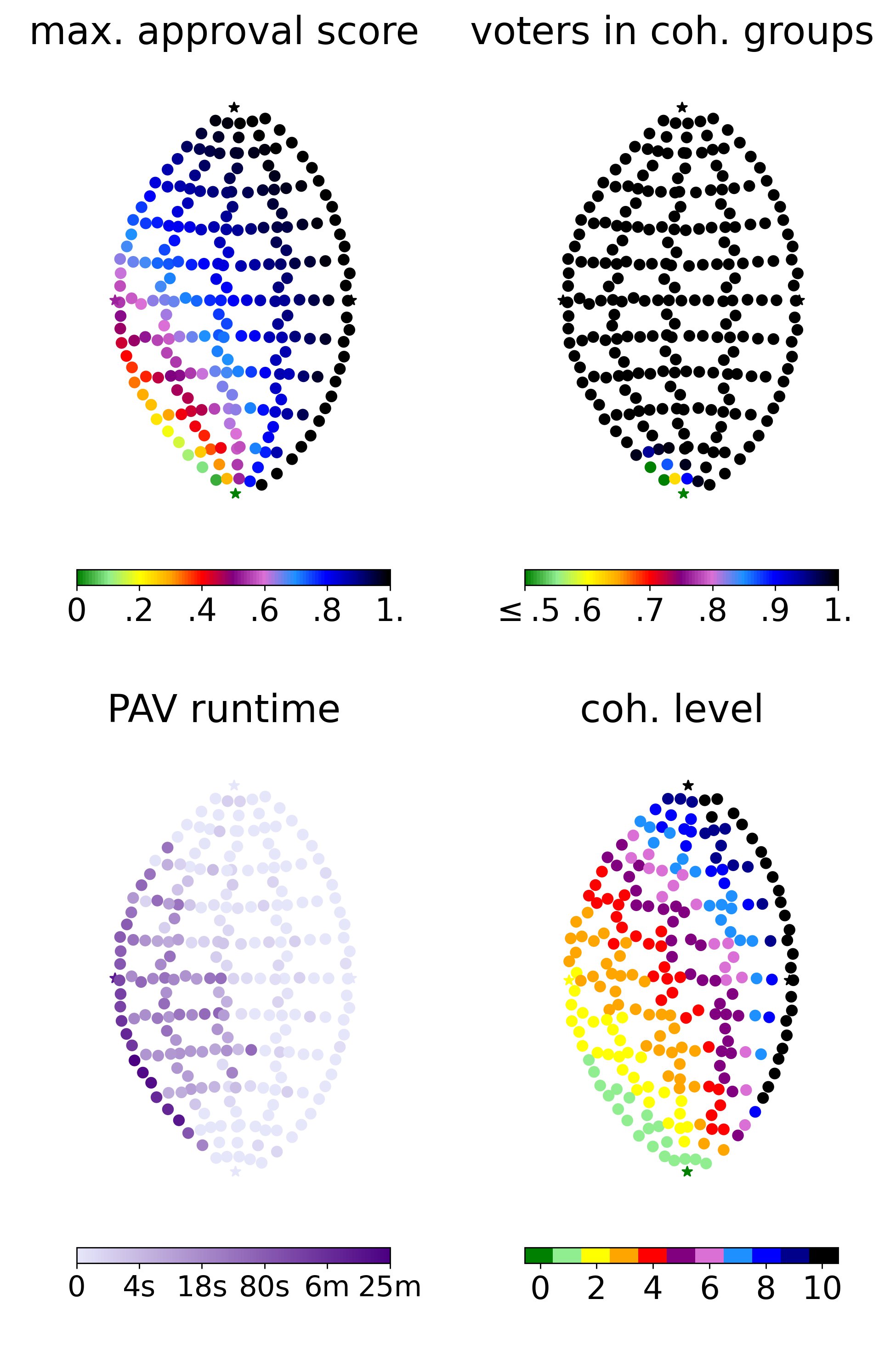}
    
    \caption{Maps for the resampling model.}
     \label{app-main-results-1}
\end{figure}

\section{Experiments}

In this section, we use the \textit{map of elections} approach
%
to analyze the quantitative properties of approval elections generated
according to our models.
%
In particular, we will see how an election's position in the grid
influences each of the properties, and what parameters to use to
%
%
generate elections with the quantitative property in a desired range.

\subsection{Experimental Design}

We use the map framework to visualize information about the following four statistics:

\begin{description}
  \item[Maximal Approval Score.] The highest approval score
  among all the candidates in a given election, normalized by the maximum
  possible score, i.e., the number of voters.

  \item[Cohesiveness Level.] The largest integer~$\ell$ such that
    there exists an~$\ell$-cohesive group (for committee size~$10$).
    To compute this, we use the algorithm based on the one provided by \cite{janeczko2022complexity}.

  \item[Voters in Cohesive Groups.] Fraction of voters that belong
    to at least one~$1$-cohesive group (for committee size~$10$).

  \item[PAV Runtime.] Runtime (in seconds) required to
  compute a winning committee under the PAV rule, by solving an
  integer linear program provided by the \texttt{abcvoting} library
  \citep{abcvoting}, using the Gurobi ILP solver.
\end{description}









We use the background dataset and six new datasets. Five of them
are generated using our statistical cultures
and consist of~$100$ candidates
and~$1000$ voters (except for the experiments related to the
cohesiveness level, where we have~$50$ candidates and~$100$ voters, due to
computation time). We have: 
\begin{itemize}
    \item $225$ elections from the noise model with Hamming distance ($25$ for each
$p \in \{0.1, 0.2, \dots, 0.9\}$ with~$\phi \in (0,1)$);

    \item $250$ elections from the disjoint model ($50$ for each
$g \in \{2,3,4,5,6\}$ with~$\phi \in (0.05, \nicefrac{1}{g})$);

    \item $225$
elections from the moving model ($25$ for each~$p \in \{0.1, 0.2, \dots, 0.9\}$ 
with~$\phi \in (0,\nicefrac{1}{100})$, and $g=1$);

    \item $200$
elections from Euclidean model ($100$ for Interval, with radius in
$(0.0025,0.25)$, and~$100$ for Square, with radius in
$(0.005, 0.5)$); these parameters are as used by
\cite{bre-fal-kac-nie2019:experimental_ejr}.

    \item $225$ elections from the truncated urn model ($25$ for each $p$~$\in$ $\{0.1,$ $0.2, \dots, 0.9\}$ with~$\alpha \in (0,1)$);

\end{itemize}

The last dataset uses real-life participatory budgeting data and contains~$44$ elections from
Pabulib~\citep{pabulib}, where for each (large enough) election we
randomly selected a subset of~$50$ candidates and~$1000$ voters
(other real-life datasets we considered had much fewer candidates).

\subsection{Experimental Results}

Our visualizations are shown in Figures~\ref{app-main-results-1}, \ref{app-main-results-ab}, \ref{app-main-results-cd}, and \ref{app-main-results-ef}.
%
%
We use the grid structure of the background dataset for comparison
with other datasets.  Notably, some of them do not fill this grid: the
disjoint model (Figure~\ref{app-main-results-ab}b) is restricted to the
lower half (i.e., the disjoint model does not yield elections with
very many approvals), the Euclidean model
(Figure~\ref{app-main-results-cd}d) is restricted to the left half (due to
the uniform distribution of points, its elections are rather ``chaotic''),
and the real-world dataset Pabulib (Figure~\ref{app-main-results-ef}f) is
placed very distinctly in the bottom left part.

\begin{figure}[]
    \centering

    \includegraphics[width=\finalwidth, trim={3 0 3 0}, clip]{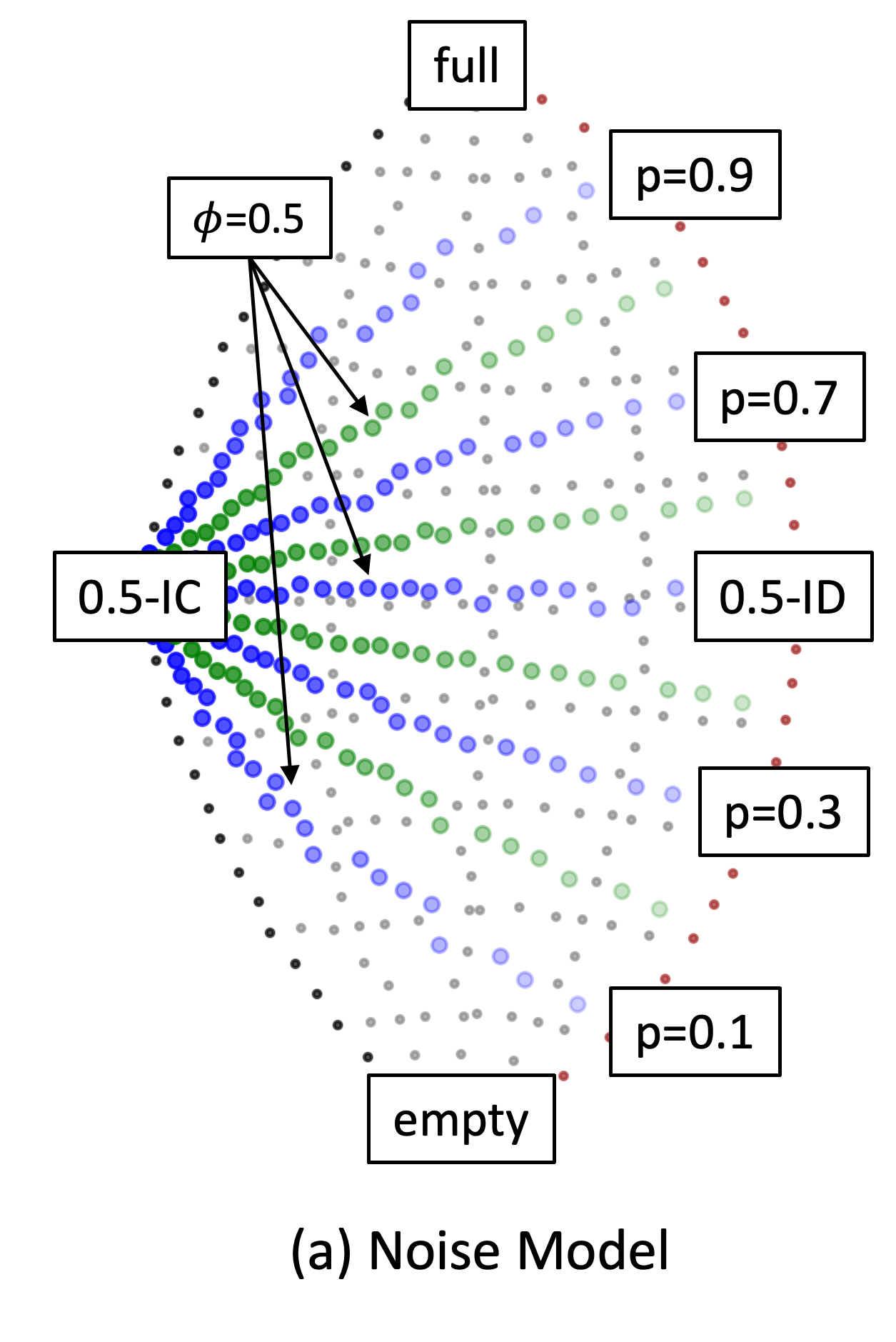}%
    \includegraphics[width=\finalwidth]{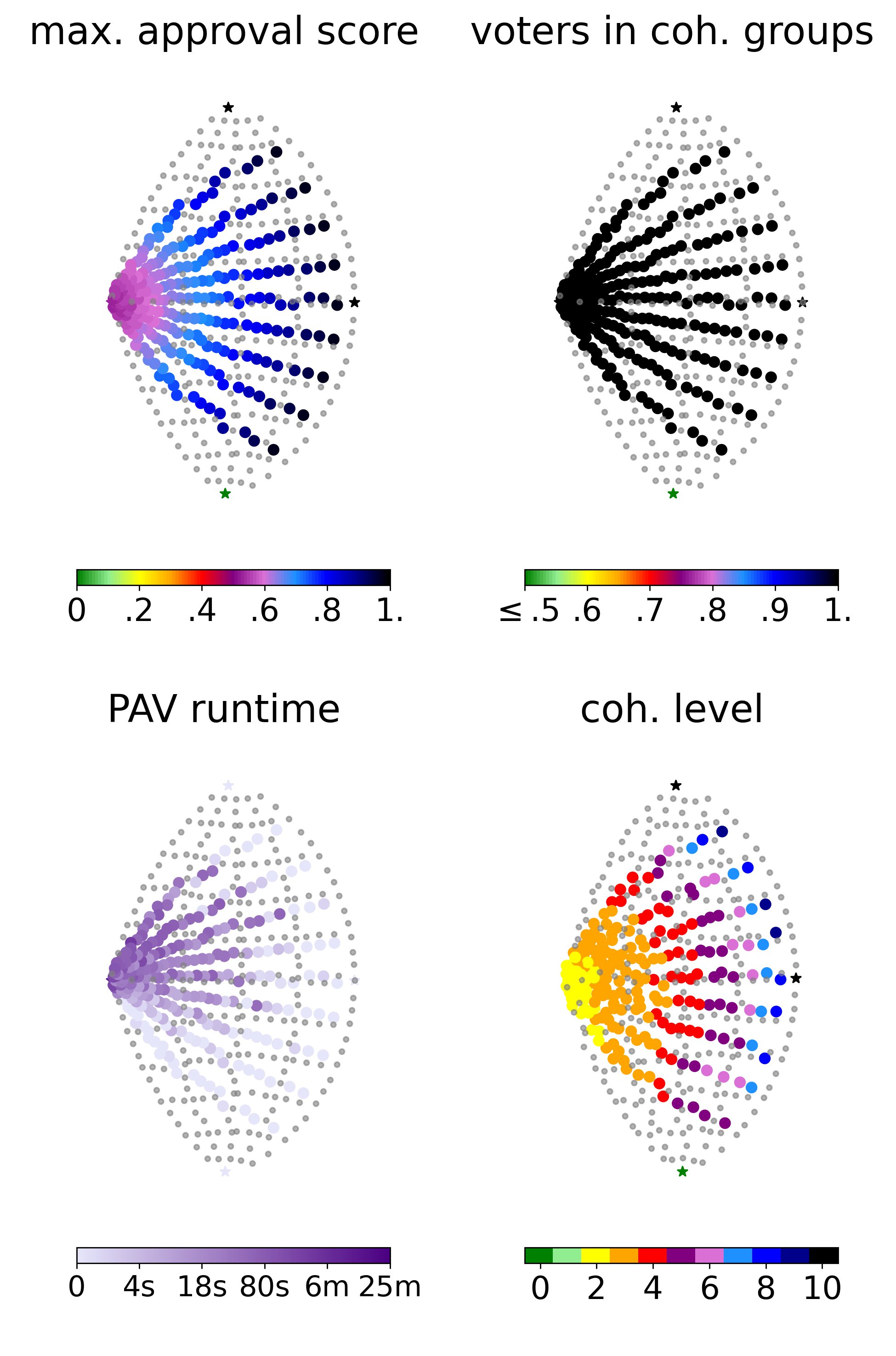}%
    \\
    \includegraphics[width=\finalwidth, trim={3 0 3 0}, clip]{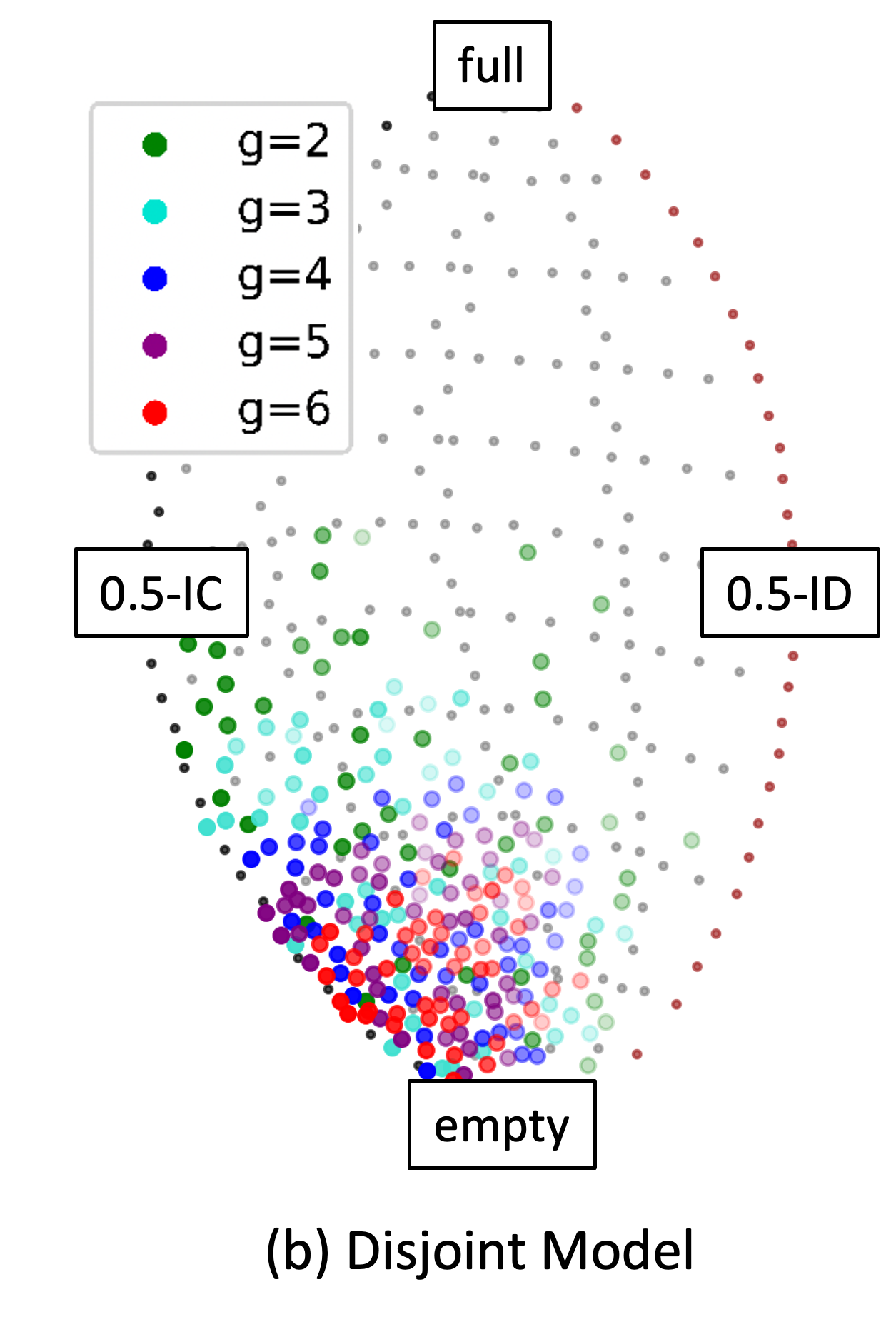}%
    \includegraphics[width=\finalwidth]{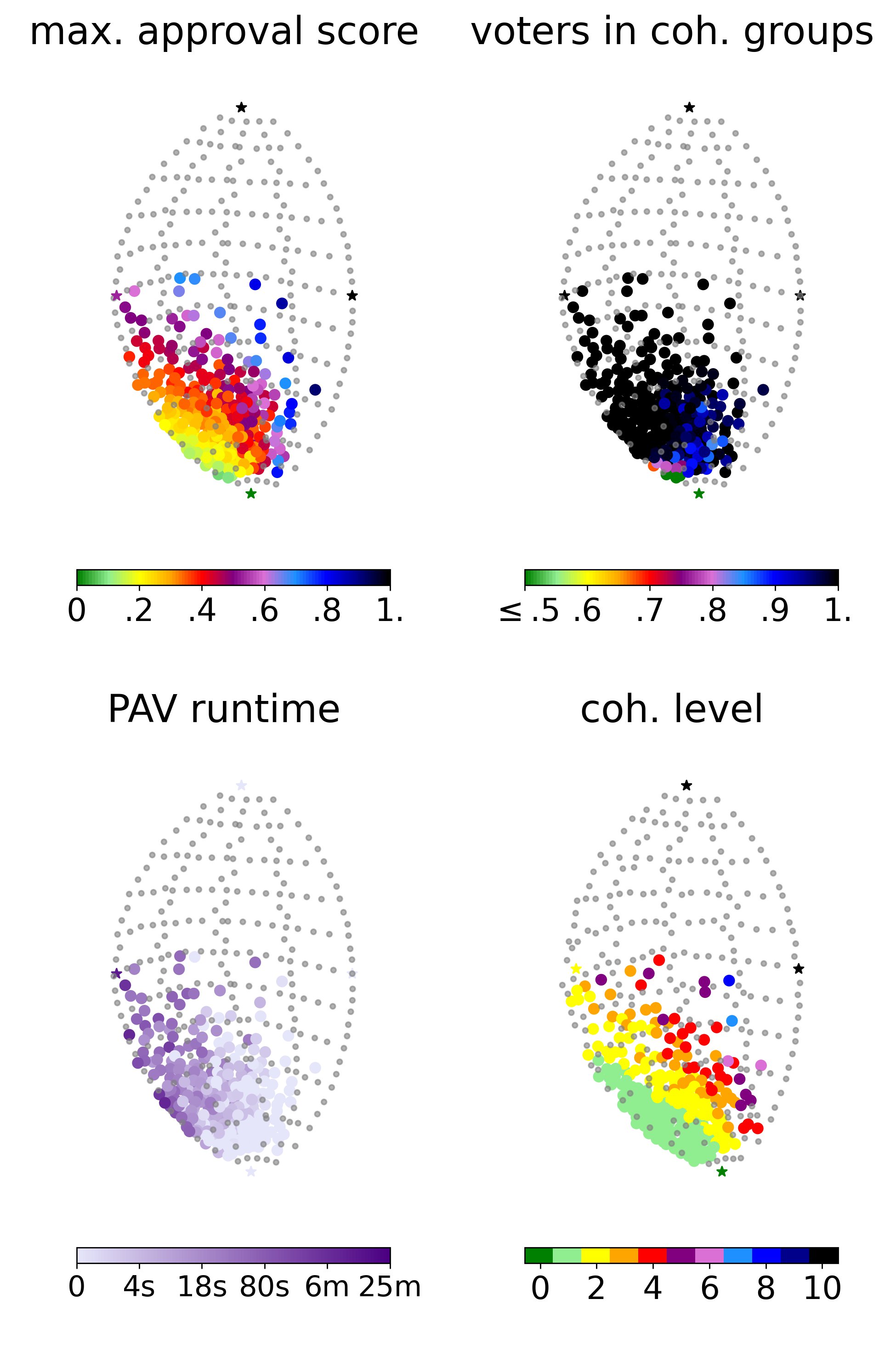}%

    \caption{Maps for (a)~the noise model and (b) the disjoint
      model. The darker a dot in the main plot is, the larger is the value of the~$\phi$ parameter.}
    \label{app-main-results-ab}
\end{figure}

\begin{figure}[]
    \centering

    \includegraphics[width=\finalwidth, trim={3 0 3 0}, clip]{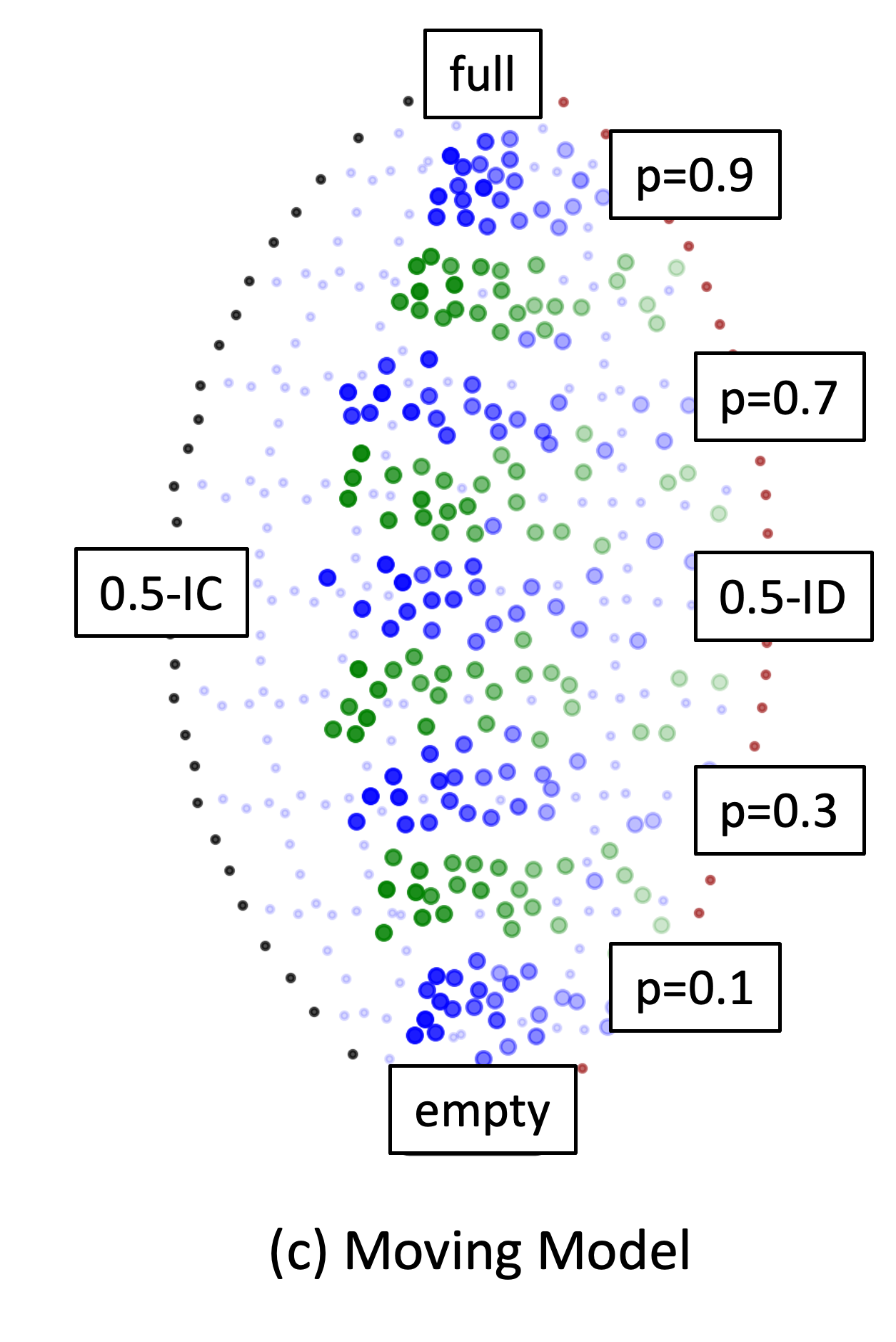}%
    \includegraphics[width=\finalwidth]{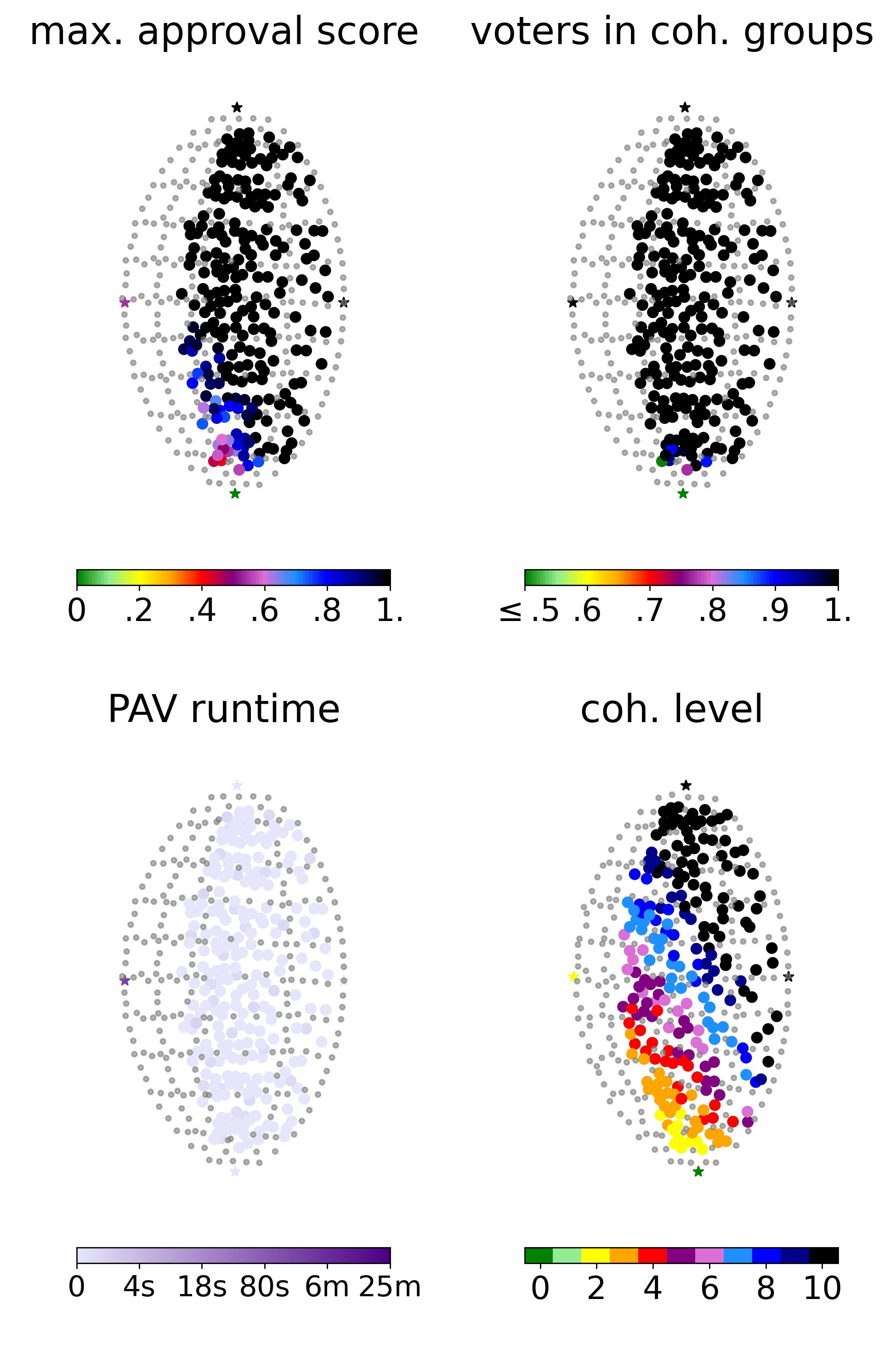}
    \\
    \includegraphics[width=\finalwidth, trim={3 0 3 0}, clip]{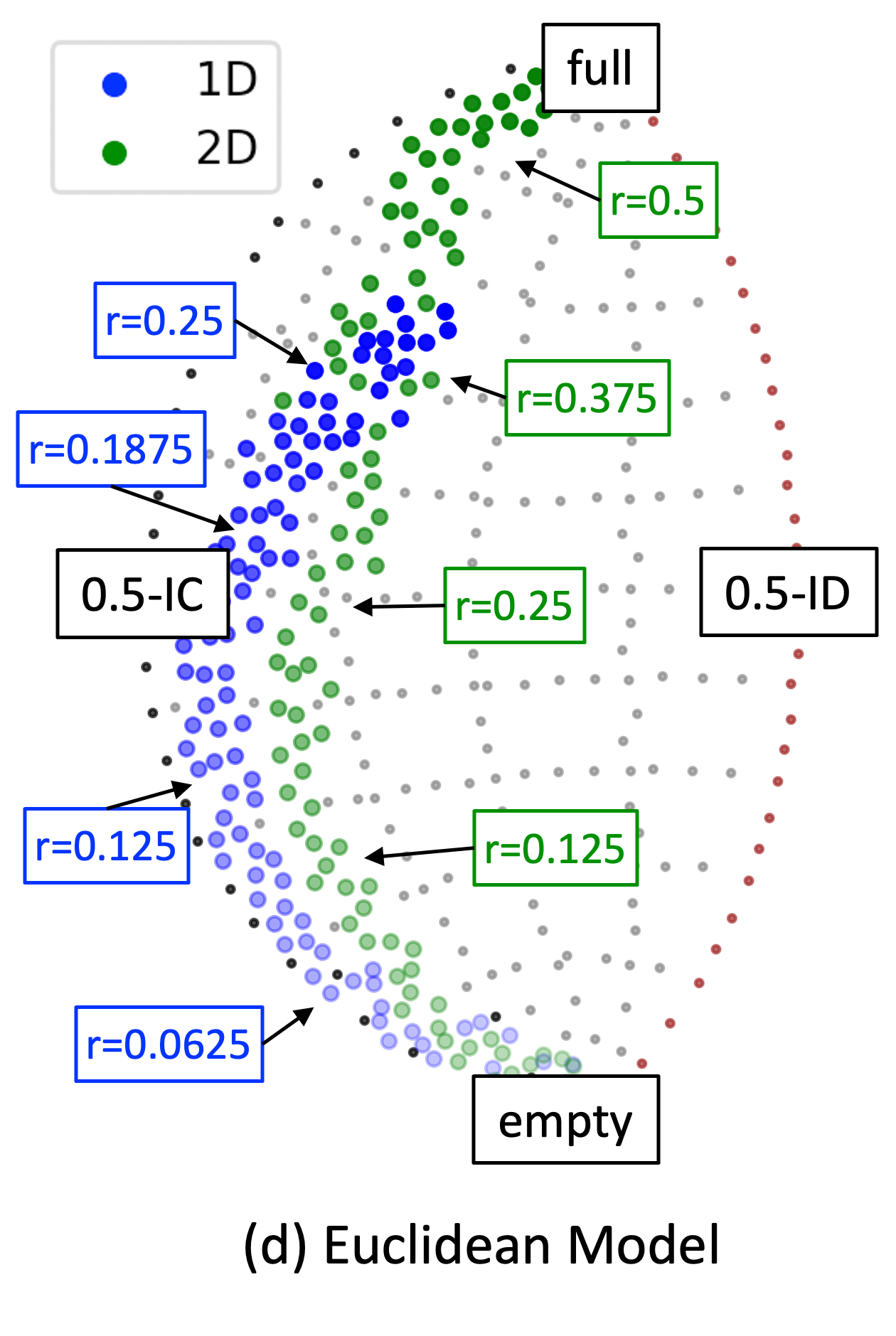}%
    \includegraphics[width=\finalwidth]{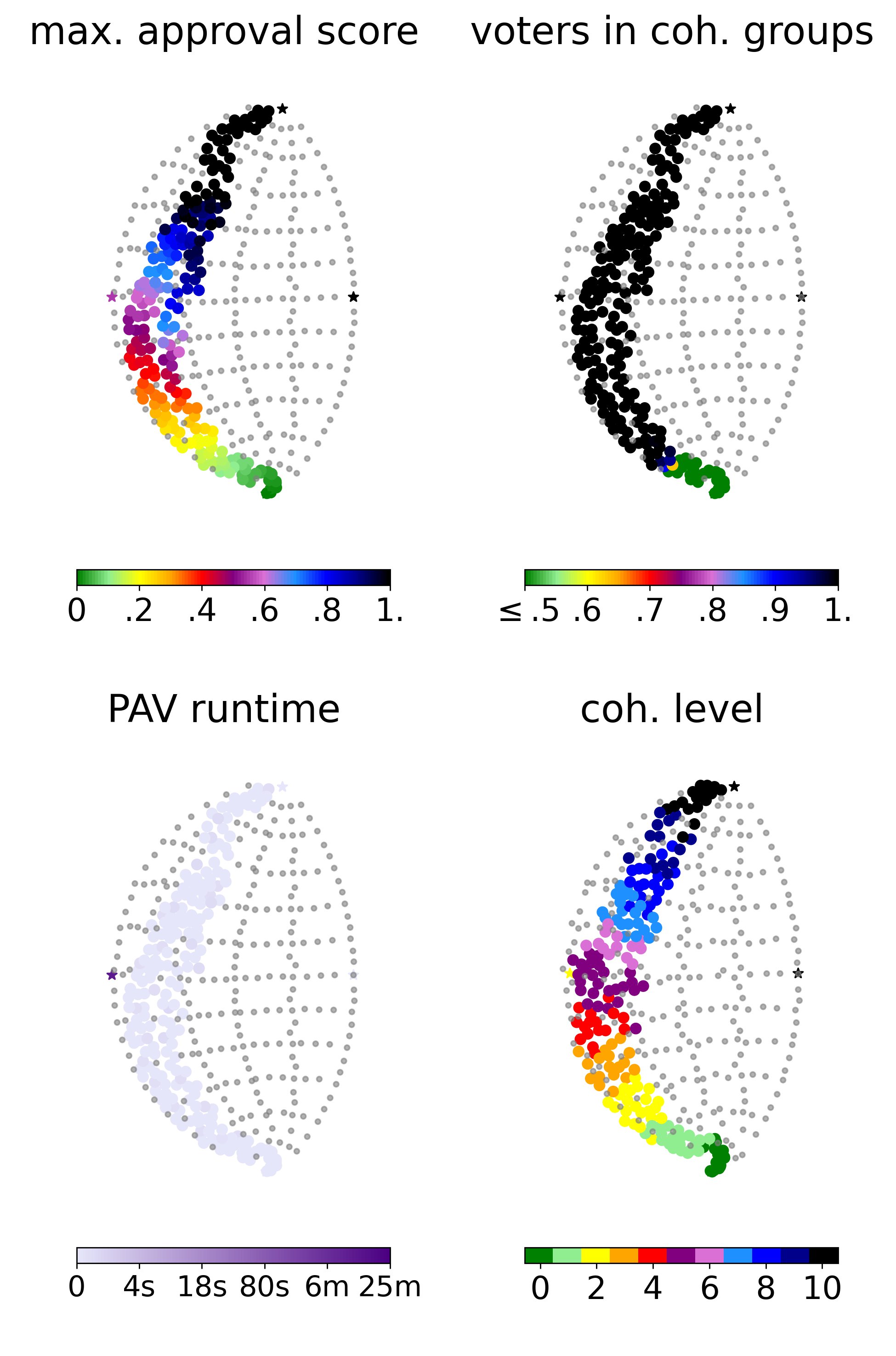}
        \caption{Maps for (c)~the moving model, (d)~the Euclidean model. The darker a dot in the main plot is, the larger is the value of~$\phi$ parameter for (c), and the larger is the length of the radius for (d).}
         \label{app-main-results-cd}
\end{figure}

\begin{figure}[]
    \centering
    \includegraphics[width=\finalwidth, trim={3 0 3 0}, clip]{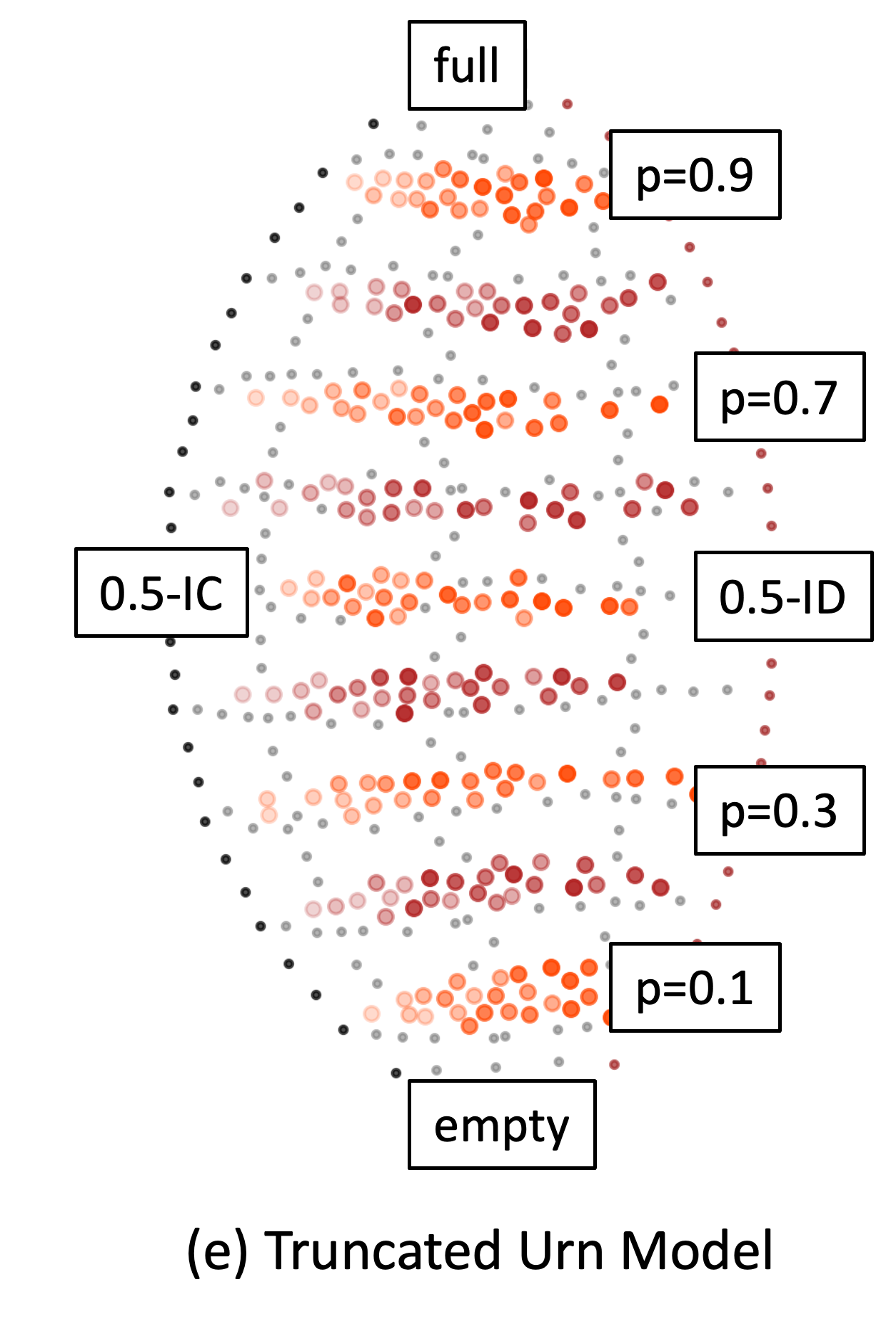}%
    \includegraphics[width=\finalwidth]{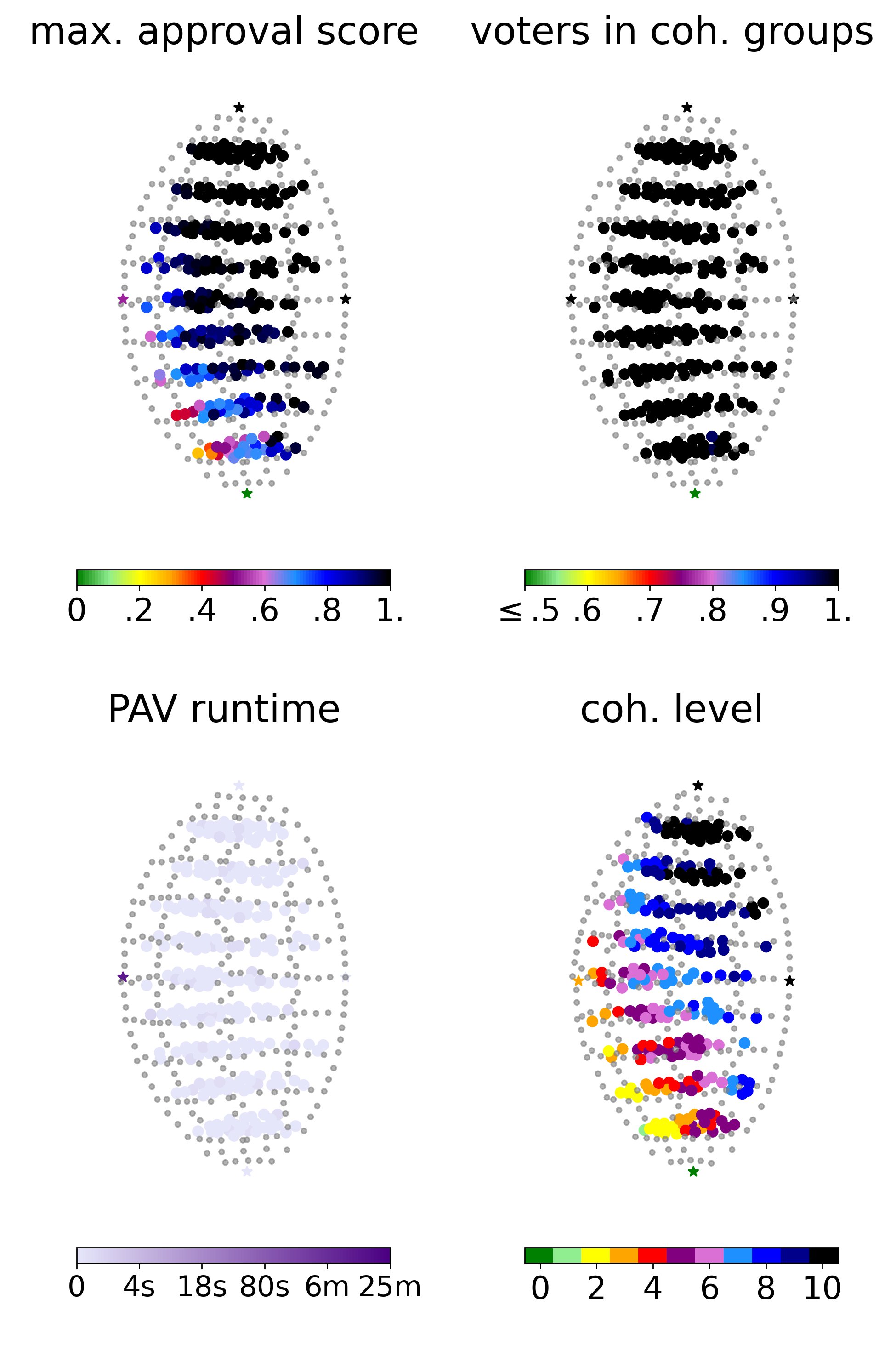}
    \\
    \includegraphics[width=\finalwidth, trim={3 0 3 0}, clip]{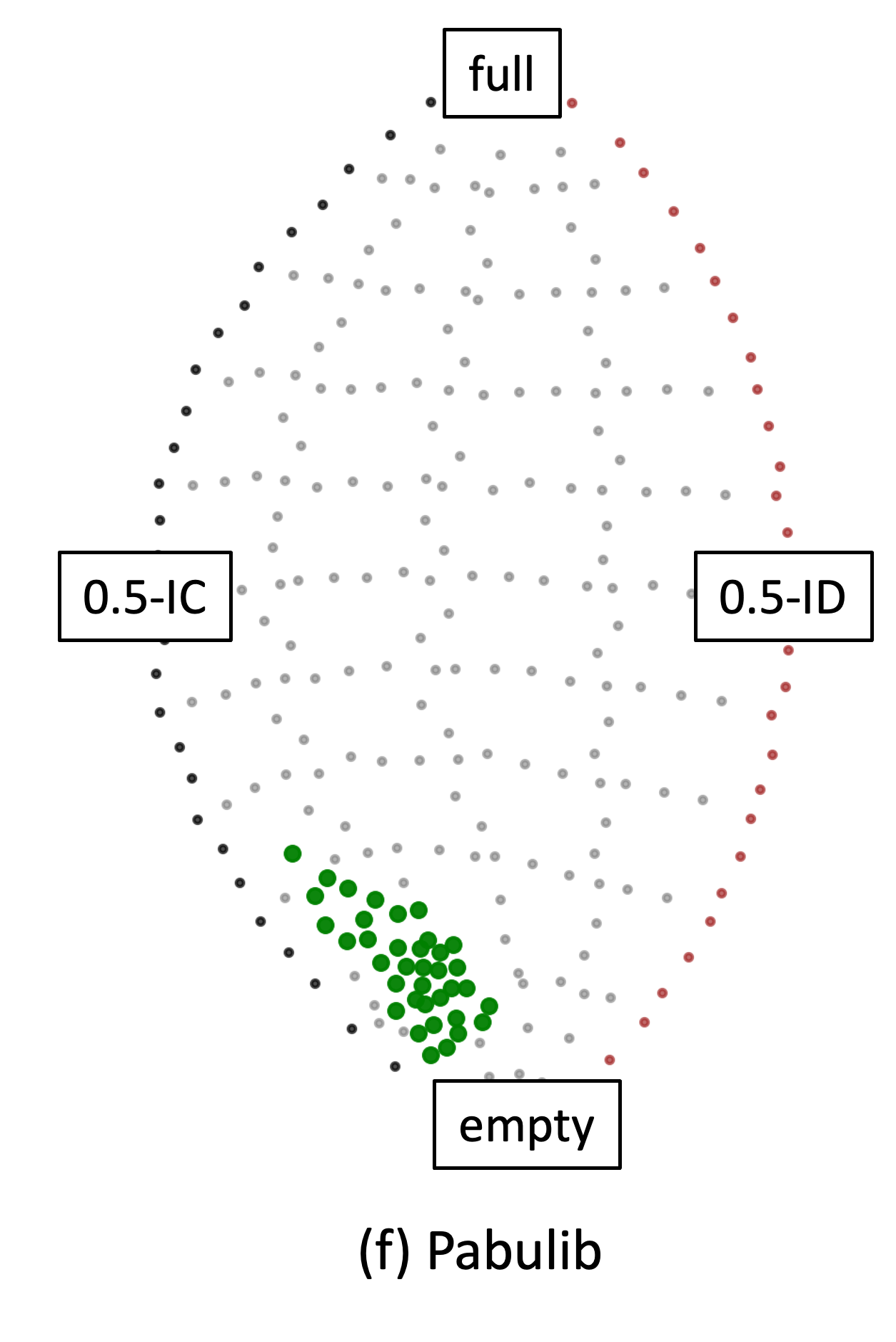}%
    \includegraphics[width=\finalwidth]{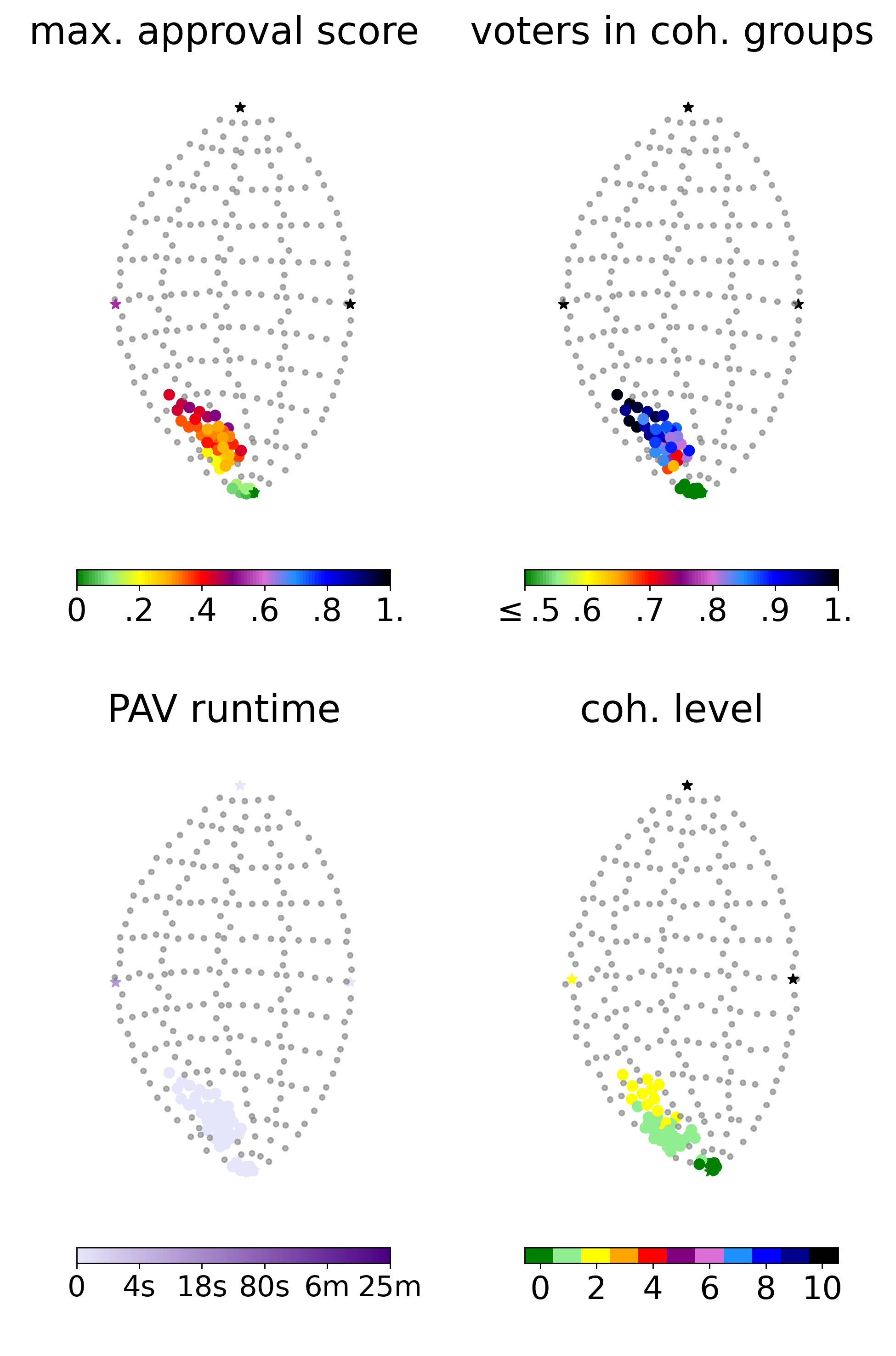}%
    \caption{Maps for (e)~the truncated urn model and (f)~Pabulib. The darker a dot in the main plot is, the larger is the value of~$\alpha$ parameter for (e).}
    \label{app-main-results-ef}
\end{figure}

To get an intuitive understanding of the four statistics, let us
consider the background dataset in Figure~\ref{app-main-results-1}.  We
see that the highest approval score value is lowest in the lower left side
and increases toward up and right. This is sensible: If the average
number of approved candidates increases, so does this statistic. Furthermore,
if voters become more homogeneous, high-scoring candidates are likely
to exist.  Moreover, regarding voters in cohesive groups, it turns out
that in most elections almost all voters belong to some 1-cohesive groups,
with the left lower part as an exception (where there are not enough
approvals to form~$1$-cohesive groups). The time needed to find a
winning committee under PAV is correlated with the distance from
0.5-IC. We see that it takes the longest to find winning committees if the election is unstructured.  Similarly to the highest approval score,
the cohesiveness level increases when moving up or right in the
diagram. Cohesive groups with levels close to the committee size only
exist in very homogeneous elections (rightmost path) and elections
with many approvals (top part).

We move on to the results for the six other datasets.
Note that each figure also contains the background dataset (gray dots)
for reference. These results help to understand the differences
between our statistical cultures.

The maximum approval score statistic provides insight into whether there is a 
candidate that is universally supported. Instances with a value close to~$1$ possess such
a candidate. In a single-winner election, this candidate is likely to be 
a clear winner. This is undesirable when simulating, for example, contested elections. Also note that in the real-world
data set (Pabulib) we do not observe such candidates.

When looking at the PAV runtime, we find some
statistical cultures that generate computationally difficult
elections, such as, e.g., the~$(p,\phi)$-resampling model with parameter values
close to~$p=0.5$ and~$\phi=1$ (0.5-IC), the noise model with
parameters~$p\in[0.5,0.9]$ and~$\phi>0.5$, and the disjoint model
with~$g=2$.  Yet, 
instances
from the real-world dataset, as well as from the Euclidean and urn ones, can be computed very quickly.\footnote{Less than 1 second on a single core (Intel Xeon Platinum 8280 CPU @ 2.70GH)
  of a 224 core machine with 6TB RAM. In contrast, the worst-case instance (0.3-IC) required 25 minutes on 13 cores.}

Concerning voters in cohesive groups, whenever this statistic is close to 1,
it is easy to satisfy most voters with at least one approved candidate
in the committee; such committees are easy to find~\citep{justifiedRepresentation}.
Since many proportional rules take special care of voters who belong
to cohesive groups, in such elections there are no voters that are at
a systematic disadvantage.  In many of our generated elections (almost)
all voters belong to~$1$-cohesive groups, but this is not the case for
the real-world, Pabulib data.
Indeed, to simulate Pabulib data well, we would likely need to
provide some new
statistical culture(s).
%

For the cohesiveness level, we see that all models generate a full
spectrum (i.e.,~$[0,10]$) of cohesiveness levels.
However, we expect realistic elections to appear in the ``lower
left'' part of our grid (with few approvals), and such elections tend
to have low cohesiveness levels.  Indeed, this is also the case for the 
Pabulib elections. Hence, it is important how proportional rules 
treat~$\ell$-cohesive groups with small~$\ell$.


\subsection{Correlation}\label{sec:correlation}
\Cref{app-main-results-1,app-main-results-ab,app-main-results-cd,app-main-results-ef} are based on the
approvalwise distance.  We 
argue that they
would not change much if we used the (computationally intractable)
isomorphic Hamming distance.  To this end, we generated~$413$ elections
with~$10$ candidates and~$50$ voters from the statistical cultures used in
the previous experiment.
The dataset we use for comparing metrics consists of:~$40$ elections from the disjoint models,~$45$ elections from the noise models with Hamming distance,~$50$ elections from moving model, ~$50$ elections from the truncated urn models,~$50$ elections from Euclidean models,~$134$ elections from resampling models, ~$20$ elections from IC,~$20$ elections from ID, and four extreme elections (i.e., 0.5-IC, 0.5-ID, Empty, Full).

\begin{figure}[]
    \centering
    \includegraphics[width=7cm]{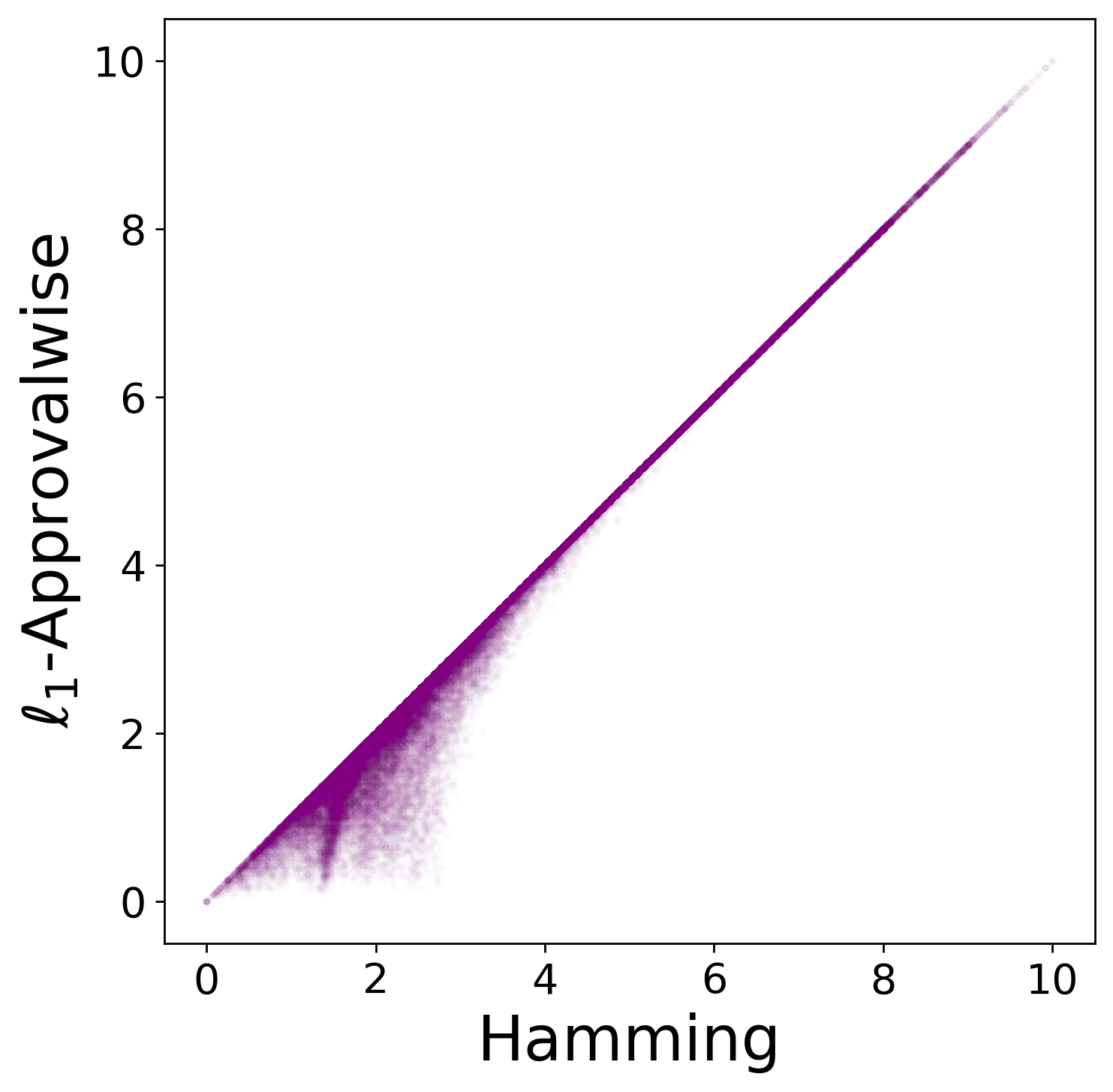}%
    \caption{Correlation between isomorphic Hamming and approvalwise metrics.}
    \label{correlation}
\end{figure}

We compare Hamming and approvalwise
distances. The results are presented in Figure~\ref{correlation}. Each
dot there represents a pair of elections, and its coordinates are the
distances between them, according to the Hamming and approvalwise
metrics. The Pearson Correlation Coefficient is~$0.9899$, and for~$67\%$ of pairs of 
elections the distances are identical. In~\Cref{tab:pcc_app_cult} we take a more fine-grained view of different models, presenting PCC individually for each of them. In each row, we present a correlation based on distances between elections, where at least one of the elections is from a given model. As we can see, when computing Hamming and approvalwise distances from ID we have a perfect correlation, while for distances from impartial culture we have the worst correlation -- which is still extremely high and equals $0.966$.

\begin{table}[]
    \centering
    \begin{tabular}{ c | c | c}
        Statistical Culture & PCC & \% equal \\
    	\midrule
        Identity  & 1.0 & 1.0 \\
        Disjoint  & 0.997 & 0.777 \\
        Moving & 0.995 & 0.652 \\
        2D Euclidean & 0.994 & 0.682 \\
        Resampling  & 0.992 & 0.713 \\
        Truncated Urn  & 0.985 & 0.490 \\
        Noise  & 0.974 & 0.579 \\
        1D Euclidean  & 0.971 & 0.494 \\
        Impartial Culture  & 0.966 & 0.556 \\
        \midrule
        \end{tabular}
    \caption{Pearson correlation coefficients between the Hamming distance and the
        approvalwise distance for each statistical culture used in our maps. The last column contains the percentage of pairs of elections for which both distances are equal.
        }
    \label{tab:pcc_app_cult}
\end{table}

\vspace{0.2cm}
\begin{conclusionbox}
The conclusions of our experiments are as follows.
\begin{itemize}
    \item The approvalwise distance is surprisingly (given how simple it is) strongly correlated with the Hamming distance, which we treat as the ideal one.
    \item Resampling model seems to be quite powerful. We can easily interpret its parameters and generate a large variety of elections.
\end{itemize}
\end{conclusionbox}

\section{Summary}

We introduced several models for generating synthetic approval elections. We believe that these models (in particular, the resampling model) will make it easier to perform future experiments that involve approval elections. We also introduced two distances between approval elections; one isomorphic \emph{ideal} one, which is precise but slow (it takes a lot of time to compute it), and other that is less precise, but fast (can be computed immediately even for instances with thousands of voters and candidates), and strongly correlated with the ideal one. We presented the applications of the map of approval elections showing how different models behave under different circumstances. Among others, we analyzed the running time of PAV rule, exhibiting regions of the map in which the time needed to compute the winning committee is the longest. Moreover, we show where some of the real-life elections lie on the map, however, a good direction and an important task for future work is to broadly study more real-life datasets with the methods proposed in this chapter.

\vspace{0.2cm}
\begin{contributionbox}
\begin{itemize}
    \item Introduction of new models for generating approval elections. In particular, introduction of the resampling model, which turned out to be very practical and is already used by other researchers~(\cite{brill2023robust,lackner2023abcvoting}).
    \item Adaptation of the map of elections framework for approval ballots.
    \item Experimental analysis of the introduced models, showing their strengths and weaknesses.
\end{itemize}
\end{contributionbox}

\chapter{Discussion \& Future Work}
\label{ch:summary}

We would like to emphasize that the main contribution of this thesis is a framework that can be used for numerous novel applications. We started a new line of research and to the date of submission of this thesis, there are already several papers using the content provided within this dissertation.


There exist many possible applications and extensions of the presented research. One possibility is to use the framework to study new types of instances. A good example of direct application of the map framework is a recent paper dedicated to the Stable Roommates and Stable Marriage instances. (This paper received Best Student Paper Award at AAMAS-2023).

\begin{itemize}
\item \emph{A Map of Diverse Synthetic Stable Roommates Instances} \\ Niclas Boehmer, Klaus Heeger, and \textbf{Stanisław Szufa}; AAMAS-\citeyear{boehmer2023map}
\end{itemize}

\noindent
Another approach is to use the map concept to visualize the election data, what was done in:

\begin{itemize}
\item \emph{Collecting, Classifying, Analyzing, and Using Real-World Ranking Data} \\ Niclas Boehmer and Nathan Schaar; AAMAS-\citeyear{boe-sch:t:datasets}
\end{itemize}

\noindent
Moreover, one can study more deeply proposed distances and aggregate representation of election associated with them. For example, in the following paper, authors focus on the analysis of position matrices.

    \begin{itemize}
\item \emph{Properties of Position Matrices and Their Elections} \\ Niclas Boehmer, Jin-Yi Cai, Piotr Faliszewski, Austen Z. Fan, Łukasz Janeczko, 
Andrzej Kaczmarczyk, and Tomasz W\c{a}s; AAAI-\citeyear{boe-cai-fal-fan-jan-kac-was:c:position-matrices}
\end{itemize}

Below, we present a list of other papers that also study similar problems.  Note that, in this chapter, we mention only those works which are either coauthored by Stanisław Szufa or by his close coworkers.

\begin{itemize}

    \item \emph{Diversity, Agreement, and Polarization in Elections} \\ Tomasz W\c{a}s, Piotr Faliszewski, Andrzej Kaczmarczyk, Krzysztof Sornat, and \textbf{Stanisław Szufa}; IJCAI-\citeyear{DBLP:conf/ijcai/Faliszewski0SSW23}
    
    \item \emph{An Experimental Comparison of Multiwinner Voting Rules on Approval Elections} \\ Piotr Faliszewski, Martin Lackner, Krzysztof Sornat, and \textbf{Stanisław Szufa}; IJCAI-\citeyear{DBLP:conf/ijcai/FaliszewskiLSS23}
    
    \item \emph{Participatory Budgeting: Data, Tools, and Analysis} \\ 
    Piotr Faliszewski, Jarosław Flis, Dominik Peters, Grzegorz Pierczyński, Piotr Skowron, Dariusz Stolicki, \textbf{Stanisław Szufa}, Nimrod Talmon; \\ IJCAI-\citeyear{fal-fli-pet-pie-sko-sto-szu-tal:c:pabulib}
    
    \item \emph{A Quantitative and Qualitative Analysis of the Robustness of (Real-World) Election Winners\footnote{Previously the paper was called \emph{On the Robustness of Winners: Counting Briberies in Elections}, and in its full arXiv version used the map of elections framework.}} \\
    Niclas Boehmer, Robert Bredereck, Piotr Faliszewski, and Rolf Niedermeier; EAAMO-\citeyear{boehmer2022quantitative}
    
    \item \emph{Discovering Consistent Subelections} \\ Łukasz Janeczko, Jérôme Lang, Grzegorz Lisowski, and \textbf{Stanisław Szufa}; To appear at AAMAS-2024
    
\end{itemize}

There are also some related problems that have not been given enough attention yet. For example, given a set of ordinal elections called $B$, find a new election, such that its distance to the closest election from $B$ is the largest possible. This will allow us to fill in the potential gaps in our map. Another problem is how to reasonably compare approval and ordinal elections, and, more generally, elections of different sizes and elections with partial preference data.

We see this thesis as an invitation to a deeper study of different elections, statistical cultures, and their relations.

\appendix

\renewcommand{\thesection}{\Alph{section}.\arabic{section}}
\setcounter{section}{0}


\chapter{Distances Between the Compass Elections}
Here, we provide missing proofs from~\Cref{subsec:analysis_compass}.
\label{apdx:supp}
\subsubsection{EMD-Positionwise}

\distemdpos*

\begin{proof}
\textbf{$\mathbf{\textbf{ID}_m}$ and $\mathbf{\textbf{UN}_m}$.} We start by computing the distance between $\ID_m$ and $\UN_m$. Note that $\UN_m$ always remains the same matrix regardless of how its columns are ordered. Thus, we can compute the distance between these two matrices using the identity permutation between the columns of the two matrices: 

{\scriptsize
\begin{align*}
    \POS(\ID_m,\UN_m) &= \sum_{i=1}^m \EMD((\ID_m)_i, (\UN_m)_i)
\\ &=\textstyle\sum_{i=1}^m (\textstyle\sum_{j=1}^{i-1} \frac{j}{m} + \textstyle\sum_{j=1}^{m-i} \frac{j}{m}) 
\\ &= \frac{1}{m} \textstyle\sum_{i=1}^m ( \frac{1 + (i-1)}{2}(i-1) + \frac{1 + (m-i)}{2}(m-i) ) 
\\ &= \frac{1}{2m} \textstyle\sum_{i=1}^m (2i^2 - 2i - 2mi + m^2  + m) 
\\ &= \frac{1}{2m} (2\frac{m(m+1)(2m+1)}{6} - m(m+1)-m^2(m+1)  + m(m^2 + m))
\\ &= \frac{1}{2m} (\frac{(m^2+m)(2m+1)}{3} - (m+1)(m+m^2)  + m(m^2 + m))
\\ &= \frac{m+1}{2} (\frac{(2m+1)}{3} - (m+1)  + m) 
\\ &= \frac{(m+1)(m-1)}{3} 
\\ &= \frac{1}{3}(m^2-1).
\end{align*}
}
In the following, we use $(*)$ when we omit some calculations analogous to the calculations for $\POS(\ID_m,\UN_m)$.

\medskip
\noindent \textbf{$\mathbf{\textbf{UN}_m}$ and $\mathbf{\textbf{ST}_m}$:} Similarly, we can also directly compute the distance between $\UN_m$ and $\ST_m$ using the identity permutation between the columns of the two matrices. In this case, all column vectors of the two matrices have in fact the same $\EMD$ distance from each other:

 \noindent
 $\POS(\UN_m,\ST_m) =  m\cdot (\frac{1}{2}+2\cdot\textstyle\sum_{i=1}^{\frac{m}{2}-1} \frac{i}{m})= \frac{m}{2}+\frac{m}{2}(\frac{m}{2}-1) = \frac{m^2}{4}.$
 
  \medskip
\noindent \textbf{$\mathbf{\textbf{UN}_m}$ and $\mathbf{\textbf{AN}_m}$:} Next, we compute the distance between $\UN_m$ and $\AN_m$ using the identity permutation between the columns of the two matrices. Recall that $\AN_m$ can be written as:
\[
  \AN_m = 0.5\begin{bmatrix}
    \ID_{\nicefrac{m}{2}} & \rID_{\nicefrac{m}{2}} \\
    \rID_{\nicefrac{m}{2}} & \ID_{\nicefrac{m}{2}}
  \end{bmatrix}.
\] Thus, it is possible to reuse our ideas from computing the distance between identity and uniformity:

  \noindent
 $\POS(\UN_m,\AN_m) = 4 \textstyle\sum_{i=1}^{\frac{m}{2}}(\textstyle\sum_{j=1}^{i-1} \frac{j}{m} + \textstyle\sum_{j=1}^{\frac{m}{2}-i} \frac{j}{m}) = (*) = \frac{2}{3}(\frac{m^2}{4}-1).$
 
 \medskip
\noindent \textbf{$\mathbf{\textbf{ID}_m}$ and $\mathbf{\textbf{ST}_m}$:}
There exist only two different types of column vectors in $\ST_m$, i.e.,  $\frac{m}{2}$ columns starting with $\frac{m}{2}$ entries of value $\frac{2}{m}$ followed by $\frac{m}{2}$ zero-entries and $\frac{m}{2}$ columns starting with $\frac{m}{2}$ zero entries followed by $\frac{m}{2}$ entries of value $\frac{2}{m}$. In $\ID_m$, $\frac{m}{2}$ columns have a one entry in the first $\frac{m}{2}$ rows and $\frac{m}{2}$ columns have a one entry in the last $\frac{m}{2}$ rows. Thus, again, the identity permutation between the columns of the two matrices minimizes the $\EMD$ distance:

  \noindent
 $\POS(\ID_m,\ST_m) = 2\cdot \POS(\ID_{\frac{m}{2}},\UN_{\frac{m}{2}}) = \frac{2}{3}(\frac{m^2}{4}-1)$
 
  \medskip
\noindent \textbf{$\mathbf{\textbf{AN}_m}$ and $\mathbf{\textbf{ST}_m}$:} We now turn to computing the distance between $\AN_m=(\an_1,\dots , \an_m)$ and $\ST_m=(\stt_1,\dots , \stt_m)$. As all column vectors of $\AN_m$ are palindromes, each column vector of $\AN_m$ has the same $\EMD$ distance to all column vectors of $\ST_m$, i.e., for $i\in [m]$ it holds that $\EMD(\an_i,\stt_j)=\EMD(\an_i,\stt_{j'})$ for all $j,j'\in [m]$. Thus, the distance between $\AN_m$ and $\ST_m$ is the same for all permutation between the columns of the two matrices. Thus, we again use the identity permutation. 
We start by computing $\EMD(\an_i,\stt_i)$ for different $i\in [m]$ separately distinguishing two cases. Let $i\in [\frac{m}{4}]$. Recall that $\an_i$ has a $0.5$ at position $i$ and position $m-i+1$ and that $\stt_i$ has a $\frac{2}{m}$ at entries $j\in [\frac{m}{2}]$. We now analyze how to transform $\an_i$ to $\stt_i$. For all $j\in [i-1]$, it is clear that it is optimal that the value $\frac{2}{m}$ moved to position $j$ comes from position $i$. The overall cost of this is $\textstyle\sum_{j=1}^{i-1} \frac{2j}{m}$. Moreover, the remaining surplus value at position $i$ (that is, $\frac{1}{2}-\frac{2i}{m}$) needs to be moved toward the end. Thus, for $j\in [i+1,\frac{m}{4}]$, we move value $\frac{2}{m}$ from position $i$ to position $j$. The overall cost of this is  $\textstyle\sum_{j=1}^{\frac{m}{4}-i} \frac{2j}{m}$. Lastly, we need to move value $\frac{2}{m}$ to positions $j\in [\frac{m}{4}+1,\frac{m}{2}]$. This needs to come from position $m-i+1$. Thus, for each $j\in [\frac{m}{4}+1,\frac{m}{2}]$, we move value $\frac{2}{m}$ from position $m-i+1$ to position $j$. The overall cost of this is $\frac{1}{2}\cdot (\frac{m}{2}-i)+\textstyle\sum_{j=1}^{\frac{m}{4}} \frac{2j}{m}=\frac{1}{2}(\frac{m}{2}-i)+\frac{m}{16}+\frac{1}{4}$ 

Now, let $i\in [\frac{m}{4}+1,\frac{m}{2}]$. For $j\in [\frac{m}{4}]$, we need to move value $\frac{2}{m}$ from position $i$ to position $j$. The overall cost of this is $\frac{1}{2}\cdot (i-\frac{m}{4}-1)+\textstyle\sum_{j=1}^{\frac{m}{4}} \frac{2j}{m}=\frac{1}{2}\cdot (i-\frac{m}{4}-1)+\frac{m}{16}+\frac{1}{4}$. For $j\in [\frac{m}{4}+1,\frac{m}{2}]$, we need to move value $\frac{2}{m}$ from position $m-i+1$ to position $j$. The overall cost of this is $\frac{1}{2}\cdot (\frac{m}{2}-i)+\textstyle\sum_{j=1}^{\frac{m}{4}} \frac{2j}{m}=\frac{1}{2}\cdot (\frac{m}{2}-i)+\frac{m}{16}+\frac{1}{4}$. 

Observing that the case $i\in [\frac{3m}{4}+1,m]$ is symmetric to $i\in [\frac{m}{4}]$ and the case $i\in [\frac{m}{2}+1,\frac{3m}{4}]$ is symmetric to $i\in [\frac{m}{4}+1,\frac{m}{2}]$ the $\EMD$ distance between $\AN_m$ and $\ST_m$ can be computed as follows:

{\scriptsize 
\begin{align*}
\POS(\AN_m,\ST_m) &= 2\cdot ( A + \frac{1}{2}\cdot (\sum_{i=1}^{\frac{m}{4}} \frac{m}{2}-i) + \frac{m}{4}\cdot (\frac{m}{16}+\frac{1}{4}) + \frac{1}{2} \cdot (\sum_{i=\frac{m}{4}+1}^{\frac{m}{2}} \cdot (i-\frac{m}{4}-1)) + \frac{m}{4}\cdot (\frac{m}{16}+\frac{1}{4}) \\
& \qquad \qquad + \frac{1}{2}\cdot (\sum_{i=\frac{m}{4}+1}^{\frac{m}{2}} \frac{m}{2}-i) + \frac{m}{4}\cdot (\frac{m}{16}+\frac{1}{4}) \\
&=\frac{m^2}{48} - \frac{1}{3} + \frac{3m^2-4m}{32}  + \frac{m}{2}\cdot (\frac{m}{16}+\frac{1}{4})+  \frac{m^2-4m}{32}+\frac{m}{2}\cdot (\frac{m}{16}+\frac{1}{4}) \\
& \qquad \qquad + \frac{m^2-4m}{32}+\frac{m}{2}\cdot (\frac{m}{16}+\frac{1}{4}) \\
&=\frac{m^2}{48}-\frac{1}{3}+\frac{3m^2-4m}{32}+\frac{3m}{2}(\frac{m}{16}+\frac{1}{4})+\frac{m^2-4m}{16} \\
&=(\frac{1}{48}+\frac{3}{32}+\frac{3}{32}+\frac{1}{16})m^2+(-\frac{4}{32}+\frac{3}{8}-\frac{4}{16})m\frac{1}{3} \\
&=\frac{13}{48}m^2-\frac{1}{3}
\end{align*}
}
\smallskip
\noindent with 

\noindent$A =  \textstyle\sum_{i=1}^{\frac{m}{4}}(\textstyle\sum_{j=1}^{i-1} \frac{2j}{m} + \textstyle\sum_{j=1}^{\frac{m}{4}-i} \frac{2j}{m}) = (*) = \frac{1}{6}(\frac{m^2}{16}-1) = \frac{1}{2}(\frac{m^2}{48} - \frac{1}{3})$

  \medskip
  
\noindent \textbf{$\mathbf{\textbf{ID}_m}$ and $\mathbf{\textbf{AN}_m}$:} Lastly, we consider $\ID_m=(\id_1,\dots , \id_m)$ and $\AN_m=(\an_1,\dots , \an_m)$. Note that, for $i\in [m]$, $\id_i$ contains a $1$ at position $i$ and $\an_i$ contains a $0.5$ at position $i$ and position $m-i$. Note further that for $i\in [\frac{m}{2}]$ it holds that $\an_i=\an_{m-i+1}$.
Fix some $i\in [\frac{m}{2}]$. For all $j\in [i,m-i+1]$ it holds that $\EMD(\an_i,\id_j)=\frac{m-2i+1}{2}$ and for all $j\in [1,i-1]\cup [m-i+2,m]$ it holds that $\EMD(\an_i,\id_j)>\frac{m-2i+1}{2}$. 
That is, for every $i\in [m]$, $\an_i$ has the same distance to all column vectors of $\ID_m$ where the one entry lies in between the two $0.5$ entries of $\an_i$ but a larger distance to all column vectors of $\ID_m$ where the one entry is above the top $0.5$ entry of $\an_i$ or below the bottom $0.5$ entry of $\an_i$. Thus, it is optimal to choose a mapping of the column vectors such that for all $i\in [m]$ it holds that $\an_i$ is mapped to a vector $\id_j$ where the one entry of $\id_j$ lies between the two $0.5$ in $\an_i$. This is, among others, achieved by the identity permutation, which we use to compute:

 \noindent
 $\POS(\ID_m,\AN_m) = 2 \textstyle\sum_{i=1}^{\frac{m}{2}} (\frac{1}{2} (m-2i+1)) = \frac{m}{2}m-\frac{m}{2}(\frac{m}{2}+1)+\frac{m}{2} = \frac{m^2}{4}$
\end{proof}


\newpage
\subsubsection{$\boldsymbol{\ell_1}$-Positionwise}

\distlonepos*

\begin{proof}

\noindent \textbf{$\mathbf{\textbf{ID}_m}$ and $\mathbf{\textbf{UN}_m}$:} Whenever computing the distances between $\UN_m$ and any other matrix, we can assume the identity permutation between the columns of the both matrices; any other permutation will produce exactly the same distance because in the $\UN_m$ matrix all columns are identical. On the diagonal, we have $m$ elements contributing $|1-\frac{1}{m}|$ to the total distance each, and all the other $m(m-1)$ elements are contributing $|0-\frac{1}{m}|$ each. Hence, the total distance is $\frac{1}{m} \cdot m(m-1) + \frac{m-1}{m} \cdot m = 2(m-1)$.

\medskip
\noindent \textbf{$\mathbf{\textbf{UN}_m}$ and $\mathbf{\textbf{ST}_m}$:} Similarly, we can also directly compute the distance between $\UN_m$ and $\ST_m$ using the identity permutation. Each element contributes $\frac{1}{m}$ (either $|\frac{1}{m}-\frac{2}{m}|$, or $|\frac{1}{m} - 0|$), hence the total distance is~$\frac{1}{m} \cdot m^2 = m$.

\medskip
\noindent \textbf{$\mathbf{\textbf{UN}_m}$ and $\mathbf{\textbf{AN}_m}$:} Again, we can directly compute the distance between $\UN_m$ and $\AN_m$ using the identity permutation. Each element on the diagonal and anti-diagonal contributes $|\frac{1}{2} - \frac{1}{m}|$ to the total distance, while all the other $m(m-2)$ elements contributes $\frac{1}{m}$ each. Therefore, the total distance is~$\frac{m-2}{2m} \cdot 2m + \frac{1}{m} \cdot m(m-2) = 2(m-2)$.

\medskip
\noindent \textbf{$\mathbf{\textbf{ID}_m}$ and $\mathbf{\textbf{ST}_m}$:} Let us assume the identity permutation. Both matrices have zeros in the upper-right and lower-left quarter, hence we focus only on the upper-left and bottom-right quarters. However, note that each of these two parts is equivalent to $\ID_{\nicefrac{m}{2}}$ for $\ID_m$, and $\UN_{\nicefrac{m}{2}}$ for $\ST_m$. Therefore, the total distance is~$2 \cdot 2(\frac{m}{2}-1) = 2(m-2)$. If we use any another permutation than identity permutation, then each candidate can contribute to the total distance either the same as for identity permutation or more (i.e., $2$ instead of $2 \cdot (1 - \frac{2}{m})$).

\medskip
\noindent \textbf{$\mathbf{\textbf{AN}_m}$ and $\mathbf{\textbf{ST}_m}$:} Let us assume the identity permutation. All elements from upper-left and lower-right quarters (but not on a diagonal) contribute~$\frac{2}{m}$ to the total distance. 
All elements from upper-right and lower-left quarters (but not on an anti-diagonal) contribute~$0$ because they are equal in both matrices.
All elements on the diagonal contribute~$\frac{1}{2} - \frac{2}{m}$, and all elements on the anti-diagonal
contribute~$\frac{1}{2}$.
Therefore, the total distance is~$\frac{2}{m} \cdot (\frac{1}{2}m^2-m)  + \frac{m-4}{2m} m + \frac{1}{2}  m = 2(m-2)$. Any other permutation will produce exactly the same distance.

\medskip
\noindent \textbf{$\mathbf{\textbf{ID}_m}$ and $\mathbf{\textbf{AN}_m}$:} Again, let us assume the identity permutation. The elements on the diagonal and anti-diagonal contributes to the total distance~$\frac{1}{2}$ (either~$|1-\frac{1}{2}|$ or $|0-\frac{1}{2}|$) each. All the other elements in both matrices are zeros, hence the total distances is~$\frac{1}{2}\cdot 2m = m$. If we use any another permutation than identity permutation, then each candidate can contribute to the total distance either the same as for identity permutation or more (i.e., $2$ instead of $1$).

\end{proof}

 
\newpage
\subsubsection{$\boldsymbol{\ell_1}$-Pairwise}

\distpair*

\begin{proof} Given the fact that each matrix is not defined on the diagonal, we omit it in our reasoning, and focus only on the other $m(m-1)$ elements. 

    \medskip
    \noindent \textbf{$\mathbf{\textbf{ID}_m}$ and $\mathbf{\textbf{UN}_m}$:} 
    In $\UN_m$ matrix all column vectors are identical, hence we do not need to worry about the candidate permutation. Therefore, calculating the distance is straightforward. If we use the identity permutation (or any other permutation), then the distances is as follows. Each element contributes $\frac{1}{2}$ (either $|1-\frac{1}{2}|$ or $|0-\frac{1}{2}|$) to the total distance. Hence, the total distance is~$\frac{1}{2}m(m-1)$.
    
    \medskip
    \noindent \textbf{$\mathbf{\textbf{UN}_m}$ and $\mathbf{\textbf{ST}_m}$:} Like for the previous distance, we do not need to worry about the permutation, and can simply assume that we use identity permutation. The values in the upper-left and lower-right quarters of both matrices are identical, so the distance between the elements in these parts is zero. As for the upper-right and lower-left quarters, each element contributes $\frac{1}{2}$ (either $|1-\frac{1}{2}|$ or $|0-\frac{1}{2}|$) like for the distances between $\ID_m$ and $\UN_m$. There are $\frac{1}{2}m^2$ such elements, hence the total distance is~$\frac{1}{4}(m^2)$.
    
    \medskip
    \noindent \textbf{$\mathbf{\textbf{ID}_m}$ and $\mathbf{\textbf{ST}_m}$:}  Let us assume the identity permutation.
    The values in the upper-right and lower-left quarters of both matrices are identical, so the distance between the elements in these parts is zero. As for the upper-right and lower-left quarters, for $\ID_m$ each of these parts is equivalent to $\ID_{\nicefrac{m}{2}}$, and for $\ST_m$ each of these part is equivalent to $\UN_{\nicefrac{m}{2}}$. Hence, the total distance is twice the distances between $\ID_{\nicefrac{m}{2}}$ and $\UN_{\nicefrac{m}{2}}$ which is $2 \cdot \nicefrac{1}{2} \cdot \frac{m}{2}(\frac{m}{2}-1)  = \frac{1}{4}m(m-2)$. If we use any another permutation than identity permutation, then each candidate can contribute to the total distance either the same as for identity permutation or more.

\end{proof}


\newpage
\subsubsection{EMD-Bordawise}

\distborda*

\begin{proof}

The calculations are as follows.

\begin{align*}
\BOR(\ID_m, \UN+m) &= n \displaystyle\mathop{\Sigma}_{i=1}^{\nicefrac{m}{2}} (i-\frac{1}{2})\frac{m}{2}  + 
             n \displaystyle\mathop{\Sigma}_{i=1}^{\nicefrac{m}{2}-1} i (\frac{m}{2} - i) \\
 &= n \left[ \frac{1}{16} m^3 + \frac{1}{16} (m-2) m^2  - \frac{1}{24} (m-2)(m-1)m \right] \\
 &= \frac{1}{12} n \cdot m (m^2-1) \text{ (for even  $m$)}, \\
\BOR(\UN_m,\ST_m) &= \frac{m}{2} \frac{m}{2} \frac{m-1}{4}n = \frac{1}{16} n \cdot m^2 (m-1), \\
\BOR(\ID_m,\ST_m) &= \BOR(\UN_m,\ID_m) - \BOR(\UN_m,\ST_m) \\
&= \frac{1}{48} n \cdot m (m^2+3m-4). \\
\end{align*}

\end{proof}


\chapter{Maps of Approval Candidates}\label{apdx:map_app_cand}


Let~$S(c)$ denote the set of supporters of candidate~$c$ (i.e., those voters that approve~$c$).
Then, for two candidates~$c$ and~$d$, their (candidate) Hamming distance is~$|S(c) \triangle S(d)|$, i.e.,  the number of
voters that approve exactly one of them. The (candidate) Jaccard distance is~$\frac{|S(c) \triangle S(d)|}{|S(c) \cup S(d)|}$.

In \Cref{fig:app_candidate_part_1,fig:app_candidate_part_2} we present the candidate maps. As before, on the left side are the results for the Hamming distance, and on the right side, are the results for the Jaccard one. We used exactly the same elections as before (i.e.,~$100$ candidates and~$1000$ voters). With the purple discs, we depict the cases where more than five candidates are identical. Numerous things which were true for the maps of votes are also true for the maps of candidates; hence, we mainly focus on the differences.
The first difference is related to the disjoint model. Here, we observe one more cloud of points than the number of groups (with an exception for the map with~$g=4$). It is because the additional group consists of candidates that were not approved in any of the initial ballots. For~$g=4$ we do not witness such a group because each candidate is a member of one of the groups.
Regarding the maps for the resampling model with~$\phi=\frac{3}{4}$, for the maps of preferences, all four maps were almost indistinguishable. However, for the maps of candidates, we see the differences between maps (especially between the first two).

For the noise model (as for the resampling model) we see a crucial difference between the candidates' and the voters' perspectives.
Candidates are divided into two groups, those that are approved in the central ballot, and those that are not.
What is interesting, although justified, is the fact that the maps for the Euclidean elections for candidates are very similar to the analogous maps for the voters.
The pictures for the urn elections are relatively chaotic. In the urn elections, the larger the~$\alpha$, the smaller the number of different votes.

For real-life elections, for the Hamming distance, the concentration in the middle of all four instances is due to numerous weak projects that are similar to each other because they were disapproved by most voters. So, the farther a project is from the center, the more approval it is likely to get. If we look at the Jaccard maps, we observe tiny clustering of points in the outskirts, which means there were some groups of similar projects; however, there were no groups of projects that were approved by a large fraction of the society. For the Jaccard distance, not much information can be gained from these pictures, which is an information in itself. It means that there is not much structure in real-life elections.

\begin{figure}
    \centering
     	\begin{minipage}{.49\textwidth}
     	        \centering
        \text{ \ \ \ Hamming}\\
        \text{}
 		    \includegraphics[width=7cm]{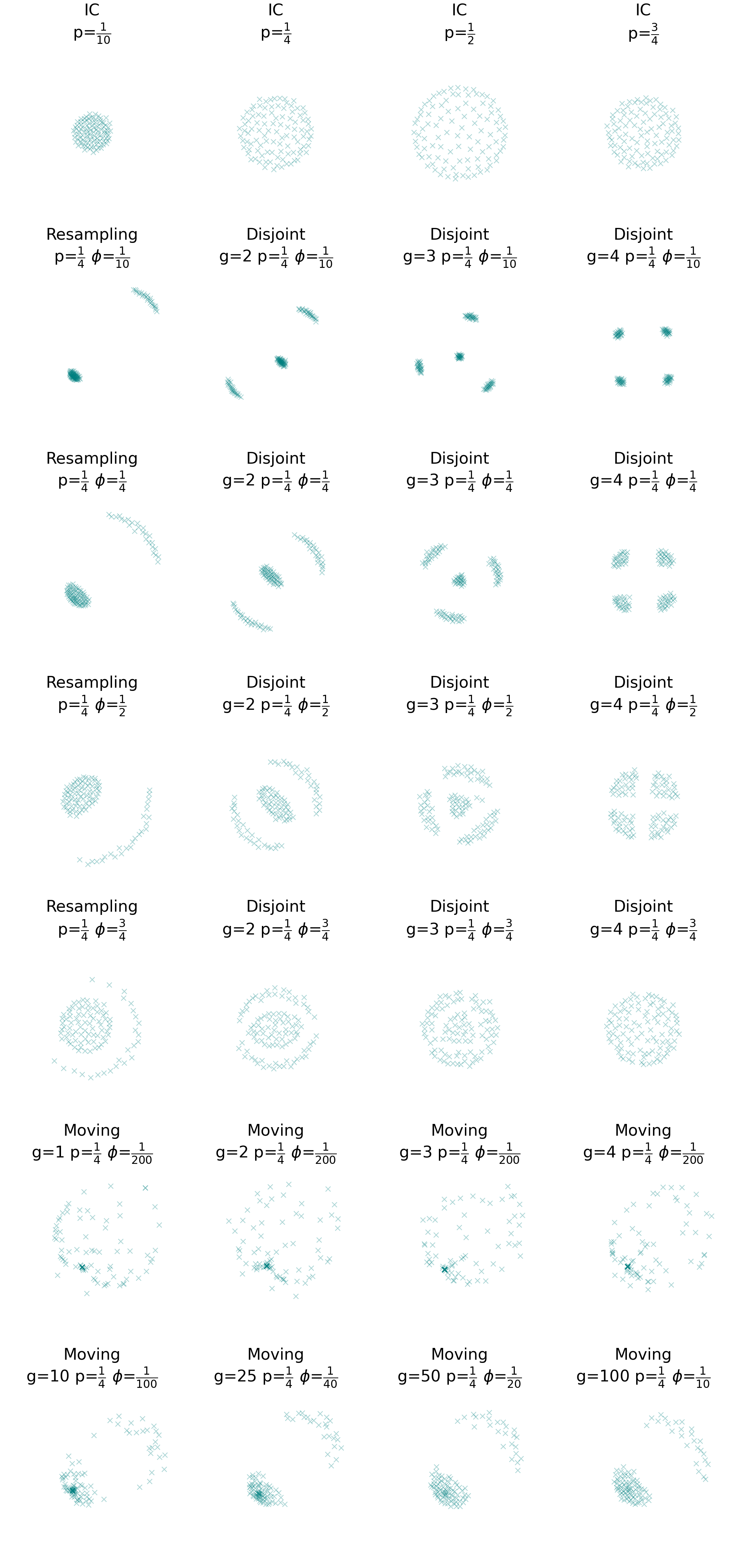}
 	\end{minipage}\hfill
 	\begin{minipage}{.49\textwidth}
 		\centering
        \text{ \ \ Jaccard}\\
        \text{}
        \includegraphics[width=7cm]{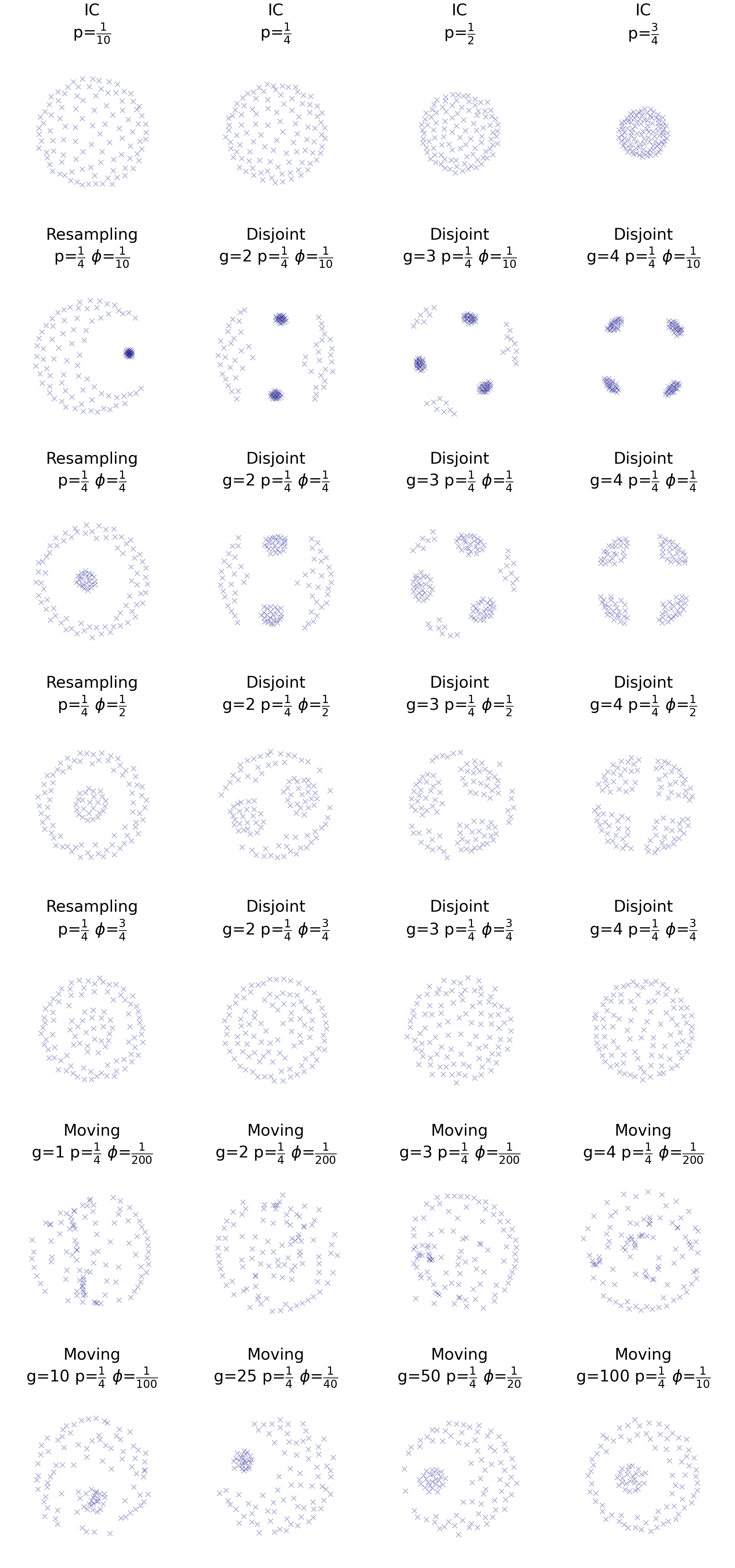}
 	\end{minipage}\hfill

    \caption{Maps of (Approval) Candidates ($100$ candidates,~$1000$ voters). On the left (teal) based on the Hamming distance, and on the right (navy) based on the Jaccard distance.}
    \label{fig:app_candidate_part_1}
\end{figure}

\begin{figure}
    \centering
     	\begin{minipage}{.49\textwidth}
     	        \centering
        \text{ \ \ \ Hamming}\\
        \text{}
 		    \includegraphics[width=7cm]{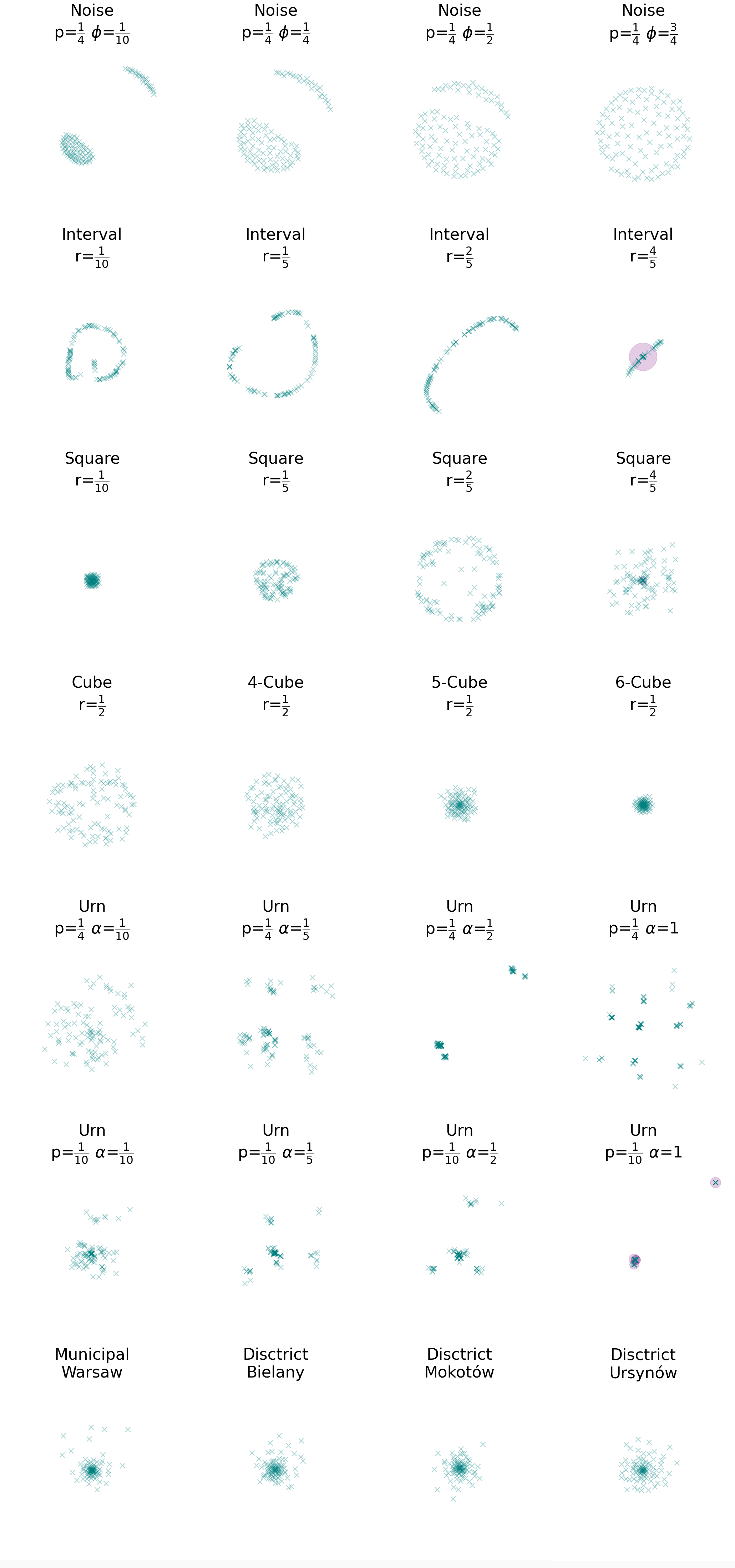}
 	\end{minipage}\hfill
 	\begin{minipage}{.49\textwidth}
 		\centering
        \text{ \ \ Jaccard}\\
        \text{}
        \includegraphics[width=7cm]{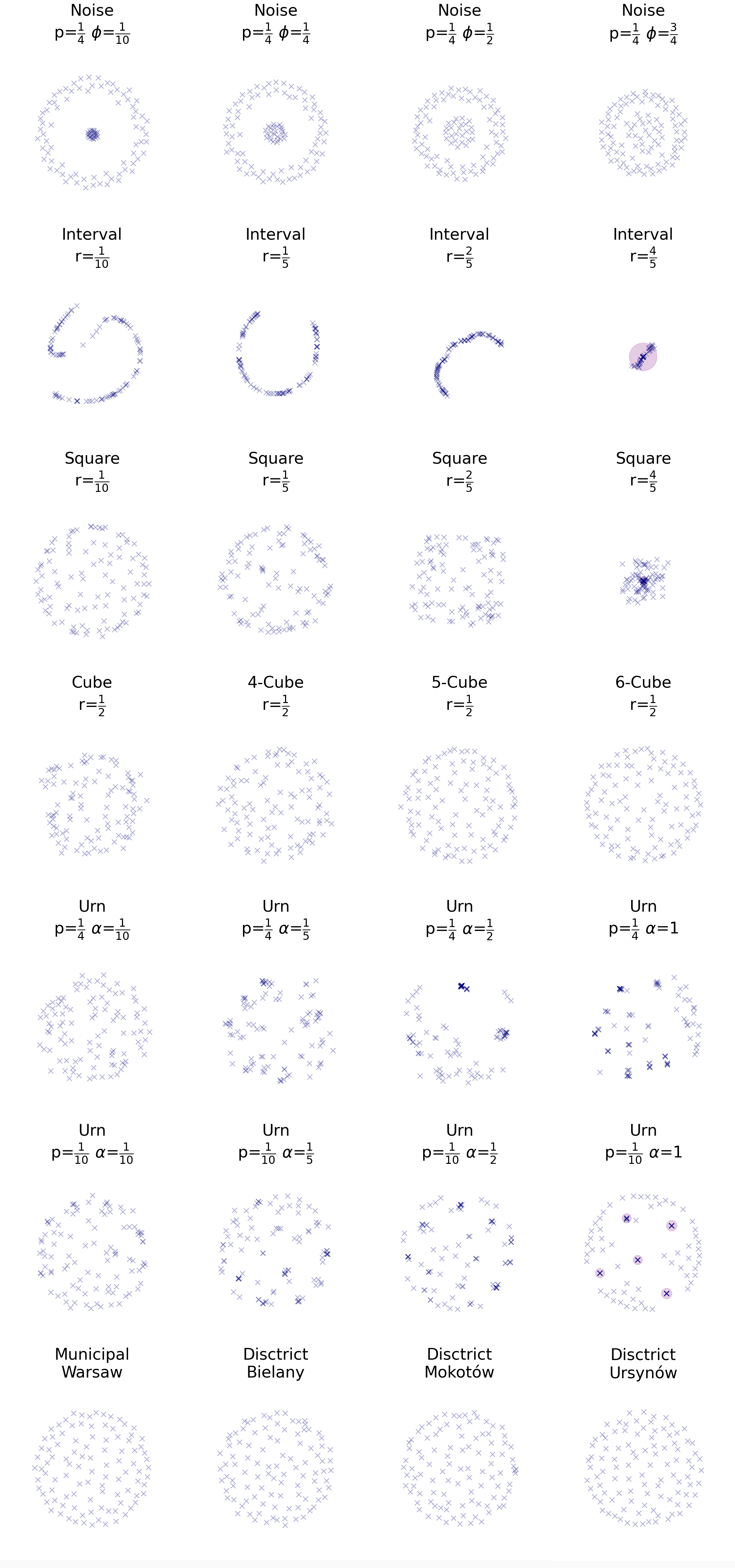}
 	\end{minipage}\hfill

    \caption{Maps of (Approval) Candidates ($100$ candidates,~$1000$ voters). On the left (teal) based on the Hamming distance, and on the right (navy) based on the Jaccard distance.}
    \label{fig:app_candidate_part_2}
\end{figure}

\bibliographystyle{plainnat}
\bibliography{bib}

\end{document}